\documentclass[11pt, oneside, hidelinks]{article}

\usepackage[margin=1in]{geometry} 

\usepackage{amsmath,amsthm,amssymb, graphicx, multicol, array,hyperref, comment, caption, subcaption, rotating, longtable, xpatch, bbm, color,mathtools, float,array,setspace,amsfonts,dcolumn,booktabs,adjustbox,bm}
\usepackage[table]{xcolor}
\usepackage{pdflscape}

\def\*#1{\mathbf{#1}}
\def\+#1{\mathbb{#1}}
\def\wh#1{\widehat{#1}}

\newcommand*{\Scale}[2][4]{\scalebox{#1}{$#2$}}
\newcommand{\I}{{-1}}
\newcommand{\T}{\top}

\newcommand{\Lp}{\left(}
\newcommand{\Rp}{\right)}
\newcommand{\Ls}{\left[}
\newcommand{\Rs}{\right]}
\newcommand{\tlo}{\mathcal{O}}
\newcommand{\tlq}{\mathcal{Q}}
\newcommand{\tls}{\mathcal{S}}
\newcommand{\dnt}{\delta_{NT}}
\newcommand{\norm}[1]{\left\lVert#1\right\rVert}
\newcommand{\RNum}[1]{\uppercase\expandafter{\romannumeral #1\relax}}
\newcommand{\calN}{\mathcal{N}}
\newcommand{\tvec}{\mathrm{vec}}
\newcommand{\covI}{\Gamma^{\textnormal{obs}}}
\newcommand{\covII}{\Gamma^{\textnormal{miss}}}
\newcommand{\covIII}{\Gamma^{\textnormal{miss, cov}}}
\newcommand{\covIS}{\Gamma^{\textnormal{obs},S}} 
\newcommand{\covIIS}{\Gamma^{\textnormal{miss},S}}

\newcommand{\covIIIS}{\Gamma^{\textnormal{miss, $S$, cov}}}
\newcommand{\Ttreat}{T_{1,i}}
\newcommand{\Tcontrol}{T_{0,i}}
\newcommand{\treat}{{(1)}}
\newcommand{\control}{{(0)}}
\newcommand{\iid}{\mathrm{iid}}
\newcommand{\diag}{\mathrm{diag}}
\newcommand{\obs}{\mathrm{obs}}

\newcommand{\Var}{\mathrm{Var}}
\newcommand{\ps}{P(W_{it}=1|S_i)}
\newcommand{\p}{P(W_{it}=1)}
\newcommand{\cov}{\mathrm{cov}}
\newcommand{\scovI}{\Sigma^{\textnormal{obs}}}
\newcommand{\scovII}{\Sigma^{\textnormal{miss}}}
\newcommand{\scovIII}{\Sigma^{\textnormal{miss, cov}}}
\newcommand{\scovIS}{\Sigma^{\textnormal{obs},S}}
\newcommand{\scovIIS}{\Sigma^{\textnormal{miss},S}}

\newcommand{\scovIIIS}{\Sigma^{\textnormal{miss, $S$, cov}}}
\newcommand{\covItreat}{\Gamma^{\textnormal{obs},(1)}}
\newcommand{\covIcontrol}{\Gamma^{\textnormal{obs}}}
\newcommand{\covIItreat}{\Gamma^{\textnormal{miss},(1)}}
\newcommand{\covIIcontrol}{\Gamma^{\textnormal{miss}}}
\newcommand{\covIItreatcontrol}{\Gamma^{\textnormal{miss}, \Delta}}
\newcommand{\covIItreattreatcontrol}{\Gamma^{\textnormal{miss, cov}, (1), \Delta}}
\newcommand{\covIIcontroltreatcontrol}{\Gamma^{\textnormal{miss, cov},(0), \Delta}}
\newcommand{\covIIItreatcontrol}{\Gamma^{\textnormal{miss, cov}, (0), (1)}}
\newcommand{\covIIIcontrolcontrol}{\Gamma^{\textnormal{miss, cov}}}

\newcommand{\spsi}{p_{it}^{S_i}}
\newcommand{\AVar}{\mathrm{AVar}}
\newcommand{\spsj}{p_{jt}^{S_j}}
\newcommand{\ACov}{\mathrm{ACov}}
\newcommand{\hatspsi}{\hat{p}_{it}^{S_i}}
\newcommand{\Cov}{\text{Cov}}
\newcommand{\Ttr}{T_{1,i}}
\newcommand{\nocontentsline}[3]{}
\newcommand{\tocless}[2]{\bgroup\let\addcontentsline=\nocontentsline#1{#2}\egroup}
\newcommand{\problim}{\mathrm{plim}}

\newtheorem{theorem}{Theorem}
\newtheorem{lemma}{Lemma}
\newtheorem{proposition}{Proposition}
\newtheorem{corollary}{Corollary}
\newtheorem*{remark}{Remark}

\newtheorem{assumpS}{Assumption}
\newtheorem{assumpG}{Assumption}

\newtheorem{assumpC}{Assumption}

\newtheorem{assumpGC}{Assumption}

\usepackage[longnamesfirst]{natbib}
\setcitestyle{open={(},close={)}}

\newcommand*{\thisdraft}{This version: \today} 
\newcommand*{\firstdraft}{First version: October 16, 2019}

\usepackage{setspace}
\usepackage{enumitem}
\setlist{nolistsep}

\setlist{noitemsep}  

\usepackage{xr}
\makeatletter
\newcommand*{\addFileDependency}[1]{
	\typeout{(#1)}
	\@addtofilelist{#1}
	\IfFileExists{#1}{}{\typeout{No file #1.}}
}
\makeatother

\newcolumntype{L}[1]{>{\raggedright\let\newline\\\arraybackslash\hspace{0pt}}m{#1}}
\newcolumntype{C}[1]{>{\centering\let\newline\\\arraybackslash\hspace{0pt}}m{#1}}
\newcolumntype{R}[1]{>{\raggedleft\let\newline\\\arraybackslash\hspace{0pt}}m{#1}}

\usepackage[margin=10pt,labelfont=bf]{caption}
\newcommand\tcaptab[1]{\captionsetup{position=top, font=normalsize, labelfont=bf, textfont=normalfont, justification=centering, margin=0mm, aboveskip=1mm, belowskip=0mm, labelsep=colon, singlelinecheck=false}\caption{#1}}
\newcommand\bnotetab[1]{\captionsetup{position=bottom, font=footnotesize,  textfont=normalfont, margin=1mm, skip=2mm, justification=justified, singlelinecheck=false}\caption*{#1}}
\newcommand\tcapfig[1]{\captionsetup{position=top, font=normalsize, labelfont=bf, textfont=normalfont, justification=centering, margin=0mm, aboveskip=2mm, belowskip=0mm, labelsep=colon, singlelinecheck=false}\caption{#1}}
\newcommand\bnotefig[1]{\captionsetup{position=bottom, font=footnotesize,  textfont=normalfont, margin=1mm, skip=2mm, justification=justified, singlelinecheck=false}\caption*{#1}}

\begin{document}

	\title{Large Dimensional Latent Factor Modeling with Missing Observations and Applications to Causal Inference\thanks{\scriptsize We thank Susan Athey, Mohsen Bayati, Guillaume Basse, Jianqing Fan, Kay Giesecke, Peter Glynn, Lisa Goldberg, Guido Imbens, Serena Ng, 
			David Simchi-Levi, seminar and conference participants at Stanford, MIT, Berkeley, Chicago Booth, Columbia, Cornell, Cornell Tech, Emory, Minnesota, USC, UCSB, UBC Sauder School of Business, University of Toronto Rotman School of Management, Boston University Questrom School of Business, 
			University of Illinois at Chicago, Stony Brook, Florida, INFORMS and the Marketplace Innovation Workshop for helpful comments.}}
	\date{\firstdraft \\ \thisdraft }
	
	\author{Ruoxuan Xiong\thanks{ \scriptsize Stanford University, Department of Management Science \& Engineering, Email: \url{rxiong@stanford.edu}.}
		\and
		Markus Pelger\thanks{\scriptsize Stanford University, Department of Management Science \& Engineering, Email: \url{mpelger@stanford.edu}.}
	}
	
	\onehalfspacing

	\begin{titlepage}
		\maketitle
		\thispagestyle{empty}
		

		\begin{abstract}
			
			This paper develops the inferential theory for latent factor models estimated from large dimensional panel data with missing observations. We propose an easy-to-use all-purpose estimator for a latent factor model by applying principal component analysis to an adjusted covariance matrix estimated from partially observed panel data. We derive the asymptotic distribution for the estimated factors, loadings and the imputed values under an approximate factor model and general missing patterns. The key application is to estimate counterfactual outcomes in causal inference from panel data. The unobserved control group is modeled as missing values, which are inferred from the latent factor model. The inferential theory for the imputed values allows us to test for individual treatment effects at any time under general adoption patterns where the units can be affected by unobserved factors.
			
			\vspace{1cm}
			

		\end{abstract}
	\end{titlepage}

	\section{Introduction}
	Large dimensional panel data with missing entries are prevalent. In causal panel data, the main focus is to estimate the unobserved potential outcomes. In financial data, stock returns can be missing before a company is listed, after its bankruptcy, or because of illiquidity. In macroeconomic datasets, panel data might be collected at different frequencies or not for all geographical locations resulting in missing entries. In the famous Netflix challenge, a majority of users' ratings for films are missing. Estimating missing entries in panel data is a fundamental problem with applications in social science, statistics, and computer science.
	
	This paper develops the inferential theory for latent factor models estimated from large dimensional panel data with missing observations. We propose a novel and easy-to-use approach to estimate a latent factor model by applying principal component analysis (PCA) to an adjusted covariance matrix, which is estimated from partially observed panel data. We derive the asymptotic normal distribution for the estimated factors, loadings, and imputed values. The key application is to estimate counterfactual outcomes for causal inference. The unobserved control group is modeled as missing values, which are inferred from the latent factor model. The inferential theory for the imputed values allows us to test for individual treatment effects at a particular time. This granular test is of practical relevance because we learn not only for whom but also when a treatment is effective.

	The inferential theory for latent factor models with missing data is important for a number of reasons. First, we show how to consistently impute the missing observations in a large dimensional panel data set, which can then be used as an input for other applications. Our confidence intervals for the imputed values can serve as a decision criterion if the imputed data should be used. Second, the distribution of the missing observations can actually be the object of interest itself. For example, the imputed values serve as the synthetic control in causal inference for which we need an asymptotic distribution theory. The inferential theory is key for deriving test statistics for treatment effects. Last but not least, we provide the complete inferential theory for the latent factors themselves, which is relevant when the factors are the object of interest and are used as input for other applications.

	Our method is very simple to adopt and works under general assumptions. We provide an ``all-purpose'' estimator that performs well under all empirically relevant missing patterns and only assumes a general approximate factor model. Our estimation consists of two simple steps, where we first apply PCA to a re-weighted covariance matrix to obtain the loadings and, in a second step, run a regression on these loadings using only the observed units to obtain the factors. The missing entries are estimated by the common components of the factor model.  Importantly, our estimator does not require the estimation of the observation pattern itself. In some cases, we might have additional information about the missing pattern. We provide a modification of our estimator that can take advantage of a probabilistic model of the missing pattern and use an inverse probability weight in the second step regression to obtain the factors. It is inspired by the inverse propensity weighted regression from causal inference that enjoys the doubly-robust property, meaning the estimator is robust to some form of omitted variable bias. Our probability weighted estimator also has similar desirable robustness properties when we omit latent factors, but it is generally less efficient than our all-purpose estimator.

	Our framework stands out by the very general patterns of missing observations that it can accommodate. We cover the common scenarios of missing at random or a simultaneous/staggered treatment adoption, where the treatment cannot be removed once implemented. Importantly, the missing pattern can depend in a general way on the unobserved factor loadings or unit-specific features. Hence, the observations can be missing because of how the units are exposed to the latent factors. Our simple all-purpose estimator does not require us to explicitly model this relationship, but takes it automatically into account. In the case of the propensity weighted estimator, we provide feasible estimators of the probability weights that result in the same distribution as the population weights.

	Deriving the inferential theory under these general conditions is a challenging problem. The missing observations have a complex effect on the asymptotic covariance matrix of the imputed entries. In particular, the asymptotic variance has an additional variance correction term compared with the fully observed panel. This term results in a larger asymptotic variance than in the fully observed case. The variance correction term arises because, in a panel with missing observations, we take averages over a different number of time periods for the different entries in the estimated covariance matrix. The variance correction term is larger if the observation pattern has many missing entries, or if it deviates more from a missing at random scheme. The propensity weighted estimator has a similar asymptotic distribution structure as our all-purpose estimator but in general a larger variance.

	Our work contributes to three distinct fields: large dimensional factor modeling, matrix completion, and causal inference. First, we extend the inferential theory of latent factors to large dimensional data with general patterns in missing entries. Second, matrix completion methods impute missing entries under the assumption of a low-rank structure, which is corrupted with noise. We provide confidence intervals for the imputed values. Lastly, the key question in causal inference is the estimation of counterfactual outcomes, i.e., what would have been the outcome if a unit had not been treated or if a unit had been treated. The unobserved counterfactual outcome can naturally be formulated as a missing observation problem. We are the first to provide a test for the point-wise treatment effect that can be heterogeneous and time-dependent under general adoption patterns where the units can be affected by unobserved factors: 

	This paper works under the framework of an approximate latent factor structure where both the cross-section dimension and time-series dimension are large. When the data is fully observed, \cite{bai2002determining} show that the factor model can be estimated with PCA applied to the covariance matrix of the data. \cite{bai2003inferential} and \cite{fan2013large} derive the consistency and asymptotic normality of the estimated factors, loadings and common components. Extensions of latent factor models with fully observed data include adding observable factors in \cite{bai2009panel}, sparse and interpretable latent factors in \cite{pelgerxiongsparse2020}, time-varying loadings in \cite{fan2016projected} and \cite{pelger2018state}, high-frequency estimation in \cite{pelger2019large} and including additional moments to estimate weak factors as in \cite{lettaupelger2019}. When a panel has missing entries, a common approach is to estimate the factor model from a subset of the data for which a balanced panel is available. This approach has two drawbacks: First, it is, in general, less efficient as our approach makes use of all the data. Second, it can lead to a biased estimate if the data is not missing at random.

	The inferential theory of large dimensional factor models with missing observations is an active area of research. Our paper is most closely related to the recent papers by \cite{jin2020factor}, \cite{bai2019matrix}, and \cite{cahan2021factor}. The papers differ in the algorithms to impute the missing observations, the generality of the missing patterns, and the proportion of required observed entries relative to the missing entries. Our main results are derived under the assumption that entries are observed at the same rate as missing entries, but we show that this assumption can be considerably relaxed. Importantly, in contrast to the other papers, our framework allows the missing pattern to depend on unit-specific features and to test for an individual treatment effect at any time for any cross-section unit or a weighted treatment effect. 
	{This is exactly what we need for the main application in causal inference.}
	\cite{jin2020factor} provide the inferential theory for the estimated factor model with the expectation-maximization (EM) algorithm under the assumption of randomly missing values. This is a major advance in the literature on using the EM algorithm to impute missing values on cross-sectional data \citep{rubin1976inference,dempster1977maximum}.\footnote{\cite{stock2002macroeconomic,banbura2014maximum,negahban2012restricted} propose to use EM algorithms to estimate the factor model from panel data with missing observations. \cite{giannone2008nowcasting,doz2011two,jungbacker2011maximum,stock2016dynamic} propose to use the state-space framework and Kalman Filtering to estimate the factor model with missing observations. \cite{gagliardini2019diagnostic} propose a simple diagnostic criterion for an approximate factor structure in large (unbalanced) panel data sets.} \cite{bai2019matrix} provide the inferential theory for the factor-based imputed values based on the innovative idea of shuffling rows and columns such that there exist fully observed TALL and WIDE blocks for estimating the factor model. Their TALL-WIDE algorithms involves two applications of principal components on the two fully observed blocks.  \cite{cahan2021factor} propose the TALL-PROJECT estimator that extends the TALL-WIDE estimator by first using only the fully observed TALL block for a PCA estimation of the factors and then obtains the loadings from a time-series regression that that uses all observed entries. They provide the inferential theory for this TALL-PROJECT estimator.
	Each of these estimators is designed for a specific observation pattern under which it performs particularly well, but might not generalize to other patterns.  In contrast, we view our estimator as a simple all-purpose estimator that can reliably impute missing data and provide the correct confidence intervals for general missing patterns and factor structures, which makes it appealing for applied researchers in causal inference. In an extensive simulation study, we show that while our estimator has a similar performance as \cite{jin2020factor} for data missing at random, and as \cite{bai2019matrix} for missing with a block structure, our estimator can have a better performance for a staggered design or when the observation pattern depends on unit-specific features.

	Our imputed values are point-wise consistent and have asymptotic normal distributions, which is relevant for the matrix completion literature that studies a similar problem. Both our paper and the matrix completion literature assume a low-rank structure in the panel data. In the matrix completion literature, the most popular method is to estimate the low-rank matrix from a convex optimization problem using a nuclear norm regularization \citep{mazumder2010spectral,negahban2011estimation,negahban2012restricted}.  
	The main results in the matrix completion literature are upper bounds for the mean-squared estimation error of the estimated matrix. However, point-wise consistency does not hold in general because the typically used nuclear norm regularization results in a bias in the estimated matrix. In their path-breaking work, \cite{chen2019inference} propose de-biased estimators and provide an inferential theory under the assumption of i.i.d. sampling and i.i.d. noise.  There is a trade-off in terms of the generality of the model and the required observations, where our work allows for the most general patterns in missing observations with a general approximate factor structure at the cost of observing entries at a higher rate than \cite{chen2019inference}. Our paper contributes to the matrix completion literature by allowing general observation patterns and dependent error structures, which is particularly relevant for applications in social science.

	Our paper allows for heterogeneous and time-dependent treatment effects of an intervention and more general intervention adoption patterns compared with the synthetic control methods in causal inference. Furthermore, our paper provides a flexible test for treatment effects.
	In comparative case studies, a key question is to estimate the counterfactual outcomes for treated units. 
	A valid control unit is ``close'' to the treatment unit except for the treatment effect. Typically synthetic controls are weighted averages of untreated units where the weights depend on unit-specific features.
	A popular model assumption is that the potential outcome is linear in observed covariates and unobserved common factors. \cite{abadie2010synthetic,abadie2015comparative},  \cite{doudchenko2016balancing}, \cite{xu2017generalized}, \cite{li2017estimation} and \cite{li2017statistical} propose to match each treated unit by weighted averages of all control units using the pretreatment observations. \cite{li2017estimation}, \cite{li2017statistical} and  \cite{masini2018counterfactual} show the inferential theory for the average treatment effect over time. These methods rely on the assumption that there is only one treated unit and the treatment effects are either constant or stationary. Another method is to regress the post-treatment outcomes for the control units on the pre-treatment outcomes and covariates and use the coefficients to predict the counterfactual outcome for the treated/control units. \cite{athey2018matrix} propose to use matrix completion methods to impute the control panel data and allow for more general treatment adoption patterns: multiple treated units and staggered treatment adoption. However, they do not provide point-wise guarantees for the imputed values. 
	In this paper, in addition to allowing for general treatment adoption patterns, we also provide the point-wise inferential theory for the imputed counterfactual outcomes. Furthermore, we can test for treatment effects even if they are heterogeneous and time-dependent. Our approach does not require a priori knowledge about which covariates describe if treated and control units are a good match. Instead, our latent loadings capture all unit-specific information in a data-driven way. The synthetic control, that we impute, is a weighted average of the untreated units, that takes all unit-specific information into account.  In causal inference, we can either model the relationship between the covariates and the outcome, or model the probabilities of missingness to estimate causal effects. Doubly robust procedures, as discussed, for example, in \cite{kang2007} combine both by using a propensity weight in regressions to mitigate the selection bias. Our propensity weighted estimator builds on this intuition. Interestingly, we prove that using the estimated feasible propensity instead of the population weights does not affect the asymptotic distribution. This observation is aligned with the results for the classical inverse propensity weighted estimator in \cite{hirano2003efficient}.


	The rest of the paper is organized as follows. Section \ref{sec:model-estimation} introduces the model and provides the simple all-purpose estimator for factors, loadings, and common components. Section \ref{sec:simple-assumptions} states the necessary assumptions for the asymptotic distribution results that are presented in Section \ref{sec:asymptotic-results}. Sections \ref{sec:propensity-weighted-estimator} and \ref{sec:feasible-est-prop-score} extend the results to the propensity weighted estimator. Section \ref{sec:test} shows how to apply our model to test treatment effects. We discuss the feasible estimation in Section \ref{sec:feasible} and how to relax the rate conditions in Section \ref{sec:generalization}. The extensive simulation in Section \ref{sec:simulation} shows the good finite sample properties, the strong performance relative to other methods, and robustness results under misspecification. The Internet Appendix collects additional simulation results and all proofs.


	\section{Model and Estimation}\label{sec:model-estimation}
	
	\subsection{Model}\label{subsec:model}
	
	Assume we partially observe a panel data set $Y$ with $T$ time periods and $N$ cross-sectional units. $Y  \in \+R^{N \times T}$ has a factor structure with $r$ common factors. We denote by $F_t \in \+R^{r}$ the latent factors, $\Lambda_i \in \+R^{r }$ the factor loadings, $C_{it} = \Lambda_i^\top F_t$ the common component, and $e_{it}$  the idiosyncratic error:
	
	\[Y_{it} = \Lambda_i^\top F_t + e_{it} \quad \text{for $i = 1, 2, \cdots, N$ and $t = 1, 2, \cdots, T$}\]
	or in vector notation,
	\[\underbrace{Y_t}_{N \times 1}  = \underbrace{\Lambda}_{N \times r} \underbrace{F_t}_{r \times 1} + \underbrace{e_t}_{N \times 1} \qquad \text{for $ t = 1, 2, \cdots, T$}.\]
	
	In an asymptotic setup where $N$ and $T$ are both large, we randomly observe some entries in $Y$. Let $W_{it} \in \{0, 1\}$ be a binary variable, where $W_{it}= 1$ indicates that the $(i,t)$-th entry is observed and $W_{it} = 0$ otherwise. In this paper, we will estimate the latent factors $F$ and loadings $\Lambda$ from the partially observed $Y$, impute the missing values, and provide the inferential theory for all estimators.

	\subsection{Missing Observations}\label{subsec:missing-obs-pattern}

	We allow for very general patterns in the missing observations.
	Figure \ref{fig:obs-pattern-illustrate} shows three important examples widely seen in empirical applications. The first one is a randomly missing pattern, that is, whether an entry is observed or not does not depend on other entries or observable covariates. For example, the observational pattern of the Netflix challenge is usually modeled as entries missing at random. The second and third ones are the observation patterns for control panels in simultaneous and staggered treatment adoptions. Once a unit adopts the treatment, it stays treated afterwards, which will be modeled as missing values. These two patterns are widely assumed in the literature on causal inference in panel data.\footnote{See \cite{candes2009exact,zhou2008large} for the Netflix challenge and \cite{athey2018matrix,athey2018design} for missing patterns used in causal inference.}
	
	\begin{figure}[t!]
		\tcapfig{Examples of patterns for missing observations}
		\centering
		\begin{tabular}{L{5cm} L{5cm} L{5cm}}
			\includegraphics[width=0.33\textwidth]{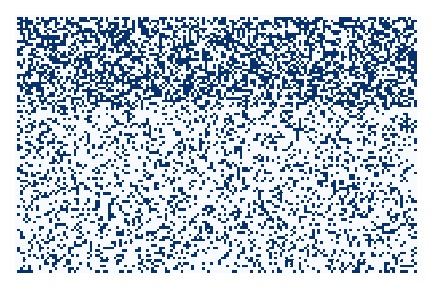} &
			\includegraphics[width=0.33\textwidth]{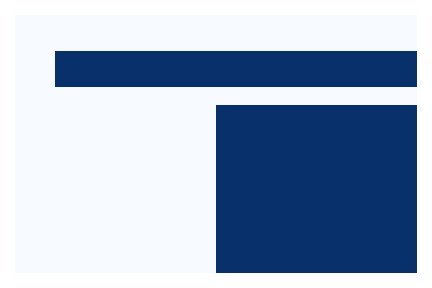} &
			\includegraphics[width=0.33\textwidth]{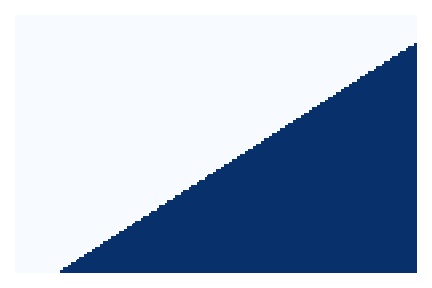} \\
			(a) Randomly missing & (b) Simultaneous treatment adoption & (c) Staggered treatment adoption
		\end{tabular}
		\bnotefig{These figures show examples of patterns of missing observations. The shaded entries indicate the missing entries for different $N\times T$ observation matrices $W$.}
		\label{fig:obs-pattern-illustrate}
	\end{figure}

	$\tlq_{ij} = \{t: W_{it} = 1 \text{ and } W_{jt} = 1 \}$ denotes the set of time periods $t$ when both units $i$ and $j$ are observed. $|\tlq_{ij}|$ is the cardinality of the set $\tlq_{ij}$. Assumption \ref{ass:obs-equal-weight} states the conditions on the observation pattern.
	
	\begin{assumpS}[Observational Pattern] \label{ass:obs-equal-weight}
		\texttt{}
		\begin{enumerate}[wide=0pt, widest=99,leftmargin=\parindent, labelsep=*]
			\item $W$ is independent of $F$ and $e$.
			\item \label{ass:add-obs}  For a given observation matrix $W$, $  \frac{|\tlq_{ij}|}{T}  \geq  \underline{q} > 0$ and  there exist constants $ q_{ij}$ and $q_{ij,kl}$  for all $i,j,k,l$
			such that  $q_{ij} =  \lim_{T \rightarrow \infty}  \frac{|\tlq_{ij}|}{T} $ and $q_{ij,kl} = \lim_{T \rightarrow \infty} \frac{|\tlq_{ij} \cap \tlq_{kl} |}{T }$.  
		\end{enumerate}
	\end{assumpS}
	
	Assumption \ref{ass:obs-equal-weight} allows very general observation patterns that can vary over time and depend on unit-specific features. In particular, the observation pattern can depend on the factor loadings that capture cross-sectional information. For the purpose of identification, we assume that the observation pattern is independent of the factors. Note that the estimator of the common components is ``symmetric'' in $N$ and $T$, and therefore we could switch the roles of $N$ and $T$ in the above assumptions. In that case, the observation pattern would be independent of the loadings but can depend on the factors. The assumption that the observation pattern is independent of the errors is closely related to the unconfoundedness assumption in \cite{rosenbaum1983central}. Assumption \ref{ass:obs-equal-weight} implicitly assumes that for any two units,  the number of time periods when both are observed is proportional to $T$. This simplifies the presentation of our results and is sufficient for most empirically relevant cases, but we will also discuss how this assumption can be relaxed. 
	
	Our framework allows for the following important examples:
	\begin{enumerate}
		\item {\it Missing at random:} $\p=p$ for all $i$ and $t$. In this case all units and times are equally likely to be observed.
		\item {\it Cross-section missing at random:} $\p=p_t$. For each $t$ each cross-sectional unit is equally likely to miss.
		\item {\it Time-series missing at random:} $\p = p_i$. For each $i$ each time observation is equally likely to miss.
        \item {\it Cross-section and time-series dependency:} $P(W_{it}=1)=p_{it}$, which allows for different probabilities for each unit and time.
		\item {\it Staggered treatment adoption:} If $W_{it}=0$ then $W_{it'}=0$ for all $t' \geq t$. This is a special case of 4. with $\p=p_{it}$. For the special case that the probability does not depend on $i$, the staggered design is a special case of cross-section missing at random $\p=p_t$.
		\item {\it Mixed frequency observations:} Each cross-section unit has a fixed known observation pattern over time. This can be modeled as one random draw for each cross-section unit to assign it to a specific pattern. A feasible model approach uses $\p=p_{t}$ as this is another special case of cross-section missing at random.
	\end{enumerate} 
	
	
	We provide an ``all-purpose'' estimator without the need to explicitly model the probability distribution $\p$. However, in some cases, we might have additional information about the missing pattern. In Section \ref{sec:propensity-weighted-estimator}, we provide a modification of our estimator that can take advantage of a model for $\p$. More specifically, we allow the cross-sectional observation pattern $\ps$ to depend on observed cross-sectional features $S = [S_i] \in \+R^{N \times K}$. These covariates $S_i$ are assumed to be time-invariant. They can be discrete, for example, an indicator variable for gender in the evaluation of a drug treatment or continuous, for example, standardized past test scores in the evaluation of an educational policy change. The cross-sectional features $S_i$ can actually be the estimated latent loadings $\Lambda_i$ themselves. We discuss how the observation probability $\ps=p_t(S_i)$ can be estimated with parametric or non-parametric estimators. While this modified estimator requires some changes to Assumption \ref{ass:obs-equal-weight}, it provides the same level of generality for the missing pattern, as discussed in Section \ref{sec:propensity-weighted-estimator}.

	\subsection{Estimator}\label{subsec:estimator}

	There are two steps to estimate the latent factor model from the partially observed panel data: First, we need to estimate the covariance matrix of the data, and second we estimate the latent factors and loadings based on the eigenvectors of the estimated covariance matrix.
	The conventional latent factor estimator without missing values applies principal component analysis to the sample covariance matrix. A natural way to deal with the missing values is to set these entries to zero. However, the conventional PCA estimator will then be biased. Our estimator correctly re-weights the entries in the covariance matrix before applying PCA.

	We first impute the missing entries by 0 and denote the imputed matrix as $\tilde Y$:\footnote{In matrix notation, we have $\tilde Y = Y \odot W$, where $\odot$ denotes the Hadamard product.}
	\[\tilde Y_{it} = Y_{it} W_{it},  \quad \text{for } i = 1, 2, \cdots, N \text{ and } t = 1, 2, \cdots, T. \]
	When some entries in $Y$ are missing, the conventional sample covariance estimator $\frac{1}{T} \tilde Y \tilde Y^\T$ is biased because the actual realizations of the missing values are not equal to zero. We propose a natural estimator of the covariance matrix, where for each entry we only use the time periods when both units are observed. This is equivalent to estimating the sample covariance matrix with $\tilde Y$, but reweighting the entries. Table \ref{tab:toyexample} is a simple example to illustrate the covariance matrix estimation if the entries are partly missing in the second half of the data.
	More generally, our sample covariance matrix estimator equals 
	\begin{eqnarray}\label{eqn:cov-est}
		\tilde\Sigma_{ij} = \frac{1}{|\tlq_{ij}|} \sum_{t \in \tlq_{ij}} Y_{it} Y_{jt}.
	\end{eqnarray}

	\begin{table}[t]
		\tcaptab{Example of covariance matrix estimation with missing entries}
		\begin{subtable}[t]{0.48\textwidth}
			\centering
			\begin{tabular}{|ccc|ccc|}
				\hline
				& &&\cellcolor{blue!25}&\cellcolor{blue!25}&\cellcolor{blue!25}\\
				$\mathbf{Y}_{1,1}$ & $\cdots$ & $\mathbf{Y}_{1,T_0}$ & \cellcolor{blue!25}$\mathbf{Y}_{1,T_0+1}$ & \cellcolor{blue!25}$\cdots$ & \cellcolor{blue!25}$ \mathbf{Y}_{1,T}$\\[2ex]
				\hline
				& &&&&\\
				$\mathbf{Y}_{2,1}$ & $\cdots$ & $\mathbf{Y}_{2,T_0}$ & $\mathbf{Y}_{2,T_0+1}$ & $\cdots$ & $\mathbf{Y}_{2,T}$ \\[2ex]
				\hline
			\end{tabular}
			\caption{Observation pattern for $Y$: Shaded entries are missing.}
			\label{tab:toy-example-obs-pattern}
		\end{subtable}
		\hspace{\fill}
		\begin{subtable}[t]{0.48\textwidth}
			\centering
			\begin{tabular}{|c|c|}
				\hline
				\cellcolor{blue!25} & \cellcolor{blue!25} \\
				\cellcolor{blue!25}$\frac{1}{T_0} \sum_{t=1}^{T_0} \mathbf{Y}_{1,t} \mathbf{Y}_{1,t}^{\top}$ & \cellcolor{blue!25}$\frac{1}{T_0} \sum_{t=1}^{T_0} \mathbf{Y}_{1,t} \mathbf{Y}_{2,t}^{\top}$ \\[2ex]
				\hline
				\cellcolor{blue!25} & \\
				\cellcolor{blue!25}$\frac{1}{T_0} \sum_{t=1}^{T_0} \mathbf{Y}_{2,t} \mathbf{Y}_{1,t}^{\top}$ & $\frac{1}{T} \sum_{t=1}^{T} \mathbf{Y}_{{2},t} \mathbf{Y}_{{2},t}^{\top}$ \\[2ex]
				\hline
			\end{tabular}
			\caption{Sample covariance matrix $\tilde\Sigma$: Shaded entries are estimated using observations up to time $T_0$}
			\label{tab:cov-est}
		\end{subtable}
		\bnotetab{This table shows an illustrative example for the covariance matrix estimation for $Y$ with missing entries. The missing entries follow a simultaneous treatment adoption pattern. For $t=T_0+1,...,T$ the first $N_0$ cross section units are missing, while the elements $N_0+1,...,N$ are observed for all $t$, i.e. $\mathbf{Y}_{1,t}= \protect \begin{pmatrix} Y_{1,t} & \cdots & Y_{N_0,t} \protect \end{pmatrix}$ and $\mathbf{Y}_{2,t}=\protect \begin{pmatrix} Y_{N_0+1,t} & \cdots & Y_{N,t} \protect \end{pmatrix}$.}
		\label{tab:toyexample}
	\end{table}

	When the data is fully observed, we can apply PCA to $\frac{1}{NT} Y  Y^\T $ to estimate the loadings. Up to rescaling, the eigenvectors of the largest eigenvalues estimate the loadings. Then, we regress $Y$ on the estimated loadings to obtain an estimate of the factors.\footnote{Alternatively, we can apply PCA to $\frac{1}{NT} Y Y^\T $ to estimate the loadings and then regress $Y^{\top}$ on the estimated loadings to estimate the factors. \cite{bai2002determining} and \cite{bai2003inferential} show that the estimated factors and loadings from this approach are consistent and asymptotically normally distributed for a fully observed panel.}  Similarly, for the partially observed data, we apply PCA to $\frac{1}{N} \tilde\Sigma$ to estimate the loadings.\footnote{We assume that the true number of factors is $r$ and has been consistently estimated as in \cite{bai2003inferential}.} Under the standard identification assumption $\tilde \Lambda^\T \tilde \Lambda/N = I_r$, we estimate the loadings $\tilde \Lambda$ as $\sqrt{N}$ times the eigenvectors of the $r$ largest eigenvalues of the sample covariance matrix, that is    
	\begin{eqnarray}\label{eqn:pca-estimate-loadings}
		\frac{1}{N} \tilde\Sigma \tilde \Lambda = \tilde \Lambda \tilde D, 
	\end{eqnarray}
	where $\tilde D$ is a diagonal matrix. 
	Then, for every time period $t$, we regress the observed $Y_{t}$ on $ \tilde \Lambda$ to estimate the factors: 
	\begin{align}\label{eqn:reg-estimate-factors}
		\tilde F_t &= \left(\sum_{i = 1}^N W_{it} \tilde \Lambda_i \tilde \Lambda_i^\T  \right)^\I  \left( \sum_{i = 1}^N W_{it} \tilde \Lambda_i Y_{it}    \right).
	\end{align}
	Interestingly, this very simple estimator automatically corrects for the impact of general observation patterns. If we have additional information that allows us to model the observation pattern as $\ps$, we propose an alternative weighted regression:
	\begin{eqnarray}\label{eqn:reg-estimate-factors-conditional}
		\tilde F^S_t =  \Bigg(\sum_{i = 1}^N \frac{W_{it}}{P(W_{it} = 1|S_i)}  \tilde  \Lambda_i \tilde  \Lambda_i^\T  \Bigg)^\I  \Bigg( \sum_{i = 1}^N \frac{W_{it}}{P(W_{it} = 1|S_i)}  \tilde  \Lambda_i  Y_{it}\Bigg).
	\end{eqnarray}
	This conditional estimator uses the weights $\frac{1}{P(W_{it} = 1|S_i)}$ in the cross-sectional regression.
	The estimator for $\tilde F^S_t$ is motivated by the inverse propensity score estimator, which is widely used in causal inference.\footnote{Since \cite{horvitz1952generalization}, weighting observations by inverse probability has been frequently used to account for missing data in mean estimation \citep{david1983nonrandom,little1986survey,little1988missing}, regression analysis \citep{robins1994estimation,robins1995semiparametric}, and causal inference \citep{hirano2003efficient}.} 
	The rationale is that the re-weighted observations correspond to a model where the data is cross-sectionally missing at random. More specifically, after re-weighting the observed data, the loadings should follow the same distribution as in the complete panel without missing observations. This could be relevant if units, that are exposed to specific factors, are more likely to miss. In the special case for cross-sectional missing at random, i.e., $\p=p_t$, the two estimators coincide.\footnote{The cross-sectional reweighting is only affected by cross-sectional differences in the probability $p_{it}$ at time $t$. Therefore we model only the dependency on a set of cross-sectional covariates $S_i$, but allow this dependency to be time-varying. We could without loss of generality consider the dependency on cross-sectional and time-varying covariates $S_{it}$. However, this would not change the cross-sectional reweighting at a specific time $t$, but would come at the cost of a more complicated notation.} We will first study the simple all-purpose estimator $\tilde F_t$ and extend it to the propensity weighted estimator $\tilde F^S_t$ in Section \ref{sec:propensity-weighted-estimator}. We show that both estimators are consistent and asymptotically normal. In most cases, $\tilde F_t$ is more efficient than the propensity score estimator, but $\tilde F^S_t$ can have desirable robustness properties under misspecification. The last step is to estimate the common component $C_{it} = \Lambda_i^\T F_t$. We use the plug-in estimator, $\tilde C_{it} = \tilde \Lambda_i^\T \tilde F_t$ respectively $\tilde C^S_{it} = \tilde \Lambda_i^\T \tilde F^S_t$. If $Y_{it}$ is not observed, we impute the missing values with $\tilde C_{it}$ or $\tilde C^S_{it}$.

	\subsection{Illustration}\label{subsec:illustration}

	We illustrate in a simple example how missing observations change the conventional PCA estimator with fully observed data. Assume that we have only one factor, and the factor, loading and residual component are i.i.d. normally distributed with $F_t \overset{i.i.d.}{\sim} \calN(0,\sigma_F^2)$, $\Lambda_i  \overset{i.i.d.}{\sim} \calN(0,1)$  and $ e_{it}  \overset{i.i.d.}{\sim} \calN(0,\sigma_e^2)$. We assume that the observations for units $1,...,N_0$ and for the times $T_0+1,...,T$ are missing according to the simultaneous adoption pattern of Table \ref{tab:toyexample}. We separate the vector of factor realizations into its first $\mathbf{F}_1= \begin{pmatrix} F_1 & \cdots & F_{T_0} \end{pmatrix}^{\top} $ and second part $\mathbf{F}_2= \begin{pmatrix} F_{T_0+1} & \cdots & F_{T} \end{pmatrix}^{\top} $ and similarly for the loadings $\mathbf{\Lambda}_{1}=\begin{pmatrix} \Lambda_1 & \cdots & \Lambda_{N_0} \end{pmatrix}^{\top}$ and $\mathbf{\Lambda}_{2}=\begin{pmatrix} \Lambda_{N_0+1} & \cdots & \Lambda_{N} \end{pmatrix}^{\top}$. Note that in this simple example $\tilde F$ and $\tilde F^S$ coincide.
	
	We start with the simplest case without error terms $e_t$ to illustrate the logic of reweighting entries. In this case the conventional covariance matrix equals
	\begin{align*}
		\frac{1}{T} \tilde Y \tilde Y^{\top} = \frac{1}{T} \begin{pmatrix} \mathbf{\Lambda}_1 \mathbf{F}_1^{\top} & 0 \\  \mathbf{\Lambda}_2 \mathbf{F}_1^{\top}& \mathbf{\Lambda}_{2} \mathbf{F}_2^{\top}  \end{pmatrix} \begin{pmatrix} \mathbf{F}_1 \mathbf{\Lambda}_1^{\top} &  \mathbf{F}_1 \mathbf{\Lambda}_{2}^{\top}  \\0& \mathbf{F}_2 \mathbf{\Lambda}_2^{\top} \end{pmatrix} = \begin{pmatrix} \sqrt{\frac{T_0}{T}} \mathbf{\Lambda}_1 \\ \mathbf{\Lambda}_2  \end{pmatrix} \left( \sigma_F^2 + o_P(1) \right) \begin{pmatrix} \sqrt{\frac{T_0}{T}} \mathbf{\Lambda}_1^{\top} & \mathbf{\Lambda}_2^{\top}  \end{pmatrix} .
	\end{align*}
	Obviously, the eigenvector of this matrix is a biased estimate of the loadings. In contrast, the eigenvector of the correctly weighted sample covariance matrix consistently estimates the loadings:
	\begin{align*}
		\tilde \Sigma = \begin{pmatrix} \mathbf{\Lambda}_1 \frac{\mathbf{F}_1^{\top} \mathbf{F}_1}{T_0} \mathbf{\Lambda}_1^{\top} &  \mathbf{\Lambda}_1 \frac{\mathbf{F}_1^{\top} \mathbf{F}_1}{T_0} \mathbf{\Lambda}_2^{\top}  \\  \mathbf{\Lambda}_2 \frac{\mathbf{F}_1^{\top} \mathbf{F}_1}{T_0} \mathbf{\Lambda}_1^{\top}  &  \mathbf{\Lambda}_2 \frac{\mathbf{F}_1^{\top} \mathbf{F}_1 + \mathbf{F}_2^{\top} \mathbf{F}_2}{T} \mathbf{\Lambda}_2^{\top}  \end{pmatrix} = \begin{pmatrix}\mathbf{\Lambda}_1 \\ \mathbf{\Lambda}_2  \end{pmatrix} \left( \sigma_F^2 + o_P(1) \right) \begin{pmatrix} \mathbf{\Lambda}_1^{\top} & \mathbf{\Lambda}_2^{\top}  \end{pmatrix}.
	\end{align*}
	The same logic carries over to the estimator of the factors. Assume that we know the population loadings, which we use here instead of the estimated loadings in the regression to estimate the factors:
	\begin{align*}
		\frac{\tilde Y^\T  {\Lambda}}{N} \left ( \frac{{\Lambda}^{\top} {\Lambda}}{N}   \right)^{-1}=\frac{1}{N} \begin{pmatrix}  \mathbf{F}_1 \mathbf{\Lambda}_1^{\top} & \mathbf{F}_2 \mathbf{\Lambda}_2^{\top} \\  0  & \mathbf{F}_2 \mathbf{\Lambda}_2^{\top} \end{pmatrix} \begin{pmatrix} \mathbf{\Lambda}_1 \\ \mathbf{\Lambda}_2 \end{pmatrix} + o_P(1) = \begin{pmatrix}  \mathbf{F}_1   \\ \mathbf{F}_2 \frac{N-N_0}{N}  \end{pmatrix} + o_P(1),
	\end{align*}
	which is a biased estimator for the second time period. The regression in Equation \eqref{eqn:reg-estimate-factors} corresponds to a weighted least square regression which provides the correct estimator:
	\begin{align*}
		\tilde{\mathbf{F}}_1 =&  \mathbf{F}_1 \frac{\mathbf{\Lambda}_1^{\top}\mathbf{\Lambda}_1 + \mathbf{\Lambda}_2^{\top}\mathbf{\Lambda}_2}{N} \left( \frac{\mathbf{\Lambda}_1^{\top}\mathbf{\Lambda}_1 + \mathbf{\Lambda}_2^{\top}\mathbf{\Lambda}_2}{N}  \right)^{-1} = \mathbf{F}_1 + o_P(1) \\
		\tilde{\mathbf{F}}_2 =&  \mathbf{F}_2 \frac{ \mathbf{\Lambda}_2^{\top}\mathbf{\Lambda}_2}{N-N_0} \left( \frac{ \mathbf{\Lambda}_2^{\top}\mathbf{\Lambda}_2}{N-N_0}  \right)^{-1}   = \mathbf{F}_2 + o_P(1) 
	\end{align*}
	The proper reweighting in the loading and factor estimation leads to an additional correction term in the asymptotic variance of the estimator. As an illustration of this additional challenge, we add the i.i.d. error term $e_{it}$ to our example. In our simplified setup our consistent estimator for the loadings $\tilde \Lambda$ has the following expansion for $i=1,...,N_0$:\footnote{The results are similar for $i > N_0$ with the expansion $\sqrt{T} \left( \tilde \Lambda_i - \Lambda_i  \right) =\sqrt{T} \left( \frac{\tilde F^{\top} \tilde F}{T}\right)^{-1} \frac{1}{N} \Big[ \mathbf{\Lambda}_1^{\top}\mathbf{\Lambda}_1  \frac{1}{T_0} \sum_{t=1}^{T_0} F_t e_{it}$ $+ \mathbf{\Lambda}_2^{\top}\mathbf{\Lambda}_2  \frac{1}{T}  \sum_{t=1}^{T} F_t e_{it} \Big]+ \sqrt{T} \left( \frac{\tilde F^{\top} \tilde F}{T}  \right)^{-1} \frac{1}{N} \mathbf{\Lambda}_1^{\top}\mathbf{\Lambda}_1   \left( \frac{\mathbf{F}_1^{\top} \mathbf{F}_1}{T_0} - \frac{F^{\top}F}{T}  \right) \Lambda_i + o_P(1)$ and asymptotic distribution $\sqrt{T} \left( \tilde \Lambda_i - \Lambda_i  \right)  \overset{d}{\rightarrow} \calN \left( 0, \text{plim}\left( \Big( \frac{T-T_0}{T_0} \frac{N_0^2}{N^2}  + 1 \Big) \frac{\sigma_e^2}{\sigma_F^2} +  2 \frac{N_0^2}{N^2}  \frac{T-T_0}{T_0}  \Lambda_i^2  \right) \right)$.
	}
	\begin{align*}
		\sqrt{T} \left( \tilde \Lambda_i - \Lambda_i  \right) = \sqrt{\frac{T}{T_0}} \left( \frac{{\tilde F^{\top}} {\tilde F}}{T}\right)^{-1} \frac{1}{\sqrt{T_0}} \sum_{t=1}^{T_0} F_t e_{it} + \sqrt{T} \left( \frac{{\tilde{F}}^{\top} {\tilde{F}}}{T}  \right)^{-1} \left( \frac{\mathbf{F}_1^{\top} \mathbf{F}_1}{T_0} - \frac{\mathbf{F}^{\top}\mathbf{F}}{T}  \right) \Lambda_i + o_P(1),
	\end{align*}
	which results in the asymptotic normal distribution
	\begin{equation}
		\sqrt{T} \left( \tilde \Lambda_i - \Lambda_i  \right)  \overset{d}{\rightarrow} \calN \left( 0, \text{plim}\left(\frac{T}{T_0} \frac{\sigma_e^2}{\sigma_F^2} + 2 \frac{T-T_0}{T_0} \Lambda_i^2 \right) \right) \qquad \text{for $i=1,...,N_0$.} \nonumber
	\end{equation}

	The second term in the asymptotic expansion is due to averaging over different number of units for different elements of the loadings. This additional variance correction term vanishes for $T_0 \rightarrow T$. Similar terms appear in the distribution of the estimators of the factors and common components. We show under general conditions how these correction terms arise in the asymptotic distribution and how to take them into account for the inferential theory.

	\section{Assumptions}\label{sec:simple-assumptions}
	
	We assume an approximate factor structure at the same level of generality as in \cite{bai2003inferential}. The factors and loadings have non-trivial time-series and cross-sectional dependency. We allow the errors to be weakly correlated in the time-series and cross-sectional dimension. The asymptotic distributions are based on general martingale central limit theorems. The general Assumptions \ref{ass:factor-model} and \ref{ass:mom-clt} are collected in the Appendix. In the main text, we present a simplified factor model with the stronger Assumptions \ref{ass:simple-factor-model} and \ref{ass:simple-moment}, which substantially simplifies the notation but conveys the main conceptual insights of the general model. It allows us to highlight the effect of missing observations.
	
	The consistency results are based on Assumption \ref{ass:simple-factor-model} that assumes that all observations are i.i.d. The key elements are that the factors and loadings are systematic in the sense that they lead to exploding eigenvalues, while the error terms are non-systematic with bounded eigenvalues in the covariance matrix of $Y$. These are standard factor model assumptions. The asymptotic distribution results require additional restrictions on the missing patterns, as stated in Assumption \ref{ass:simple-moment}. 
	
	\begin{assumpS}[Simplified Factor Model]\label{ass:simple-factor-model}
		\texttt{} \\
		There exists a positive constant $M<\infty$ such that:
		\begin{enumerate}[wide=0pt, widest=99,leftmargin=\parindent, labelsep=*]
			\item Factors: $F_t \stackrel{\iid}{\sim} (0, \Sigma_F )$, $\+E[\norm{F_t}^4]  \leq M$, and $\+E \norm{F_t F_t^\T - \Sigma_F}^{2+\epsilon} \leq M$ for some $\epsilon \in (0, 1)$.
			\item Factor loadings:  
			$\Lambda_i \stackrel{\iid}{\sim} (0, \Sigma_{\Lambda} )$ and
			$\+E[\norm{\Lambda_i}^4]  \leq M$.
			\item Errors: $e_{it} \stackrel{\iid}{\sim} (0, \sigma_e^2)$, $\+E [ e_{it}^8] \leq M$.
			\item Independence: $F$, $\Lambda$ and $e$ are independent. 
			\item Eigenvalues: The eigenvalues of $\Sigma_\Lambda \Sigma_F$ are distinct.
		\end{enumerate}
	\end{assumpS}

	\begin{assumpS}[Moments of Simplified Factor Model]\label{ass:simple-moment}
		\texttt{}
		\begin{enumerate}[wide=0pt, widest=99,leftmargin=\parindent, labelsep=*]
			\item Systematic loadings: $\frac{1}{N} \sum_{i=1}^N \Lambda_i \Lambda_i^{\top} W_{it} \overset{P}{\rightarrow} \Sigma_{\Lambda,t}$ for some positive definite matrix $\Sigma_{\Lambda, t}$ for any $t$.
			\item  Dependency in missing pattern: \label{ass:add-obs-equal-weight} $\frac{1}{N^2} \sum_{i=1}^N \sum_{l=1}^N \frac{q_{ij,lj}}{q_{ij} q_{lj}} \overset{P}{\rightarrow}  \omega_{jj} $, $\lim_{N\rightarrow \infty} \frac{1}{N^3} \sum_{i=1}^N \sum_{l = 1}^N  \sum_{k = 1}^N   \frac{q_{li,kj}}{q_{li}q_{kj}}  \overset{P}{\rightarrow}   \omega_j$ and $\lim_{N\rightarrow \infty} \frac{1}{N^4} \sum_{i=1}^N \sum_{l = 1}^N \sum_{j = 1}^N \sum_{k = 1}^N   \frac{q_{li,kj}}{q_{li}q_{kj}}  \overset{P}{\rightarrow}  \omega $ for all $j$ and some constants $\omega_{jj}, \omega_j, \omega$.
		\end{enumerate}
	\end{assumpS}
	Assumption \ref{ass:simple-moment} has two key elements. First, the full rank assumption of $\Sigma_{\Lambda, t}$ captures that the factor loadings are systematic for the observed entries. Second, the number of observed units at every time period $t$ is proportional to $N$ and different units share a number of observed entries that is proportional to $T$. The impact of the missing pattern on the asymptotic variances of the estimators is captured by the three key parameters $\omega, \omega_j$ and $\omega_{jj}$. Note that by construction these constants satisfy $\omega_{jj}, \omega_j, \omega \geq 1$. If the observations are missing at random with probability $p$, then $\omega_{jj} = \frac{1}{p}$, $\omega_j = 1$ and $\omega= 1$. 
	
	As stated in Proposition \ref{prop:simple-assump-imply-general-assump} in the Appendix, the simplified model is just a special case of the general approximate factor model specified by Assumptions \ref{ass:factor-model} and \ref{ass:mom-clt}. The simplified Assumption \ref{ass:simple-factor-model} implies the general Assumption \ref{ass:factor-model}, while Assumption \ref{ass:simple-moment} combined with the other simplified assumptions implies the general Assumption \ref{ass:mom-clt}.

	We assume that the number of factors $r$ is consistently estimated. For example the criteria developed in \cite{bai2002determining} can be extended to our case of missing values based on the various bounds and expansions that we derive in this paper. A promising alternative would be to extend the cross-validation estimator of \cite{jin2020factor} or an eigenvalue ratio argument as in \cite{ahn2013} to general missing patterns. Given a consistent estimator for the number of factors, we can treat $r$ as known.


	\section{Asymptotic Results}\label{sec:asymptotic-results}	
	\subsection{Consistency}\label{subsec:consistency}

	We first show the consistency of our estimators. Our analysis starts with plugging $\tilde Y = (\Lambda^\T F + e) \odot W$ into Equation (\ref{eqn:pca-estimate-loadings}) which yields the following decomposition:
	\begin{eqnarray*}
		\nonumber \tilde{\Lambda}_j =\!\!\!\!\!\!\!\! && \underbrace{\frac{1}{NT} \tilde{D}^\I \sum_{i=1}^N \tilde{\Lambda}_i \Lambda_i^\T F^\T \diag(W_i \odot W_j) F \Lambda_j/ q_{ij} }_{H_j \Lambda_j }  + \underbrace{\frac{1}{NT} \tilde{D}^\I \sum_{i=1}^N \tilde{\Lambda}_i e_i^\T \diag(W_i \odot W_j) F \Lambda_j/q_{ij} }_{(a)}   \\
		&&  + \underbrace{\frac{1}{NT} \tilde{D}^\I \sum_{i=1}^N   \tilde{\Lambda}_i \Lambda_i^\T F^\T \diag(W_i \odot W_j) e_j/q_{ij} }_{(b)}  + \underbrace{\frac{1}{NT} \tilde{D}^\I  \sum_{i=1}^N\tilde{\Lambda}_i e_i^\T \diag(W_i \odot W_j) e_j/q_{ij}  }_{(c)}.
	\end{eqnarray*}
	Similar to \cite{bai2002determining} this decomposition relates the estimated loadings to the population loadings, $\tilde \Lambda_j = H_j \Lambda_j + (a) + (b) + (c)$, up to a rotation matrix $H_j = \frac{1}{NT} \tilde{D}^\I \sum_{i=1}^N \tilde{\Lambda}_i \Lambda_i^\T F^\T \diag(W_i \odot W_j) F/q_{ij}$. The key difference to factor analysis with fully observed data is that this rotation matrix can be different for different units $j$. However, the estimation of the factors is based on a projection on the loading space and hence implicitly requires the same rotation matrix for all loadings. 
	
	We consider for all units a unified rotation matrix defined as $H = \frac{1}{NT} \tilde{D}^\I \tilde{\Lambda}^\T \Lambda F^\T F$ which is essentially the same conventional rotation matrix as in \cite{bai2002determining}. This yields the decomposition 
	\[\tilde{\Lambda}_j - H \Lambda_j = \tilde{\Lambda}_j - H_j \Lambda_j + (H_j - H) \Lambda_j = (a) + (b) + (c) +  (H_j - H) \Lambda_j.  \]
	We show that the cross-section averages of the square of $(a)$, $(b)$ and $(c)$ converge to 0 at the rate  $O_P \Lp \min\Lp \frac{1}{N}, \frac{1}{T} \Rp\Rp$. 
	The key difference compared with the fully observed factor analysis is the last term. 
	If $\frac{1}{T} F^\T F \xrightarrow{P} \Sigma_F$ and $\frac{1}{|\tlq_{ij}|} \sum_{t \in \tlq_{ij}} F_t F_t^\T \xrightarrow{P} \Sigma_F$, we can show that $H_j - H = O_P \Lp \min \Lp  \frac{1}{\sqrt{N}}, \frac{1}{\sqrt{T}} \Rp \Rp$. This rate is sufficiently fast to obtain consistency, but will contribute to the asymptotic normal distribution. Note that the correction term $H_j - H$ is a fundamental problem for any estimator that makes use of all observations.\footnote{The estimator in \cite{bai2019matrix} can avoid this term by neglecting partially observed entries, which means that in general, they are using less information. The estimator of \cite{bai2019matrix} is optimized for the block structure of a simultaneous adoption pattern. It runs two PCA estimates for the block with full cross-sectional observations and the block with full time-series observations. Hence, they can infer the ``local'' rotation matrices for each block and rotate the estimates to avoid the correction term $H_j - H$. \cite{cahan2021factor} leverage the block structure, and in the first step run PCA to estimate factors on the block with full-series observations, so that the correction term $H_j - H$ is avoided. If we seek to use full observations in the first step, which is what we propose in this paper, then the correction term $H_j - H$ cannot be avoided.
	}

	The next theorem shows the consistency of the estimated loadings.
	
	\begin{theorem}\label{thm:consistency-same-H}
		Define $\delta_{NT} = \min (N, T)$.
		Under Assumptions \ref{ass:obs-equal-weight} and \ref{ass:factor-model} it holds that
		\begin{eqnarray}\label{eqn:consistency}
			\delta_{NT} \Lp \frac{1}{N} \sum_{j=1}^N \norm{\tilde{\Lambda}_j - H \Lambda_j}^2 \Rp = O_P(1),
		\end{eqnarray}
		where $H = \frac{1}{NT} \tilde{D}^\I \tilde{\Lambda}^\T \Lambda F^\T F$.  
	\end{theorem}
	
	Theorem \ref{thm:consistency-same-H} states that the complete loading matrix can be consistently estimated up to an appropriate rotation as $N,T \rightarrow \infty$ even if we only observe an incomplete panel matrix.
	The convergence rate is the same rate as for the fully observed panel in \cite{bai2002determining}. Theorem \ref{thm:consistency-same-H} is based on the assumption that the observed entries are representative of the missing entries and hence provide a consistent estimation. Theorem \ref{thm:consistency-same-H} is a critical intermediate step to show the asymptotic normality of the estimated factor model in the next section.

	\subsection{Asymptotic Normality}\label{subsec:normality}

	The factors, loadings, and common components are asymptotically normally distributed. Indeed, Theorem \ref{theorem:asy-normal-equal-weight} states that the asymptotic distributions have two parts: First, we recover the asymptotic variance that is identical to the conventional PCA in \cite{bai2003inferential} under the same rate conditions. These are the expression when we set the additional correction terms $\covII_{\Lambda,j}$ and $\covII_{F,t} $ to zero. However, in the presence of missing values, these correction terms are necessary to capture the additional uncertainty. 
	Theorem \ref{theorem:asy-normal-equal-weight}  also includes the asymptotic expansions that lead to the normal distributions. As stated in the previous section, the difference between the unit-specific rotation $H_j$ and the ``global'' rotation matrix $H$ contributes to the distribution and leads to the variance correction terms $\covII_{\Lambda,j}$ and $\covII_{F,t}$. As expected, this variance correction is increasing in the number of missing observations. We want to emphasize again that this type of variance correction is a conceptual issue that cannot be avoided when making use of all observed entries.\footnote{In the asymptotic distribution, we apply the rotation matrices to the estimated loadings and factors instead of their population values as in \cite{bai2003inferential}. Obviously, these two representations are equivalent and can be easily transformed into each other. Our choice of representation was made for exposition purposes only.}

	\begin{theorem}\label{theorem:asy-normal-equal-weight}
		Under Assumptions \ref{ass:obs-equal-weight}, \ref{ass:factor-model} and \ref{ass:mom-clt} and for $N,T \rightarrow \infty$ we have for each $j$ and $t$:
		\begin{enumerate}[wide=0pt, widest=99,leftmargin=\parindent, labelsep=*]
			\item For $\sqrt{T}/N \rightarrow 0$ the asymptotic distribution of the loadings is 
			\begin{align}
				\nonumber  \sqrt{T}  \Sigma_{\Lambda,j}^{-1/2} &(H^{-1}   \tilde{\Lambda}_j - \Lambda_j )   =  \Sigma_{\Lambda,j}^{-1/2}  \left( \frac{1}{T}  F^{\top}F \right)^{-1} \left( \frac{1}{N} \Lambda^{\top} \Lambda \right)^{-1}  \Bigg[\bigg( \frac{1}{N} \sum_{i=1}^N \Lambda_i \Lambda_i^{\top}  \sqrt{\frac{T}{|Q_{ij}|}} \frac{1}{\sqrt{|Q_{ij}|} } \sum_{t \in Q_{ij}} F_t e_{jt} \bigg) \\
				&+ \bigg( \frac{1}{N} \sum_{i=1}^N \Lambda_i \Lambda_i^{\top} \sqrt{T} \Big( \frac{1}{|Q_{ij}|} \sum_{t \in Q_{ij}} F_t F_t^{\top} - \frac{1}{T} F^{\top} F \Big) \bigg)  \Lambda_j  \Bigg] + o_P(1) 
				\xrightarrow{d} \calN  ( 0, I_r ) \label{eqn:asy-var-lam},
			\end{align}
			where 
			\begin{align*}
				\Sigma_{\Lambda,j} = \Sigma_{F}^\I \Sigma_{\Lambda}^\I   \big[ \covI_{\Lambda,j} + \covII_{\Lambda,j} \big]\Sigma_{\Lambda}^\I  \Sigma_{F}^\I
			\end{align*}
			with $\covII_{\Lambda,j} = h_j(\Lambda_j)$. $\covI_{\Lambda,j}$ and the function $h_j(\cdot) $ are defined in Assumptions \ref{ass:mom-clt}.\ref{ass:asy-normal-main-term-thm-loading} and \ref{ass:mom-clt}.\ref{ass:asy-normal-add-term-thm-loading}.
			\item For $\sqrt{N}/T \rightarrow 0$ and $\sqrt{T}/N \rightarrow 0$, the asymptotic distribution of the factors is 
			\begin{align}
				\nonumber \sqrt{\delta_{NT}} &\Sigma_{F,t}^{-1/2}  (H^\T \tilde F_t -  F_t ) =  \Sigma_{F,t}^{-1/2}  \bigg( \frac{1}{N} \sum_{i=1}^N W_{it} \Lambda_i \Lambda_i^{\top} \bigg)^{-1}  \bigg(  \sqrt{\frac{\delta_{NT}}{N}}  \frac{1}{\sqrt{N}} \sum_{i=1}^N W_{it} \Lambda_i e_{it} \bigg) \\  
				&+   \Sigma_{F,t}^{-1/2}  \bigg( \frac{1}{N} \sum_{i=1}^N W_{it} \Lambda_i \Lambda_i^{\top} \bigg)^{-1}  \bigg( \frac{\sqrt{\delta_{NT}}}{N} \sum_{i=1}^N W_{it} \left( H^{-1} \tilde \Lambda_i - \Lambda_i \right) \Lambda_i^{\top} F_t \bigg) + o_P(1) 
				\xrightarrow{d} \calN \bigg(  0,  I_r  \bigg), \label{eqn:asy-var-F-equal-weight}
			\end{align}
			where 
			\begin{align*}
				\Sigma_{F,t} =  \Sigma_{\Lambda,t}^\I 
				\Big[ \frac{\delta_{NT}}{N} \covI_{F,t}  +\frac{\delta_{NT}}{T}  \covII_{F,t}  \Big] \Sigma_{\Lambda,t}^\I,  
			\end{align*} 
			with $\covII_{F,t} = g_t(F_t)$. $ \covI_{F,t}$ and the function $g_t(\cdot)$ are defined in Assumptions \ref{ass:mom-clt}.\ref{ass:asy-normal-main-term-thm-factor} and \ref{ass:mom-clt}.\ref{ass:asy-normal-add-term-thm-loading}.
			\item The asymptotic distribution of the common component is
			\begin{align}
				\sqrt{\delta_{NT}} \Sigma_{C,it}^{-1/2} (\tilde C_{jt} - C_{jt})  =& \sqrt{\delta_{NT}} \Sigma_{C,it}^{-1/2} \left( H^{-1} \tilde \Lambda_j - \Lambda_j \right)^{\top} F_t + \sqrt{\delta_{NT}} \Sigma_{C,it}^{-1/2} \Lambda_j^{\top} \left( H^{\top} \tilde F_t - F_t \right) + o_P(1) \nonumber\\ \xrightarrow{d}& \calN (0, 1), \label{eqn:asy-var-C-equal-weight}
			\end{align} 
			where 
			\begin{align*}
				\Sigma_{C,it} =& \frac{\delta_{NT}}{T}   F_t^{\top}  \Sigma_F^\I \Sigma_\Lambda^\I \left( \covI_{\Lambda,j} +  \covII_{\Lambda,j} \right)  \Sigma_\Lambda^\I \Sigma_F^\I  F_t  
				+ \frac{\delta_{NT}}{T}   \Lambda_j^\T \Sigma_{\Lambda,t}^\I \covII_{F,t}  \Sigma_{\Lambda,t}^\I   \Lambda_j  \\&+ \frac{\delta_{NT}}{N}   \Lambda_j^{\top}  \Sigma_{\Lambda,t}^\I \covI_{F,t} \Sigma_{\Lambda,t}^\I\Lambda_j -    2  \frac{\delta_{NT}}{T}  \Lambda_j^{\top} \Sigma_{\Lambda,t}^\I \covIII_{\Lambda, F, j, t}   \Sigma_\Lambda^\I   \Sigma_F^\I  F_t. 
			\end{align*}	
			$ \covIII_{\Lambda, F, j, t} = g^{\cov}_{j,t}(\Lambda_j, F_t)$, and the function $ g^{\cov}_{j,t}(\cdot,\cdot)$ is defined in Assumption \ref{ass:mom-clt}.\ref{ass:asy-normal-add-term-thm-loading}.
		\end{enumerate}
	\end{theorem}

	Importantly, the estimator for the factors has a different convergence rate compared to the conventional estimator on fully observed data. The asymptotic distribution of the factors is determined by two terms with different convergence rates, $\frac{\delta_{NT}}{N} \covI_{F,t}  +\frac{\delta_{NT}}{T}   \covII_{F,t}$. With a fully observed panel $ \covII_{F,t}$ would disappear, and the factors would converge at a rate of $\sqrt{N}$. However, with observations that are not missing at random, the difference between $H_j$ and $H$, that appears in the loading expansion $ H^{-1} \tilde \Lambda_j - \Lambda_j $ and has a convergence rate of $\sqrt{T}$, also contributes to the asymptotic distribution of estimated factors, which results in the overall rate $\sqrt{\delta_{NT}}$.

	The asymptotic distribution of common components depends on the estimation error of the estimated loadings and factors. In the asymptotic distribution of the estimated loadings and factors, the conventional part with asymptotic variances $\covI_{\Lambda,j}$ and $\covI_{F,t}$ is asymptotically independent as argued in \cite{bai2003inferential}. However, the second part with the asymptotic variances $\covII_{\Lambda,j}$ and $\covII_{F,t}$ that captures the difference between $H_j$ and $H$ is in general correlated, and hence their covariance $\covIII_{\Lambda, F, j, t}$ contributes to the asymptotic variance of common components as stated in Equation \eqref{eqn:asy-var-C-equal-weight}.
	
	The correction terms in the asymptotic variances are determined by the functions $h_j(.)$, $g_t(.)$ and $g^{\cov}_{j,t}(.,.)$. These are quadratic functions in the elements of $\Lambda_i$ and $F_t$, which depend on the moments of the factor model. As the loadings and factors are random, it implies that the asymptotic variances themselves are random. This complicates the analysis, but our assumptions ensure that the normalized estimates converge to a standard normal distribution.

	The distribution results of Theorem \ref{theorem:asy-normal-equal-weight} simplify under Assumptions \ref{ass:simple-factor-model} and \ref{ass:simple-moment}, and we can provide explicit expressions for the asymptotic variances. If we assume in addition that the proportions of observed time-series ($q_{ij}$ and $q_{ij,kl}$) are independent of the second moment of the loadings $\Lambda_i \Lambda_i^\T$, we can further separate the effect of missing patterns from the properties of the factor model.
	
	
	\begin{corollary}\label{corollary:asy-normal-equal-weight}
		Suppose Assumptions \ref{ass:obs-equal-weight}, \ref{ass:simple-factor-model} and \ref{ass:simple-moment} hold and $N,T \rightarrow \infty$. Then Theorem \ref{theorem:asy-normal-equal-weight} holds. If in addition, $q_{ij}$ and $q_{ij,kl}$ are independent of $\Lambda_m \Lambda_m^\T$ for all $i, j, k, l, m$, then the asymptotic variances simplify as follows with the weights $\omega, \omega_{j}$ and $\omega_{jj}$ defined in Assumption \ref{ass:simple-moment}:
		\begin{enumerate}[wide=0pt, widest=99,leftmargin=\parindent, labelsep=*]
			\item 
			The asymptotic variance of the loadings in formula \eqref{eqn:asy-var-lam} simplifies to 
			\[\Sigma_{\Lambda,j} = \omega_{jj} \cdot  \scovI_{\Lambda}+ (\omega_{jj}-1) \scovII_{\Lambda,j}, \]
			where 
			\[\scovI_{\Lambda}= \Sigma_F^{-1} \sigma_e^2, \;\; \quad \scovII_{\Lambda,j}= \Sigma_F^{-1} \Sigma_{\Lambda}^{-1}  \big( \Lambda_j^{\top} \otimes \Sigma_{\Lambda}  \big) \Xi_F \big( \Lambda_j \otimes \Sigma_{\Lambda} \big) \Sigma_{\Lambda}^{-1} \Sigma_F^{-1}, \]
			and $\+E[\tvec(F_t F_t^\T - \Sigma_{F})  \tvec(F_t F_t^\T - \Sigma_{F})^\T] = \Xi_F$.
			\item 
			The asymptotic variance of the factors in formula \eqref{eqn:asy-var-F-equal-weight} simplifies to 
			\[\Sigma_{F,t} = \frac{\delta_{NT}}{N}  \scovI_{F,t}  + \frac{\delta_{NT}}{T} (\omega-1) \scovII_{F,t},  \]
			where 
			\[
			\scovI_{F,t} =  \Sigma_{\Lambda,t}^{-1} \sigma_e^2,  \quad \quad \scovII_{F,t} = \Sigma_{\Lambda,t}^{-1} \big( I_r \otimes (F_t^{\top} \Sigma_F^{-1} \Sigma_{\Lambda}^{-1}) \big) ( \Sigma_{\Lambda,t}  \otimes  \Sigma_\Lambda )  \Xi_F    (  \Sigma_{\Lambda,t} \otimes \Sigma_\Lambda)  \big( I_r \otimes (\Sigma_{\Lambda}^{-1} \Sigma_F^{-1} F_t )\big) \Sigma_{\Lambda,t}^{-1}.
			\]
			\item 
			The asymptotic variance of the common component in formula \eqref{eqn:asy-var-C-equal-weight} simplifies to
			\begin{align*}
				\Sigma_{C,it} = & \frac{\delta_{NT}}{T} \bigg[ F_t^{\top} \big( \omega_{jj} \cdot  \scovI_{\Lambda} + (\omega_{jj}-1) \cdot \scovII_{\Lambda,j} \big)F_t + (\omega-1) \Lambda_j^\T \scovII_{F,t} \Lambda_j      - 2  ( \omega_j - 1  )\Lambda_j^{\top}  \scovIII_{\Lambda,F,j,t} F_{t}  \bigg] \\ &  + \frac{\delta_{NT}}{N}   \Lambda_j^{\top} \scovI_{F,t} \Lambda_j, 
			\end{align*} where
			\[ \scovIII_{\Lambda,F,j,t}  = \Sigma_{\Lambda, t}^\I \big(I_r \otimes (    F_t^\T \Sigma_F^\I  \Sigma_\Lambda^\I ) \big)    ( \Sigma_{\Lambda,t}  \otimes  \Sigma_\Lambda )         \Xi_F  ( \Lambda_j \otimes  \Sigma_\Lambda)   \Sigma_\Lambda^\I \Sigma_F^{-1}  .    \]
		\end{enumerate}
	\end{corollary}

	The simplified model provides a clear interpretation of the effect of missing data. Importantly, the parameters $\omega, \omega_{j}$ and $\omega_{jj}$, that depend only on the missing pattern, but not on the factor model, determine the weights of correction terms. The asymptotic covariance of the loadings is a weighted combination of the variance of an OLS regression of the population factors $F$ on $Y_j$ and the correction term. The weight $\omega_{jj} \geq 1$ depends on the number of the observed entries and the similarities in observation patterns for different units. Without missing data, it equals $\omega_{jj}=1$ and the correction term disappears. If the data is observed uniformly at random with probability $p$, the weight equals $\omega_{jj}=1/p$ which is increasing in the proportion of missing observations. 
	
	Similarly, the asymptotic variance of the factors has two components: the variance of an OLS regression of the population loadings on $Y_t$ using only observed entries, and the correction term. The weight $\omega \geq 1$ increases the scale of the correction term. When all entries are observed, or all entries are observed cross-sectionally at random (with either the same or different probabilities), then $\omega = 1$, the correction term vanishes, and the asymptotic variance only depends on $\scovI_{F,t}$. If the missing pattern does not depend on the loadings, then $\Sigma_{\Lambda,t}=p_t \Sigma_{\Lambda}$ and $\scovI_{F,t}$ simplifies to $\frac{1}{p_t} \Sigma_{\Lambda}^{-1} \sigma_e^2$ which is the variance of an OLS regression of the population loadings on $Y_t$ scaled by the inverse proportion of observed entries at time $t$. 
	
	The distribution of the common component depends on all three parameters $\omega, \omega_{j}$ and $\omega_{jj}$. If all entries are observed at random, then $\omega_j = 1$ and the contribution of the loading and factor distribution to the common component are separated similar to the conventional PCA setup in \cite{bai2003inferential}. In this case, only the two terms $\omega_{jj} F_t^\T \scovI_\Lambda F_t$ and $\Lambda_j^\T \scovI_{F,t} \Lambda_j$ remain in the asymptotic variance.

		\begin{remark}
			%
			If all entries are observed at random with equal probability, we can use the approach of \cite{jin2020factor} to estimate the factor model and impute the missing entries. We compare the efficiency of our approach with the one of \cite{jin2020factor}. For a direct comparison, we follow the order of estimation in \cite{jin2020factor} and switch the role of factors and loadings in our all-purpose estimator: We first estimate the factors from the time-series sample covariance matrix, and then estimate the loadings from a time-series regression of the observed outcomes on the estimated factors. 
			\begin{proposition}\label{prop:compare-su}
				Suppose Assumptions \ref{ass:obs-equal-weight}, \ref{ass:simple-factor-model} and \ref{ass:simple-moment} hold and that every entry is randomly missing with observed probability $p$. We switch the role of factors and loadings in the all-purpose estimator. As  $N,T \rightarrow \infty$, it holds that:
				\begin{enumerate}
					\item The estimated factors $\tilde{F}_t$ are asymptotically the same as the initial estimates of factors in \cite{jin2020factor}, but are asymptotically less efficient than the iterated estimates of factors  in \cite{jin2020factor}. 
					\item The estimated loadings $\tilde{\Lambda}_i$ are asymptotically more efficient than the initial estimates of loadings in \cite{jin2020factor}, but are asymptotically the same as the iterated estimates of loadings in \cite{jin2020factor}.
				\end{enumerate}
			\end{proposition}
		\end{remark}

	\section{Propensity Weighted Estimator}\label{sec:propensity-weighted-estimator}
	
	We provide the assumptions and general distribution theory for the propensity weighted estimator for the factors $\tilde{F}^S_t$ defined in Equation \eqref{eqn:reg-estimate-factors-conditional}. This conditional estimator uses the weights $\frac{1}{P(W_{it} = 1|S_i)}$ in the cross-sectional regression to obtain the factors. We allow the observation probability to depend on observed cross-sectional features $S = [S_i] \in \+R^{N \times K}$ that explain why certain units are more likely to be observed than other units. This conditional setup requires some modifications of the previous assumptions. In addition to Assumption \ref{ass:obs-equal-weight} we require the following assumption:
	
	\begin{assumpC}[Conditional Observational Pattern] \label{ass:obs}
		\texttt{}
		\begin{enumerate}[wide=0pt, widest=99,leftmargin=\parindent, labelsep=*]
			\item $W$ is independent of $\Lambda$ conditional on $S$. 
			\item For any $i$ and $ j$ satisfying $i \neq j$, and for any $t$ and $s$, $W_{it}$ is independent of $W_{js}$ conditional on $S_i$ and $S_j$ where $t$ and $s$ can be the same. The probability of $W_{it} = 1$ depends on $S_i$ and satisfies $P(W_{it}=1|S_i) \geq \underline{p} > 0$.  
		\end{enumerate}
	\end{assumpC}

	We assume $S$ contains all the information in $\Lambda$ that is predictive for the observation pattern. In other words, $W$ is independent of $\Lambda$ conditional on $S$, as stated in Assumption \ref{ass:obs}.1. This is closely related to the {unconfoundedness} assumption in causal inference. It also assumes that the conditional probability $P(W_{it}=1|S_i) $ is bounded away from 0, which implies that the number of observed cross-sectional and time-series entries is proportional to $N$ and $T$, respectively. This corresponds to the {overlap assumption} in causal inference.\footnote{See \citep{rosenbaum1983central} for the connection to unconfoundedness and the overlap assumption. We assume $P(W_{it} = 1|S)$ is bounded away from 0, such that $\frac{1}{P(W_{it}=1|S)}$ does not diverge, which is equivalent to the overlap assumption in causal inference.} Note that it is straightforward to include the covariates of ``neighbor units'' in $S_i$ to allow for network effects.

	We replace Assumptions \ref{ass:factor-model} and \ref{ass:mom-clt} by their conditional counterpart Assumptions \ref{ass:factor-model-conditional} and \ref{ass:mom-clt-conditional} which have a similar level of generality. These are required for the asymptotic normality of $\tilde{F}^S_t$ and $\tilde C^S_{it}$. As before, we collect the Assumptions \ref{ass:factor-model-conditional} and \ref{ass:mom-clt-conditional} for a general approximate factor model in the Appendix and present the assumptions for a simplified factor model in the main text, which are sufficient to convey all conceptual insights. 
	
	\begin{assumpC}[Conditional Factor Model]\label{ass:simple-factor-model-conditional}
		\texttt{}
		\begin{enumerate}[wide=0pt, widest=99,leftmargin=\parindent, labelsep=*]
			\item $S$ is independent of $F$ and $e$. 
			\item For any $i$, $\Lambda_i$ is independent of $S_j$ conditional on $S_i$ for $j \neq i$. Moreover,  for any $i$ and $ j$ satisfying $i \neq j$, $\Lambda_i$ is independent of $\Lambda_j$ conditional on $S_i$ and $S_j$.
		\end{enumerate}
	\end{assumpC}

	\setcounter{assumpC}{2}
	\begin{assumpC}[Moments of Conditional Factor Model]\label{ass:simple-moment-conditional}
		\texttt{}
		\begin{enumerate}[wide=0pt, widest=99,leftmargin=\parindent, labelsep=*]
			\item $\+E[\norm{\Lambda_i}^8|S]  \leq \overline{\Lambda} < \infty$. 
			\item Systematic loadings: $\lim_{N \rightarrow \infty} \frac{1}{N} \sum_{i = 1}^N \frac{1}{P(W_{it}=1|S_i) } \+E[\Lambda_i \Lambda_i^\T | S_i]  \xrightarrow{P} \Sigma_{\Lambda, S,t} $ for every $t$ for some positive definite matrix $\Sigma_{\Lambda,S,t}$.
		\end{enumerate}
	\end{assumpC}
	
	Under Assumption \ref{ass:simple-factor-model-conditional}, $\frac{1}{N} \sum_{i = 1}^N \frac{W_{it}}{P(W_{it} = 1|S_i)}  \tilde  \Lambda_i \tilde  \Lambda_i^\T $ converges in probability to an identity matrix which is the same limit as the loading estimates in conventional PCA without missing data. The assumption that $S$ is independent of $F$ and $e$ is conceptually similar to the assumption that $\Lambda$ is independent of $F$ and $e$, where the latter is standard in the literature on large dimensional factor modeling. The additional moment conditions in Assumption \ref{ass:simple-moment-conditional} are required for the asymptotic distribution, where $\Sigma_{\Lambda,S,t}$ appears in the asymptotic covariances of $\tilde{F}^S_t$ and $\tilde C^S_{it}$. 
	
	Proposition \ref{prop:simple-assump-imply-general-assump} in the Appendix shows that the simplified model is just a special case of the general approximate factor model specified by Assumptions \ref{ass:factor-model-conditional} and \ref{ass:mom-clt-conditional}. The simplified Assumption \ref{ass:simple-factor-model} combined with Assumptions \ref{ass:obs-equal-weight}, \ref{ass:obs} and \ref{ass:simple-factor-model} imply the general conditional Assumption \ref{ass:factor-model}, while Assumption \ref{ass:simple-moment} combined with the other simplified Assumptions \ref{ass:obs-equal-weight}, \ref{ass:obs}, \ref{ass:simple-factor-model}, \ref{ass:simple-moment}.2 and \ref{ass:simple-factor-model-conditional} imply the general Assumption \ref{ass:mom-clt-conditional}.

\subsection{Asymptotic Normality}
The propensity weighted estimator only differs in the distribution of the factors and common components. Both $\tilde F^S_t$ and $\tilde C^S_{it}$ follow a normal distribution, but in most cases have a larger asymptotic variance than the estimators $\tilde F_t$ and $\tilde C_{it}$. The loadings are not affected by the propensity score weighting.

\begin{theorem}\label{theorem:asy-normal}
	Under Assumptions \ref{ass:obs-equal-weight}, \ref{ass:obs}, \ref{ass:factor-model}, \ref{ass:factor-model-conditional} and \ref{ass:mom-clt-conditional} and for $N,T \rightarrow \infty$ we have for each $j$ and $t$:
	\begin{enumerate}[wide=0pt, widest=99,leftmargin=\parindent, labelsep=*]
		\item The asymptotic distribution of the loadings is the same as in Theorem \ref{theorem:asy-normal-equal-weight}.
		\item For ${\sqrt{N}}/{T} \rightarrow 0$ and $\sqrt{T}/N \rightarrow 0$, the asymptotic distribution of the factors is 
		\begin{align}
			\sqrt{\delta_{NT}} (\Sigma_{F,t}^S)^{-1/2} (H^\T \tilde F^S_t - F_t )  \xrightarrow{d} \calN(0,1),   \label{eqn:asy-var-F}
		\end{align}
		where 
		\[ \Sigma_{F,t}^S =  \Sigma_{\Lambda}^\I \Big[ \frac{\delta_{NT}}{N}  \covIS_{F,t}   + \frac{\delta_{NT}}{T} \covIIS_{F,t}  \Big]   \Sigma_{\Lambda}^\I,   \]
		with $\covIIS_{F,t} = g^S_{t}(F_t)$. $\covIS_{F,t}$ and $g^S_t(\cdot)$ are defined in Assumptions \ref{ass:mom-clt-conditional}.\ref{ass:asy-normal-main-term-thm-factor-conditional} and \ref{ass:mom-clt-conditional}.\ref{ass:asy-normal-add-term-thm-loading-conditional}. 
		\item The asymptotic distribution of the common components is
		\begin{align}
			\sqrt{\delta_{NT}} (\Sigma_{C,it}^S)^{-1/2}  (\tilde C^S_{jt} - C_{jt})  \xrightarrow{d} 
			\calN (0,1), \label{eqn:asy-var-C}
		\end{align}
		where 
		\begin{align*}
			\Sigma_{C,it}^S =  & \frac{\delta_{NT}}{T} F_t^\T \Sigma_F^\I \Sigma_\Lambda^\I ( \covI_{\Lambda,j}   +\covII_{\Lambda,j}  )  \Sigma_\Lambda^\I \Sigma_F^\I F_t +  \frac{\delta_{NT}}{T}  \Lambda_j^\T \Sigma_\Lambda^\I \covIIS_{F,t} \Sigma_\Lambda^\I \Lambda_j 
			\\ & + \frac{\delta_{NT}}{N} \Lambda_j^\T \Sigma_\Lambda^\I  \Gamma^{\obs, S}_{F,t}   \Sigma_\Lambda^\I \Lambda_j -   2 \cdot \frac{\delta_{NT}}{T}  \Lambda_j^{\top} \Sigma_{\Lambda}^\I \covIIIS_{\Lambda, F, j, t}  \Sigma_\Lambda^\I   \Sigma_F^\I  F_t ,
		\end{align*}
		with $\covIIIS_{\Lambda, F, j, t} = g^{\cov,S}_{j, t}(\Lambda_j, F_t)$, and the function $g^{\cov,S}_{j, t}(\cdot,\cdot)$ is defined in Assumption \ref{ass:mom-clt-conditional}.\ref{ass:asy-normal-add-term-thm-loading-conditional}. 
	\end{enumerate}
\end{theorem}

The distribution results have the same general structure as in Theorem \ref{theorem:asy-normal-equal-weight}. However, there are two key differences. First, the outer matrices in the variance of $\tilde F^S_t$ are $\Sigma_{\Lambda}^{-1}$ while they depend on the observational pattern in $\Sigma_{\Lambda,t}^\I$ in Equation \eqref{eqn:asy-var-F-equal-weight}. Second, the middle terms $\covIS_{F,t}$ and $\covIIS_{F,t}$ may depend on $\ps$. The same structure carries over to the common component. In the case of generalized least squares regressions, it is straightforward to compare the asymptotic covariances for different weights and to determine an efficient estimator. With missing observations, the problem becomes more challenging as the asymptotic covariances depend on two matrices for the factor estimates and three terms for the common components. For the general models in Theorems \ref{theorem:asy-normal-equal-weight} and \ref{theorem:asy-normal} we cannot state which estimator is more efficient without imposing additional structure. However, for the simplified model, we can rank the efficiency of the two estimators.


\begin{corollary}\label{corollary:asy-normal}
	Suppose Assumptions \ref{ass:obs-equal-weight}, \ref{ass:obs}, \ref{ass:simple-factor-model}, \ref{ass:simple-moment}.2, \ref{ass:simple-factor-model-conditional}, and \ref{ass:simple-moment-conditional} hold and $N,T \rightarrow \infty$. Then Theorem \ref{theorem:asy-normal} holds. If in addition, $q_{ij}$ and $q_{ij,kl}$ are independent of $\Lambda_m \Lambda_m^\T$ for all $i, j, k, l, m$, then the asymptotic variances simplify as follows with the weights $\omega, \omega_{j}$ and $\omega_{jj}$ defined in Assumption \ref{ass:simple-moment}:
	\begin{enumerate}[wide=0pt, widest=99,leftmargin=\parindent, labelsep=*]
		\item 
		The asymptotic variance of the factors in formula \eqref{eqn:asy-var-F} simplifies to 
		\[\Sigma_{F,t}^S = \frac{\delta_{NT}}{N}  \scovIS_{F} + \frac{\delta_{NT}}{T} ( \omega - 1) \scovIIS_{F,t}, \]
		where
		\begin{align*}
			\scovIS_{F,t}=&   \Sigma_\Lambda^\I  \Sigma_{\Lambda,S,t}  \Sigma_\Lambda^\I  \sigma_e^2, \\
			\scovIIS_{F,t}  =& \Sigma_\Lambda^\I \big( I_r \otimes (F_t^{\top} \Sigma_F^{-1} \Sigma_{\Lambda}^{-1}) \big)  \left( \Sigma_{\Lambda} \otimes \Sigma_{\Lambda} \right) \Xi_F \left( \Sigma_{\Lambda} \otimes \Sigma_{\Lambda} \right) \big( I_r \otimes (\Sigma_{\Lambda}^{-1} \Sigma_F^{-1} F_t )\big)  \Sigma_\Lambda^\I.
		\end{align*}
		\item 
		The asymptotic variance of the common component in formula \ref{eqn:asy-var-C} simplifies to 
		\begin{align*}
			\Sigma_{C,it}^S =&	\frac{\delta_{NT}}{T} \bigg[  F_t^{\top}  \Big( \omega_{jj}  \scovI_{\Lambda} + (\omega_{jj}-1)  \scovII_{\Lambda,j} \Big) F_t + (\omega-1)   \Lambda_j^{\top} \scovIIS_{F,t} \Lambda_j   \\
			& \quad \quad  - 2 (\omega_j- 1) \Lambda_j^{\top} \scovIIIS_{\Lambda,F,j,t} F_t \bigg] + \frac{\delta_{NT}}{N}  \Lambda_j^{\top}  \scovIS_{F,t} \Lambda_j, 
		\end{align*} 
		where
		\[\scovIIIS_{\Lambda,F,j,t}  = \Sigma_{\Lambda}^\I \big(I_r \otimes (    F_t^\T \Sigma_F^\I  \Sigma_\Lambda^\I ) \big)    ( \Sigma_{\Lambda}  \otimes  \Sigma_\Lambda )         \Xi_F  ( \Lambda_j \otimes  \Sigma_\Lambda)   \Sigma_\Lambda^\I \Sigma_F^{-1}.  \]
		\item $\tilde F_t^S$ is weakly less efficient than $\tilde F_t$, if $S$ is independent of $\Lambda$. In the case of only one factor, i.e. $r=1$, $\tilde F_t^S$ is weakly less efficient than $\tilde F_t$ for any $S$.
	\end{enumerate}
\end{corollary} 
An interesting observation is that $\scovIIS_{F,t}$  and $\scovIIIS_{\Lambda,F,j,t}$ depend neither on the observation pattern nor on $S$. This is because $\frac{1}{\ps}$ removes the asymptotic dependency between $W_{it}$ and $\Lambda_i$. Hence, this part of the asymptotic distribution has a complete separation between the missing observation pattern captured by the weights $\omega, \omega_{j}$ and $\omega_{jj}$ and distribution terms that depend only on the factor model. However, $\scovIS_{F,t}$ depends on $\ps$ as this component comes from a probability weighted least square regression of the population loadings on the observed entries in $Y$, which is different from the corresponding OLS regression in Corollary \ref{corollary:asy-normal-equal-weight}.1.

The key observation is that $\tilde F_t^S$ and as a consequence also $\tilde C_{it}^S$ seem to be in many cases less efficient than $\tilde F_t$ and $\tilde C_{it}$, which means that the asymptotic variances of the all-purpose estimator are less than or equal to those of the propensity weighted estimator. In the case of only one factor it holds that $ \scovIIS_{F,t} = \scovII_{F,t} $. Not surprisingly, it holds that $\scovIS_{F,t} \geq \scovI_{F,t}$ as in the case of i.i.d. errors, an OLS regression is the most efficient linear estimator. This result can also be derived from $\Sigma_{\Lambda,S,t}/\Sigma_{\Lambda}^2 \geq 1/\Sigma_{\Lambda,t}$, which follows from the Cauchy-Schwartz inequality. In the case of multiple factors, we take advantage of the concavity of the average weighted by $1/\ps$ to prove the efficiency relationship. In simulations we confirm that when the loadings depend on $S$, it is possible that $ \scovIIS_{F,t} < \scovII_{F,t} $, which can result in minor efficiency gains for $\tilde F_t^S$. For a general residual covariance matrix, the efficiency results are more complex. \cite{pelgerxiongsparse2020} show that for fully observed data under certain assumptions, the optimal weight in the factor regression is the inverse residual covariance matrix or equivalently PCA, applied to a covariance matrix re-weighted by the square-root of the inverse residual covariance matrix, is the most efficient estimator. Hence, if the propensity weight is close to the inverse residual covariance matrix, it lowers the first term $\scovIS_{F,t}$. However, the effect on the second term is more complex, and hence there are in general cases where $\tilde F_t^S$ can be more efficient than $\tilde F_t$. In simulations, we show that for a correctly specified model, the estimates of $\tilde F_t$ and $\tilde F_t^S$ are close, but $\tilde F_t$ is generally more precise. However, the ``doubly-robust'' estimator $\tilde F_t^S$ seems to be less affected by various forms of misspecification, e.g., omitted factors, weak factors, or a nonlinear factor model. Hence, $\tilde F_t^S$ might be appealing because it is more robust but not based on efficiency arguments.

\subsection{Robustness to Model Misspecification}
The propensity weighted regressions can be robust to the selection bias from omitting factors. In the causal inference literature regressions weighted by propensity scores have been used in the estimation of causal effects to reduce the bias that arises from omitting regressors or misspecifying the outcome model. However, as the propensity weighted regressions have a larger variance, regressions without the propensity weights seem to be preferred for correctly specified models.\footnote{\cite{robins1994estimation,robins1995semiparametric} among others discuss the reduction of bias from omitting regressors. \cite{robins2000inference,kang2007,robins2007comment} show the large variance of propensity weighted regressions, and \cite{freedman2008weighting} suggests unweighted regressions for correctly specified models} In this section, we illustrate that this logic carries over to our latent factor model setup. 

Our setup differs from classical causal inference as we estimate the covariates as latent factors from the data. However, we can have a situation similar to omitted variables if we estimate too few latent factors, the factors are weak, or the population model is nonlinear. In the previous section we have shown that the propensity weighted estimator is in general less efficient than our all-purpose estimator when we use the correct number of factors $r$.\footnote{Given our distribution theory, it is relatively straightforward to extend the consistent estimator of \cite{bai2002determining} for the number of factors to the more general case with missing data. However, most existing estimators for the number of latent factors explicitly or implicitly depend on choice parameters, which implies that in practice it is possible to use too few factors \citep{pelger2019large,lettau2018factors}.} However, when we use too few latent factors, our weighted estimator can have a smaller selection bias, and hence be preferable. We illustrate the general logic of this result with an example and confirm it with extensive simulations in Section \ref{sec:robust}. 


We assume a two-factor model $Y_{it} = \Lambda_{i1} F_{t1}   +\Lambda_{i2}  F_{t2}  + e_{it}$, where 
\begin{align*}\begin{bmatrix}
		F_{t1} \\ F_{t2}
	\end{bmatrix} \stackrel{\iid}{\sim} \mathcal{N} \bigg( \begin{bmatrix}
		0 \\ 0
	\end{bmatrix}, \begin{bmatrix}
		\sigma_F^2 & 0 \\ 0 &  \sigma_F^2
	\end{bmatrix} \bigg), \qquad \begin{bmatrix}
		\Lambda_{i1} \\ \Lambda_{i2}
	\end{bmatrix}\stackrel{\iid}{\sim} \mathcal{N} \bigg( \begin{bmatrix}
		\mu_\Lambda \\ 0
	\end{bmatrix}, \begin{bmatrix}
		\sigma_{\Lambda}^2 & 0 \\ 0 &  \sigma_{\Lambda}^2
	\end{bmatrix} \bigg), \qquad e_{it}  \stackrel{\iid}{\sim} \mathcal{N}(0,\sigma_e^2).
\end{align*}
The key assumption is that the observation pattern depends on the loadings. In order to have a transparent example, we assume the following pattern. The first $T_0$ time periods are fully observed. After time $T_0$, whether a unit is observed or not depends on an indicator variable $S_i $, defined as $S_i = \boldsymbol{1}_{\Lambda_{i1}  > c_1,  \Lambda_{i2} > c_2}$ for some $c_1 > 0$ and $c_2 > 0$. Suppose $P(W_{it} = 1|S_i = 1) = p$ and  $P(W_{it} = 1|S_i = 0) = 1 - p$ for some $p > 0.5$. In other words, units with large loadings are more likely to be observed. We would get similar results if large loadings are more likely to be missing. Without loss of generality, we can shuffle the units and obtain the observation pattern in Table \ref{tab:toy-example-obs-pattern}. In the following, we estimate the factor model from the shuffled outcome matrix, where the outcomes for the first $N_0$ units after time $T_0$ are missing, and the outcomes for the last $N-N_0$ units are fully observed.

Assume that we omit one factor and estimate only a one-factor model with both, the simple and propensity weighted, estimators. In this case, our factor model is misspecified. Without loss of generality, we can set $\mu_\Lambda^2 + \sigma_{\Lambda}^2 = 1$.  The estimated loading vector is consistent and the same for both estimators:
\[\tilde{\*\Lambda}_1 = \*\Lambda_{1}  + o_p(1).\]

We compare for both approaches the estimates of the first factor and the common components from time $T_0+1$ to $T$ (the two approaches coincide from time 1 to $T_0$ as all units are fully observed). For the simple regression and for $T_0 < t < T$, we can show that
\begin{align*}
	\tilde{F}_{t1} =&  \bigg( \sum_{i = N_0+1}^N Y_{it} \tilde\Lambda_{i1}  \bigg) \bigg( \sum_{i = N_0+1}^N  \tilde\Lambda^2_{i1}  \bigg)^\I 
	=F_{t1} +  \gamma F_{t2} + o_p(1),
\end{align*}
where $\gamma = \lim_{N_0, N \rightarrow \infty} \frac{\sum_{i = N_0+1}^N \Lambda_{i1}  \Lambda_{i2} }{\sum_{i = N_0+1}^N \Lambda_{i1}^2 }$. The key element is that $\gamma \neq 0$ because of the dependence of the missing pattern on $\Lambda$. In our example, both $\Lambda_{i1} $ and $ \Lambda_{i2}$ tend to have large values on the observed units, and hence are not asymptotically orthogonal on the subset of observed data. A non-zero $\gamma$ creates a selection bias similar to the conventional omitted variable bias.\footnote{\cite{simon1954spurious} refers to this type of selection bias as a spurious correlation.}   

Since  the common component is $C_{it} = \Lambda_{i1} F_{t1}   +\Lambda_{i2}  F_{t2}$, the estimated common components have an error of $\tilde{C}_{it} - C_{it}  =( \gamma \Lambda_{i1}  +\Lambda_{i2}  ) F_{t2}   + o_p(1)$. Hence, the mean squared error from $T_0+1$ to $T$ equals  
\begin{align*}
	\frac{1}{N(T- T_0)} \sum_{i=1}^{N} \sum_{t = T_0 + 1}^T (\tilde{C}_{it} - C_{it}  )^2  =& \frac{1}{N(T- T_0)} \sum_{i=1}^{N} \sum_{t = T_0 + 1}^T \bigg( \gamma \Lambda_{i1}+ \Lambda_{i2} \bigg)^2  F_{t2}^2 + o_p(1) \\ =& (\gamma^2 +\sigma_{\Lambda}^2) \sigma_F^2 + o_p(1)
\end{align*}
as $\frac{1}{N} \sum_{i=1}^N \Lambda_{i1} \Lambda_{i2} \rightarrow 0$. In contrast, the propensity weighted regression for $T_0 < t < T$ equals
\begin{align*}
	\tilde{F}^S_{t1} = & \bigg( \sum_{i = N_0+1}^N \frac{1}{P(W_{it} = 1|S_i)} Y_{it} \tilde\Lambda_{i1}  \bigg) \bigg( \sum_{i = N_0+1}^N \frac{1}{P(W_{it} = 1|S_i)} \tilde\Lambda^2_{i1}  \bigg)^\I 
	=   F_{t1}  + o_p(1).
\end{align*}
Thus, the propensity weighted regression corrects for the selection bias and is the same as if $\gamma$ would have been zero. The estimated common components have an error of $\tilde{C}^S_{it} - C_{it} = \Lambda_{i2}  F_{t2}  + o_p(1)$, which implies that the mean squared error from time $T_0+1$ to $T$ is 
\begin{align*}
	\frac{1}{N(T- T_0)} \sum_{i=1}^{N} \sum_{t = T_0 + 1}^T (\tilde{C}^S_{it} - {C}_{it})^2  
	=  \frac{1}{N(T- T_0)} \sum_{i=1}^{N} \sum_{t = T_0 + 1}^T \Lambda_{i2}^2  F_{t2}^2 + o_p(1) = \sigma_{\Lambda}^2 \sigma_F^2 + o_p(1).
\end{align*}

Importantly, $\tilde{C}^S_{it} $ has a smaller mean squared error than $\tilde{C}_{it} $. In summary, when the estimated factor model has omitted factor(s) and is misspecified, the propensity weighted estimator could adjust for the selection bias and reduce the estimation error.

\section{Feasible Estimator of the Probability Weighting}\label{sec:feasible-est-prop-score}

We provide feasible estimators for $P(W_{it}=1|S_i)$ which we need in Equation \eqref{eqn:reg-estimate-factors-conditional} to estimate the factors, and we show that the asymptotic distribution of factors is not affected by using the estimated weights instead of their population counterpart. While in (stratified) randomized experiments, researchers decide and therefore know the treatment assignment probability given covariates, $P(W_{it}=1|S_i)$, the probability distribution of the missing pattern in observational studies generally needs to be estimated, which can affect the distribution theory for the latent factor model. Here we provide conditions under which the previously derived results continue to hold with a feasible estimator. To simplify notation denote by $p_{it}=P(W_{it}=1|S_i)$ the propensity score and its estimate by $\hat p_{it}=\wh P(W_{it}=1|S_i)$. The feasible estimator for the factors $\hat F_t^S$ replaces $p_{it}$ by $\hat p_{it}$ in Equation \eqref{eqn:reg-estimate-factors-conditional}, which yields the following decomposition:
\begin{align*}
	\hat F_t^S =& \left(\sum_{i = 1}^N \frac{W_{it}}{\hat p_{it}}  \tilde  \Lambda_i \tilde  \Lambda_i^\T  \right)^\I  \left( \sum_{i = 1}^N \frac{W_{it} }{\hat p_{it}} Y_{it} \tilde{\Lambda}_i \right) =  \tilde F_t^S +   \left(\frac{1}{N} \sum_{i = 1}^N \frac{W_{it}}{ p_{it}}  \tilde  \Lambda_i \tilde  \Lambda_i^\T  \right)^\I \left(\frac{1}{N} \sum_{i = 1}^N  \frac{p_{it}-\hat p_{it}}{\hat p_{it}  }  \frac{W_{it} }{p_{it}} Y_{it} \tilde{\Lambda}_i \right) \\
	&+\left(\frac{1}{N} \sum_{i = 1}^N \frac{W_{it}}{ p_{it}}  \tilde  \Lambda_i \tilde  \Lambda_i^\T  \right)^\I \left(\frac{1}{N} \sum_{i = 1}^N \frac{\hat p_{it} - p_{it}}{\hat p_{it}} \frac{W_{it}}{ p_{it}}  \tilde  \Lambda_i \tilde  \Lambda_i^\T  \right) \cdot \hat F_t^S.
\end{align*}

Under weak assumptions on $\hat p_{it}$, that are satisfied for feasible estimators of the most empirically relevant observation patterns, the additional term $\hat F^S_t-\tilde F_t^S$ can be neglected in the asymptotic distribution.
\begin{theorem}\label{thm:feasible-estimator}
	We replace the propensity score in $p_{it}$  in Equation \eqref{eqn:reg-estimate-factors-conditional} by its estimate $\hat p_{it}$.
	\begin{enumerate}[wide=0pt, widest=99,leftmargin=\parindent, labelsep=*]
		\item The estimates of the loadings do not depend on the propensity score. Hence, Theorem \ref{thm:consistency-same-H}  and the asymptotic distribution of the loadings in Theorem \ref{theorem:asy-normal} continue to hold independently of $\hat p_{it}$.
		\item The following holds for the distribution of the factors and common components. 
		\begin{enumerate}
			\item If $\max_i |\hat p_{it} - p_{it} | = o_P (1)$, then the factors and common components are estimated consistently pointwise under the assumptions of Theorem \ref{theorem:asy-normal}.
			\item If $\max_i |\hat p_{it} - p_{it} | = o_P \left( \frac{1}{N^{1/4}} \right)$, then Theorem \ref{theorem:asy-normal} continues to hold as it is. 
		\end{enumerate}
	\end{enumerate}
\end{theorem}

\begin{table}[t!]
	\tcaptab{Examples of feasible estimators of the probability weight}\label{tab:examples}
	\footnotesize{
		\setlength{\tabcolsep}{4pt} 
		\renewcommand{\arraystretch}{0.5} 
		\begin{tabular}{@{}lllll@{}}
			\toprule
			Description & $\ps$ & Estimator & Asymptotic distribution & Effect on  \\ 
			&  &  & of $\hat p_{it}- p_{it}$ & distribution\\ \midrule
			Time-series missing at random& $p(S_i)$ & logit on full panel $W$  &  $ O_P \left( \frac{1}{\sqrt{NT}} \right)$ &  no \\
			(parametric)&  &  &  &   \\
			Time-series missing at random & $p(S_i)$ & kernel on full panel $W$  &  $O_P \left( \frac{1}{\sqrt{NTh}} \right)$ &  no \\ 
			(non-parametric)&  &  &  &   \\
			Cross-section and time-series & $p_t(S_i)$ & logit on $W_t$    & $O_P \Lp \frac{1}{\sqrt{N}} \Rp$   & no \\
			dependency (parametric) &  &  &  &   \\
			Cross-section and time-series & $p_t(s)$  & $\hat p_t(s) = \frac{|\tlo_{s,t}|}{N_s}$     & $\frac{1}{\sqrt{N_s}} \calN \left (0, p_t(s) (1- p_t(s)) \right)$   & no  \\
			dependency (discrete $S$) &  &  &  &   \\
			Staggered treatment adoption & $p_t(S_i)$ & hazard rate model    & $O_P \left(\frac{1}{\sqrt{N}} \right)$   & no  \\
			with $S$ (parametric) &  &  &  &   \\
			\\ \bottomrule
		\end{tabular}
	}
	\bnotetab{This table shows feasible estimators for the probability of the most important cases of missing patterns. We propose examples of feasible estimators. The asymptotic distribution for $\hat p_{it}- p_{it}$ is given under suitable assumptions and for $S_i$ being i.i.d and sub-Gaussian. The main text includes additional details. The effect on distribution refers to the asymptotic distribution of the factors and common components in Theorem \ref{theorem:asy-normal}. The exact details are described in Theorem \ref{thm:feasible-estimator}.}
\end{table}

We discuss feasible estimators for the most important cases of missing patterns which are summarized in Table \ref{tab:examples}. Obviously, we only need to consider the case where $p_{it}$ varies for different cross-sectional units as otherwise the estimator simplifies to our estimator in Equation \eqref{eqn:pca-estimate-loadings}. For simplicity these examples assume that $S_i$ are $i.i.d.$  and sub-Gaussian but can be generalized to weak dependency patterns. The simplest case is missing at random only in the time-series dimension, that is $\ps = p(S_i)$ for some parametric or non-parametric function $p(.)$. A relevant example is the estimation of $p(S_i)$ with a logit model on the full panel $W$ which has the convergence rate $\hat p(S_i)=p(S_i)+ O_P \left( \frac{1}{\sqrt{NT}} \right)$ and a uniform bound of order $\frac{\log(NT)}{\sqrt{NT}}$. Hence, Theorem \ref{thm:feasible-estimator}.2(b) applies. If $p(S_i)$ is estimated non-parametrically with a kernel with bandwidth $h$, the convergence rate is typically $\sqrt{NT h}$ with a uniform bound of order $\frac{\log(NTh)}{\sqrt{NTh}}$, which does not change the distribution results if $Th$ is sufficiently large. In the more complex model $\ps=p_t(S_i)$ the observations probability depends on the  cross-section and time-series information. A relevant example for a parametric model is a logit model estimated on $W_t$ for each $t$ separately with a convergence rate of $\sqrt{N}$. Under weak assumptions on $S_i$, the uniform convergence bound in Theorem \ref{thm:feasible-estimator}.2(b) holds. 

An important special case are discrete values for $S$, that is, the covariates $S$ take only finitely many values. An example for a binary variable $S$ would be gender, when male or female individuals have different probabilities to be treated. If the probabilities for the different discrete outcomes of $S$ are bounded away from zero, then the estimator $P(W_{it}=1|S_i=s)=p_t(s)$ simplifies to $p_t$, but just averaged over the cross-section units for which $S_i=s$. In more detail, consider the estimator $\hat p_t(s) = \frac{|\tlo_{s,t}|}{N_s}$ where $N_{s} = \sum_{i = 1}^N \mathbf{1}(S = s)$  and $\tlo_{s,t} = \{i: W_{it} = 1 \text{ and } S = s\}$. Then, $\sqrt{N_s} \left( \hat p_t(s) - p_t(s) \right)  \xrightarrow{d} \calN \left (0, p_t(s) (1- p_t(s)) \right)$. If $N_s$ is sufficiently large, for example proportional to $N$, then the feasible estimator does not change the distribution results in Theorem \ref{theorem:asy-normal}. These estimators directly carry over to staggered treatment adoption. The staggered design can also be modeled with a parametric hazard model $\ps=p(t,S_i)$, which under appropriate assumptions converges at the rate $\sqrt{N}$ as well. In summary, for all these cases the distribution results are not affected by using a feasible estimator for the propensity score. 

As previously mentioned, we allow $S_i=\Lambda_i$. This is appealing as $\Lambda$ is by construction capturing the unit-specific features and hence should account for the differences in cross-sectional observation patterns. As the estimator $\hat \Lambda$ does not depend on the probability weights, it can be used in the estimation of $P(W_{it}|\Lambda_i)$. Theorem \ref{theorem:asy-normal} states that the estimation error of $\hat \Lambda_i$ is of the order $O_P \Lp \frac{1}{\sqrt{N}} \Rp$. While the consistency results for the factors and common components continue to hold, we need some additional weak assumptions on the tail behavior of the loadings and error terms to satisfy the uniform condition in Theorem \ref{thm:feasible-estimator}.2(b).

\section{Tests of Treatment Effects}\label{sec:test}

The key application of our the asymptotic distribution theory is to test causal effects. The fundamental problem in causal inference is that we observe an outcome either for the control or the treated data, but not for both at the same time. The unknown counterfactual of what the treated observations could have been without treatment can be naturally modeled as a data imputation problem. In this section, we consider the case where once a unit adopts the treatment, it stays treated afterward, for example, the simultaneous and staggered treatment designs illustrated in Figure \ref{fig:obs-pattern-illustrate}. Given the general missing patterns that we allow for, the generalization to more complex adoption patterns is straightforward. We denote by $\Tcontrol$ and $\Ttreat$ the number of control and treated time periods for unit $i$ where their sum adds up to $\Tcontrol+\Ttreat=T$. The superscripts $(0)$ and $(1)$ denote the observations for control and treated observations.

The individual treatment effect measures the difference between the treated and control outcomes:
\begin{align*}
	\tau_{it} = Y_{it}^{(1)} - Y_{it}^{(0)} \qquad \text{for $t > \Tcontrol$, $i=1,...,N$,}
\end{align*} 
where by construction for a specific time $t$ and unit $i$ we only observe either $Y_{it}^{(1)}$ or $Y_{it}^{(0)}$, but not both. Average treatment effects can be estimated by an average over time or the cross-section of the individual treatment effects. We assume that the data has a factor structure which results in a model of the form
\begin{align}
	Y_{it} = \tau_{it} D_{it} + \Lambda_i^{\top} F_t + e_{it}, \label{eqn:treatment}
\end{align}
where $D_{it}=1$ is a treatment indicator. Note that this model is very general and captures many relevant models as special cases. The factor structure includes interactive fixed effects as in Bai (2009). Simple time- and cross-sectional fixed effects are a special case for constant loadings respectively factors. The factors can be either observed covariates or latent factors. One of the main challenges in estimating a treatment effect is to control for all relevant covariates. Failure in doing so results in an omitted variable bias in the treatment effect estimation as discussed among others in \cite{gobillon2016regional}. The strength of our latent factor model is that we can avoid this problem by automatically including all relevant covariates in a data driven way. Note that our latent factor model can also account for some uncertainty in the functional form of the dependency on the factors. For example, if $Y_{it}$ is a polynomial function of a factor, this could be captured by including additional latent factors as described for example in \cite{pelger2018state}. A generalization of Equation \eqref{eqn:treatment} adds additional observed covariates $X_{it} \in \mathbbm R^{K}$ to $Y_{it}$ which yields $Y_{it} = \tau_{it} D_{it} + \Lambda_i^{\top} F_t + X_{it}^{\top} b + e_{it}$. If these observed covariates follow a factor structure $X_{it} = {\Lambda^{X}_i}^{\top} F^X_t + e^X_{it}$, it puts us back into the framework of Equation \eqref{eqn:treatment}. Otherwise it is straightforward to include general observable covariates $X_{it}$ by studying the residual $Y_{it} - X_{it} \hat b$, where $\hat b$ is estimated by a regression on the control group.\footnote{Using the residuals  $Y_{it} - X_{it} \hat b$ for the factor analysis and treatment effect analysis with our method generally adds another covariance term to the asymptotic covariance matrix. This term comes from the regression to obtain $\hat b$ and is straightforward to include. Here we focus on the conceptually more challenging problem of dealing with the unobserved factors.}

We only observe $Y_{it}^{(1)}$ for the treated group and could obtain the counterfactual outcome $Y_{it}^{(0)}$ from the imputed value $\hat Y_{it}^{(0)} =\hat C_{it}^{(0)}$, where $\hat C_{it}^{(0)}$ is the common component estimated only from the untreated control data. This is the same setup as in \cite{bai2019matrix}. Given our asymptotic distribution theory for the common component, we can provide the asymptotic distribution of the individual and average treatment effects analogously to \cite{bai2019matrix}. A shortcoming of estimating the individual treatment effect by $Y_{it}^{(1)} - \hat{Y}_{it}^{(0)}$ is that the observed treated observations $Y_{it}^{(1)}$ contain an idiosyncratic error $e_{it}$. Hence, it is not possible to test for individual treatment effects without imposing very strong additional assumptions on the error. For sufficiently large $T-T_0$, this error component can be averaged out in the average treatment effect.

We impose slightly stronger assumptions on the structure of the treatment effect which will allow us to derive substantially stronger results. Assume that the treatment effect has also a factor structure, that is $\tau_{it}=\left( {\Lambda_i^{\tau}} \right)^{\top} F_t^{\tau}$. In this case we can represent the problem as
\begin{align}
	Y_{it}^{(1)} =\left(  \Lambda_{i}^{(1)}\right)^{\top} F_t^{(1)} + e_{it} \qquad Y_{it}^{(0)} =\left(  \Lambda_{i}^{(0)} \right)^{\top}F_t^{(0)} + e_{it}, \label{eqn:treatmentfactor}
\end{align}
where the factor structure subsumes the treatment effect. Hence, the individual treatment effect is equivalent to the difference in the common components between the treated and control:
\begin{align*}
	\tau_{it}= Y_{it}^{(1)} - Y_{it}^{(0)}=C_{it}^{(1)} -C_{it}^{(0)} \qquad \text{for $t > \Tcontrol$, $i=1,...,N$.}
\end{align*} 
Fundamentally, we are testing if the treatment changes the underlying factor structure. Hence, we can test if the treatment changes interactive fixed effects. This is a very general setup that allows for time and cross-sectional heterogeneity in the treatment effect, while the treatment itself can depend on the latent cross-sectional covariates modeled by the loadings.

In the following we consider three different treatment effects:
\begin{enumerate}[wide=0pt, widest=99,leftmargin=\parindent, labelsep=*]
	\item Individual treatment effect: $\tau_{it} =  C_{it}^\treat - C_{it}^\control$
	\item Average treatment effect over time: $ \tau_i=\frac{1}{\Ttreat} \sum_{t=\Tcontrol+1}^{T} \tau_{it}$
	\item Weighted average treatment effect: $\tau_{\beta,i}= \beta_i^\treat-\beta_{i}^\control$ where $\beta_i$ are the regression coefficients on some covariates $Z$:
	\[\beta_i^\control = (Z^\T Z)^\I Z^\T C^\control_{i,(\Tcontrol+1):T} \quad \text{ and } \quad \beta_i^\treat = (Z^\T Z)^\I Z^\T C^\treat_{i,(\Tcontrol+1):T}.  \]
\end{enumerate}
Here, $C^\control_{i,(\Tcontrol+1):T}  = \begin{bmatrix}C_{i, \Tcontrol+1}^\control & \cdots & C_{iT}^\control \end{bmatrix}^\T \in \+R^{\Ttreat}$ denotes the observations for $t >\Tcontrol$. The weighted average treatment effect $\tau_{\beta,i}$ generalizes the average treatment effect $\tau_i$, which is a special case for $Z = \vec{1}$. Both tests for the individual treatment effect and the weighted average treatment effect cannot be obtained with conventional estimators, but are important to answer economic questions. For example, in our companion paper \cite{pelgerxiongpublication2020}, we test if pricing anomalies of investment strategies as measured by their pricing errors persist after these strategies have been published in academic journals. In this problem the treatment is the publication of an investment strategy in a journal and the treatment effect is measured by a change in regression coefficients. More specifically, the pricing error corresponds to the intercept in a regression of the excess returns of the strategies on a set of benchmark risk factors. A simple average treatment effect would not be sufficient to study this question.

For each of the three treatment effects we derive the asymptotic distribution under the null-hypothesis of no effect, which allows us to run one-sided or two-sided hypothesis tests. For example, the two-sided hypothesis test for the weighted average treatment effect takes the form
\begin{eqnarray}\label{hypo:average}
	\mathcal{H}_0:  \tau_{\beta,i} = 0,  \qquad  \mathcal{H}_1:  \tau_{\beta,i} \neq 0. 
\end{eqnarray}
This is the hypothesis we test in our simulation and the empirical companion paper. The problem formulated in Equation \eqref{eqn:treatmentfactor} can be solved by applying our latent factor model estimation twice: First, we estimate $C_{it}^{(1)}$ from the treated data with the control observations as missing values. Second, we estimate $C_{it}^{(0)} $ from the control data, while the treated observations are viewed as missing. The inferential theory follows readily from Theorems \ref{theorem:asy-normal-equal-weight} and \ref{theorem:asy-normal}. The asymptotic variance for the individual treatment effect $\tau_{it}$ is the sum of the asymptotic variances of $\hat C_{it}^{(0)} $ and $\hat C_{it}^{(1)} $ and a covariance term based on the correction terms for the control and treated. While the calculations are tedious, they are a direct consequence of the distribution results that we have derived. The average treatment effects follow then from the results of the individual treatment effects. In this section, we want to focus on a special case, which we consider the most relevant from a practical perspective.

In most causal inference applications, such as the empirical study in our companion paper and \cite{abadie2010synthetic,abadie2015comparative}, the majority of observations are control observations. Hence, it might be infeasible to estimate a latent factor model only from the treated data as required in Equation \eqref{eqn:treatmentfactor}. For example in the simultaneous treatment case in Table \ref{tab:toyexample}, we can estimate a latent factor for the control, but not for the treated. Hence, we impose the additional assumption that the control and treated panel share the same underlying factors, while the loadings can be different, that is,
\begin{align}
	Y^\control_{it} = (\Lambda_i^\control)^\T F_t + e_{it},  \qquad  Y^\treat_{it} = (\Lambda_i^\treat)^\T F_t + e_{it} . \label{eqn:treatmentloading}
\end{align}
This implies that the treatment can only affect the loadings. This is still a very general setup as the loadings and factors are latent. For example, a model based on Equation \eqref{eqn:treatmentfactor} with one factor that changes on the treated data, can be captured in Equation \eqref{eqn:treatmentloading} by a two-factor model where the corresponding loadings change on the treated data. 

First, we estimate the factor model from the incomplete control panel $Y^\control$ and obtain $\tilde C^{(0)}=( \tilde \Lambda^{(0)})^{\top} \tilde F$. Second, we use an ordinary least squares regression to estimate the loadings for the treated $\tilde \Lambda_i^\treat$, 
\begin{align}\label{eqn:reg-est-loading}
	\textstyle \tilde{\Lambda}_i^\treat = \Lp \sum_{t = \Tcontrol+1}^{T} \tilde F_t \tilde F_t^\T \Rp^\I \sum_{t = \Tcontrol+1}^{T} \tilde F_t Y^\treat_{it}, 
\end{align}
which yields an estimate for the common components for the treated panel $\tilde{C}^\treat_{it} = (\tilde{\Lambda}_i^\treat)^\T \tilde{F}_t.$\footnote{If units switch between treatment and control, we can modify Equation (\ref{eqn:reg-est-loading}) to $\tilde \Lambda_i^{\treat} = \Lp \sum_{t \in \tls_i} \tilde F_t \tilde F_t^\T \Rp^\I \sum_{t \in  \tls_i} \tilde F_t Y^{\treat}_{it}$, where $\tls_i$ it the set of indices for the treated observations. } 

The following theorem shows the asymptotic distributions for $\tilde{C}^\treat_{it}$, the individual treatment effect, and the weighted average treatment effect. The asymptotic distributions allow us to construct test statistics for various treatment effects.

	\begin{theorem}\label{theorem:ate-same-factor}
		Suppose Assumptions \ref{ass:obs-equal-weight}, \ref{ass:factor-model}, \ref{ass:mom-clt} and \ref{ass:add-factor} hold and the control and treated panel share the same factors. For $\delta_{NT_i}= \min(N, \Ttreat) $, as  $\delta_{NT_i} \rightarrow \infty$ the following holds:
		\begin{enumerate}[wide=0pt, widest=99,leftmargin=\parindent, labelsep=*]
			\item The asymptotic distribution for the common component is 
			\begin{align}
				\nonumber & \sqrt{\delta_{NT_i}} (\Sigma^{(1)}_{C,it})^{-1/2} (\tilde{C}^\treat_{it} - C^\treat_{it} ) \xrightarrow{d} \calN (0,1),
			\end{align}
			where 
			\begin{align*}
				\Sigma^{(1)}_{C,it} =& F_t^\T \Sigma_{F}^\I \bigg[  \frac{\delta_{NT_i}}{\Ttreat} \covItreat_{\Lambda,i} + \frac{\delta_{NT_i}}{T} \covIItreat_{\Lambda,i} \bigg]   \Sigma_{F}^\I  F_t + (\Lambda_i^\treat )^\T \Sigma_{\Lambda, t}^\I \bigg[ \frac{\delta_{NT_i}}{N} \covIcontrol_{F,t} + \frac{\delta_{NT_i}}{T} \covIIcontrol_{F,t}  \bigg] \Sigma_{\Lambda, t}^\I \Lambda_i^\treat \\
				& \quad  - 2 \cdot \frac{\delta_{NT_i}}{T} F_t^\T \Sigma_{F}^\I \covIIItreatcontrol_{\Lambda,F,i,t}  \Sigma_{\Lambda, t}^\I \Lambda_i^\treat, 
			\end{align*}
			with $ \covIcontrol_{F,t}$ and $ \covIIcontrol_{F,t}$ given in Theorem \ref{theorem:asy-normal-equal-weight}, $\covItreat_{\Lambda,i} = \Sigma_{F,e_i} $, \\ $\covIItreat_{\Lambda,i}  = \Sigma_{\Lambda}^\I \Big[ \frac{1}{\Ttreat^2} \sum_{u,s = \Tcontrol+1}^T   g_{u,s}(\Sigma_{\Lambda,u}^\I \Lambda_i^\treat,\Sigma_{\Lambda,s}^\I \Lambda_i^\treat)  \Big] \Sigma_{\Lambda}^\I   $, \\ 
			$\covIIItreatcontrol_{\Lambda,F,i,t}  = \Sigma_{\Lambda}^\I \Big[ \frac{1}{\Ttreat} \sum_{u = \Tcontrol+1}^T    g_{u,s}(\Sigma_{\Lambda,u}^\I \Lambda_i^\treat,\Sigma_\Lambda^\I  \Sigma_F^\I F_t) \Big]$, and the function $g_{u,s}(\cdot,\cdot)$ is defined in Assumption \ref{ass:add-factor}.\\
			\item The asymptotic distribution for the individual treatment effect is
			\begin{align}
				\sqrt{\delta_{NT_i}} (\Sigma_{\tau,it})^{-1/2}  \Big((\tilde{C}^\treat_{it} - C^\treat_{it} ) - (\tilde{C}^\control_{it} - C^\control_{it} )\Big) \xrightarrow{d} \calN (0,1)\label{eqn:ite-asy-normal}
			\end{align}
			where 
			\begin{align*}
				\Sigma_{\tau,it} =& F_t^\T \Sigma_{F}^\I  \Gamma_{\Lambda,i}^{\textnormal{obs,miss}}   \Sigma_{F}^\I  F_t + \left(\Lambda_i^\treat - \Lambda_i^\control \right)^\T  \Gamma_{F,t}^{\textnormal{obs,miss}}  \left(\Lambda_i^\treat - \Lambda_i^\control \right) \\&
				+ 2 \cdot  F_t^\T \Sigma_{F}^\I  \Gamma_{\Lambda,F,i,t}^{\textnormal{miss,cov,diff}}  \left(\Lambda_i^\treat - \Lambda_i^\control \right) 
			\end{align*}
			with 
			$\covIcontrol_{\Lambda,i} $, $\covIIcontrol_{\Lambda,i}$ and $\covIIIcontrolcontrol_{\Lambda, F, i,t}$ given in Theorem \ref{theorem:asy-normal-equal-weight}, and $ \Gamma_{F,t}^{\textnormal{obs,miss}}= \Sigma_{\Lambda, t}^\I \bigg[ \frac{\delta_{NT_i}}{N} \covIcontrol_{F,t} + \frac{\delta_{NT_i}}{T} \covIIcontrol_{F,t}  \bigg]  \Sigma_{\Lambda, t}^\I $, \\
			$ \Gamma_{\Lambda,i}^{\textnormal{obs,miss}} =  \frac{\delta_{NT_i}}{T} \Sigma_{\Lambda}^\I \big[\covIcontrol_{\Lambda,i}   + \covIIcontrol_{\Lambda,i} \big]   \Sigma_{\Lambda}^\I +  \frac{\delta_{NT_i}}{\Ttreat} \covItreat_{\Lambda,i} + \frac{\delta_{NT_i}}{T} \covIItreat_{\Lambda,i}  -  \frac{\delta_{NT_i}}{ T} \big( \covIIItreatcontrol_{\Lambda,\Lambda,i}  + (\covIIItreatcontrol_{\Lambda,\Lambda,i})^\T\big)$, \\   	
			$\Gamma_{\Lambda,F,i,t}^{\textnormal{miss,cov,diff}} = \frac{\delta_{NT_i}}{T} \left( \Sigma_{\Lambda}^\I   \covIIIcontrolcontrol_{\Lambda,F,i,t} -  \covIIItreatcontrol_{\Lambda,F,i,t} \right)\Sigma_{\Lambda, t}^\I $,\\
			$\covIIItreatcontrol_{\Lambda,\Lambda,i} = \Sigma_{\Lambda}^\I \Big[ \frac{\delta_{NT_i}}{\Ttreat } \sum_{s = \Tcontrol+1}^T g^\cov_{i,s}(\Lambda_i^\control,\Sigma_{\Lambda,s}^\I \Lambda_i^\treat )  \Big] \Sigma_{\Lambda}^\I    $, and the function
			and the function $g^\cov_{i,s}(\cdot,\cdot)$ is defined in Assumption \ref{ass:mom-clt}.\ref{ass:asy-normal-add-term-thm-loading}.
			\\
			\item 	If, in addition, it holds for $Z \in \+R^{\Ttreat \times L }$, $Z^\T Z/\Ttreat \xrightarrow{P} \Sigma_{Z}$ and $\frac{1}{\Ttreat} \sum_{t = \Tcontrol+1}^T Z_t F_t^\T \xrightarrow{P} \Sigma_{F,Z}$, then the asymptotic distribution for the weighted average treatment effect is
			\begin{align}
				\sqrt{\delta_{NT_i}} (\Sigma_{\beta,t})^{-1/2} \Big(( \tilde \beta_i^\treat - \beta_i^\treat) - ( \tilde \beta_i^\control - \beta_i^\control) \Big) \xrightarrow{d} \calN(0,1) 	\label{eqn:ate-asy-normal}
			\end{align}
			with 
			\begin{align*}
				\Sigma_{\beta,t} =& \Sigma_Z^\I  \Sigma_{F,Z} \Sigma_{F}^\I  \Gamma_{\Lambda,i}^{\textnormal{obs,miss}} \Sigma_{F}^\I  \Sigma_{F,Z}^\T \Sigma_Z^\I + \Sigma_Z^\I \covIItreatcontrol_{Z,i}  \Sigma_Z^\I  \\		
				&+  \frac{\delta_{NT_i}}{ T} \Sigma_Z^\I \bigg[ \Sigma_{F,Z} \Sigma_{F}^\I \Sigma_{\Lambda}^\I   \covIIcontroltreatcontrol_{\Lambda,Z,i} +   (\covIIcontroltreatcontrol_{\Lambda,Z,i})^\T \cdot \Sigma_{\Lambda}^\I  \Sigma_{F}^\I \Sigma_{F,Z}^\T \bigg] \Sigma_Z^\I \\
				& - \frac{\delta_{NT_i}}{ T}\Sigma_Z^\I   \bigg[\Sigma_{F,Z} \Sigma_{F}^\I\cdot  \covIItreattreatcontrol_{\Lambda,Z,i} +   (\covIItreattreatcontrol_{\Lambda,Z,i})^\T \cdot \Sigma_{F}^\I  \Sigma_{F,Z}^\T \bigg] \Sigma_Z^\I  ,
			\end{align*}
			with 
			$\covIIcontroltreatcontrol_{\Lambda,Z,i} = \Big[ \frac{1}{\Ttreat} \sum_{s = \Tcontrol+1}^T  g^\cov_{i,s}(\Lambda_i^\control,\Sigma_{\Lambda, s}^\I (\Lambda_i^\treat - \Lambda_i^\control) ) \Big] \Sigma_{\Lambda}^\I  \Sigma_{F}^\I \Sigma_{F,Z}$,  \\
			$\covIItreattreatcontrol_{\Lambda,Z,i} = \Sigma_{\Lambda}^\I \Big[ \frac{1}{\Ttreat^2}  \sum_{u,s = \Tcontrol+1}^T g_{u,s}(\Sigma_{\Lambda,u}^\I \Lambda_i^\treat,\Sigma_{\Lambda, s}^\I (\Lambda_i^\treat - \Lambda_i^\control))    \Big] \Sigma_{\Lambda}^\I  \Sigma_{F}^\I \Sigma_{F,Z}$, \\ and
			$\covIItreatcontrol_{\Lambda,Z,i} = \Sigma_{F,Z}^\T \Sigma_{F}^\I  \Sigma_{\Lambda}^\I \Big[ \frac{1}{\Ttreat^2} \sum_{u,s = \Tcontrol+1}^T g_{u,s}(\Sigma_{\Lambda, u}^\I (\Lambda_i^\treat - \Lambda_i^\control),\Sigma_{\Lambda, s}^\I (\Lambda_i^\treat - \Lambda_i^\control) )    \Big] \Sigma_{\Lambda}^\I  \Sigma_{F}^\I \Sigma_{F,Z}$, $g^\cov_{i,s}(\cdot,\cdot)$ and the functions $g_{u,s}(\cdot,\cdot)$ are defined in Assumptions \ref{ass:mom-clt}.\ref{ass:asy-normal-add-term-thm-loading} and \ref{ass:add-factor}. 
		\end{enumerate}
		Suppose  Assumptions \ref{ass:obs-equal-weight}, \ref{ass:obs}, \ref{ass:factor-model}, \ref{ass:factor-model-conditional}, \ref{ass:mom-clt-conditional} and \ref{ass:add-factor-conditional} hold. The above three results hold for the propensity weighted estimator after replacing $\Gamma_{F,t}^{\textnormal{obs}}$, $\Gamma_{\Lambda,i}^{\textnormal{obs}}$, $\Gamma_{F,t}^{\textnormal{miss}}$ , $\Gamma_{\Lambda,i}^{\textnormal{miss}}$, $\Sigma_{\Lambda,t}$, $g^\cov_{i,s}(\cdot,\cdot)$ and $g_{u,s}(\cdot,\cdot)$ with $\Gamma_{F,t}^{\textnormal{obs},S}$, $\Gamma_{\Lambda,i}^{\textnormal{obs},S}$, $\Gamma_{F,t}^{\textnormal{miss},S}$ , $\Gamma_{\Lambda,i}^{\textnormal{miss},S}$, $\Sigma_{\Lambda}$, $g^{\cov,S}_{i,s}(\cdot,\cdot)$ and $g^S_{u,s}(\cdot,\cdot)$.    
	\end{theorem}

The results of Theorem \ref{theorem:ate-same-factor} are a consequence of Theorems \ref{theorem:asy-normal-equal-weight} and \ref{theorem:asy-normal}. The challenge arises from correctly capturing the asymptotic covariance between the estimated treated and control common components. This additional covariance term is due to the correction terms from the missing observations. In Theorem \ref{theorem:ate-same-factor}, we impose the additional Assumption \ref{ass:add-factor} for the general estimator and Assumption \ref{ass:add-factor-conditional} for the probability-weighted estimator. Both simply state that the conventional central limit theorems based on the weak dependencies in the errors apply to the subset of treated time periods. These conditions are automatically satisfied in our simplified model and thus can be neglected, as stated in Proposition \ref{prop:simple-assump-imply-general-assump} in the Appendix.

\section{Feasible Estimation and Testing}\label{sec:feasible}

Theorems \ref{theorem:asy-normal-equal-weight}, \ref{theorem:asy-normal} and \ref{theorem:ate-same-factor} are formulated with respect to the asymptotic covariances based on the population model. In order to use them in practice we need feasible estimators for the covariance terms. We propose to use the plug-in estimators $\tilde F_t$, $\tilde{\Lambda}_i$ and $\tilde e_{it} = Y_{it} - \tilde{\Lambda}_i^\T \tilde F_t$ for $(H^\I)^\T F_t $, $H \Lambda_i$ and $e_{it}$. All moments are based on these three objects. For example $\hat \Sigma_F:= \frac{1}{T} \tilde F^{\top} \tilde F$ consistently estimates $(H^\I)^\T  \Sigma_F (H^\I)$. The rotation matrix $H$ can be ignored in the estimated covariances of the common components and the treatment effects as it cancels out. It is only the distribution of the loadings and factors that are estimated up to a rotation matrix. The challenge is to deal with the time and cross-sectional dependency in the residuals. We impose the additional assumption that the time-series and cross-section covariance matrices of the errors $e_{it}$ are sparse in the sense that only a finite number of row elements are non-zero and we know the indices of the non-zero elements. More specifically we define
\begin{align*}
	\mathcal{E}_t = \{ i,j : \+E[e_{it}e_{jt}] \neq 0 \} \qquad \mathcal{E}= \{ i,j,s,t : \+E[e_{it}e_{js}] \neq 0 \} 
\end{align*}
and assume that $|\mathcal{E}_t |= O(N)$ and $| \mathcal{E} | = O(NT)$.
The estimator for $H \covI_{\Lambda,j} H^\T$ and $H \covI_{F,t} H^\T$ depend on the dependency structure in the residuals and we propose the plug-in estimator based on only the non-zero moments of the residuals: 
\begin{align*}
	\widehat \covI_{F,t} &= \frac{1}{N} \sum_{i=1}^N \sum_{j=1}^N W_{it} W_{jt} \tilde \Lambda_i \tilde \Lambda_j^{\top} \tilde e_{it} \tilde e_{jt} \mathbbm 1_{\{i,j \in \mathcal{E}_t \}} \\
	\widehat \covI_{\Lambda,j} &= \frac{T}{N^2} \sum_{i=1}^N \sum_{k=1}^N \tilde \Lambda_i \tilde \Lambda_i^{\top} \frac{1}{|\tlq_{ij}|  |\tlq_{kj}| }  \sum_{t,s \in \tlq_{ij}}  \tilde F_t \tilde F_s^{\top}   \tilde \Lambda_k \tilde \Lambda_k^{\top}  \tilde e_{it} \tilde e_{ks}  \mathbbm 1_{\{i,k,s,t \in \mathcal{E} \}.}
\end{align*}

The estimators are analogous for the probability-weighted estimator. A special case is the estimation approach in \cite{bai2003inferential} that assumes independence of the residuals over time and the cross-section and hence only uses the diagonal entries of the residual covariance and autocovariance matrix. Instead of assuming knowledge of the non-zero entries, it is possible to generalize the estimator similar to \cite{fan2013large} and estimate the non-zero entries with a thresholding estimation approach. We propose a HAC estimator for $\covII_{\Lambda,j}$, $\covII_{F,t}$ and $\covIII_{\Lambda, F, j, t}$
	to account for the time-series dependency in the factors similar to \cite{bai2003inferential}. 




\begin{proposition}\label{prop:cov}
	Suppose that the assumptions of Theorems \ref{theorem:asy-normal-equal-weight}, \ref{theorem:asy-normal} or \ref{theorem:ate-same-factor} hold. In addition, we assume that the time-series and cross-section covariance matrices of the errors $e_{it}$ are sparse in the sense that $|\mathcal{E}_t |= O(N)$ and $| \mathcal{E} | = O(NT)$ and we know the non-zero elements. Then, the plug-in estimators of the asymptotic covariances in Theorems \ref{theorem:asy-normal-equal-weight}, \ref{theorem:asy-normal} and \ref{theorem:ate-same-factor} are consistent and the asymptotic statements in the respective theorems continue to hold with the estimated covariance matrices. 
\end{proposition}
Hence, the treatment effects normalized by their estimated standard deviations follow asymptotically a standard normal distribution, and we obtain feasible test statistics for the various treatment effects. 	

%
%
%
%
%
%
%
%
\section{Generalization of the Missing Patterns}\label{sec:generalization}

Our results can be generalized to the case where the number of observed entries is not proportional to $N$ or $T$ but grows at a strictly smaller rate. The general arguments of the proofs stay the same but we need to carefully account for the convergence rates of each term based on the set $\tlq_{ij}$. The mean squared consistency of the estimated loadings in Theorem \ref{thm:consistency-same-H} generalizes to
\[ \frac{1}{N} \sum_{j = 1}^{N} \norm{\tilde{\Lambda_j}  - H \Lambda_j}^2 = O_p \Bigg( \max \bigg( \frac{1}{N}, \,\,  \frac{1}{N^2} \sum_{i = 1}^{N} \sum_{j = 1}^N \frac{1}{|\tlq_{ij} | }  \bigg)\Bigg). \]
Moreover, we can show the asymptotic normality of the estimated loadings $\tilde \Lambda$, factors $\tilde F$ from the equally weighted regression \eqref{eqn:reg-estimate-factors}, and common components $\tilde C$ under similar assumptions as those in Theorem \ref{theorem:asy-normal-equal-weight}. The estimated loadings $\tilde{\Lambda}_j$ are asymptotically normal with convergence rate 
\[ H^\I \tilde{\Lambda}_j -  \Lambda_j = O_p \Bigg( \bigg[\max \bigg(  \frac{1}{N} \sum_{i = 1}^{N} \frac{1}{|\tlq_{ij}|}, \,\,   \frac{1}{N^2} \sum_{i = 1}^{N} \sum_{l=1}^{N} \frac{|\tlq_{ij} \cap \tlq_{lj} |}{|\tlq_{ij}| |\tlq_{lj}| } \bigg) \bigg]^{1/2} \Bigg),  \]
where the second term is closely related to $\omega_{jj}$ defined in Assumption \ref{ass:simple-moment}.
The estimated factors $\tilde{F}_t$ are asymptotically normal with convergence rate
\[H^\T \tilde{F}_t - F_t =   O_p \Bigg( \bigg[\max \bigg(  \frac{1}{\sum_{i = 1}^{N} W_{it}}, \,\,   \frac{1}{N^4} \sum_{i = 1}^{N} \sum_{j=1}^{N} \sum_{k = 1}^{N} \sum_{l = 1}^{N} \frac{|\tlq_{ij} \cap \tlq_{kl} |}{|\tlq_{ij}| |\tlq_{kl}| } \bigg) \bigg]^{1/2} \Bigg),   \]
where the second term is closely related to $\omega$ defined in Assumption \ref{ass:simple-moment}.
Similarly, by combining the rates of estimated factors and loadings, the estimated common components $\tilde{C}_{it}$  have an asymptotic normal distribution with rate
\begin{align*}
	\tilde{C}_{jt} - C_{jt}  =& O_p \Bigg( \bigg[\max \bigg(  \frac{1}{N} \sum_{i = 1}^{N} \frac{1}{|\tlq_{ij}|}, \,\,   \frac{1}{N^2} \sum_{i = 1}^{N} \sum_{l=1}^{N} \frac{|\tlq_{ij} \cap \tlq_{lj} |}{|\tlq_{ij}| |\tlq_{lj}| }, \,\, \frac{1}{\sum_{i = 1}^{N} W_{it}}, \\
	& \quad    \frac{1}{N^4} \sum_{i = 1}^{N} \sum_{j=1}^{N} \sum_{k = 1}^{N} \sum_{l = 1}^{N} \frac{|\tlq_{ij} \cap \tlq_{kl} |}{|\tlq_{ij}| |\tlq_{kl}| }, \,\,    \frac{1}{N^3} \sum_{i = 1}^{N} \sum_{k = 1}^{N} \sum_{l = 1}^{N} \frac{|\tlq_{ij} \cap \tlq_{kl} |}{|\tlq_{ij}| |\tlq_{kl}| }\bigg) \bigg]^{1/2} \Bigg). 
\end{align*}
The last term is closely related to $\omega_j$ defined in Assumption \ref{ass:simple-moment}.
The expression for the asymptotic covariances of the estimators become more complex. The proofs for the consistency and asymptotic normality for the general case, when observed entries are not proportional to $N$ and $T$, are very similar to the proofs of Theorems \ref{thm:consistency-same-H} and \ref{theorem:asy-normal-equal-weight}, but just require carefully keeping track of the convergence rates of each term.\footnote{The proofs are available upon request.}

We illustrate the more general convergence rates in the simultaneous treatment observation pattern in Table \ref{tab:toy-example-obs-pattern}, where we can provide explicit expressions for the different rates. The mean square consistency result of the loadings simplifies to
\[ \frac{1}{N} \sum_{j = 1}^{N} \norm{\tilde{\Lambda}_j  - H \Lambda_j}^2 = O_p \Bigg( \max \bigg( \frac{1}{N}, \,\, \frac{N_0}{N T_0}, \,\, \frac{1}{T} \bigg)\Bigg). \]
We obtain two different convergence rates for the estimated loadings:
\begin{align*}
	H^\I \tilde{\Lambda}_j -  \Lambda_j  = \begin{cases}
		O_P \Lp \frac{1}{\sqrt{T_0} }  \Rp & j \leq N_0 \\
		O_p  \Lp \max \Big( \sqrt{\frac{N_0}{N T_0 }} , \frac{1}{\sqrt{T}} \Big)
		\Rp & j > N_0. 
	\end{cases}
\end{align*}
Similarly, the estimated factors have two different convergence rates depending on which time block we consider:
\begin{align*}
	H^\T \tilde{F}_t - F_t  = \begin{cases}
		O_P \Lp \max \Big(  \frac{1}{\sqrt{N}}, \frac{N_0}{N \sqrt{T_0 }}, \frac{1}{\sqrt{T}}  \Big) \Rp & t \leq T_0 \\
		O_P \Lp \max \Big(\frac{1}{\sqrt{N - N_0}}, \frac{N_0}{N \sqrt{T_0}}, \frac{1}{\sqrt{T}} \Big) \Rp & t > T_0. 
	\end{cases}
\end{align*}
This results in four different convergence rates for each block for the estimated common components:
\begin{align*}
	\tilde{C}_{jt} - C_{jt} = \begin{cases}
		O_P \Lp \max \Big(  \frac{1}{\sqrt{N}}, \frac{1}{\sqrt{T_0}}  \Big) \Rp &  j \leq N_0, \,\,  t \leq T_0 \\
		O_P \Lp \max \Big(  \frac{1}{\sqrt{N}}, \sqrt{\frac{N_0}{N T_0}}, \frac{1}{\sqrt{T}}  \Big) \Rp &  j > N_0, \,\,  t \leq T_0 \\
		O_P \Lp \max \Big(  \frac{1}{\sqrt{N - N_0}}, \frac{1}{\sqrt{T_0}}  \Big) \Rp &  j \leq N_0, \,\,  t > T_0 \\
		O_P \Lp \max \Big(  \frac{1}{\sqrt{N - N_0}}, \sqrt{\frac{N_0}{N T_0}}, \frac{1}{\sqrt{T}}  \Big) \Rp &  j > N_0, \,\,  t > T_0. 
	\end{cases}
\end{align*}
\section{Simulation}\label{sec:simulation}

\subsection{Asymptotic Distributions}\label{subsec:simulation-asymptotics}

In this section, we demonstrate the finite sample properties of our asymptotic results for both the observed entries and the missing entries. We confirm the theoretical distribution results for the estimated factor, loadings, common components, and treatment effects. We generate the data from a one-factor model $X_{it} = \Lambda_i  F_t + e_{it}$, where $F_t \stackrel{\iid}{\sim}  \calN(0,1)$, $\Lambda_i \stackrel{\iid}{\sim}  \calN(0,1)$ and $e_{it} \stackrel{\iid}{\sim}  \calN(0,1)$. The observation pattern depends on unit-specific characteristics $S_i = \mathbbm{1}(\Lambda_i \geq 0 )$, which are a function of the factor loadings. We study two observation patterns which are illustrated in Figures \ref{fig:obs-pattern-illustrate}(a) and \ref{fig:obs-pattern-illustrate}(b):
\begin{enumerate}
	\item {\it Missing at random}: Entries are observed independently with probability 0.75 if $S_i = 1$, and 0.5 if $S_i = 0$.	
	\item {\it Simultaneous treatment adoption}:  Once a unit adopts treatment, it stays treated afterward. For the units with $S_i = 1$, $25\% $ randomly selected units adopt the treatment from time $0.75 \cdot T$ and the remaining $75\% $ units stay in the control group until the end. For the units with $S_i= 0$, $62.5\%$ randomly selected units adopt the treatment from time $0.375\cdot T$ and the remaining $37.5\%$  units stay in the control group until the end.  We model the treated data as missing. 
\end{enumerate}	

\begin{figure}[h!]
	\tcapfig{Histograms of Standardized Loadings, Factors, and Common Components}
	\centering
	\begin{subfigure}{0.53\textwidth}
		\centering
		\includegraphics[width=1\linewidth]{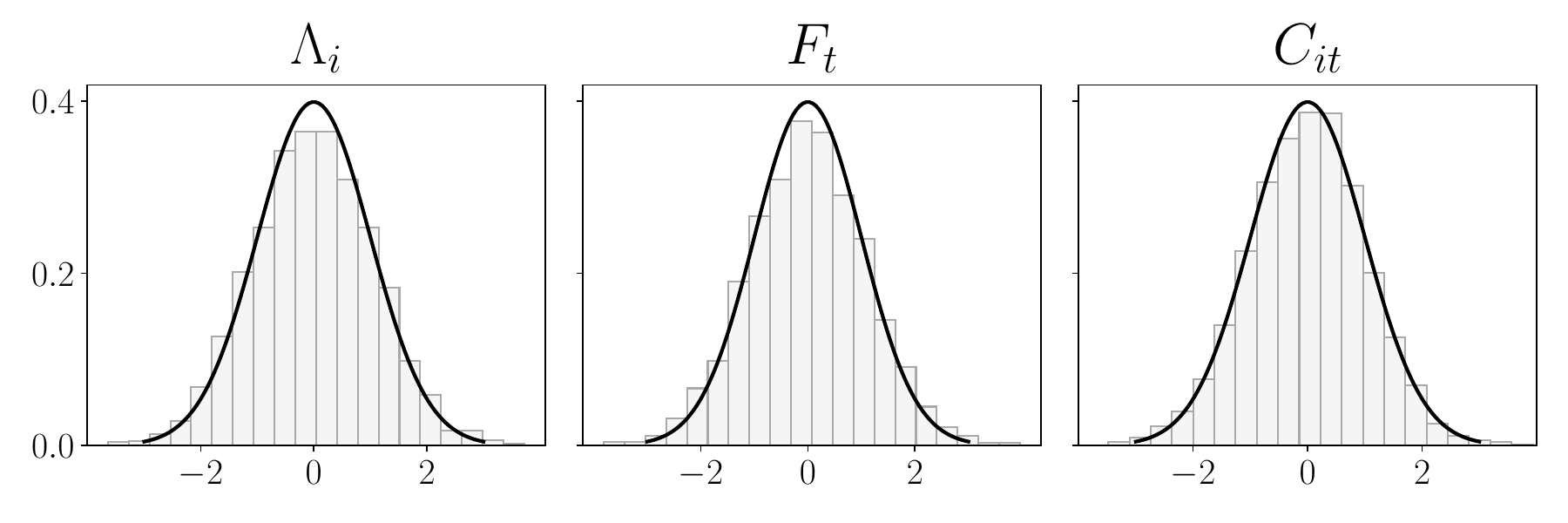}
		\caption{$Y_{it}$ is observed (missing at random)}
		\label{fig:hist-random-obs}
	\end{subfigure}%
	\begin{subfigure}{0.53\textwidth}
		\centering
		\includegraphics[width=1\linewidth]{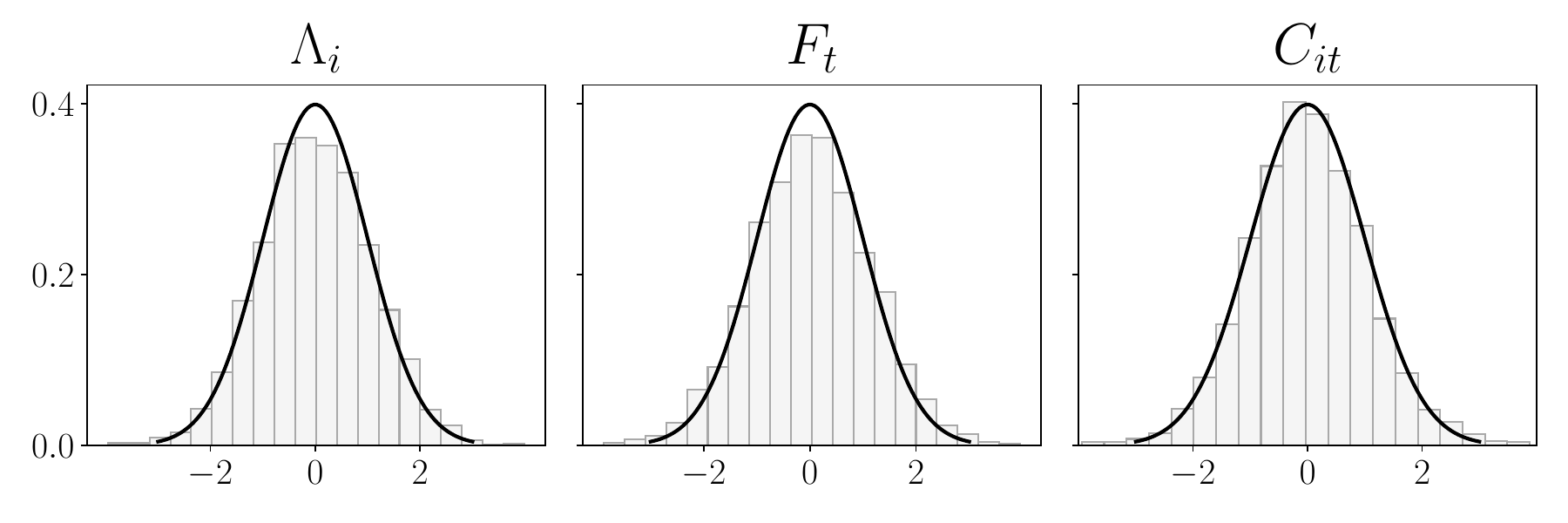}
		\caption{$Y_{it}$ is missing (missing at random)}
		\label{fig:hist-random-miss}
	\end{subfigure}
	\begin{subfigure}{0.53\textwidth}
		\centering
		\includegraphics[width=1\linewidth]{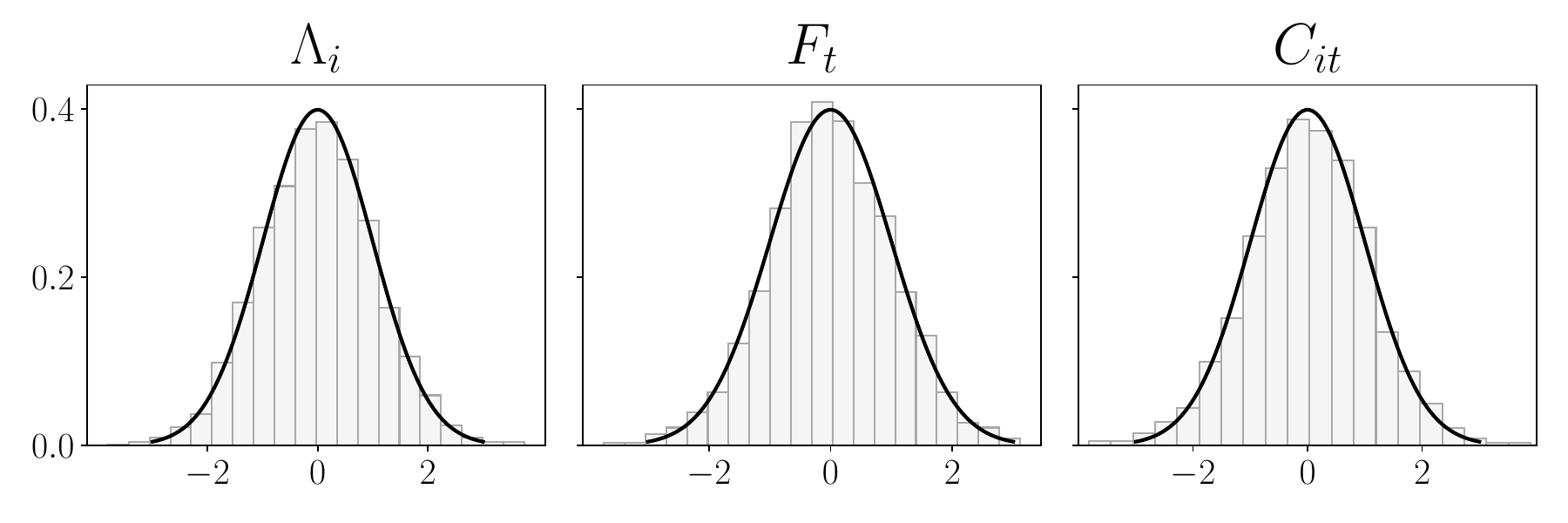}
		\caption{$Y_{it}$ is observed (simultaneous treatment adoption)}
		\label{fig:hist-simul-obs}
	\end{subfigure}%
	\begin{subfigure}{0.53\textwidth}
		\centering
		\includegraphics[width=1\linewidth]{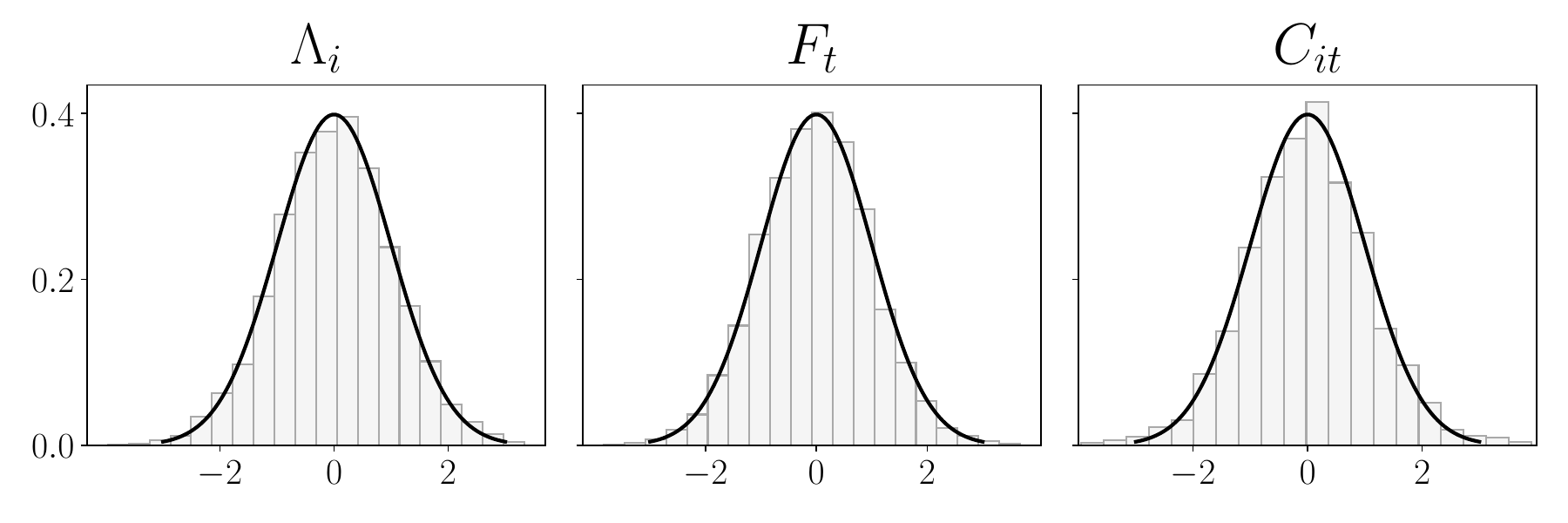}
		\caption{$Y_{it}$ is missing (simultaneous treatment adoption)}
		\label{fig:hist-simul-miss}
	\end{subfigure}
	\label{fig:hist-all}
	\bnotefig{These figures show the histograms of estimated standardized loadings, factors, and common components normalized by their estimated standard devisions, where $N = 100$ and $T = 150$. The normal density function is superimposed on the histograms. The results are based on 2,000 Monte Carlo simulations. The Internet Appendix collects the histograms for other specifications of $N$ and $T$.}
\end{figure}

\begin{figure}[h!]
	\tcapfig{Histograms of Standardized Control and Treated Common Components, Individual and Average Treatment Effects}
	\centering
	\begin{subfigure}{0.7\textwidth}
		\centering
		\includegraphics[width=1\linewidth]{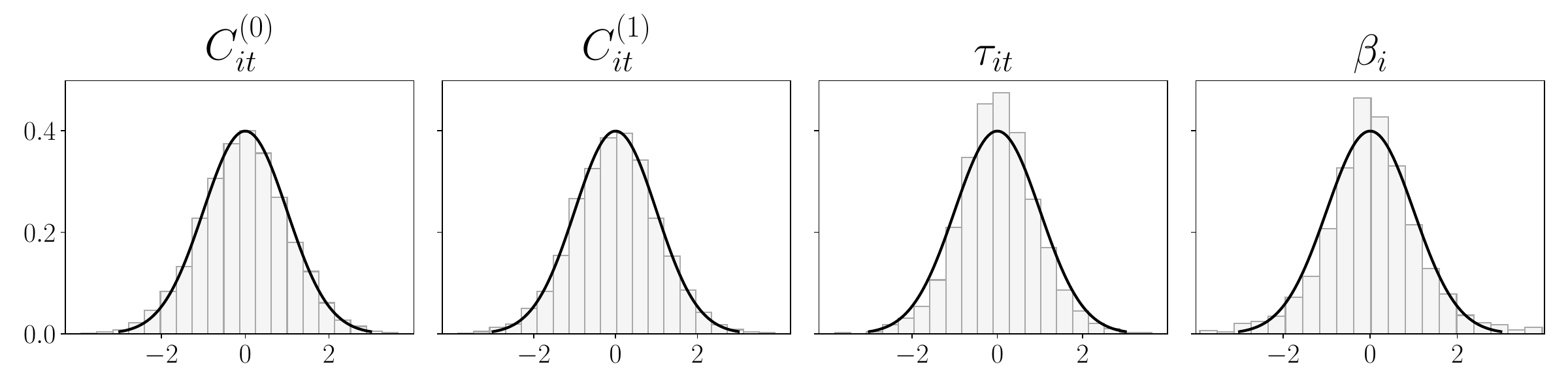}
		\caption{Treatment effect $\tau = 0$}
	\end{subfigure}
	\begin{subfigure}{0.7\textwidth}
		\centering
		\includegraphics[width=1\linewidth]{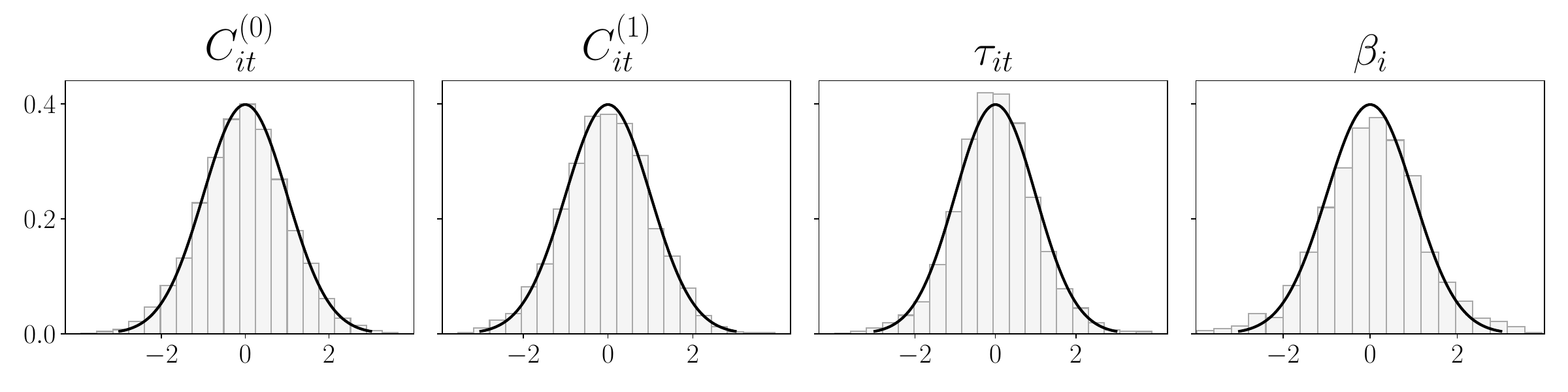}
		\caption{Treatment effect $\tau = 0.25$}
	\end{subfigure}
	\bnotefig{These figures show the histograms of estimated control and treated common components, individual and average treatment effect ($Z = \vec{1}$) normalized by their estimated standard deviations, where $N = 100$ and $T = 150$. The normal density function is superimposed on the histograms. The observation pattern follows the simultaneous treatment adoption pattern. The results are based on 2,000 Monte Carlo simulations. The Internet Appendix collects the histograms for other specifications of $N$ and $T$.}
	\label{fig:hist-treatment}
\end{figure}

\begin{table}[h!]
	\tcaptab{Statistical Power of Treatment Effect Tests}
	\centering
	\begin{tabular}{lrc|rrrr|rrrr}
		\toprule
		&      & {} & \multicolumn{4}{c|}{$\tilde{C}^\treat_{it}- \tilde{C}^\control_{it} $} & \multicolumn{4}{c}{$\tilde{\beta}^\treat_i - \tilde{\beta}^\control_i$} \\
		&      & $\Lambda^\treat_i - \Lambda^\control_i$ &  0.25 &  0.50 &  1.00 &  2.00 &  0.25 &  0.50 &  1.00 &  2.00 \\
		$N$ & $T$ & $\mu_F$ &       &       &       &       &       &       &       &       \\
		\midrule
100                     & 100  & 0.900 & 0.171 & 0.450 & 0.902 & 0.991 & 0.198 & 0.440 & 0.864 & 0.968 \\
                        &      & 0.000 & 0.271 & 0.660 & 0.946 & 0.996 & 0.271 & 0.654 & 0.939 & 0.996 \\
                        &      & 0.500 & 0.347 & 0.835 & 0.981 & 1.000 & 0.345 & 0.831 & 0.979 & 0.998 \\
                        &      & 0.000 & 0.466 & 0.906 & 0.991 & 1.000 & 0.481 & 0.910 & 0.991 & 1.000 \\
		\midrule
250                     & 100  & 0.900 & 0.175 & 0.464 & 0.896 & 0.994 & 0.165 & 0.456 & 0.866 & 0.981 \\
                        &      & 0.000 & 0.273 & 0.722 & 0.959 & 0.998 & 0.271 & 0.731 & 0.954 & 0.998 \\
                        & 500  & 0.900 & 0.572 & 0.954 & 1.000 & 1.000 & 0.558 & 0.931 & 0.983 & 0.989 \\
                        &      & 0.000 & 0.764 & 0.970 & 1.000 & 1.000 & 0.772 & 0.979 & 1.000 & 1.000 \\
		\midrule
500                     & 500  & 0.900 & 0.610 & 0.975 & 1.000 & 1.000 & 0.591 & 0.958 & 0.998 & 1.000 \\
                        &      & 0.000 & 0.805 & 0.987 & 1.000 & 1.000 & 0.809 & 0.987 & 1.000 & 1.000 \\
                        & 1000 & 0.900 & 0.860 & 0.992 & 1.000 & 1.000 & 0.848 & 0.983 & 0.998 & 0.998 \\
                        &      & 0.000 & 0.959 & 1.000 & 1.000 & 1.000 & 0.957 & 1.000 & 1.000 & 1.000 \\
		\bottomrule
	\end{tabular}

	\bnotetab{This table shows the proportion of test statistics of the treatment effect that reject the null hypotheses $\mathcal{H}_0: C^\treat_{it} - C^\control_{it} = 0$ or $\mathcal{H}_0:  \beta^\treat_i - \beta^\control_i = 0$, where $ \beta^\treat_i = \frac{1}{\Ttreat} \sum_{\Tcontrol+1}^T C^\treat_{it}$ and  $ \beta^\control_i = \frac{1}{\Ttreat} \sum_{\Tcontrol+1}^T C^\control_{it}$. We consider a 95\% confidence level (the test statistics are within $[-1.96, 1.96]$) over 500 Monte Carlo simulations . The test statistics normalize $\tilde{C}^\treat_{it}- \tilde{C}^\control_{it} $ and $\tilde{\beta}^\treat_i - \tilde{\beta}^\control_i$ with their estimated standard deviation from Equations \eqref{eqn:ite-asy-normal} and \eqref{eqn:ate-asy-normal}. The estimated standard deviations are estimated under the null hypothesis of $\Lambda^\treat_i - \Lambda^\control_i = 0$. The observation pattern follows the simultaneous treatment adoption pattern. The proportion of acceptance decreases with $N, T, \mu_F$ and $\tilde{\beta}^\treat_i - \tilde{\beta}^\control_i$, implying that the statistical power increases with the data dimensionality and the scale of the treatment effect. The Internet Appendix collects additional robustness tests confirming the same findings for different specifications and also showing that the statistical power increases with the proportion of observed entries in the data}
	\label{tab:treatment-power-main-text}
\end{table}

To conserve space, we report here the distribution results for the regression based estimator based on Equation \eqref{eqn:reg-estimate-factors}, but the results extend to the propensity-weighted estimator. Figure \ref{fig:hist-all} shows the histograms of standardized factors, loadings, and common components for randomly selected observed entries and missing entries based on Theorem \ref{theorem:asy-normal-equal-weight}. The histograms match the standard normal density function very well and support the validity of our asymptotic results in finite samples.

Figure \ref{fig:hist-treatment} confirms that our treatment test in Theorem \ref{theorem:ate-same-factor} has the correct size. The control data follows our benchmark one-factor model. We assume a constant treatment effect, i.e., $\Lambda_i^\treat = \Lambda_i^\control + \tau$, where $\tau$ is set to 0 or 0.25. Figure \ref{fig:hist-treatment} shows the histograms of standardized common components for treated and control, the individual treatment effect, and an equally weighted treatment effect for randomly selected units and times. As expected, the histograms support the validity of our asymptotic results in finite samples.

Table \ref{tab:treatment-power-main-text} demonstrates the statistical power of our tests for individual and average treatment effects, where the null hypotheses are $\mathcal{H}_0:  \beta^\treat_i - \beta^\control_i = 0$ with equal weights for all time periods, i.e.,  $\beta^\treat_i=\tau^\treat_i$ and $\beta^\control_i=\tau^\control_i$. The power increases with the data dimensionality ($N$ and $T$) and the scale of treatment effect that is determined by the mean of the factor $\mu_F$ and the difference between the control and treated loadings $\Lambda_i^\treat - \Lambda_i^\control$. The null hypothesis implies $\Lambda_i^\treat - \Lambda_i^\control=0$, which we use in the estimation of the asymptotic variance. This slightly improves the power, but the results in the Internet Appendix show that we also have good power properties without imposing the null hypothesis in the estimation of the asymptotic covariances. Moreover, the statistical power increases with the proportion of observed entries, as shown in the comparison between Tables \ref{tab:treatment-power-main-text} and \ref{tab:treatment-power-fewer-obs-null} in the Internet Appendix. 

\subsection{Robustness to Missing Patterns}

In this section, we show that our benchmark regression-based estimator (denoted as XP) and propensity-weighted estimator (denoted as $\text{XP}_{\text{PROP}}$) perform well under a variety of missing patterns. As reference we also include the estimators of \cite{jin2020factor} (denoted as JMS) and \cite{bai2019matrix} (denoted as BN). Each of the two estimators is designed for a specific observation pattern and hence provides a natural reference level for that specific pattern. \cite{jin2020factor} assume that observations are missing at random, while \cite{bai2019matrix} is tailored to an observation pattern with a block structure after proper reshuffling. These are the four estimation approaches that provide an inferential theory for imputed common components in an approximate factor model and were available at the time of submission of this paper.


We generate the data from a two-factor model $X_{it} = \Lambda_i^{\top} F_t + e_{it}$, where $F_t \stackrel{\iid}{\sim}  \calN(0,I_2)$, $\Lambda_i \stackrel{\iid}{\sim}  \calN(0,I_2)$ and $e_{it} \stackrel{\iid}{\sim}  \calN(0,1)$.   We consider six different observation patterns. The first three cases are (1) missing uniformly at random, (2) simultaneous treatment adoption, and (3) staggered treatment adoption. Then, we allow the observation pattern for these three cases to depend on a unit-specific characteristic defined as $S_i = \mathbbm{1}(\Lambda_{i,2} \geq 0 )$. Hence, case four to six are (4) missing at random conditional on $S_i$, (5) simultaneous treatment adoption conditional on $S_i$, (6) staggered treatment adoption conditional on $S_i$. Table \ref{tab:comparison-jms-bn} contains figures showing the observation patterns and their detailed descriptions. Note that these are all practically relevant patterns, in particular the staggered treatment adoption that appears in our empirical companion paper and is prevalent in empirical applications. 

\begin{table}[h!]
	\centering
	\tcaptab{Robustness to Missing Patterns}
		\begin{minipage}{0.17\linewidth}
			\vspace{0.7cm}
			\flushright
			\includegraphics[width=15mm]{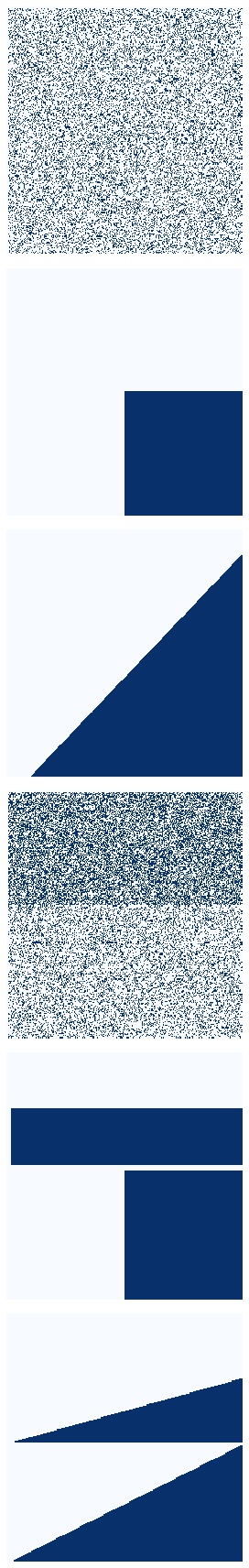}
		\end{minipage}\hfill
		\begin{minipage}{0.8\linewidth}
			\begin{tabular}{l|l|cccc}
				\toprule
				Observation Pattern & $W_{it}$ &  XP  &  $\text{XP}_{\text{PROP}}$ &  JMS &      BN \\
				\midrule
				Random & obs &   {\bf 0.015} & {\bf 0.015} & 0.023 & -- \\
				& miss &   {\bf 0.015} & {\bf 0.015} & 0.021 & -- \\
				& all &   {\bf 0.015} & {\bf 0.015} & 0.023 & -- \\
				\midrule
				Simultaneous & obs &   { \bf 0.012} & {\bf 0.012} & 0.124 &   {\bf 0.012} \\
				& miss &   {  0.020} & {  0.020} & 0.184 &   {\bf 0.017} \\
				& all &   {  0.014} & {  0.014} & 0.139 &   {\bf 0.013} \\
				\midrule
				Staggered & obs &   {\bf 0.017} & {\bf 0.017} & 0.366 &   0.073 \\
				& miss &   {\bf 0.043} & {\bf 0.043} & 0.318 &   0.087 \\
				& all &   {\bf 0.027} & {\bf 0.027} & 0.347 &   0.078 \\
				\midrule
				Random & obs &   {\bf 0.019} & 0.020 & 0.077 & -- \\
				$W$ depends on $S$ & miss &   {\bf 0.024} & {\bf 0.024} & 0.067 & -- \\
				& all &   {\bf 0.021} & {\bf 0.021} & 0.073 & -- \\
				\midrule
				Simultaneous& obs &   {\bf 0.032} & 0.040 & 0.703 &   0.141 \\
				$W$ depends on $S$ & miss &   {\bf 0.231} & 0.256 & 0.521 &   0.279 \\
				& all &   {\bf 0.129} & 0.145 & 0.615 &   0.209 \\
				\midrule
				Staggered & obs &   {\bf 0.016} & 0.018 & 0.272 &   0.117 \\
				$W$ depends on $S$ & miss &   {\bf 0.064} & 0.069 & 0.346 &   0.186 \\
				& all &   {\bf 0.033} & 0.036 & 0.299 &   0.142 \\
				\bottomrule
			\end{tabular}
	\end{minipage} 
	\bnotetab{This table reports the relative MSE of XP (our benchmark method), $\textnormal{XP}_{\textnormal{PROP}}$ (our propensity-weighted method), JMS \citep{jin2020factor}, and BN \citep{bai2019matrix} on observed, missing and all entries, $N = 250$, $T = 250$. The figures on the left show patterns of missing observations with the shaded entries indicating the missing entries. Bold numbers indicate the best relative model performance. We generate a two-factor model and a unit-specific characteristic $S_i = \mathbbm{1}(\Lambda_{i,2} \geq 0 )$. The observation patterns are generated as follows. (1) {\it Missing uniformly at random}: Entries are observed independently with probability $p = 0.75$. (2) {\it Simultaneous treatment adoption}:  $50\% $ randomly selected units adopt the treatment from time $0.5 \cdot T$ and the remaining $50\% $ units stay in the control group until the end. (3) {\it Staggered treatment adoption}:  All units are in the control group for $t < 0.1\cdot T$. At time $0.1 \cdot T \leq  t \leq T$, $\frac{t-0.1 \cdot T}{T} \%$ units are in the treated group. The remaining $10\%$ units stay in the control group until the end. (4) {\it Missing at random conditional on $S_i$}: Entries are observed independently with probability $p_{it} = 0.75$ $S_i = 1$, and $p_{it} = 0.5$ if $S_i = 0$. (5) {\it Simultaneous treatment adoption conditional on $S_i$}: For the units with $S_i = 1$, $95\% $ units adopt the treatment from time $0.5 \cdot T$ and $5\% $ units stay in the control group until the end. For the units with $S_i= 0$, $50\%$ units adopt the treatment from time $0.02\cdot T$ and  $50\%$ units stay in the control group until the end. (6) {\it Staggered treatment adoption conditional on $S_i$}:  All units are in the control group for $t < 0.02\cdot T$. For the units with $S_i = 1$, at time $0.02 \cdot T \leq  t \leq T$, $\frac{t-0.02 \cdot T}{T} \%$ units are in the treated group with the remaining $2\%$ staying in the control group until the end. For the units with $S_i = 0$, at time $0.02 \cdot T \leq  t \leq T$, $\frac{t-0.02 \cdot T}{1.96 T} \%$ units are in the treated group with the remaining $50\%$ units staying in the control group until the end. We run 100 Monte Carlo simulations. The Internet Appendix collects additional robustness results with the same findings.}
	\label{tab:comparison-jms-bn}
\end{table}

Table \ref{tab:comparison-jms-bn} compares the performance of estimating the common components. We report the normalized mean squared error (MSE) of the four methods for observed, missing and all units defined as follows:
\[\text{relative MSE}_{\mathcal{S}} = \frac{\sum_{(i,t) \in \mathcal{S}} \big(\tilde{C}_{it} - C_{it}\big)^2 }{\sum_{(i,t) \in \mathcal{S}} C_{it}^2}, \]
where $\mathcal{S}$ is either the set observed, missing or all observations. 

First, and most importantly, our benchmark estimator shows excellent performance for all observation patterns. Our estimator has the smallest or at least a very similar small MSE compared to the other methods, as indicated by the bold numbers. Hence, we view our approach as a simple and reliable all-purpose estimator. 
Our propensity-weighted estimator is very close to the benchmark estimator but performs slightly worse. This is in line with our theoretical result that propensity weighting is generally less efficient.

In the case of missing at random conditional or unconditional on $S_i$, our methods have the smallest MSE. \cite{jin2020factor} also have a small MSE as long as the observation pattern does not depend on $S_i$ as their method is designed for missing uniformly at random. Missing at random violates the assumptions of \cite{bai2019matrix} and their estimator not applicable as there not sufficiently large blocks of fully observed entries. 


In the case of simultaneous treatment adoption, \cite{bai2019matrix} has the smallest MSE as their method is tailored to this case. Interestingly, our method as an all-purpose estimator is very close to \cite{bai2019matrix}. When the observation pattern depends on $S$, it can shrink the size of fully observed blocks, which increases the importance of using all observed entries resulting in the smallest MSE for our method.
	In the case of simultaneous treatment adoption, the assumptions in \cite{jin2020factor} are violated, which is reflected in the larger MSE.

Our methods have the smallest MSE for the case of staggered treatment adoption that is prevalent in empirical applications \citep{athey2018design}. This holds whether the observation pattern depends or does not depend on $S$. In contrast to \cite{bai2019matrix}, we use all observed entries in the estimation, which provides a more efficient estimator. Note that in this simulation example the fully observed blocks are very small, and hence, similar to the missing at random case, the assumptions in \cite{bai2019matrix} might not be satisfied. As the assumptions of \cite{jin2020factor} are violated, their imputation results in larger errors. 

The Internet Appendix shows that the findings are robust to the size of the panel and the parameters of the observation patterns. We also compare the MSE of the various methods after iterations in Tables \ref{tab:comparison-jms-bn-iter1}-\ref{tab:comparison-jms-bn-iter3} in the Internet Appendix. In more detail, we first impute the missing values with different methods. In the second step, we apply PCA to the full panel with imputed values to estimate the factor model and update the imputed values with the estimated common components. The observed entries stay the same. This process is repeated for multiple iterations. Note that this iterated estimation approach is actually a different estimation approach by itself. The four methods provide different starting values for the same iterative estimation approach that is based on a fixed-point argument. Importantly, there is no inferential theory for iterative estimators under general patterns.\footnote{While \cite{jin2020factor} consider iterations, their asymptotic results only hold for missing at random. \cite{bai2019matrix} provide distribution results for a different iteration that is not making use of all observations and therefore only has a minor effect.} 
Hence, if the goal is to estimate treatment effects, these iterative estimators cannot be used. Since our methods start with a value that has a smaller MSE, our methods, in general, converge faster (often already after three iterations) and also have a small MSE for a fixed number of iterations. Our results are robust to the choice of $N$ and $T$ and we present the corresponding results for $N = 100$ and $T = 150$ in Tables \ref{tab:comparison-jms-bn-iter0-N100T150}-\ref{tab:comparison-jms-bn-iter3-N100T150} in the Internet Appendix. In summary, if the goal is to only minimize the imputation error without an inferential theory, the iterative estimation generally improves the results, but the relative performance of the different estimation approaches without iteration carries over to the iteration setup.

\begin{table}[h!]
	\tcaptab{Benchmark and Propensity-Weighted Estimator for Weak and Missing Factors}
	\centering
	\setlength{\tabcolsep}{6pt} 
	\renewcommand{\arraystretch}{1.0} 
	{\small
		\begin{tabular}{l|rr|rr||rr|rr}
			\toprule
			$k$ estimated factors& \multicolumn{4}{c||}{1} & \multicolumn{4}{c}{2} \\
			\midrule 
			$[\mu_{F,1}, \mu_{F,2}]$ & \multicolumn{2}{c|}{[1,1]} & \multicolumn{2}{c||}{[5,0.5]} & \multicolumn{2}{c|}{[1,1]} & \multicolumn{2}{c}{[5, 0.5]}  \\ 
			$[\sigma_{F,1}, \sigma_{F,2}]$ & \multicolumn{2}{c|}{[1,1]} & \multicolumn{2}{c||}{[5,0.5]} & \multicolumn{2}{c|}{[1,1]} & \multicolumn{2}{c}{[5, 0.5]}  \\ \midrule
			Method &            XP & $\text{XP}_{\text{PROP}}$ &                XP & $\text{XP}_{\text{PROP}}$ &            XP & $\text{XP}_{\text{PROP}}$ &                XP & $\text{XP}_{\text{PROP}}$ \\
			\midrule
			obs $C^{(0)}_{it}$                  &         {\bf 0.227} &   0.251 &             { 0.011} &   { 0.011} &         { 0.014} &   { 0.014} &             {\bf 0.002} &   0.003 \\
			miss $C^{(0)}_{it}$                 &         0.478 &   {\bf 0.288} &             { 0.007} &   { 0.007} &         {\bf 0.044} &   0.045 &             0.026 &   {\bf 0.023} \\
			all $C^{(0)}_{it}$                  &         0.314 &   {\bf 0.264} &             { 0.009} &   { 0.009} &         {\bf 0.024} &   0.025 &             0.014 &   {\bf 0.012} \\
			obs $C^{(0)}_{it}$($S=1$)           &         {\bf 0.184} &   0.254 &             {\bf 0.755} &   0.761 &         { 0.013} &   { 0.013} &             {\bf 0.122} &   0.125 \\
			miss $C^{(0)}_{it}$($S=1$)          &         {\bf 0.046} &   0.261 &             {\bf 0.751} &   0.769 &         { 0.019} &   { 0.019} &             {\bf 0.123} &   0.132 \\
			obs $C^{(0)}_{it}$($S=0$)           &         0.304 &   {\bf 0.268} &             0.001 &   {\bf 0.000} &         { 0.016} &   { 0.016} &             0.001 &   {\bf 0.001} \\
			miss $C^{(0)}_{it}$($S=0$)          &         0.721 &   {\bf 0.308} &             0.003 &   {\bf 0.002} &         { 0.059} &   { 0.059} &             0.025 &   {\bf 0.022} \\
			obs $C^{(1)}_{it}$                  &         0.402 &   {\bf 0.278} &             0.007 &   {\bf 0.006} &         0.037 &   {\bf 0.036} &             {\bf 0.002} &   0.003 \\
			$C^{(1)}_{it} - C^{(0)}_{it}$       &         0.481 &   {\bf 0.294} &             0.008 &   {\bf 0.007} &         { 0.052} &   { 0.052} &             0.026 &   {\bf 0.023} \\
			$\beta^{(1)}_{i} - \beta^{(0)}_{i}$ &         0.168 &   {\bf 0.032} &             { 0.002} &   { 0.002} &         {\bf 0.012} &   0.013 &             0.008 &   {\bf 0.007} \\
			ATE                                 &         0.090 &   {\bf 0.026} &             {\bf 0.006} &   0.007 &         0.009 &   {\bf 0.008} &             0.012 &   {\bf 0.011} \\
			\bottomrule
	\end{tabular}}
	\label{tab:omit}
	\bnotetab{This table compares the percentage errors for various estimates with the benchmark estimator (XP) and the propensity weighted estimator $\text{XP}_{\text{PROP}}$ for omitted and weak factors. The data is simulated with a two-factor model and a simultaneous treatment adoption for different means and variances of the latent factors. {\bf For $k=1$ one factor is omitted in the estimation as the population model is a two-factor model. For $[\sigma_{F,1}, \sigma_{F,2}]=[5, 0.5]$ the second factor is weak}. In more detail: $Y_{it}^\control = \Lambda_{i,1}^\control F_{t,1}  + \Lambda_{i,2}^\control F_{t,2}  + e_{it}^\control$ and $Y_{it}^\treat=\Lambda_{i,1}^\treat F_{t,1}  + \Lambda_{i,2}^\treat F_{t,2}  + e_{it}^\treat$. The first half of the cross-section depends on the first factor, while the second half depends on the second factor: For $i = 1, \cdots, N/2$, $\Lambda^\control_{i,1} \sim \calN(0,1)$, $\Lambda^\treat_{i,1} = \Lambda_{i,1}^\control + \calN(0.2,1)$  and $\Lambda^\treat_{i,2} = \Lambda^\control_{i,2} = 0$, and for $i = N/2+1, \cdots, N$, $\Lambda^\treat_{i,1} = \Lambda^\control_{i,1} = 0$, $\Lambda^\control_{i,2} \sim \calN(0,1)$ and $\Lambda^\treat_{i,2} = \Lambda^\control_{i,2} +  \calN(0.2, 1)$. Let $N = 250$, $T = 250$ and $e_{it} \stackrel{\iid}{\sim}  \calN(0,1)$. The observation pattern depends on an observed unit specific variable defined as $S_i = \mathbbm{1}(\Lambda_{i,2}^\control \neq 0 )$ which only depends on the loadings of the second factor. Once a unit adopts treatment, it stays treated afterwards. For the units with $S_i = 1$, $50\%$ randomly selected units adopt the treatment from time $0.5\cdot T$ and the remaining $50\%$  units stay in the control group until the end. For the units with $S_i= 0$, $90\%$ randomly selected units adopt the treatment from time $0.5\cdot T$ and the remaining $10\%$  units stay in the control group until the end. We report the relative MSE for common components for observed and unobserved treated and control common components. We also report the results conditional on the characteristic $S_i$ and the relative MSE of $\beta^{(1)}_{i} - \beta^{(0)}_{i}$ capturing the average treatment effect over time for each unit and ATE which is the relative MSE of the overall average treatment effect $\sum_{(i,t): W_{it} = 0} \big( \hat{C}^\treat_{it} - \hat{C}^\control_{it} \big)$. The results are generated from 1,000 Monte Carlo simulations. The results show that $\text{XP}_{\text{PROP}}$ can be a more robust estimator for missing observations under misspecification (omitted or weak factors).
	}
\end{table}

\subsection{Misspecification and Robustness of Propensity-Weighted Estimator}\label{sec:robust}

In this section, we show that the propensity-weighted estimator can have desirable robustness properties under misspecification. Our results are motivated by insights from causal inference that propose doubly robust estimation procedures for missing values, as discussed, for example, in \cite{kang2007}. In causal inference, we can either model the relationship between the covariates and the outcome or model the probabilities of missingness to estimate causal effects. Doubly robust procedures combine both by using a propensity weight in regressions to mitigate the selection bias. Their potential advantage is that they can provide reliable estimates in the case of omitted variables. Our setup differs from classical causal inference as we estimate the covariates as latent factors from the data. However, we can have a situation similar to omitted variables if we estimate too few latent factors, the factors are weak, or the population model is nonlinear.

We compare our benchmark estimator (XP) and propensity-weighted estimator ($\text{XP}_{\text{PROP}}$) under two types of model misspecification. In Table \ref{tab:omit}, we consider the case of omitted factors. The population model is generated by a two-factor model, but we only estimate one latent factor. In this case, the propensity-weighted estimator can perform better than the benchmark estimator. However, when the model is correctly specified, and we estimate two factors, the benchmark estimator dominates. When the second factor is weak in the sense that its variance and corresponding eigenvalue are very small, the situation is similar to an omitted factor. In this case, it is possible that the propensity-weighted estimator performs better even if we estimate the correct number of latent factors. Note that weak factors are also a form of misspecification, as discussed in \cite{onatski2012}. In this simulation, observations are more likely to miss if they are exposed to the omitted or weak second factor. Hence, the robustness gains of the propensity-weighted estimator arise for the missing data and the treatment effects.

The case of omitted latent factors shares similarity with the case of a misspecified functional form. In Table \ref{tab:exp} in the Online Appendix we generate the data from a non-linear one-factor model. Under certain assumptions it is possible to approximate a non-linear transformation as a linear function of appropriate basis functions. Such an approximation can be formulated as a linear latent multi-factor model, where the additional factors are non-linear transformations of the underlying one-factor model. Hence, some form of functional model misspecification can be corrected by using more latent factors. In our example, the non-linearity is very well approximated by three latent factors for the simple and propensity-weighted estimator. However, if we use only one or two latent factors, the propensity-weighted estimator is more robust to the misspecification. While we do not provide a formal non-parametric theory, our simulation suggests that a non-linear misspecification can share similar features with the case of omitted factors. Therefore, if a researcher suspects some form of model misspecification, the propensity-weighted estimator can be a useful alternative.\footnote{For a non-linear factor model, \cite{feng2020causal} proposes a local PCA method that uses a linear model approximation in a local neighborhood. Our argument is based on a global approximation of the non-linear functional relationship, where the additional latent factors serve as additional basis functions.}

\section{Conclusion}
This paper develops the inferential theory for latent factor models estimated from large dimensional panel data with missing observations. Our paper stands out by the generality of the missing patterns that we allow for. We propose two estimators for the latent factor model: a simple all-purpose estimator and an extension to a probability-weighted estimator. Our all-purpose estimator is easy to use while it performs well under a variety of missing patterns. The propensity weighted estimator is an alternative that is less efficient for correctly specified models but can be more robust to certain forms of misspecification. The key application of our asymptotic distribution theory is to test causal treatment effects. We provide a test for the point-wise treatment effect that can be heterogeneous and time-dependent under general adoption patterns where the units can be affected by unobserved factors.

\section{Appendix}


\paragraph{Notation.} Let $M<\infty$ denote a generic constant. Let $\norm{v} $ denote the vector norm and $\norm{A} = trace(A^\top A)^{1/2}$ the Frobenius norm of matrix $A$. 

\subsection*{General Assumptions}\label{sec:general-assumptions}


\begin{assumpG}[Factor Model] \label{ass:factor-model}
	\texttt{}
	\begin{enumerate}[wide=0pt, widest=99,leftmargin=\parindent, labelsep=*]
		\item \label{ass:factor} Factors: $\forall \, t$, $\+E[\norm{F_t}^4]  \leq \bar{F} < \infty$. There exists some positive definite $r \times r$ matrix $\Sigma_F$, such that  $\frac{1}{T} \sum_{t=1}^T F_t F_t^\T \xrightarrow{P} \Sigma_F$ and $\+E\norm{\sqrt{T} \Lp \frac{1}{T} \sum_{t =1}^T F_t F_t^\T - \Sigma_F \Rp}^2  \leq M$. Furthermore, for any $\tlq_{ij}$, $\frac{1}{|\tlq_{ij}|} \sum_{t \in \tlq_{ij}} F_t F_t^\T \xrightarrow{P} \Sigma_F$ and $\+E\norm{\sqrt{|\tlq_{ij}|} \Lp \frac{1}{|\tlq_{ij}|} \sum_{t \in \tlq_{ij}} F_t F_t^\T - \Sigma_F \Rp}^2  \leq M$. 
		\item \label{ass:loading} Factor loadings: loadings are random, independent of factors and errors. $\forall \, t$, $\+E[\norm{\Lambda_i}^4]  \leq \bar{\Lambda} < \infty$. 
		There exists some positive definite $r\times r$ matrix $\Sigma_\Lambda$ such that $\frac{1}{N} \sum_{i = 1}^N \Lambda_i \Lambda_i^\T \xrightarrow{P} \Sigma_{\Lambda}  $ and
		$\+E \norm{\sqrt{N} \Lp  \frac{1}{N} \sum_{i=1}^N \Lambda_i \Lambda_i^\T  - \Sigma_\Lambda \Rp} \leq M $. 
		\item \label{ass:error} Time and cross-section dependence and heteroskedasticity of errors: There exists a positive constant $M < \infty$, such that for all $N$ and $T$:
		\begin{enumerate}[wide=0pt, widest=99,leftmargin=\parindent, labelsep=*]
			\item $\+E[e_{it}] = 0$, $\+E|e_{it}|^8 \leq M$.
			\item $\+E[e_{is}e_{it}] = \gamma_{st,i}$ with $|\gamma_{st,i}| \leq \gamma_{st}$ for some $\gamma_{st}$ and all $i$. For all $t$, $\sum_{s=1}^T \gamma_{st} \leq M$.
			\item $\+E[e_{it}e_{jt}] = \tau_{ij,t}$ with $|\tau_{ij,t}| \leq \tau_{ij}$  for some $\tau_{ij}$ and all $t$. For all $i$, $\sum_{j = 1}^N \tau_{ij} \leq M$.
			\item $\+E[e_{it}e_{js}] = \tau_{ij,ts}$ and $\sum_{j=1}^N \sum_{s = 1}^T |\tau_{ij,ts}| \leq M$ for all $i$ and $t$.
			\item For all $i$ and $j$, $\+E \left\vert \frac{1}{|\tlq_{ij}|^{1/2}} \sum_{t \in \tlq_{ij}} \Lp e_{it}e_{jt} -  \+E[e_{it}e_{jt}] \Rp \right\vert^4 \leq M$.
		\end{enumerate}
		\item \label{ass:factor-error} Weak dependence between factor and idiosyncratic errors: for every $(i,j)$,
		\[\+E \norm{\frac{1}{\sqrt{|\tlq_{ij}|}} \sum_{i \in \tlq_{ij}} F_t e_{it} }^4 \leq M. \]
		\item \label{ass:eigen} Eigenvalues: The eigenvalues of $\Sigma_{\Lambda} \Sigma_F$ are distinct.
	\end{enumerate}
\end{assumpG}

\begin{assumpG}[Moments and Central Limit Theorems]\label{ass:mom-clt}
	\texttt{}
	For all $N$ and $T$, 
	\begin{enumerate}[wide=0pt, widest=99,leftmargin=\parindent, labelsep=*]
		\item $\textstyle \+E\Ls \norm{\sqrt{\frac{T}{N}}  \sum_{i=1}^N \frac{1}{|\tlq_{ij}|} \sum_{s \in \tlq_{ij}} \phi_{i,st} \Lp e_{is}e_{js} - \+E[e_{is}e_{js} ] \Rp }^2\Rs \leq M$, where $\phi_{i,st} = W_{it}  F_s$, $ \Lambda_i$, $W_{it} \Lambda_i$,  for every $j$ and $t$.
		\item $\textstyle \+E \Ls \norm{\sqrt{\frac{T}{N}} \sum_{i=1}^N\frac{\phi_{it} }{|\tlq_{ij}|} \sum_{t \in \tlq_{ij}} F_t^\T e_{it} }^2 \Rs \leq M$ for every $t$ and for  $\phi_{it} = \Lambda_i$, $W_{it}$, $W_{it} \Lambda_i$. 
		\item \label{ass:asy-normal-main-term-thm-loading} 
		$\frac{\sqrt{T}}{N} \sum_{i = 1}^N  \Lambda_i \Lambda_i^\T  \frac{1}{|\tlq_{ij}|} \sum_{t \in \tlq_{ij}} F_t e_{it} \xrightarrow{d} \calN (0, \covI_{\Lambda,j}) $ for every $j$.
		\item \label{ass:asy-normal-main-term-thm-factor} 
		$\frac{1}{\sqrt{N}} \sum_{i = 1}^N W_{it} \Lambda_i  e_{it} \xrightarrow{d} \calN(0, \covI_{F,t})$ for every $t$.
		\item \label{ass:asy-normal-add-term-thm-loading} We define the filtration $\mathcal{G}^t = \sigma(\cup_{s= 1}^T \mathcal{G}^t_{Ts} )$ with $\mathcal{G}^t_{Ts} = \sigma(\{W_{ij}, j \leq s,  \text{all } i \}, \Lambda, v_t)$ generated by $\{W_{ij}, j \leq s,  \text{all } i \}$, $\Lambda$ and $v_t$, which is given by $v_t  = \Sigma_{\Lambda}^\I \Sigma_{F}^\I F_t$. For every $i$ and $t$, and $u_i = \Lambda_i$, it holds   
				\begin{align*}
					\sqrt{T} \begin{bmatrix}
						X_{i} u_i \\ \mathbf{X}_t v_t 
					\end{bmatrix} \rightarrow \calN \Bigg( 0, \begin{bmatrix}
						h_i(u_i) & g^{\cov}_{i,t}(u_i, v_t)^\T \\  g^{\cov}_{i,t}(u_i, v_t) & g_t(v_t ) 
					\end{bmatrix} \Bigg)
					\quad \mathcal{G}^t-\text{stably},    
				\end{align*}
				where  $X_{i} = \frac{1}{N} \sum_{l = 1}^N \Lambda_l \Lambda_l^\T  \Big( \frac{1}{|\tlq_{li}|} \sum_{s \in \tlq_{li}} F_s F_s^\T - \frac{1}{T} \sum_{s = 1}^T F_s F_s^\T \Big)$ and $\mathbf{X}_{t} = \frac{1}{N} \sum_{i = 1}^N W_{it} X_i \Lambda_i \Lambda_i^\T$.
			\item \label{ass:add-mom-three-sum-lam-err} $\+E\Ls \norm{ \sqrt{\frac{T}{N}} \sum_{i = 1}^N    \Lp  \frac{1}{|\tlq_{li}|} \sum_{s \in \tlq_{li}} F_s F_s^\T - \frac{1}{T} \sum_{s=1}^T F_s F_s^\T \Rp \Lambda_i W_{it} e_{it} }^2 \Rs\leq M$ for every $l$.
		\end{enumerate}
	\end{assumpG}
	
	\begin{assumpG}[Additional Assumptions on Factor Model]\label{ass:add-factor}
		As $\Ttreat \rightarrow \infty$, it holds
		\begin{enumerate}[wide=0pt, widest=99,leftmargin=\parindent, labelsep=*]
			\item 
			$\frac{1}{\sqrt{\Ttreat}} \sum_{\Tcontrol+1}^{T} F_t e_{it} \xrightarrow{d} \calN (0, \Sigma_{F,e_i})$.
			\item $ \+E \Ls \norm{\frac{1}{\sqrt{N \Ttreat}}  \sum_{t = \Tcontrol+1}^T \sum_{j = 1}^N W_{jt} \Lambda_j  e_{jt} }^2 \Rs \leq M$ and \\  $ \+E \Ls \norm{\frac{1}{\sqrt{N \Ttreat}}   \sum_{t = \Tcontrol+1}^T \sum_{j = 1}^N Z_t F_t^\T W_{jt} \Lambda_j  e_{jt} }^2 \Rs \leq M$ for every $i$, $Z \in \+R^{\Ttreat \times L }$ and $\norm{Z_t} \leq M$.
			\item Assumption \ref{ass:mom-clt}.\ref{ass:asy-normal-add-term-thm-loading} holds for $v_t$ equal to $\Sigma_{\Lambda,t}^\I \Lambda^\treat_i$ and $\Sigma_{\Lambda,t}^\I  (\Lambda^\treat_i- \Lambda^\control_i)$ under the filtration $\mathcal{G} = \sigma(\cup_{s= 1}^T \mathcal{G}_{Ts} )$ with  $\mathcal{G}_{Ts} = \sigma(\{W_{ij}, j \leq s,  \text{all } i \}, \Lambda)$ generated by $\{W_{ij}, j \leq s,  \text{all } i \}$ and $\Lambda$.
				\item For $v_t$ equal to $\Sigma_{\Lambda,t}^\I \Lambda^\treat_i$, $\Sigma_{\Lambda,t}^\I  (\Lambda^\treat_i- \Lambda^\control_i)$ or $\Sigma_{\Lambda}^\I \Sigma_{F}^\I F_t$, the elements of the random vector $\*X_{\Tcontrol+1} v_{\Tcontrol+1} , \cdots, \*X_T v_{T}$ are jointly $\mathcal{G}^t$-stably normal with $\mathrm{ACov}(\*X_t v_t, \*X_s v_s) = g^{\cov}_{t,s}(v_t, v_s)$ for all $\Tcontrol \leq t, s \leq T$, where $\mathcal{G}^t$ is defined as $\mathcal{G}^t = \sigma(\cup_{s= 1}^T \mathcal{G}^t_{Ts} )$, $\mathcal{G}^t_{Ts} = \sigma(\{W_{ij}, j \leq s,  \text{all } i \}, \Lambda)$ if $v_t \neq\Sigma_{\Lambda}^\I \Sigma_{F}^\I F_t  $, and  $\mathcal{G}^t_{Ts} = \sigma(\{W_{ij}, j \leq s,  \text{all } i \}, \Lambda, v_t)$ otherwise.
		\end{enumerate}
	\end{assumpG}

	\begin{assumpGC}[Conditional Factor Model] \label{ass:factor-model-conditional}
		\texttt{}
		\begin{enumerate}[wide=0pt, widest=99,leftmargin=\parindent, labelsep=*]
			\item \label{ass:loading-conditional} Factor loadings: $\+E[\norm{\Lambda_i}^4|S]  \leq \bar{\Lambda} < \infty$.
			There exists some positive definite $r\times r$ matrix $\Sigma_\Lambda$ such that $\frac{1}{N} \sum_{i = 1}^N \frac{W_{it}}{P(W_{it}=1|S_i)} \Lambda_i \Lambda_i^\T \xrightarrow{P} \Sigma_{\Lambda}  $ and  $\+E \norm{\sqrt{N}\Lp \frac{1}{N}  \sum_{i = 1}^N \frac{1}{P(W_{it}=1|S_i)} W_{it}\Lambda_i \Lambda_i^\T - \Sigma_{\Lambda} \Rp }  \leq M$.
		\end{enumerate}
	\end{assumpGC}
	
	\begin{assumpGC}[Conditional Moments and Central Limit Theorems]\label{ass:mom-clt-conditional}
		\texttt{}
		$S$ is independent of $F$ and $e$ and $\+E[\norm{\Lambda_i}^6|S] \leq \bar{\Lambda}$. 
		For all $N$ and $T$,
		\begin{enumerate}[wide=0pt, widest=99,leftmargin=\parindent, labelsep=*]
			\item $\textstyle \+E\Ls \norm{\sqrt{\frac{T}{N}}  \sum_{i=1}^N \frac{1}{|\tlq_{ij}|} \sum_{s \in \tlq_{ij}}  \phi_{i,st} \Lp e_{is}e_{js} - \+E[e_{is}e_{js} ] \Rp }^2  \Rs \leq M$ for every $j$ and $t$, \\ where $\phi_{i,st} = \frac{W_{it}F_s}{\ps }$, $ \Lambda_i$,  $\frac{W_{it}}{\ps}  \Lambda_i$.
			\item $\textstyle \+E \Ls \norm{\sqrt{\frac{T}{N}} \sum_{i=1}^N\frac{\phi_{it} }{|\tlq_{ij}|} \sum_{t \in \tlq_{ij}} F_t^\T e_{it} }^2 \Rs \leq M$ for every $t$ and for  $\phi_{it} = \Lambda_i$, $\frac{W_{it}}{\ps}  \Lambda_i$. 
			\item \label{ass:asy-normal-main-term-thm-loading-conditional} 
			$\frac{\sqrt{T}}{N} \sum_{i = 1}^N  \Lambda_i \Lambda_i^\T  \frac{1}{|\tlq_{ij}|} \sum_{t \in \tlq_{ij}} F_t e_{it} \xrightarrow{d} \calN (0, \covI_{\Lambda,j}) $ for every $j$.
			\item \label{ass:asy-normal-main-term-thm-factor-conditional} 
			$\frac{1}{\sqrt{N}} \sum_{i = 1}^N \frac{W_{it}}{P(W_{it} = 1|S_i)} \Lambda_i  e_{it} \xrightarrow{d} \calN(0, \Gamma^{\obs, S}_{F,t})$ for every $t$.
			\item \label{ass:asy-normal-add-term-thm-loading-conditional}
			We define the filtration $\mathcal{G}^t = \sigma(\cup_{s= 1}^T \mathcal{G}^t_{Ts} )$ with $\mathcal{G}^t_{Ts} = \sigma(\{W_{ij}, j \leq s,  \text{all } i \}, \Lambda, v_t)$ generated by $\{W_{ij}, j \leq s,  \text{all } i \}$, $\Lambda$ and $v_t$, which is given by $v_t  = \Sigma_{\Lambda}^\I \Sigma_{F}^\I F_t$. For every $i$ and $t$, and $u_i = \Lambda_i$, it holds
				\begin{align*}
					\sqrt{T} \begin{bmatrix}
						X_{i} u_i \\ \mathbf{X}_t^S v_t
					\end{bmatrix} \rightarrow \calN \Bigg( 0, \begin{bmatrix}
						h_i(u_i) & g^{\cov,S}_{i, t}(u_i,v_t)^\T  \\g^{\cov,S}_{i, t}(u_i,v_t)  & g^S_t(v_t)
					\end{bmatrix} \Bigg)
					\quad \mathcal{G}^t-\text{stably},    
				\end{align*}
				where  $X_{i} = \frac{1}{N} \sum_{l = 1}^N \Lambda_l \Lambda_l^\T  \Big( \frac{1}{|\tlq_{li}|} \sum_{s \in \tlq_{li}} F_s F_s^\T - \frac{1}{T} \sum_{s = 1}^T F_s F_s^\T \Big) $ and $\mathbf{X}_{t} = \frac{1}{N} \sum_{i = 1}^N \frac{W_{it}}{\ps}  X_i \Lambda_i \Lambda_i^\T $.
			\item \label{ass:add-mom-three-sum-lam-err-conditional} $\+E\Ls \norm{ \sqrt{\frac{T}{N}} \sum_{i = 1}^N    \Lp  \frac{1}{|\tlq_{li}|} \sum_{s \in \tlq_{li}} F_s F_s^\T - \frac{1}{T} \sum_{s=1}^T F_s F_s^\T \Rp \frac{W_{it}}{P(W_{it} = 1|S_i)} \Lambda_i   e_{it}}^2  \Rs\leq M$ for every $l$.
		\end{enumerate}
	\end{assumpGC}

	\begin{assumpGC}[Additional Assumptions on Factor Model]\label{ass:add-factor-conditional}
		As $\Ttreat \rightarrow \infty$, it holds
		\begin{enumerate}[wide=0pt, widest=99,leftmargin=\parindent, labelsep=*]
			\item 
			$\frac{1}{\sqrt{\Ttreat}} \sum_{\Tcontrol+1}^{T} F_t e_{it} \xrightarrow{d} \calN (0, \Sigma_{F,e_i})$.
			\item $ \+E \Ls \norm{\frac{1}{\sqrt{N \Ttreat}}  \sum_{t = \Tcontrol+1}^T \sum_{j = 1}^N \frac{W_{jt}}{\ps}  \Lambda_j  e_{jt} }^2 \Rs \leq M$ and \\  $ \+E \Ls \norm{\frac{1}{\sqrt{N \Ttreat}}   \sum_{t = \Tcontrol+1}^T \sum_{j = 1}^N Z_t F_t^\T \frac{W_{jt}}{\ps} \Lambda_j  e_{jt} }^2 \Rs \leq M$ for every $i$, $Z \in \+R^{\Ttreat \times L }$ and $\norm{Z_t} \leq M$.
			\item Assumption \ref{ass:mom-clt-conditional}.\ref{ass:asy-normal-add-term-thm-loading-conditional} holds for $v_t$ equal to $\Sigma_{\Lambda,t}^\I \Lambda^\treat_i$ and $\Sigma_{\Lambda,t}^\I  (\Lambda^\treat_i- \Lambda^\control_i)$ under the filtration $\mathcal{G} = \sigma(\cup_{s= 1}^T \mathcal{G}_{Ts} )$ with  $\mathcal{G}_{Ts} = \sigma(\{W_{ij}, j \leq s,  \text{all } i \}, \Lambda)$ generated by $\{W_{ij}, j \leq s,  \text{all } i \}$ and $\Lambda$. 
				\item For $v_t$ equal to $\Sigma_{\Lambda,t}^\I \Lambda^\treat_i$, $\Sigma_{\Lambda,t}^\I  (\Lambda^\treat_i- \Lambda^\control_i)$ or $\Sigma_{\Lambda}^\I \Sigma_{F}^\I F_t$, the elements of the random vector $\*X^S_{\Tcontrol+1} v_{\Tcontrol+1} , \cdots, \*X^S_T v_{T}$ are jointly $\mathcal{G}^t$-stably normal with $\mathrm{ACov}(\*X^S_t v_t, \*X^S_s v_s) = g^{S}_{t,s}(v_t, v_s)$ for all $\Tcontrol \leq t, s \leq T$, where $\mathcal{G}^t$ is defined as $\mathcal{G}^t = \sigma(\cup_{s= 1}^T \mathcal{G}^t_{Ts} )$, $\mathcal{G}^t_{Ts} = \sigma(\{W_{ij}, j \leq s,  \text{all } i \}, \Lambda)$ if $v_t \neq\Sigma_{\Lambda}^\I \Sigma_{F}^\I F_t  $, and  $\mathcal{G}^t_{Ts} = \sigma(\{W_{ij}, j \leq s,  \text{all } i \}, \Lambda, v_t)$ otherwise. 
		\end{enumerate}
	\end{assumpGC}

	Assumption \ref{ass:factor-model} describes an approximate factor structure and is at a similar level of generality as \cite{bai2003inferential}: 
	(1) Assumption \ref{ass:factor-model}.\ref{ass:factor}  ensures that each factor has a nontrivial contribution to the variation in $X$.
	(2) We assume loadings are random but independent of factors and errors in Assumption \ref{ass:factor-model}.\ref{ass:loading}. We could study a factor model conditioned on some particular realization of the loadings, and the analysis would essentially be equivalent to that under the assumption that loadings are nonrandom.
	(3) Assumption \ref{ass:factor-model}.\ref{ass:error} allows errors to be time-series and cross-sectionally weakly correlated.
	(4) Assumption \ref{ass:factor-model}.\ref{ass:factor-error} allows factors and idiosyncratic errors to be weakly correlated.
	(5) Assumption \ref{ass:factor-model}.\ref{ass:eigen} guarantees that each loading and factor can be uniquely identified up to some rotation matrix.
	Additionally, we assume that these aspects also hold if we look at a subset of all time periods (the subset is denoted as $\tlq_{ij}$ in Assumption \ref{ass:factor-model}). Together with Assumption \ref{ass:obs}.2, our covariance matrix estimator \eqref{eqn:cov-est} using incomplete observations has similar properties as the conventional covariance matrix estimator $\frac{1}{T} X X^\T$ using full observations. For example, both  $ \frac{1}{|\tlq_{ij}|} \sum_{t \in \tlq_{ij}} X_{it} X_{jt}$ and $\frac{1}{T} \sum_{t=1}^T X_{it} X_{jt}$ are consistent estimators for  $\Sigma_{ij}$. Moreover, the top $r$ eigenvalues estimated from both matrices are consistent as shown in Lemma \ref{lemma:eigenvalue} in the Internet Appendix, which is the foundation for developing the inferential theory of the factor model estimated from Equation \eqref{eqn:cov-est}.
	
	Assumption \ref{ass:mom-clt} is not required to show the consistency of loadings and factors but is only used to show the asymptotic normality of the estimators. Assumption \ref{ass:mom-clt}.1-4 are closely related to the moment and CLT assumptions in \cite{bai2003inferential}. The first two parts in Assumptions \ref{ass:mom-clt} restrict the second moments of certain averages. The 3rd and 4th points state the necessary central limit theorems. $\frac{\sqrt{T}}{N} \sum_{i = 1}^N  \Lambda_i \Lambda_i^\T  \frac{1}{|\tlq_{ij}|} \sum_{t \in \tlq_{ij}} F_t e_{it} \xrightarrow{d} \calN(0, \Phi_{j})$ is one of the leading terms in the asymptotic distribution of the estimated loadings $\tilde \Lambda_j$. However, $\frac{1}{|\tlq_{ij}|} \sum_{t \in \tlq_{ij}} F_t e_{it}$ varies with $j$ so we cannot separately average over the cross-sectional and time dimension as in the conventional framework. Point 5 is specific to the missing value problem and introduces the correction terms that appear in the asymptotic distribution. They are due to the fact that our estimator averages over different number of observations for different entries in the covariance matrix. 
	Assumption \ref{ass:mom-clt}.\ref{ass:asy-normal-add-term-thm-loading} assumes a central limit theorem for $X_i$ and $\mathbf{X}_t$. The usual CLT of the form 
		\begin{align*}
			\sqrt{T} \begin{bmatrix}
				\tvec(X_i) \\ \tvec(\mathbf{X}_t)
			\end{bmatrix} \xrightarrow{d} \calN \Bigg( 0, \begin{bmatrix}
				\Phi_i & (\mathbf{\Phi}^{\cov}_{i,t})^\T  \\ \mathbf{\Phi}^{\cov}_{i,t}  & \mathbf{\Phi}_t
			\end{bmatrix} \Bigg),    
		\end{align*}
		is not sufficient as $X_i$ and $\mathbf{X}_t$ are multiplied with the random variables $u_i$ and $v_t$ in $\sqrt{T} \begin{bmatrix}
			(X_{i} u_i)^\T & (\mathbf{X}_t^S v_t)^\T
		\end{bmatrix}$. The asymptotic variances of these products are quadratic functions in the elements of those random variables given by $h_i(u_i)$ and $g_t(v_t)$ and take the form of $h_i(u_i) = (u_i^\T \otimes I_r) \Phi_i (u_i \otimes I_r)$ and $g_t(v_t) = (v_t^\T \otimes I_r) \mathbf{\Phi}_t (v_t \otimes I_r)$ respectively.

Assumption \ref{ass:mom-clt}.\ref{ass:asy-normal-add-term-thm-loading} requires a central limit theorem for stable convergence in law which is stronger than the conventional central limit theorem for convergence in distribution. The reason is that the asymptotic variance in Assumption \ref{ass:mom-clt}.\ref{ass:asy-normal-add-term-thm-loading} depends on both $\Lambda_i$ and $F_t$, which are random variables. Hence, we deal with a mixed normal limit and stable convergence in law ensures that the normal distribution of the central limit theorem will be independent of $\Lambda_i$ and $F_t$. Because of the stable convergence in law result, the estimated factors and common components normalized by their random standard deviation will converge to a standard normal distribution. In more detail, Assumption \ref{ass:mom-clt}.\ref{ass:asy-normal-add-term-thm-loading} implies that $X_i$ and $\mathbf{X}_t$ jointly converge $\mathcal{G}^t$-stably for $(N, T) \rightarrow \infty$ to a mixed normal distribution, whose asymptotic variance is random but measurable with respect to the sigma-field $\mathcal{G}^t$. Assumption \ref{ass:mom-clt}.\ref{ass:asy-normal-add-term-thm-loading} is used in Theorem \ref{theorem:asy-normal-equal-weight} to show the asymptotic distribution of the variance correction term whose asymptotic variance is random. Our simplified factor model specified by Assumption \ref{ass:simple-factor-model} is sufficient to guarantee a central limit theorem for stable convergence in law. Proposition \ref{prop:simple-assump-imply-general-assump} shows that the simplified model implies Assumption \ref{ass:mom-clt}.\ref{ass:asy-normal-add-term-thm-loading}.

Assumptions \ref{ass:factor-model-conditional} and \ref{ass:mom-clt-conditional} are the corresponding assumptions for the propensity-weighted estimator with a similar level of generality. The additional Assumptions \ref{ass:add-factor} and \ref{ass:add-factor-conditional} are only needed for the treatment effect tests. The simplified assumptions imply the general assumptions as stated in Proposition \ref{prop:simple-assump-imply-general-assump}.

	\begin{proposition}\label{prop:simple-assump-imply-general-assump}The simplified model is a special case of the general model:
		\begin{enumerate}
			\item  Assumptions \ref{ass:factor-model} and \ref{ass:mom-clt} are satisfied in the simplified model:
			\vspace{-0.2cm}
			\begin{enumerate}[wide=0pt, widest=99,leftmargin=\parindent, labelsep=*]
				\item Assumptions \ref{ass:obs-equal-weight} and \ref{ass:simple-factor-model} imply Assumption \ref{ass:factor-model}.
				\item Assumptions \ref{ass:obs-equal-weight}, \ref{ass:simple-factor-model} and \ref{ass:simple-moment} imply Assumption \ref{ass:mom-clt}.
			\end{enumerate}
				\item Assumptions \ref{ass:factor-model-conditional} and \ref{ass:mom-clt-conditional} are satisfied in the simplified conditional model:
				\begin{enumerate}[wide=0pt, widest=99,leftmargin=\parindent, labelsep=*]
					\item Assumptions \ref{ass:obs-equal-weight}, \ref{ass:obs}, \ref{ass:simple-factor-model} and \ref{ass:simple-factor-model-conditional} imply  \ref{ass:factor-model-conditional}.
					\item Assumptions \ref{ass:obs-equal-weight}, \ref{ass:obs}, \ref{ass:simple-factor-model}, \ref{ass:simple-moment}.2, \ref{ass:simple-factor-model-conditional} and \ref{ass:simple-moment-conditional} imply Assumption \ref{ass:mom-clt-conditional}.
				\end{enumerate}
					\item	Assumptions \ref{ass:add-factor} and \ref{ass:add-factor-conditional} are satisfied in the simplified model.
					Specifically, 
					\begin{enumerate}
						\item Assumptions \ref{ass:obs-equal-weight}, \ref{ass:simple-factor-model} and \ref{ass:simple-moment}  imply Assumption \ref{ass:add-factor}. 
						\item Assumptions \ref{ass:obs-equal-weight}, \ref{ass:obs}, \ref{ass:simple-factor-model}, \ref{ass:simple-moment}.2, \ref{ass:simple-factor-model-conditional} and \ref{ass:simple-moment-conditional} imply Assumption \ref{ass:add-factor-conditional}.
					\end{enumerate}
				\end{enumerate}
			\end{proposition}



	\singlespacing
	\bibliographystyle{econometrica}
	\let\oldbibliography\thebibliography
	\renewcommand{\thebibliography}[1]{%
		\oldbibliography{#1}%
		\setlength{\itemsep}{4pt}%
	}
	{\footnotesize
		\bibliography{reference}
	}
	
	\onehalfspacing

	\newpage
	
	\onehalfspacing
	\begin{titlepage}
	
		\centering{\huge Internet Appendix to Large Dimensional Latent Factor Modeling with Missing Observations and Applications to Causal Inference}
		
		\thispagestyle{empty}
		
		\vspace{2cm}
		

		\begin{abstract}
			
			The Internet Appendix collects the proofs and additional results that support the main text. We show in simulations that our estimators perform well relative to alternative estimators and can be improved even further with an iterative approach. We also confirm that the distribution results, statistical power and robustness to misspecification hold under a variety of simulation setups. Lastly, we collect the detailed proofs for all the theoretical statements. 
			
			\vspace{1cm}
			
			\noindent\textbf{Keywords:} Factor Analysis, Principal Components, Synthetic Control, Causal Inference, Treatment Effect, Missing Entry, Large-Dimensional Panel Data, Large $N$ and $T$, Matrix Completion

			\noindent\textbf{JEL classification:} C14, C38, C55, G12
		\end{abstract}
	\end{titlepage}
	
	\onehalfspacing
	\setcounter{section}{0}
	

	\section{Simulation Results}

	\subsection{Robustness to Missing Patterns with Iterative Estimation}

	\begin{table}[H]
		\centering
		\tcaptab{Comparison with \cite{jin2020factor} and \cite{bai2019matrix} with One Iteration}
		\begin{minipage}{0.17\linewidth}
			\vspace{0.7cm}
			\flushright
			\includegraphics[width=15mm]{plots/obs_pattern/obs_pattern.jpg}
		\end{minipage}\hfill
		\begin{minipage}{0.8\linewidth}
			\begin{tabular}{l|l|cccc}
				\toprule
				Observation Pattern & $W_{it}$ &  XP  &  $\text{XP}_{\text{PROP}}$ &  JMS &      BN \\
				\midrule			
				Random & obs &   {\bf 0.011} & {\bf 0.011} & {\bf 0.011} & -- \\
				& miss &   {\bf 0.012} & {\bf 0.012} & {\bf 0.012} & -- \\
				& all &   {\bf 0.011} & {\bf 0.011} & {\bf 0.012} & -- \\ \midrule
				Simultaneous & obs &   {\bf 0.011} & {\bf 0.011} & 0.021 &  {\bf 0.011} \\
				& miss &   { 0.019} & { 0.019} & 0.092 &  {\bf 0.017} \\
				& all &   { 0.013} & 0.013 & 0.039 &  {\bf 0.012} \\ \midrule
				Staggered & obs &   {\bf 0.015} & {\bf 0.015} & 0.036 &  0.022 \\
				& miss &   {\bf 0.039} & {\bf 0.039} & 0.220 &  0.050 \\
				& all &   {\bf 0.024} & {\bf 0.024} & 0.111 &  0.033 \\ \midrule
				Random  & obs &   {\bf 0.014} & {\bf 0.014} & 0.020 & -- \\
				$W$ depends on $S$& miss &   {\bf 0.017} & {\bf 0.017} & 0.026 & -- \\
				& all &   {\bf 0.015} & {\bf 0.015} & 0.022 & -- \\ \midrule
				Simultaneous & obs &   {\bf 0.019} & 0.023 & 0.028 &  0.059 \\
				$W$ depends on $S$& miss &   0.224 & 0.239 & 0.483 &  {\bf 0.151} \\ 
				& all &   0.118 & 0.128 & 0.251 &  {\bf 0.104} \\ \midrule
				Staggered  & obs &   {\bf 0.014} & 0.015 & 0.032 &  0.029 \\
				$W$ depends on $S$& miss &   {\bf 0.058} & 0.061 & 0.250 &  0.097 \\
				& all &   {\bf 0.030} & 0.031 & 0.111 &  0.053 \\
				\bottomrule
			\end{tabular}
		\end{minipage}
		\bnotetab{This table reports the relative MSE of XP (our benchmark method), $\textnormal{XP}_{\textnormal{PROP}}$ (our propensity-weighted method), JMS \citep{jin2020factor}, and BN \citep{bai2019matrix} on observed, missing and all entries, $N = 250$, $T = 250$. The figures on the left show patterns of missing observations with the shaded entries indicating the missing entries. Bold numbers indicate the best relative model performance. The data is generated as in Table \ref{tab:comparison-jms-bn}. We first impute the missing values with the different methods. In a second step we apply PCA to the full panel with imputed values to estimate the factor model and update the imputed values with the common components. The observed entries stay the same. This process is repeated for multiple iterations. Here we consider one iteration.}
		\label{tab:comparison-jms-bn-iter1}
	\end{table}

	\begin{table}[H]
		\centering
		\tcaptab{Comparison with \cite{jin2020factor} and \cite{bai2019matrix} with Two Iterations}
		\begin{minipage}{0.17\linewidth}
			\vspace{0.7cm}
			\flushright
			\includegraphics[width=15mm]{plots/obs_pattern/obs_pattern.jpg}
		\end{minipage}\hfill
		\begin{minipage}{0.8\linewidth}
			\begin{tabular}{l|l|cccc}
				\toprule
				Observation Pattern & $W_{it}$ &  XP  &  $\text{XP}_{\text{PROP}}$ &  JMS &      BN \\
				\midrule			
				Random & obs &   {\bf 0.011} & {\bf 0.011} & {\bf 0.011} &  -- \\
				& miss &   {\bf 0.011} & {\bf 0.011} & {\bf 0.011} &  -- \\
				& all &   {\bf 0.011} & {\bf 0.011} & {\bf 0.011} &  -- \\ \midrule
				Simultaneous & obs &   {\bf 0.011} & {\bf 0.011} & 0.015 &  {\bf 0.011} \\
				& miss &   0.018 & 0.018 & 0.053 &  {\bf 0.017} \\
				& all &   0.013 & 0.013 & 0.025 &  {\bf 0.012} \\ \midrule
				Staggered & obs &   {\bf 0.014} & {\bf 0.014} & 0.028 &  0.016 \\
				& miss &   {\bf 0.036} & {\bf 0.036} & 0.163 &  0.040 \\
				& all &   {\bf 0.023} & {\bf 0.023} & 0.083 &  0.026 \\ \midrule
				Random  & obs &   {\bf 0.013} & {\bf 0.013} & 0.014 & -- \\
				$W$ depends on $S$& miss &   {\bf 0.015} & {\bf 0.015} & 0.018 & -- \\
				& all &   {\bf 0.014} & {\bf 0.014} & 0.016 & -- \\ \midrule
				Simultaneous & obs &   {\bf 0.017} & 0.020 & 0.025 &  0.050 \\
				$W$ depends on $S$& miss &   0.220 & 0.230 & 0.456 &  {\bf 0.128} \\
				& all &   0.116 & 0.122 & 0.236 &  {\bf 0.088} \\ \midrule
				Staggered & obs &   {\bf 0.014} & {\bf 0.014} & 0.024 &  0.019 \\
				$W$ depends on $S$& miss &   {\bf 0.054} & 0.056 & 0.195 &  0.072 \\
				& all &   {\bf 0.028} & 0.029 & 0.086 &  0.038 \\
				\bottomrule
			\end{tabular}
		\end{minipage}
		\bnotetab{This table reports the relative MSE of XP (our benchmark method), $\textnormal{XP}_{\textnormal{PROP}}$ (our propensity-weighted method), JMS \citep{jin2020factor}, and BN \citep{bai2019matrix} on observed, missing and all entries, $N = 250$, $T = 250$. The figures on the left show patterns of missing observations with the shaded entries indicating the missing entries. Bold numbers indicate the best relative model performance. The data is generated as in Table \ref{tab:comparison-jms-bn}.  We first impute the missing values with the different methods. In a second step we apply PCA to the full panel with imputed values to estimate the factor model and update the imputed values with the common components. The observed entries stay the same. This process is repeated for multiple iterations. Here we consider two iterations.}
		
		\label{tab:comparison-jms-bn-iter2}
	\end{table}

	\begin{table}[H]
		\centering
		\tcaptab{Comparison with \cite{jin2020factor} and \cite{bai2019matrix} with Three Iterations}
		\begin{minipage}{0.17\linewidth}
			\vspace{0.7cm}
			\flushright
			\includegraphics[width=15mm]{plots/obs_pattern/obs_pattern.jpg}
		\end{minipage}\hfill
		\begin{minipage}{0.8\linewidth}
			\begin{tabular}{l|l|cccc}
				\toprule
				Observation Pattern & $W_{it}$ &  XP  &  $\text{XP}_{\text{PROP}}$ &  JMS &      BN \\
				\midrule			
				Random & obs &   {\bf 0.011} & {\bf 0.011} & {\bf 0.011} &  -- \\
				& miss &   {\bf 0.011} & {\bf 0.011} & {\bf 0.011} &  -- \\
				& all &   {\bf 0.011} & {\bf 0.011} & {\bf 0.011} &  -- \\ \midrule
				Simultaneous & obs &   {\bf 0.011} & {\bf 0.011} & 0.013 &  {\bf 0.011} \\
				& miss &   0.018 & 0.018 & 0.035 &  {\bf 0.017} \\
				& all &   0.013 & 0.013 & 0.019 &  {\bf 0.012} \\ \midrule
				Staggered & obs &   {\bf 0.014} & {\bf 0.014} & 0.023 &  0.015 \\
				& miss &   {\bf 0.035} & {\bf 0.035} & 0.126 &  0.036 \\
				& all &   {\bf 0.022} & {\bf 0.022} & 0.065 &  0.023 \\ \midrule
				Random  & obs &   {\bf 0.013} & {\bf 0.013} & {\bf 0.013} & -- \\
				$W$ depends on $S$ & miss &   {\bf 0.014} & {\bf 0.014} & 0.015 & -- \\
				& all &   {\bf 0.013} & {\bf 0.013} & 0.014 & -- \\ \midrule
				Simultaneous & obs &   {\bf 0.017} & 0.018 & 0.023 &  0.048 \\
				$W$ depends on $S$ & miss &   0.217 & 0.224 & 0.435 &  {\bf 0.122} \\
				& all &   0.114 & 0.119 & 0.224 &  {\bf 0.084} \\ \midrule
				Staggered & obs &   {\bf 0.014} & {\bf 0.014} & 0.021 &  0.016 \\
				$W$ depends on $S$ & miss &   {\bf 0.051} & 0.052 & 0.159 &  0.059 \\
				& all &   {\bf 0.027} & {\bf 0.027} & 0.070 &  0.032 \\
				\bottomrule
			\end{tabular}
		\end{minipage}
		\bnotetab{This table reports the relative MSE of XP (our benchmark method), $\textnormal{XP}_{\textnormal{PROP}}$ (our propensity-weighted method), JMS \citep{jin2020factor}, and BN \citep{bai2019matrix} on observed, missing and all entries, $N = 250$, $T = 250$. The figures on the left show patterns of missing observations with the shaded entries indicating the missing entries. Bold numbers indicate the best relative model performance. The data is generated as in Table \ref{tab:comparison-jms-bn}. We first impute the missing values with the different methods. In a second step we apply PCA to the full panel with imputed values to estimate the factor model and update the imputed values with the common components. The observed entries stay the same. This process is repeated for multiple iterations. Here we consider three iterations.}
		\label{tab:comparison-jms-bn-iter3}
	\end{table}

	\begin{table}[H]
		\centering
		\tcaptab{Comparison with \cite{jin2020factor} and \cite{bai2019matrix} ($N = 100, T = 150$)}
		\begin{minipage}{0.17\linewidth}
			\vspace{0.7cm}
			\flushright
			\includegraphics[width=23.5mm]{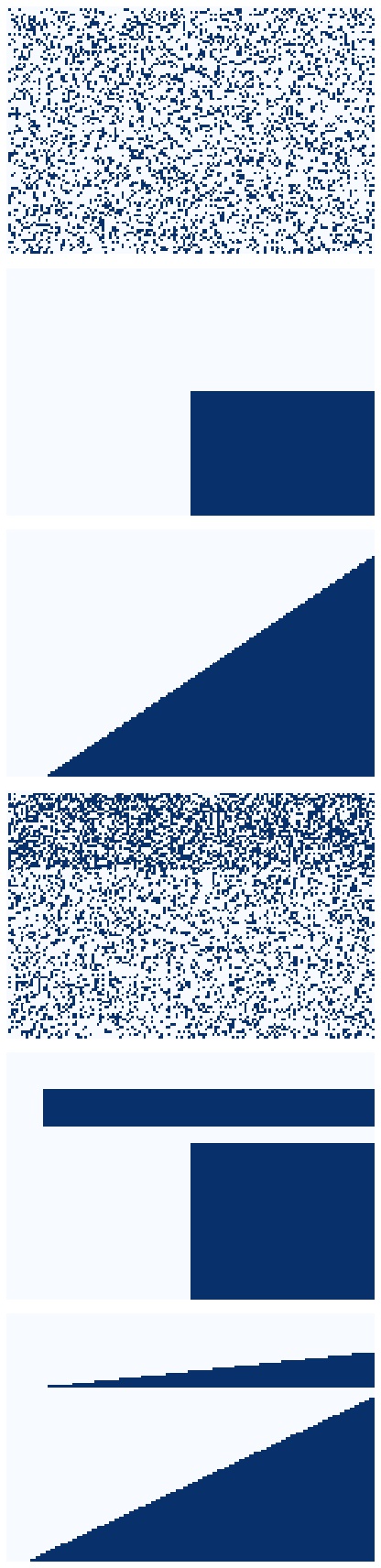}
		\end{minipage}\hfill
		\begin{minipage}{0.8\linewidth}
			\begin{tabular}{l|l|cccc}
				\toprule
				Observation Pattern & $W_{it}$ &  XP  &  $\text{XP}_{\text{PROP}}$ &  JMS &      BN \\
				\midrule			
				Random & obs &   {\bf 0.028} & {\bf 0.028} & 0.046 & -- \\
				& miss &   {\bf 0.031} & {\bf 0.031} & 0.044 & -- \\ 
				& all &   {\bf 0.029} & {\bf 0.029} & 0.046 & -- \\ \midrule
				Simultaneous & obs &   {\bf 0.025} & {\bf 0.025} & 0.142 &  {\bf 0.025} \\
				& miss &   0.042 & 0.042 & 0.202 &  {\bf 0.037} \\
				& all &   0.029 & 0.029 & 0.158 &  {\bf 0.028} \\ \midrule
				Staggered & obs &   {\bf 0.033} & {\bf 0.033} & 0.388 &  0.163 \\
				& miss &   {\bf 0.087} & {\bf 0.087} & 0.341 &  0.205 \\
				& all &   {\bf 0.054} & {\bf 0.054} & 0.370 &  0.179 \\ \midrule
				Random & obs &   {\bf 0.035} &  0.036 & 0.089 & -- \\
				$W$ depends on $S$ & miss &   {\bf 0.043} & {\bf 0.043} & 0.094 & -- \\
				& all &   {\bf 0.038} & {\bf 0.038} & 0.091 & -- \\ \midrule
				Simultaneous  & obs &   {\bf 0.038} & 0.045 & 0.630 &  0.092 \\ 
				$W$ depends on $S$ & miss &   {\bf 0.113} & 0.144 & 0.417 &  0.130 \\
				& all &   {\bf 0.072} & 0.091 & 0.534 &  0.109 \\ \midrule
				Staggered & obs &   {\bf 0.032} & 0.034 & 0.338 &  0.127 \\
				$W$ depends on $S$ 	& miss &   {\bf 0.091} & 0.101 & 0.329 &  0.184 \\
				& all &   {\bf 0.053} & 0.058 & 0.336 &  0.147 \\
				\bottomrule
			\end{tabular}
		\end{minipage}
		\bnotetab{This table reports the relative MSE of XP (our benchmark method), $\textnormal{XP}_{\textnormal{PROP}}$ (our propensity-weighted method), JMS \citep{jin2020factor}, and BN \citep{bai2019matrix} on observed, missing and all entries, $N = 100$, $T = 150$. The figures on the left show patterns of missing observations with the shaded entries indicating the missing entries. Bold numbers indicate the best relative model performance. We generate a two-factor model and a unit-specific characteristic $S_i = \mathbbm{1}(\Lambda_{i,2} \geq -0.5 )$. The observation patterns are generated as follows. (1) {\it Missing uniformly at random}: Entries are observed independently with probability $p = 0.75$. (2) {\it Simultaneous treatment adoption}:  $50\% $ randomly selected units adopt the treatment from time $0.5 \cdot T$ and the remaining $50\% $ units stay in the control group until the end. (3) {\it Staggered treatment adoption}:  All units are in the control group for $t < 0.1\cdot T$. At time $0.1 \cdot T \leq  t \leq T$, $\frac{t-0.1 \cdot T}{T} \%$ units are in the treated group. The remaining $10\%$ units stay in the control group until the end. (4) {\it Missing at random conditional on $S_i$}: Entries are observed independently with probability $p_{it} = 0.75$ $S_i = 1$, and $p_{it} = 0.5$ if $S_i = 0$. (5) {\it Simultaneous treatment adoption conditional on $S_i$}: For the units with $S_i = 1$, $90\% $ units adopt the treatment from time $0.5 \cdot T$ and $10\% $ units stay in the control group until the end. For the units with $S_i= 0$, $50\%$ units adopt the treatment from time $0.02\cdot T$ and  $50\%$ units stay in the control group until the end. (6) {\it Staggered treatment adoption conditional on $S_i$}:  All units are in the control group for $t < 0.05\cdot T$. For the units with $S_i = 1$, at time $0.05 \cdot T \leq  t \leq T$, $\frac{t-0.05 \cdot T}{T} \%$ units are in the treated group with the remaining $5\%$ units staying in the control group until the end. For the units with $S_i = 0$, at time $0.05 \cdot T \leq  t \leq T$, $\frac{t-0.05 \cdot T}{1.96 T} \%$ units are in the treated group with the remaining $50\%$ units to stay in the control group until the end. We run 100 Monte Carlo simulations.}
		\label{tab:comparison-jms-bn-iter0-N100T150}
	\end{table}

	\begin{table}[H]
		\centering
		\tcaptab{Comparison with \cite{jin2020factor} and \cite{bai2019matrix} with One Iteration ($N = 100, T = 150$)}
		\begin{minipage}{0.17\linewidth}
			\vspace{0.7cm}
			\flushright
			\includegraphics[width=23.5mm]{plots/obs_pattern/obs_pattern_N_100_T_150.jpg}
		\end{minipage}\hfill
		\begin{minipage}{0.8\linewidth}
			\begin{tabular}{l|l|cccc}
				\toprule
				Observation Pattern & $W_{it}$ &  XP  &  $\text{XP}_{\text{PROP}}$ &  JMS &      BN \\
				\midrule			
				Random & obs &   {\bf 0.023} & {\bf 0.023} & 0.024 & -- \\
				& miss &   {\bf 0.025} & {\bf 0.025} & 0.026 & -- \\
				& all &   {\bf 0.023} & {\bf 0.023} & 0.024 & -- \\ \midrule
				Simultaneous & obs &   \textbf{0.023} & \textbf{0.023} & 0.032 & \textbf{0.023} \\
				& miss &   {\bf 0.039} & {\bf 0.039} & 0.109 & 0.037 \\
				& all &   0.027 & 0.027 & 0.052 & {\bf 0.026} \\ \midrule
				Staggered & obs &   {\bf 0.030} & {\bf 0.030} & 0.048 & 0.047 \\
				& miss &   {\bf 0.080} & {\bf 0.080} & 0.245 & 0.124 \\
				& all &   {\bf 0.050} & {\bf 0.050} & 0.127 & 0.077 \\ \midrule
				Random & obs &   {\bf 0.027} & {\bf 0.027} & 0.033 & -- \\
				$W$ depends on $S$ & miss &  \textbf{0.032} & \textbf{0.032} & 0.046 & -- \\
				& all &   {\bf 0.029} & {\bf 0.029} & 0.037 & -- \\ \midrule
				Simultaneous & obs &   {\bf 0.034} & 0.038 & 0.055 & 0.043 \\
				$W$ depends on $S$ & miss &   0.107 & 0.128 & 0.328 & {\bf 0.080} \\
				& all &   0.067 & 0.079 & 0.181 & {\bf 0.060} \\ \midrule
				Staggered & obs &   {\bf 0.028} & 0.029 & 0.043 & 0.041 \\ 
				$W$ depends on $S$ & miss &   {\bf 0.083} & 0.089 & 0.235 & 0.106 \\
				& all &   {\bf 0.048} & 0.051 & 0.113 & 0.064 \\
				\bottomrule
			\end{tabular}
		\end{minipage}
		\bnotetab{This table reports the relative MSE of XP (our benchmark method), $\textnormal{XP}_{\textnormal{PROP}}$ (our propensity-weighted method), JMS \citep{jin2020factor}, and BN \citep{bai2019matrix} on observed, missing and all entries, $N = 100$, $T = 150$. The figures on the left show patterns of missing observations with the shaded entries indicating the missing entries. Bold numbers indicate the best relative model performance. The data is generated as in Figure \ref{tab:comparison-jms-bn-iter0-N100T150}. We first impute the missing values with the different methods. In a second step we apply PCA to the full panel with imputed values to estimate the factor model and update the imputed values with the common components. The observed entries stay the same. This process is repeated for multiple iterations. Here we consider one iteration.}
		\label{tab:comparison-jms-bn-iter1-N100T150}
	\end{table}

	\begin{table}[H]
		\centering
		\tcaptab{Comparison with \cite{jin2020factor} and \cite{bai2019matrix} with Two Iterations ($N = 100, T = 150$)}
		\begin{minipage}{0.17\linewidth}
			\vspace{0.7cm}
			\flushright
			\includegraphics[width=23.5mm]{plots/obs_pattern/obs_pattern_N_100_T_150.jpg}
		\end{minipage}\hfill
		\begin{minipage}{0.8\linewidth}
			\begin{tabular}{l|l|cccc}
				\toprule
				Observation Pattern & $W_{it}$ &  XP  &  $\text{XP}_{\text{PROP}}$ &  JMS &      BN \\
				\midrule			
				Random & obs &   {\bf 0.023} & {\bf 0.023} & 0.023 & -- \\
				& miss &   {\bf 0.024} & {\bf 0.024} & 0.024 & -- \\
				& all &   {\bf 0.023} & {\bf 0.023} & 0.023 & --\\ \midrule
				Simultaneous & obs &   {\bf 0.023} & {\bf 0.023} & 0.027 & {\bf 0.023} \\
				& miss &    0.038 &  0.038 & 0.070 & {\bf 0.037} \\
				& all &    0.027 & 0.027 & 0.038 & {\bf 0.026 }\\ \midrule
				Staggered & obs &   {\bf 0.029} & {\bf 0.029} & 0.040 & 0.035 \\
				& miss &   {\bf 0.077} & {\bf 0.077} & 0.189 & 0.099 \\
				& all &   {\bf 0.048} & {\bf 0.048} & 0.100 & 0.061 \\ \midrule
				Random  & obs &   {\bf 0.026} & {\bf 0.026} & 0.027 & -- \\
				$W$ depends on $S$& miss &   {\bf 0.029} & {\bf 0.029} & 0.033 & -- \\
				& all &   {\bf 0.027} & {\bf 0.027} & 0.029 & -- \\ \midrule
				Simultaneous & obs &   {\bf 0.033} & 0.036 & 0.048 & 0.038 \\
				$W$ depends on $S$& miss &   0.102 & 0.118 & 0.271 & {\bf 0.069} \\
				& all &   0.065 & 0.073 & 0.151 & {\bf 0.052} \\ \midrule
				Staggered & obs &   {\bf 0.028} & {\bf 0.028} & 0.036 & 0.032 \\
				$W$ depends on $S$& miss &   {\bf 0.078} & 0.082 & 0.182 & 0.084 \\
				& all &   {\bf 0.046} & 0.047 & 0.089 & 0.051 \\
				\bottomrule
			\end{tabular}
		\end{minipage}
		\bnotetab{This table reports the relative MSE of XP (our benchmark method), $\textnormal{XP}_{\textnormal{PROP}}$ (our propensity-weighted method), JMS \citep{jin2020factor}, and BN \citep{bai2019matrix} on observed, missing and all entries, $N = 100$, $T = 150$. The figures on the left show patterns of missing observations with the shaded entries indicating the missing entries. Bold numbers indicate the best relative model performance. The data is generated as in Figure \ref{tab:comparison-jms-bn-iter0-N100T150}. We first impute the missing values with the different methods. In a second step we apply PCA to the full panel with imputed values to estimate the factor model and update the imputed values with the common components. The observed entries stay the same. This process is repeated for multiple iterations. Here we consider two iterations.}
		\label{tab:comparison-jms-bn-iter2-N100T150}
	\end{table}

	\begin{table}[H]
		\centering
		\tcaptab{Comparison with \cite{jin2020factor} and \cite{bai2019matrix} with Three  Iterations ($N = 100, T = 150$)}
		\begin{minipage}{0.17\linewidth}
			\vspace{0.7cm}
			\flushright
			\includegraphics[width=23.5mm]{plots/obs_pattern/obs_pattern_N_100_T_150.jpg}
		\end{minipage}\hfill
		\begin{minipage}{0.8\linewidth}
			\begin{tabular}{l|l|cccc}
				\toprule
				Observation Pattern & $W_{it}$ &  XP  &  $\text{XP}_{\text{PROP}}$ &  JMS &      BN \\
				\midrule	
				Random & obs &   {\bf 0.023} & {\bf 0.023} & {\bf 0.023} & -- \\
				& miss &   {\bf 0.024} & {\bf 0.024 }& {\bf 0.024} & -- \\
				& all &   {\bf 0.023} & {\bf 0.023} & {\bf 0.023} & -- \\  \midrule
				Simultaneous & obs &   {\bf 0.023} & {\bf 0.023} & 0.024 & {\bf 0.023} \\
				& miss &   0.038 & 0.038 & 0.053 & {\bf 0.037} \\
				& all &   {\bf 0.026} & {\bf 0.026} & 0.031 & {\bf 0.026} \\  \midrule
				Staggered & obs &   {\bf 0.029} & {\bf 0.029} & 0.036 & 0.032 \\
				& miss &   {\bf 0.074} & {\bf 0.074} & 0.153 & 0.088 \\
				& all &   {\bf 0.047} & {\bf 0.047} & 0.083 & 0.054 \\  \midrule
				Random& obs &   {\bf 0.026} & {\bf 0.026} & {\bf 0.026} & -- \\
				$W$ depends on $S$ & miss &   {\bf 0.028} & {\bf 0.028} & 0.029 & -- \\
				& all &   {\bf 0.027} & {\bf 0.027} & {\bf 0.027} & -- \\  \midrule
				Simultaneous & obs &   {\bf 0.033} & 0.034 & 0.043 & 0.037 \\
				$W$ depends on $S$ & miss &   0.099 & 0.110 & 0.229 & {\bf 0.066} \\
				& all &   0.063 & 0.069 & 0.129 & {\bf 0.050} \\  \midrule
				Staggered & obs &   {\bf 0.027} & 0.028 & 0.033 & 0.029 \\
				$W$ depends on $S$ & miss &   {\bf 0.074} & 0.077 & 0.148 & {\bf 0.074} \\
				& all &   {\bf 0.044} & 0.045 & 0.074 & 0.045 \\		
				\bottomrule
			\end{tabular}
		\end{minipage}
		\bnotetab{This table reports the relative MSE of XP (our benchmark method), $\textnormal{XP}_{\textnormal{PROP}}$ (our propensity-weighted method), JMS \citep{jin2020factor}, and BN \citep{bai2019matrix} on observed, missing and all entries, $N = 100$, $T = 150$. The figures on the left show patterns of missing observations with the shaded entries indicating the missing entries. Bold numbers indicate the best relative model performance. The data is generated as in Figure \ref{tab:comparison-jms-bn-iter0-N100T150}. We first impute the missing values with the different methods. In a second step we apply PCA to the full panel with imputed values to estimate the factor model and update the imputed values with the common components. The observed entries stay the same. This process is repeated for multiple iterations. Here we consider three iterations.}
		\label{tab:comparison-jms-bn-iter3-N100T150}
	\end{table}
	
	\pagebreak
	
	\subsection{Asymptotic Distribution}
	
	\begin{figure}[H]
		\tcapfig{Histograms of Standardized Loadings, Factors, and Common Components for Missing at Random}
		\centering
		\begin{subfigure}{0.53\textwidth}
			\centering
			\includegraphics[width=1\linewidth]{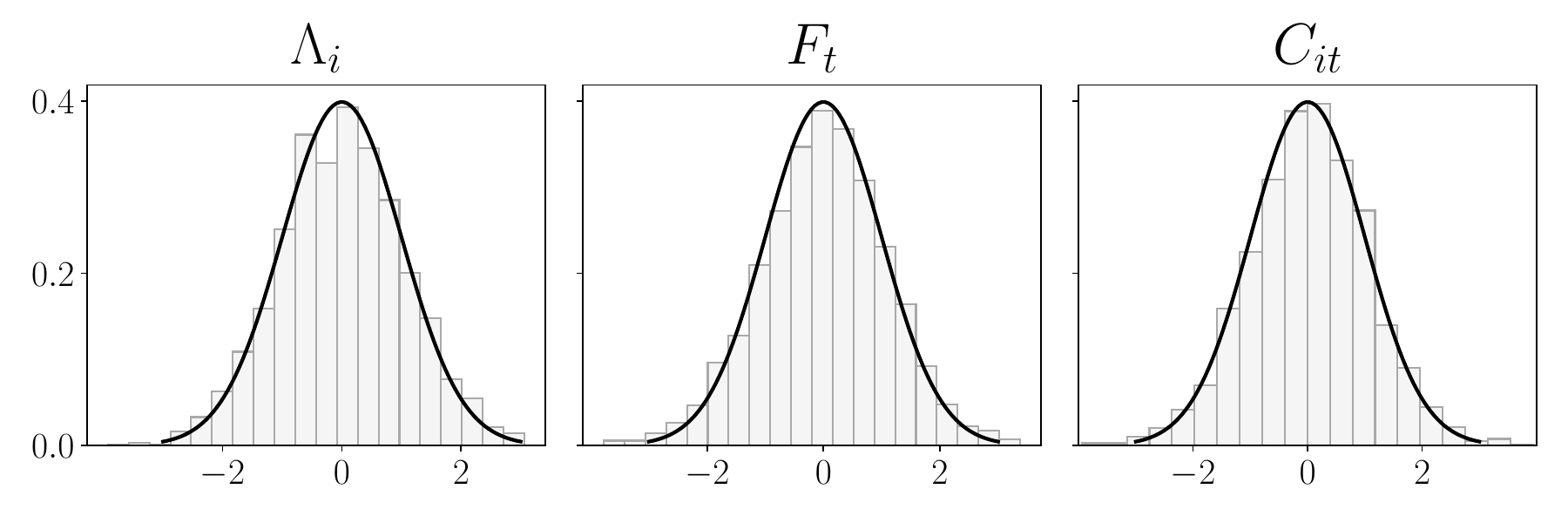}
			\caption{$Y_{it}$ is observed ($N = 250, T = 250$)}
		\end{subfigure}%
		\begin{subfigure}{0.53\textwidth}
			\centering
			\includegraphics[width=1\linewidth]{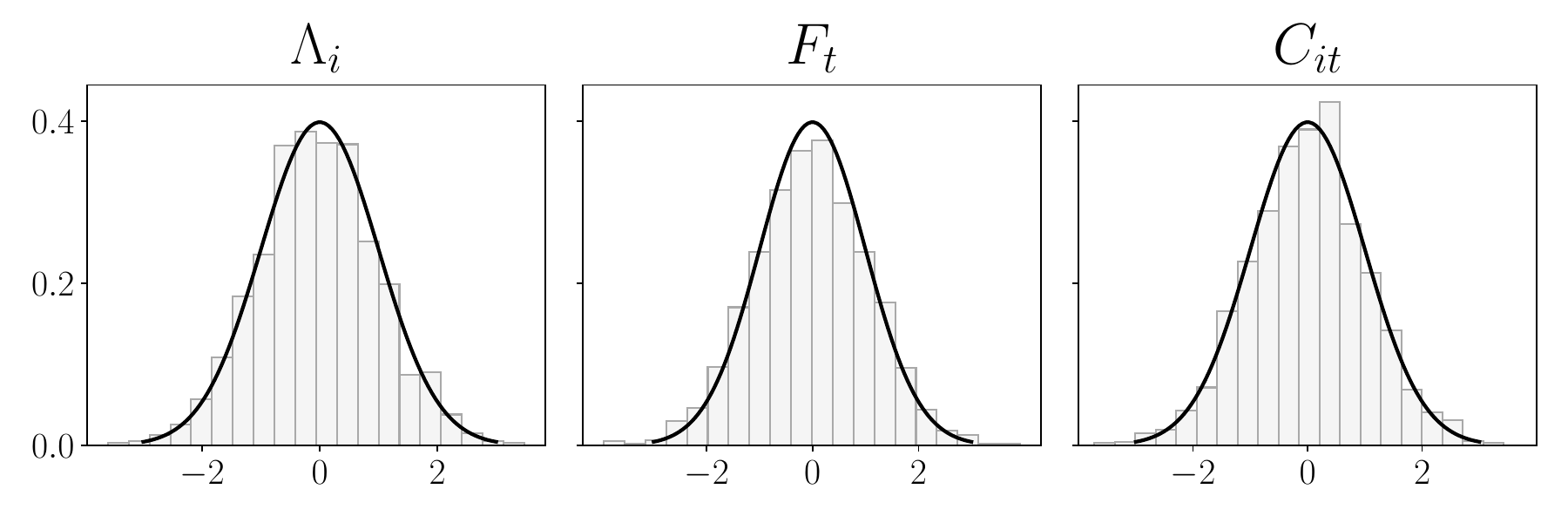}
			\caption{$Y_{it}$ is missing ($N = 250, T = 250$)}
		\end{subfigure}
		\begin{subfigure}{0.53\textwidth}
			\centering
			\includegraphics[width=1\linewidth]{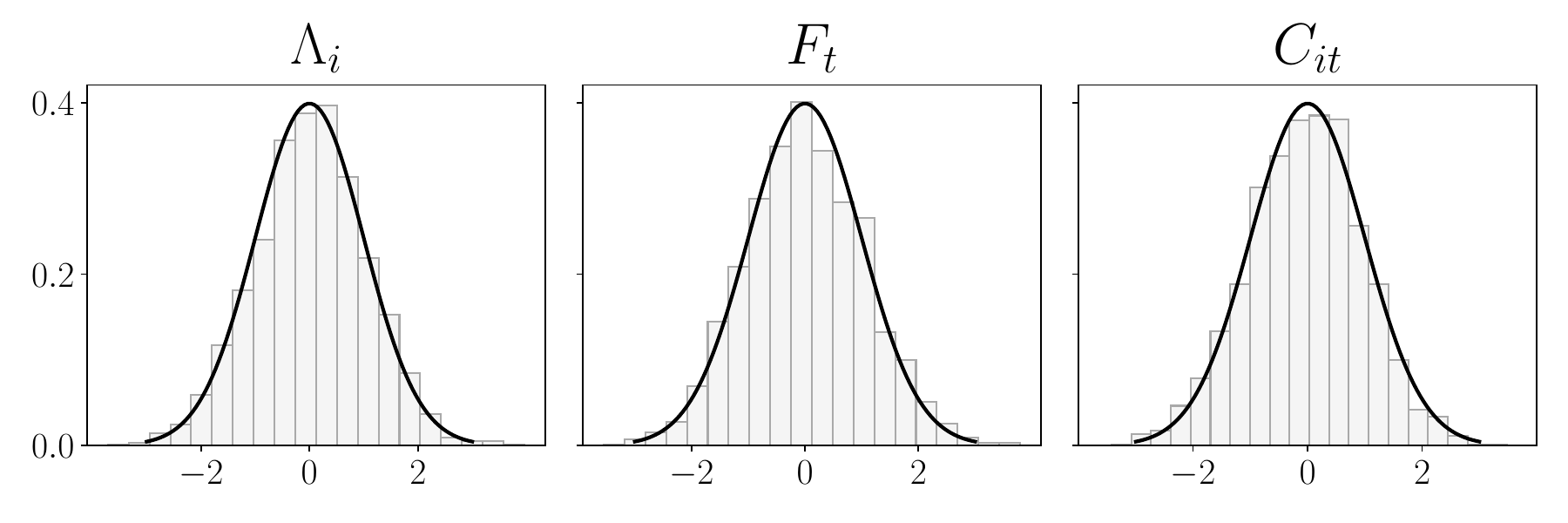}
			\caption{$Y_{it}$ is observed ($N = 250, T = 500$)}
		\end{subfigure}%
		\begin{subfigure}{0.53\textwidth}
			\centering
			\includegraphics[width=1\linewidth]{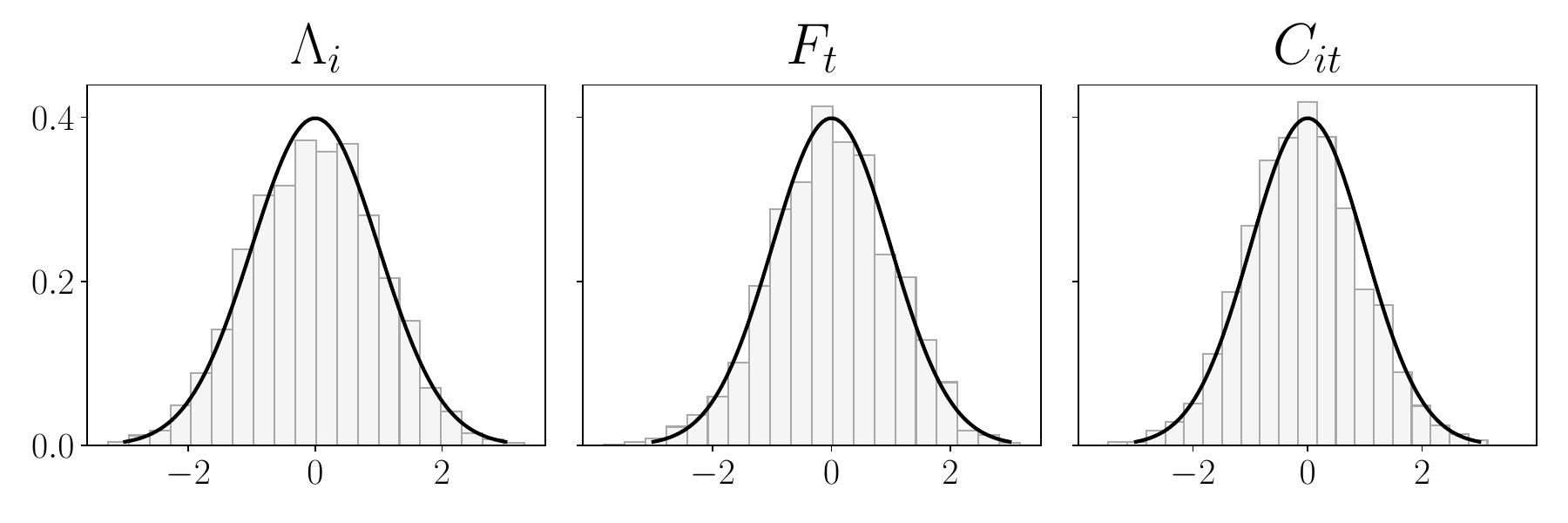}
			\caption{$Y_{it}$ is missing ($N = 250, T = 500$)}
		\end{subfigure}
		\begin{subfigure}{0.53\textwidth}
			\centering
			\includegraphics[width=1\linewidth]{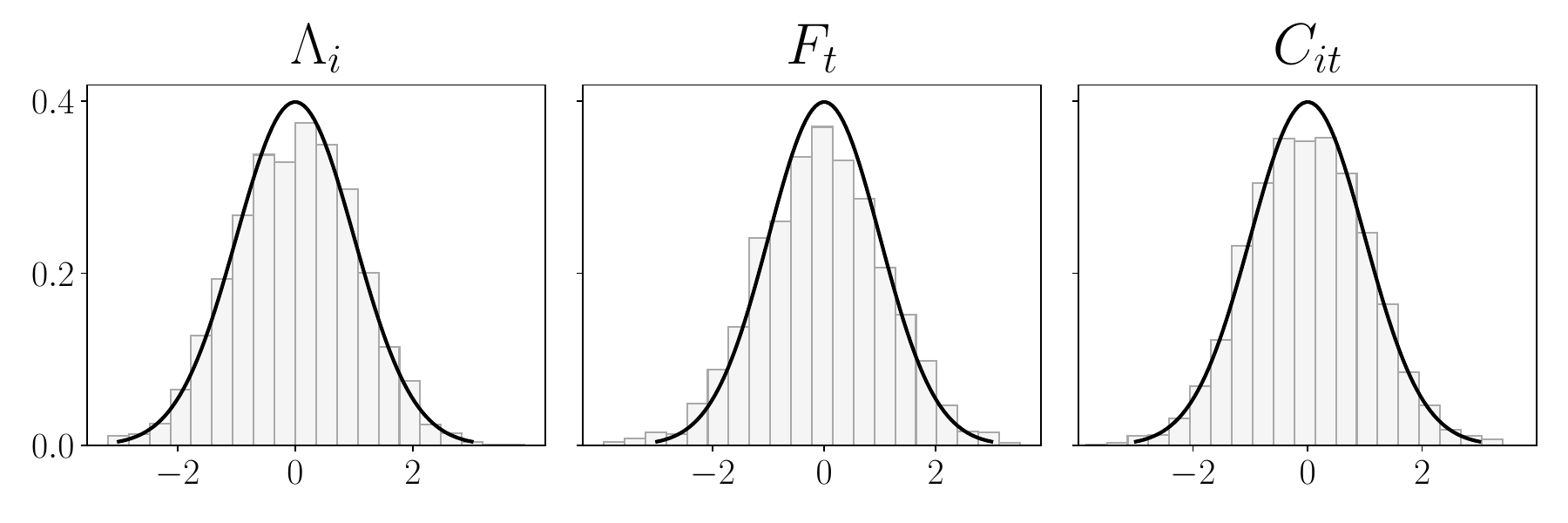}
			\caption{$Y_{it}$ is observed ($N = 500, T = 250$)}
		\end{subfigure}%
		\begin{subfigure}{0.53\textwidth}
			\centering
			\includegraphics[width=1\linewidth]{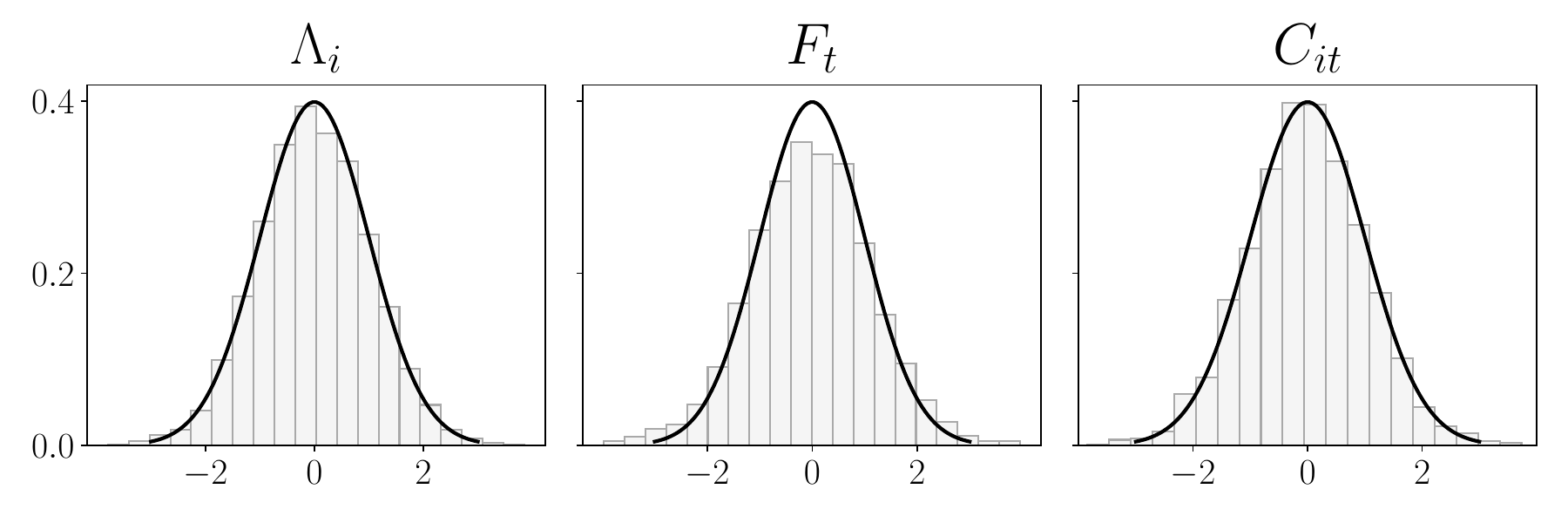}
			\caption{$Y_{it}$ is missing ($N = 500, T = 250$)}
		\end{subfigure}
		\bnotefig{These figures show histograms of estimated standardized loadings, factors, and common components, where entries are missing at random. The normal density function is superimposed on the histograms. The data is simulated with a one-factor model $X_{it} = \Lambda_i \cdot F_t + e_{it}$, where $F_t \stackrel{\text{i.i.d.}}{\sim}  N(0,1)$, $\Lambda_i \stackrel{\text{i.i.d.}}{\sim}  N(0,1)$ and $e_{it} \stackrel{\text{i.i.d.}}{\sim}  N(0,1)$.  The observation pattern depends on an observed state variable defined as $S_i = \mathbbm{1}(\Lambda_i \geq 0 )$. Entries are observed independently with probability 0.75 if $S_i = 1$, and 0.5 if $S_i = 0$. We run 2,000 Monte Carlo simulations.}
		\label{fig:hist-random}
	\end{figure}

	\begin{figure}[H]
		\tcapfig{Histograms of Standardized Loadings, Factors, and Common Components for Simultaneous Treatment Adoption}
		\centering
		\begin{subfigure}{0.53\textwidth}
			\centering
			\includegraphics[width=1\linewidth]{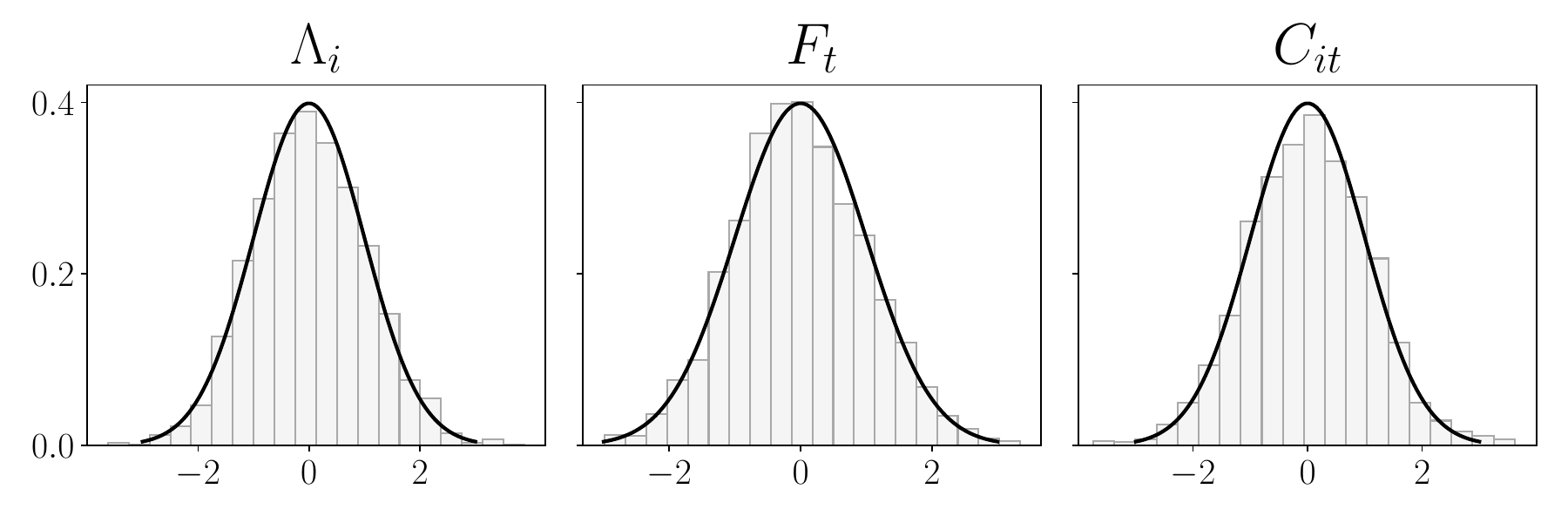}
			\caption{$Y_{it}$ is observed ($N = 250, T = 250$)}
		\end{subfigure}%
		\begin{subfigure}{0.53\textwidth}
			\centering
			\includegraphics[width=1\linewidth]{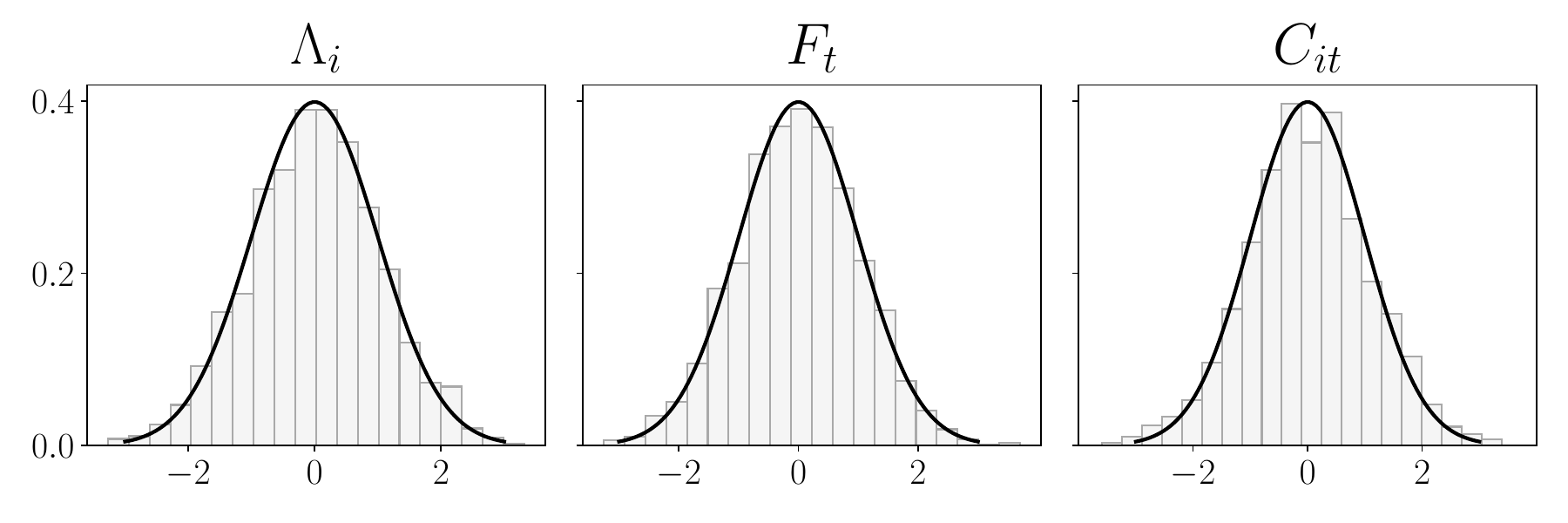}
			\caption{$Y_{it}$ is missing ($N = 250, T = 250$)}
		\end{subfigure}
		\begin{subfigure}{0.53\textwidth}
			\centering
			\includegraphics[width=1\linewidth]{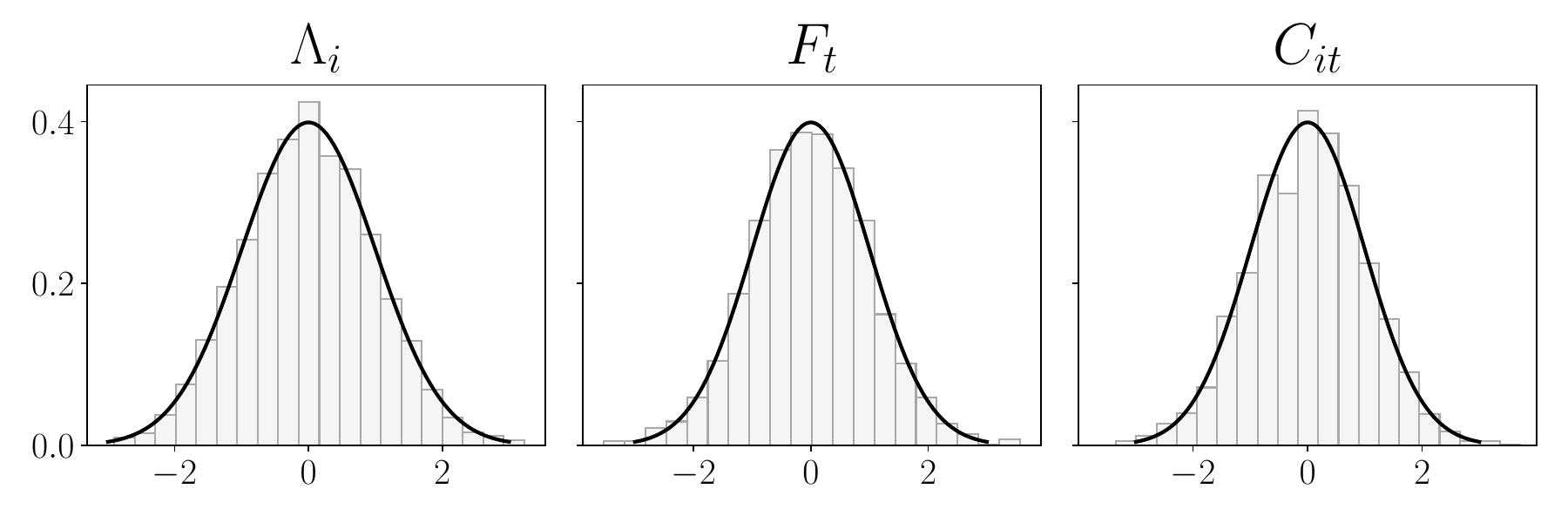}
			\caption{$Y_{it}$ is observed ($N = 250, T = 500$)}
		\end{subfigure}%
		\begin{subfigure}{0.53\textwidth}
			\centering
			\includegraphics[width=1\linewidth]{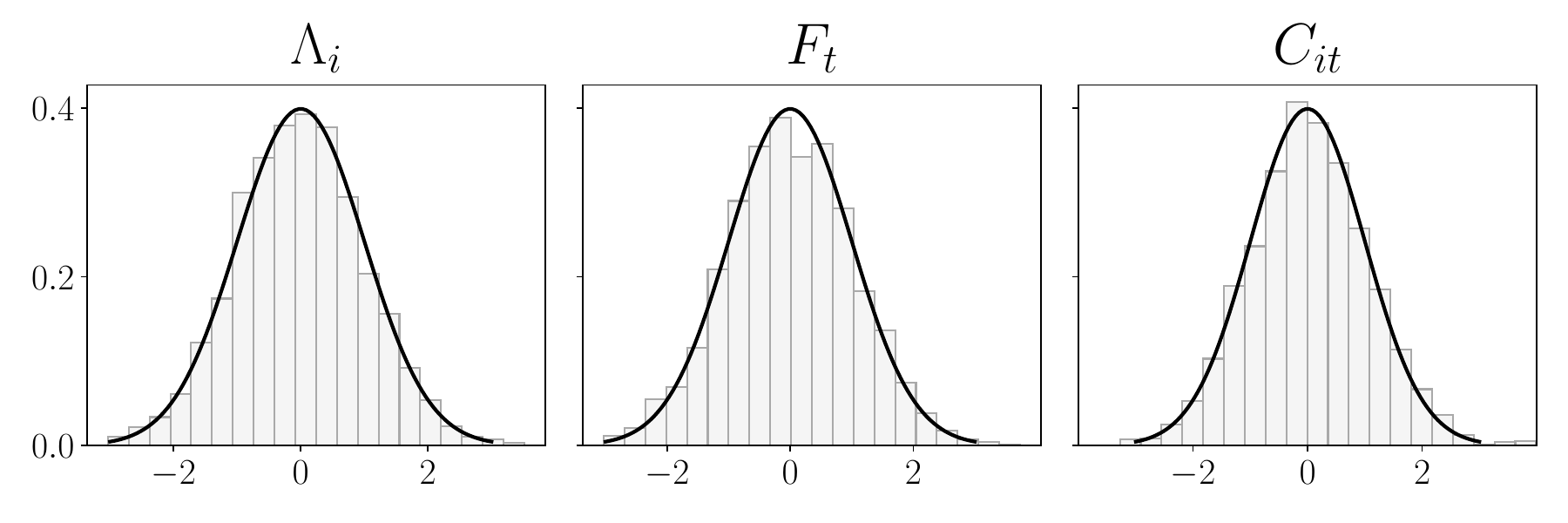}
			\caption{$Y_{it}$ is missing ($N = 250, T = 500$)}
		\end{subfigure}
		\begin{subfigure}{0.53\textwidth}
			\centering
			\includegraphics[width=1\linewidth]{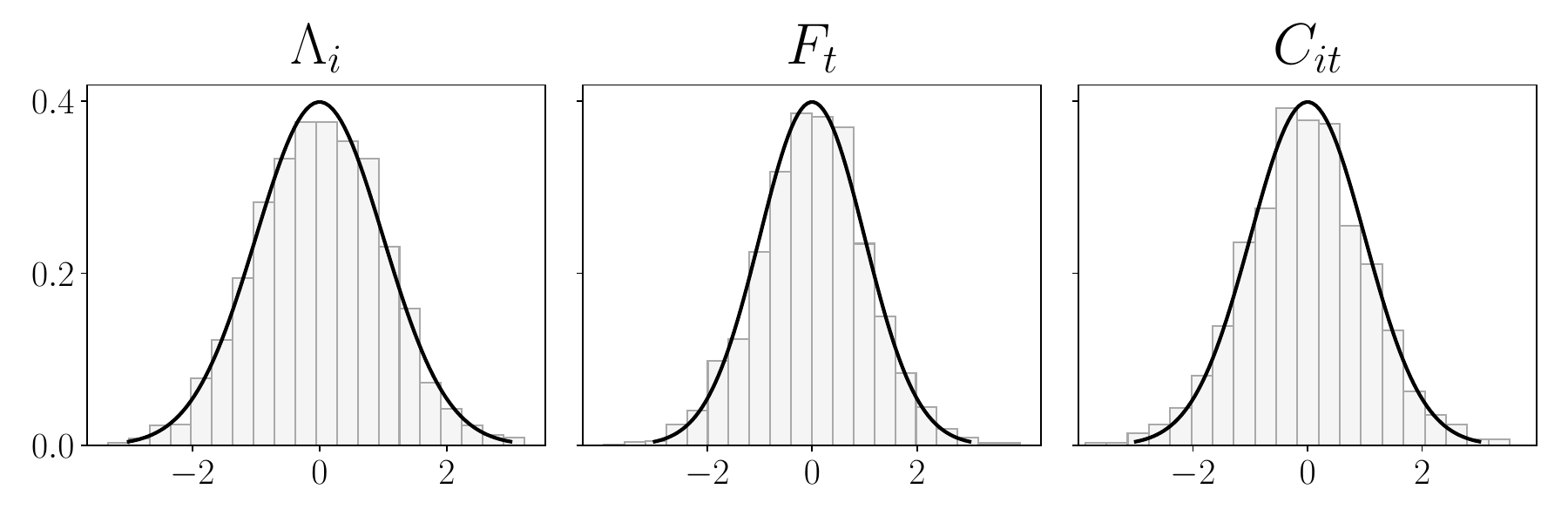}
			\caption{$Y_{it}$ is observed ($N = 500, T = 250$)}
		\end{subfigure}%
		\begin{subfigure}{0.53\textwidth}
			\centering		\includegraphics[width=1\linewidth]{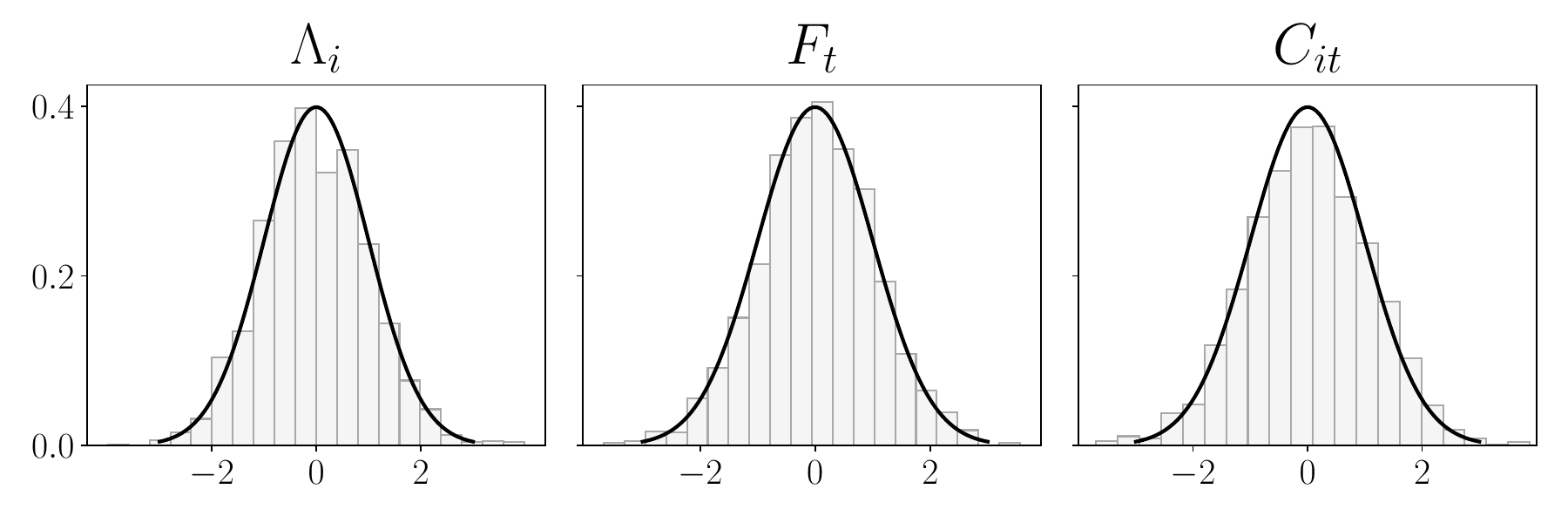}
			\caption{$Y_{it}$ is missing ($N = 500, T = 250$)}
		\end{subfigure}
		\bnotefig{These figures show histograms of estimated standardized loadings, factors, and common components, where the observation pattern is the control panel in the simultaneous treatment adoption case. The normal density function is superimposed on the histograms. The data is simulated with a one-factor model $X_{it} = \Lambda_i \cdot F_t + e_{it}$, where $F_t \stackrel{\text{i.i.d.}}{\sim}  N(0,1)$, $\Lambda_i \stackrel{\text{i.i.d.}}{\sim}  N(0,1)$ and $e_{it} \stackrel{\text{i.i.d.}}{\sim}  N(0,1)$.  The observation pattern depends on an observed state variable defined as $S_i = \mathbbm{1}(\Lambda_i \geq 0 )$. Simultaneous treatment adoption:  Once a unit adopts treatment, it stays treated afterwards. For the units with $S_i = 1$, $25\% $ randomly selected units adopt the treatment from time $0.75 \cdot T$ and the remaining $75\% $ units stay in the control group until the end. For the units with $S_i= 0$, $62.5\%$ randomly selected units adopt the treatment from time $0.375\cdot T$ and the remaining $37.5\%$  units stay in the control group until the end.  We estimate the factor model from the control panel, where a unit's observations are missing if it adopts treatment. We run 2,000 Monte Carlo simulations.}
		\label{fig:hist-simul}
	\end{figure}

	\begin{figure}[H]
		\tcapfig{Histograms of Standardized Control and Treated Common Components, Individual and Average Treatment Effects}
		\centering
		\begin{subfigure}{0.65\textwidth}
			\centering
			\includegraphics[width=1\linewidth]{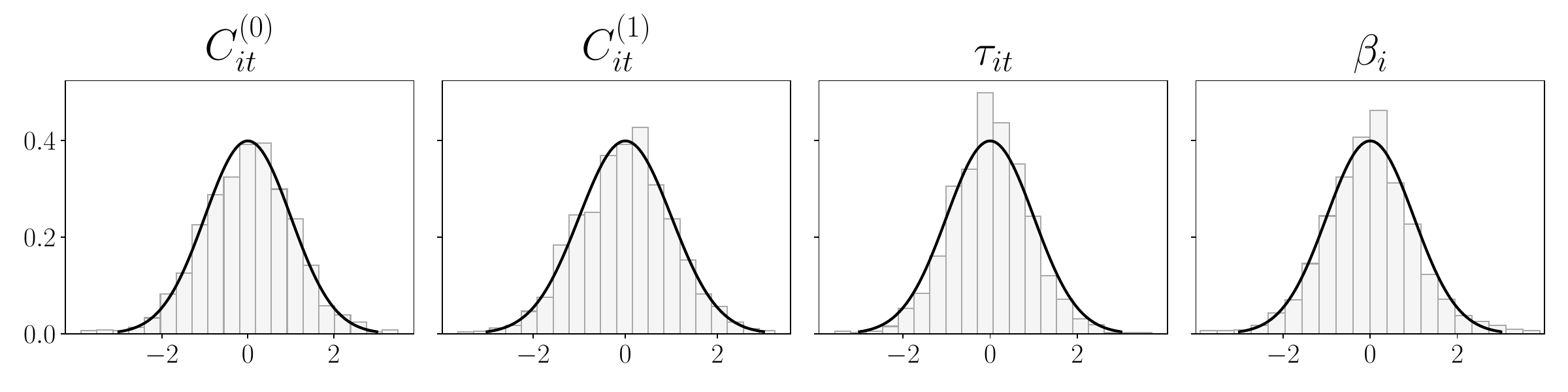}
			\caption{Treatment effect $\tau = 0$ ($N = 250, T = 250$)}
		\end{subfigure}
		\begin{subfigure}{0.65\textwidth}
			\centering
			\includegraphics[width=1\linewidth]{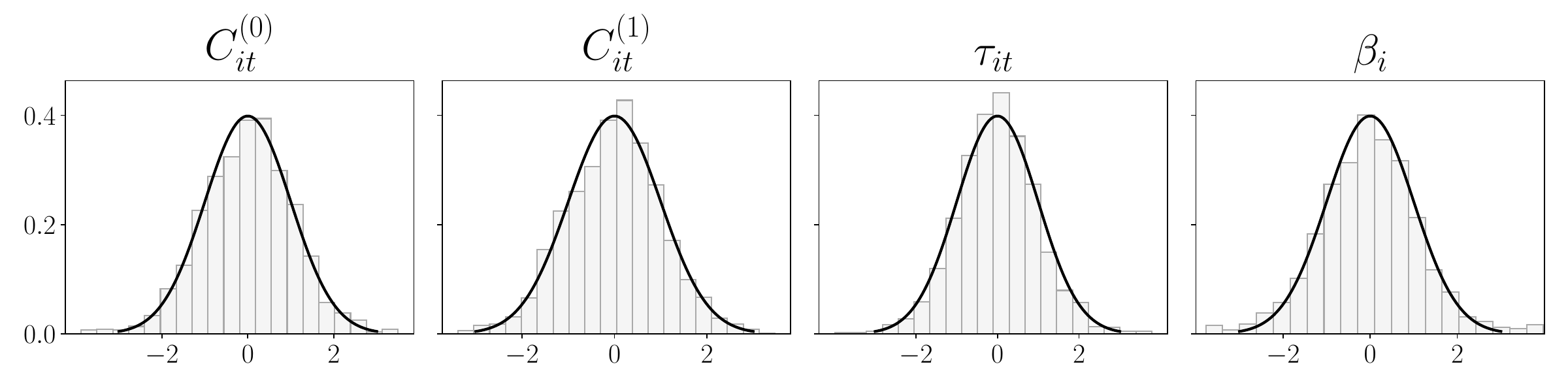}
			\caption{Treatment effect $\tau = 0.25$ ($N = 250, T = 250$)}
		\end{subfigure}
		\begin{subfigure}{0.65\textwidth}
			\centering
			\includegraphics[width=1\linewidth]{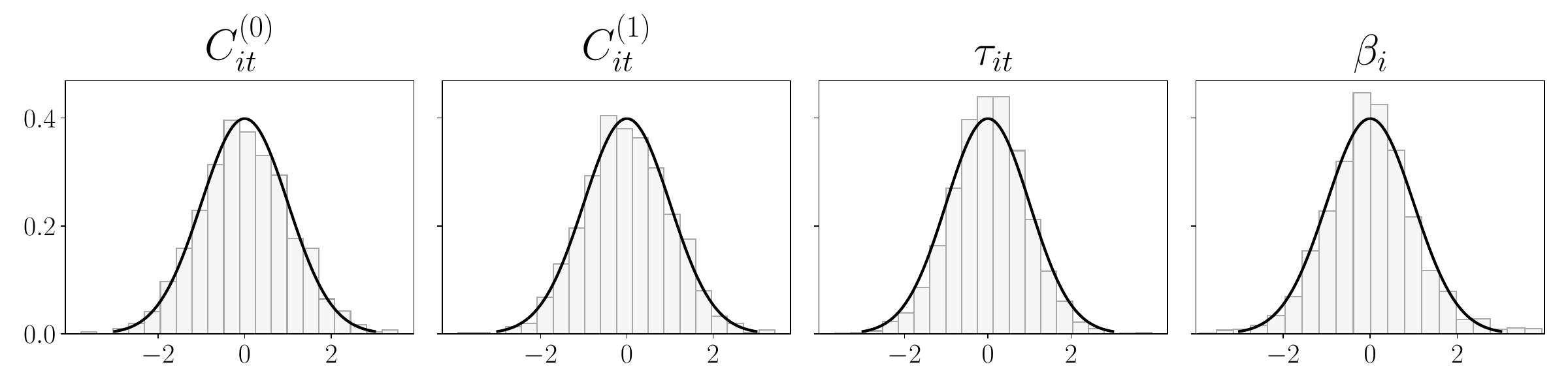}
			\caption{Treatment effect $\tau = 0$ ($N = 250, T = 500$)}
		\end{subfigure}
		\begin{subfigure}{0.65\textwidth}
			\centering
			\includegraphics[width=1\linewidth]{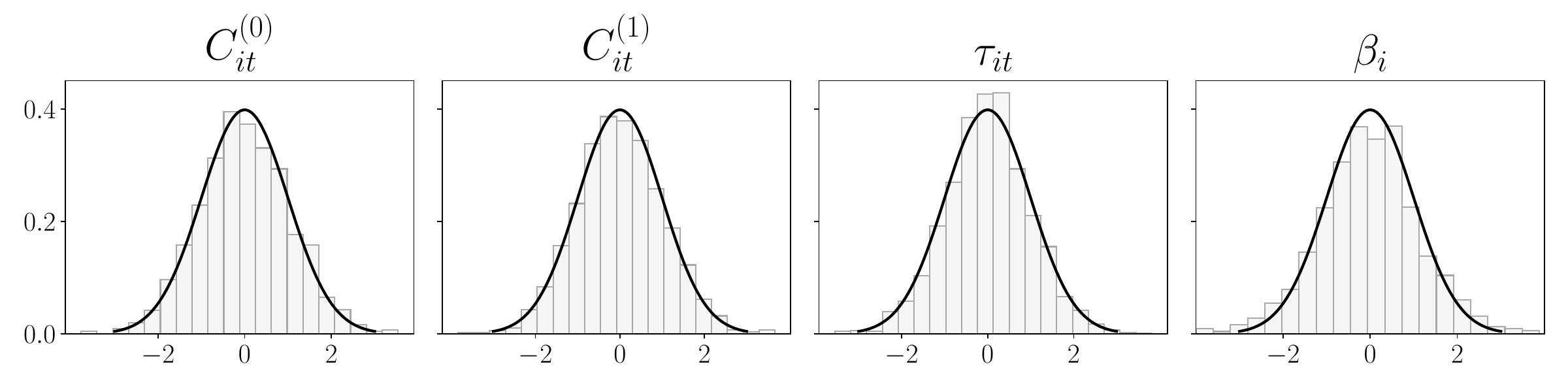}
			\caption{Treatment effect $\tau = 0.25$ ($N = 250, T = 500$)}
		\end{subfigure}
		\begin{subfigure}{0.65\textwidth}
			\centering
			\includegraphics[width=1\linewidth]{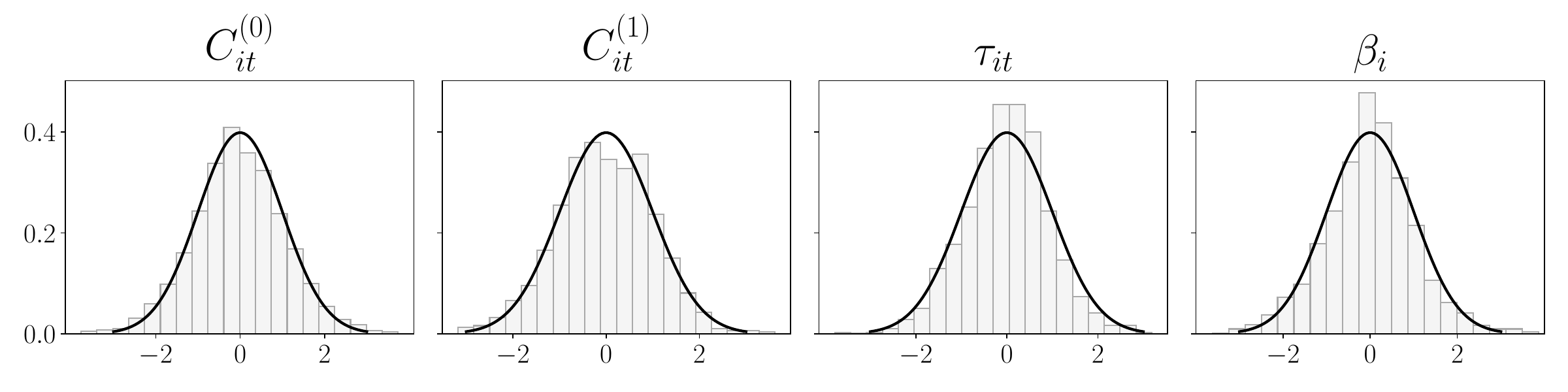}
			\caption{Treatment effect $\tau = 0$ ($N = 500, T = 250$)}
		\end{subfigure}
		\begin{subfigure}{0.65\textwidth}
			\centering
			\includegraphics[width=1\linewidth]{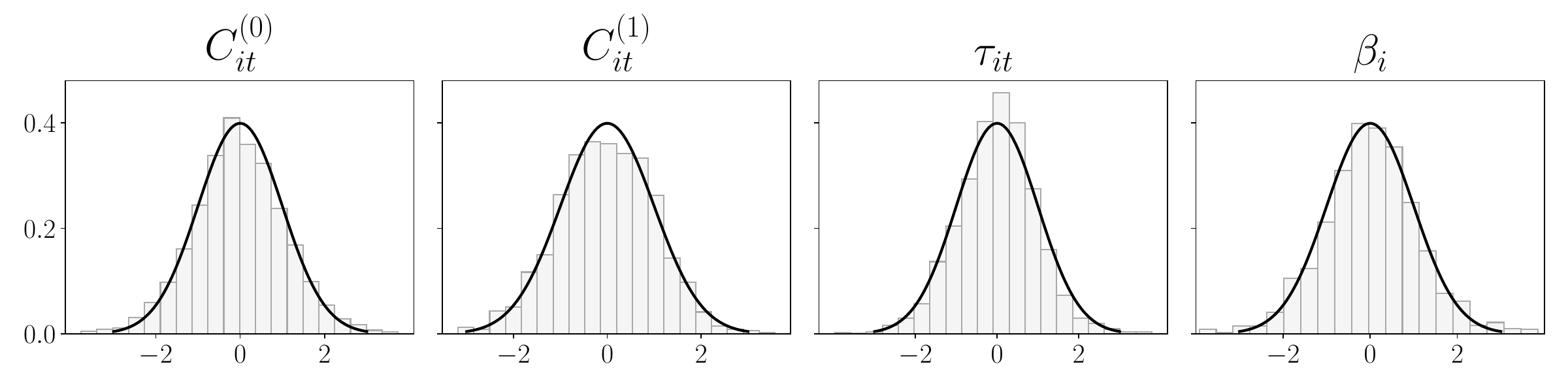}
			\caption{Treatment effect $\tau = 0.25$ ($N = 500, T = 250$)}
		\end{subfigure}
		\bnotefig{These figures show histograms of estimated control and treated common components, individual and average treatment effect ($Z = \vec{1}$) for different combinations of $N$ and $T$. The normal density function is superimposed on the histograms. We run 2,000 Monte Carlo simulations.}
		\label{fig:hist-treatment-appendix}
	\end{figure}

	\subsection{Statistical Power of Treatment Effect Tests}

	\begin{table}[H]
		\tcaptab{Statistical Power of Treatment Effect Tests}
		\centering
		{\small
			\begin{tabular}{lrc|rrrr|rrrr}
				\toprule
				&      & {} & \multicolumn{4}{c|}{$\tilde{C}^\treat_{it}- \tilde{C}^\control_{it} $} & \multicolumn{4}{c}{$\tilde{\beta}^\treat_i - \tilde{\beta}^\control_i$} \\
				&      & $\Lambda^\treat_i - \Lambda^\control_i$ &  0.25 &  0.50 &  1.00 &  2.00 &  0.25 &  0.50 &  1.00 &  2.00 \\
				$N$ & $T$ & $\mu_F$ &       &       &       &       &       &       &       &       \\
				\midrule
				100                     & 100  & 0.1  & 0.171 & 0.450 & 0.902 & 0.991 & 0.198 & 0.440 & 0.864 & 0.968 \\
                        &      & 0.25 & 0.162 & 0.484 & 0.898 & 0.991 & 0.174 & 0.457 & 0.862 & 0.972 \\
                        &      & 0.5  & 0.179 & 0.543 & 0.901 & 0.991 & 0.189 & 0.524 & 0.907 & 0.987 \\
                        &      & 1    & 0.271 & 0.660 & 0.946 & 0.996 & 0.271 & 0.654 & 0.939 & 0.996 \\
                        & 250  & 0.1  & 0.306 & 0.794 & 0.979 & 1.000 & 0.298 & 0.758 & 0.950 & 0.985 \\
                        &      & 0.25 & 0.298 & 0.795 & 0.977 & 1.000 & 0.292 & 0.780 & 0.960 & 0.989 \\
                        &      & 0.5  & 0.347 & 0.835 & 0.981 & 1.000 & 0.345 & 0.831 & 0.979 & 0.998 \\
                        &      & 1    & 0.466 & 0.906 & 0.991 & 1.000 & 0.481 & 0.910 & 0.991 & 1.000 \\
\midrule 
250                     & 100  & 0.1  & 0.175 & 0.464 & 0.896 & 0.994 & 0.165 & 0.456 & 0.866 & 0.981 \\
                        &      & 0.25 & 0.188 & 0.494 & 0.909 & 0.994 & 0.190 & 0.479 & 0.885 & 0.974 \\
                        &      & 0.5  & 0.203 & 0.539 & 0.906 & 0.994 & 0.199 & 0.535 & 0.897 & 0.989 \\
                        &      & 1    & 0.273 & 0.722 & 0.959 & 0.998 & 0.271 & 0.731 & 0.954 & 0.998 \\
                        & 250  & 0.1  & 0.302 & 0.811 & 0.989 & 1.000 & 0.297 & 0.791 & 0.966 & 0.991 \\
                        &      & 0.25 & 0.327 & 0.817 & 0.989 & 1.000 & 0.325 & 0.815 & 0.985 & 1.000 \\
                        &      & 0.5  & 0.366 & 0.857 & 0.994 & 1.000 & 0.355 & 0.853 & 0.991 & 1.000 \\
                        &      & 1    & 0.521 & 0.916 & 0.996 & 1.000 & 0.528 & 0.918 & 0.996 & 1.000 \\
                        & 500  & 0.1  & 0.572 & 0.954 & 1.000 & 1.000 & 0.558 & 0.931 & 0.983 & 0.989 \\
                        &      & 0.25 & 0.575 & 0.962 & 1.000 & 1.000 & 0.577 & 0.960 & 1.000 & 1.000 \\
                        &      & 0.5  & 0.635 & 0.964 & 1.000 & 1.000 & 0.633 & 0.966 & 1.000 & 1.000 \\
                        &      & 1    & 0.764 & 0.970 & 1.000 & 1.000 & 0.772 & 0.979 & 1.000 & 1.000 \\
                        & 1000 & 0.1  & 0.796 & 0.987 & 1.000 & 1.000 & 0.777 & 0.981 & 0.996 & 0.998 \\
                        &      & 0.25 & 0.812 & 0.989 & 1.000 & 1.000 & 0.814 & 0.989 & 1.000 & 1.000 \\
                        &      & 0.5  & 0.850 & 0.994 & 1.000 & 1.000 & 0.854 & 0.994 & 1.000 & 1.000 \\
                        &      & 1    & 0.918 & 0.996 & 1.000 & 1.000 & 0.920 & 0.996 & 1.000 & 1.000 \\
\midrule 
500                     & 500  & 0.1  & 0.610 & 0.975 & 1.000 & 1.000 & 0.591 & 0.958 & 0.998 & 1.000 \\
                        &      & 0.25 & 0.610 & 0.966 & 1.000 & 1.000 & 0.604 & 0.969 & 1.000 & 1.000 \\
                        &      & 0.5  & 0.649 & 0.979 & 1.000 & 1.000 & 0.644 & 0.972 & 1.000 & 1.000 \\
                        &      & 1    & 0.805 & 0.987 & 1.000 & 1.000 & 0.809 & 0.987 & 1.000 & 1.000 \\
                        & 1000 & 0.1  & 0.860 & 0.992 & 1.000 & 1.000 & 0.848 & 0.983 & 0.998 & 0.998 \\
                        &      & 0.25 & 0.874 & 0.996 & 1.000 & 1.000 & 0.877 & 0.996 & 1.000 & 1.000 \\
                        &      & 0.5  & 0.898 & 0.996 & 1.000 & 1.000 & 0.898 & 0.998 & 1.000 & 1.000 \\
                        &      & 1    & 0.959 & 1.000 & 1.000 & 1.000 & 0.957 & 1.000 & 1.000 & 1.000 \\
				\bottomrule
		\end{tabular}}
		\bnotetab{This table shows the proportion of test statistics of the treatment effect that reject the null hypotheses $\mathcal{H}_0: C^\treat_{it} - C^\control_{it} = 0$ or $\mathcal{H}_0:  \beta^\treat_i - \beta^\control_i = 0$, where $ \beta^\treat_i = \frac{1}{\Ttreat} \sum_{\Tcontrol+1}^T C^\treat_{it}$ and  $ \beta^\control_i = \frac{1}{\Ttreat} \sum_{\Tcontrol+1}^T C^\control_{it}$. We consider a 95\% confidence level (the test statistics are within $[-1.96, 1.96]$) over 500 Monte Carlo simulations . The test statistics normalize $\tilde{C}^\treat_{it}- \tilde{C}^\control_{it} $ and $\tilde{\beta}^\treat_i - \tilde{\beta}^\control_i$ by their estimated standard deviation from Equations \eqref{eqn:ite-asy-normal} and \eqref{eqn:ate-asy-normal}. {\bf The estimated standard deviations are estimated under the null hypothesis of $\Lambda^\treat_i - \Lambda^\control_i = 0$}. The observation pattern follows the simultaneous treatment adoption pattern.  The proportion of acceptance decreases with $N, T, \mu_F$ and $\tilde{\beta}^\treat_i - \tilde{\beta}^\control_i$, implying that the statistical power increases with the data dimensionality, the scale of the treatment effect and the proportion of observed entries in the data.}
		\label{tab:treatment-power-null}
	\end{table}

	\begin{table}[H]
		\tcaptab{Statistical Power of Treatment Effect Tests with Fewer Observed Values}
		\centering
		{\small
			\begin{tabular}{lrc|rrrr|rrrr}
				\toprule
				&      & {} & \multicolumn{4}{c|}{$\tilde{C}^\treat_{it}- \tilde{C}^\control_{it} $} & \multicolumn{4}{c}{$\tilde{\beta}^\treat_i - \tilde{\beta}^\control_i$} \\
				&      & $\Lambda^\treat_i - \Lambda^\control_i$ &  0.25 &  0.50 &  1.00 &  2.00 &  0.25 &  0.50 &  1.00 &  2.00 \\
				$N$ & $T$ & $\mu_F$ &       &       &       &       &       &       &       &       \\
				\midrule
				100                     & 100  & 0.1  & 0.167 & 0.439 & 0.789 & 0.931 & 0.200 & 0.437 & 0.755 & 0.896 \\
                        &      & 0.25 & 0.177 & 0.473 & 0.790 & 0.944 & 0.188 & 0.452 & 0.770 & 0.928 \\
                        &      & 0.5  & 0.178 & 0.463 & 0.784 & 0.934 & 0.176 & 0.459 & 0.773 & 0.927 \\
                        &      & 1    & 0.247 & 0.619 & 0.860 & 0.960 & 0.242 & 0.621 & 0.844 & 0.948 \\
                        & 250  & 0.1  & 0.327 & 0.744 & 0.934 & 0.984 & 0.331 & 0.717 & 0.922 & 0.968 \\
                        &      & 0.25 & 0.347 & 0.742 & 0.924 & 0.984 & 0.347 & 0.726 & 0.924 & 0.984 \\
                        &      & 0.5  & 0.363 & 0.764 & 0.933 & 0.984 & 0.369 & 0.757 & 0.931 & 0.984 \\
                        &      & 1    & 0.514 & 0.814 & 0.947 & 0.986 & 0.498 & 0.816 & 0.952 & 0.988 \\
\midrule
250                     & 100  & 0.1  & 0.175 & 0.465 & 0.809 & 0.939 & 0.195 & 0.465 & 0.762 & 0.919 \\
                        &      & 0.25 & 0.175 & 0.466 & 0.809 & 0.944 & 0.188 & 0.457 & 0.787 & 0.933 \\
                        &      & 0.5  & 0.182 & 0.487 & 0.814 & 0.946 & 0.182 & 0.476 & 0.809 & 0.953 \\
                        &      & 1    & 0.241 & 0.629 & 0.862 & 0.957 & 0.245 & 0.611 & 0.857 & 0.957 \\
                        & 250  & 0.1  & 0.304 & 0.733 & 0.933 & 0.991 & 0.300 & 0.703 & 0.918 & 0.980 \\
                        &      & 0.25 & 0.310 & 0.746 & 0.938 & 0.986 & 0.299 & 0.730 & 0.938 & 0.989 \\
                        &      & 0.5  & 0.341 & 0.742 & 0.944 & 0.989 & 0.330 & 0.746 & 0.937 & 0.991 \\
                        &      & 1    & 0.463 & 0.803 & 0.958 & 0.988 & 0.452 & 0.810 & 0.960 & 0.991 \\
                        & 500  & 0.1  & 0.553 & 0.872 & 0.973 & 0.998 & 0.542 & 0.867 & 0.973 & 0.991 \\
                        &      & 0.25 & 0.560 & 0.865 & 0.962 & 0.995 & 0.558 & 0.863 & 0.968 & 0.998 \\
                        &      & 0.5  & 0.609 & 0.873 & 0.975 & 0.998 & 0.609 & 0.873 & 0.977 & 1.000 \\
                        &      & 1    & 0.738 & 0.937 & 0.990 & 1.000 & 0.736 & 0.930 & 0.988 & 1.000 \\
                        & 1000 & 0.1  & 0.731 & 0.944 & 0.993 & 1.000 & 0.718 & 0.935 & 0.989 & 1.000 \\
                        &      & 0.25 & 0.725 & 0.952 & 0.998 & 1.000 & 0.721 & 0.948 & 0.993 & 1.000 \\
                        &      & 0.5  & 0.748 & 0.949 & 0.998 & 1.000 & 0.748 & 0.947 & 0.998 & 1.000 \\
                        &      & 1    & 0.802 & 0.962 & 0.995 & 1.000 & 0.798 & 0.965 & 0.995 & 1.000 \\
\midrule
500                     & 500  & 0.1  & 0.555 & 0.882 & 0.977 & 0.998 & 0.560 & 0.877 & 0.963 & 0.995 \\
                        &      & 0.25 & 0.569 & 0.886 & 0.982 & 0.998 & 0.562 & 0.888 & 0.982 & 0.998 \\
                        &      & 0.5  & 0.606 & 0.891 & 0.984 & 1.000 & 0.595 & 0.889 & 0.984 & 0.998 \\
                        &      & 1    & 0.731 & 0.920 & 0.990 & 1.000 & 0.733 & 0.917 & 0.990 & 0.998 \\
                        & 1000 & 0.1  & 0.760 & 0.958 & 0.996 & 1.000 & 0.758 & 0.960 & 0.998 & 1.000 \\
                        &      & 0.25 & 0.770 & 0.958 & 0.998 & 1.000 & 0.774 & 0.962 & 1.000 & 1.000 \\
                        &      & 0.5  & 0.803 & 0.967 & 0.998 & 1.000 & 0.799 & 0.964 & 1.000 & 1.000 \\
                        &      & 1    & 0.869 & 0.981 & 1.000 & 1.000 & 0.871 & 0.981 & 1.000 & 1.000 \\
                        &      & 1    & 0.886 & 0.927 & 0.932 & 0.932 & 0.962 & 1.000 & 1.000 & 1.000 \\
				\bottomrule
		\end{tabular}}
		\bnotetab{This table shows the proportion of test statistics of the treatment effect that reject the null hypotheses $\mathcal{H}_0: C^\treat_{it} - C^\control_{it} = 0$ or $\mathcal{H}_0:  \beta^\treat_i - \beta^\control_i = 0$, where $ \beta^\treat_i = \frac{1}{\Ttreat} \sum_{\Tcontrol+1}^T C^\treat_{it}$ and  $ \beta^\control_i = \frac{1}{\Ttreat} \sum_{\Tcontrol+1}^T C^\control_{it}$. We consider a 95\% confidence level (the test statistics are within $[-1.96, 1.96]$) over 500 Monte Carlo simulations . The test statistics normalize $\tilde{C}^\treat_{it}- \tilde{C}^\control_{it} $ and $\tilde{\beta}^\treat_i - \tilde{\beta}^\control_i$ by their estimated standard deviation from Equations \eqref{eqn:ite-asy-normal} and \eqref{eqn:ate-asy-normal}. {\bf The estimated standard deviations are estimated under the null hypothesis of $\Lambda^\treat_i - \Lambda^\control_i = 0$}. The observation pattern follows the simultaneous treatment adoption pattern, but has fewer control observations: The observed state variable is defined as $S_i = \mathbbm{1}(\Lambda_i \geq 0 )$. For the units with $S_i = 1$, $50\% $ randomly selected units adopt the treatment from time $0.5 \cdot T$ and the remaining $50\% $ units stay in the control group until the end. For the units with $S_i= 0$, $75\%$ randomly selected units adopt the treatment from time $0.25\cdot T$ and the remaining $25\%$  units stay in the control group until the end. The proportion of acceptance decreases with $N, T, \mu_F$ and $\tilde{\beta}^\treat_i - \tilde{\beta}^\control_i$, implying that the statistical power increases with the data dimensionality, the scale of the treatment effect and the proportion of observed entries in the data.}
		\label{tab:treatment-power-fewer-obs-null}
	\end{table}

	\begin{table}[H]
		\tcaptab{Statistical Power of Treatment Effect Tests under Alternative Estimated Standard Deviations}
		\centering
		{\small
			\begin{tabular}{lrc|rrrr|rrrr}
				\toprule
				&      & {} & \multicolumn{4}{c|}{$\tilde{C}^\treat_{it}- \tilde{C}^\control_{it} $} & \multicolumn{4}{c}{$\tilde{\beta}^\treat_i - \tilde{\beta}^\control_i$} \\
				&      & $\Lambda^\treat_i - \Lambda^\control_i$ &  0.25 &  0.50 &  1.00 &  2.00 &  0.25 &  0.50 &  1.00 &  2.00 \\
				$N$ & $T$ & $\mu_F$ &       &       &       &       &       &       &       &       \\
				\midrule
				100                     & 100  & 0.1  & 0.095 & 0.299 & 0.658 & 0.770 & 0.211 & 0.476 & 0.886 & 0.975 \\
                        &      & 0.25 & 0.106 & 0.313 & 0.686 & 0.781 & 0.196 & 0.490 & 0.904 & 0.980 \\
                        &      & 0.5  & 0.127 & 0.386 & 0.703 & 0.797 & 0.218 & 0.569 & 0.937 & 0.987 \\
                        &      & 1    & 0.220 & 0.531 & 0.788 & 0.836 & 0.302 & 0.710 & 0.973 & 0.995 \\
                        & 250  & 0.1  & 0.184 & 0.571 & 0.763 & 0.802 & 0.311 & 0.794 & 0.962 & 0.990 \\
                        &      & 0.25 & 0.200 & 0.579 & 0.773 & 0.816 & 0.307 & 0.798 & 0.979 & 0.995 \\
                        &      & 0.5  & 0.233 & 0.645 & 0.771 & 0.800 & 0.359 & 0.860 & 0.987 & 0.997 \\
                        &      & 1    & 0.374 & 0.759 & 0.858 & 0.868 & 0.511 & 0.940 & 0.998 & 1.000 \\
\midrule
250                     & 100  & 0.1  & 0.124 & 0.370 & 0.739 & 0.843 & 0.194 & 0.492 & 0.898 & 0.981 \\
                        &      & 0.25 & 0.134 & 0.378 & 0.726 & 0.828 & 0.210 & 0.516 & 0.908 & 0.981 \\
                        &      & 0.5  & 0.167 & 0.454 & 0.782 & 0.872 & 0.238 & 0.577 & 0.940 & 0.995 \\
                        &      & 1    & 0.239 & 0.653 & 0.877 & 0.932 & 0.321 & 0.780 & 0.980 & 1.000 \\
                        & 250  & 0.1  & 0.216 & 0.630 & 0.832 & 0.847 & 0.324 & 0.833 & 0.981 & 0.995 \\
                        &      & 0.25 & 0.226 & 0.658 & 0.848 & 0.853 & 0.345 & 0.845 & 0.991 & 1.000 \\
                        &      & 0.5  & 0.270 & 0.730 & 0.889 & 0.909 & 0.404 & 0.887 & 0.998 & 1.000 \\
                        &      & 1    & 0.448 & 0.831 & 0.907 & 0.916 & 0.550 & 0.937 & 0.998 & 1.000 \\
                        & 500  & 0.1  & 0.441 & 0.776 & 0.843 & 0.844 & 0.590 & 0.946 & 0.989 & 0.993 \\
                        &      & 0.25 & 0.449 & 0.795 & 0.855 & 0.858 & 0.605 & 0.973 & 1.000 & 1.000 \\
                        &      & 0.5  & 0.527 & 0.819 & 0.880 & 0.886 & 0.665 & 0.983 & 1.000 & 1.000 \\
                        &      & 1    & 0.695 & 0.902 & 0.932 & 0.939 & 0.798 & 0.983 & 1.000 & 1.000 \\
                        & 1000 & 0.1  & 0.645 & 0.851 & 0.870 & 0.872 & 0.798 & 0.986 & 1.000 & 1.000 \\
                        &      & 0.25 & 0.669 & 0.850 & 0.866 & 0.866 & 0.833 & 0.994 & 1.000 & 1.000 \\
                        &      & 0.5  & 0.728 & 0.892 & 0.900 & 0.901 & 0.880 & 0.994 & 1.000 & 1.000 \\
                        &      & 1    & 0.818 & 0.899 & 0.914 & 0.917 & 0.934 & 0.996 & 1.000 & 1.000 \\
\midrule
500                     & 500  & 0.1  & 0.496 & 0.878 & 0.910 & 0.910 & 0.619 & 0.971 & 0.998 & 1.000 \\
                        &      & 0.25 & 0.510 & 0.848 & 0.892 & 0.896 & 0.628 & 0.975 & 1.000 & 1.000 \\
                        &      & 0.5  & 0.567 & 0.870 & 0.914 & 0.914 & 0.664 & 0.984 & 1.000 & 1.000 \\
                        &      & 1    & 0.745 & 0.932 & 0.947 & 0.949 & 0.839 & 0.992 & 1.000 & 1.000 \\
                        & 1000 & 0.1  & 0.751 & 0.888 & 0.901 & 0.901 & 0.855 & 0.992 & 0.998 & 1.000 \\
                        &      & 0.25 & 0.752 & 0.902 & 0.910 & 0.911 & 0.883 & 1.000 & 1.000 & 1.000 \\
                        &      & 0.5  & 0.794 & 0.894 & 0.906 & 0.907 & 0.916 & 1.000 & 1.000 & 1.000 \\
                        &      & 1    & 0.886 & 0.927 & 0.932 & 0.932 & 0.962 & 1.000 & 1.000 & 1.000 \\
				\bottomrule
		\end{tabular}}
		\bnotetab{This table shows the proportion of test statistics of the treatment effect that  reject the null hypotheses $\mathcal{H}_0: C^\treat_{it} - C^\control_{it} = 0$ or $\mathcal{H}_0:  \beta^\treat_i - \beta^\control_i = 0$, where $ \beta^\treat_i = \frac{1}{\Ttreat} \sum_{\Tcontrol+1}^T C^\treat_{it}$ and  $ \beta^\control_i = \frac{1}{\Ttreat} \sum_{\Tcontrol+1}^T C^\control_{it}$. We consider a 95\% confidence level (the test statistics are within $[-1.96, 1.96]$) over 500 Monte Carlo simulations . The test statistics normalize $\tilde{C}^\treat_{it}- \tilde{C}^\control_{it} $ and $\tilde{\beta}^\treat_i - \tilde{\beta}^\control_i$ by their estimated standard deviation from Equations \eqref{eqn:ite-asy-normal} and \eqref{eqn:ate-asy-normal}. {\bf The estimated standard deviations are estimated without imposing the null hypothesis of $\Lambda^\treat_i - \Lambda^\control_i = 0$}. The observation pattern follows the simultaneous treatment adoption pattern.  The proportion of acceptance decreases with $N, T, \mu_F$ and $\tilde{\beta}^\treat_i - \tilde{\beta}^\control_i$, implying that the statistical power increases with the data dimensionality, the scale of the treatment effect and the proportion of observed entries in the data. The power decreases slightly if the standard deviations are estimated without imposing the null hypothesis of $\Lambda^\treat_i - \Lambda^\control_i = 0$.}
		\label{tab:treatment-power}
	\end{table}

	\begin{table}[H]
		\tcaptab{Statistical Power of Treatment Effect Tests under Alternative Estimated Standard Deviations with Fewer Observed Values}
		\centering
		{\small
			\begin{tabular}{lrc|rrrr|rrrr}
				\toprule
				&      & {} & \multicolumn{4}{c|}{$\tilde{C}^\treat_{it}- \tilde{C}^\control_{it} $} & \multicolumn{4}{c}{$\tilde{\beta}^\treat_i - \tilde{\beta}^\control_i$} \\
				&      & $\Lambda^\treat_i - \Lambda^\control_i$ &  0.25 &  0.50 &  1.00 &  2.00 &  0.25 &  0.50 &  1.00 &  2.00 \\
				$N$ & $T$ & $\mu_F$ &       &       &       &       &       &       &       &       \\
				\midrule
				100                     & 100  & 0.1  & 0.096 & 0.270 & 0.520 & 0.651 & 0.204 & 0.452 & 0.774 & 0.912 \\
                        &      & 0.25 & 0.098 & 0.300 & 0.532 & 0.659 & 0.192 & 0.469 & 0.797 & 0.933 \\
                        &      & 0.5  & 0.108 & 0.306 & 0.552 & 0.707 & 0.176 & 0.486 & 0.801 & 0.945 \\
                        &      & 1    & 0.153 & 0.481 & 0.712 & 0.813 & 0.249 & 0.633 & 0.881 & 0.959 \\
                        & 250  & 0.1  & 0.198 & 0.529 & 0.692 & 0.747 & 0.352 & 0.742 & 0.932 & 0.970 \\
                        &      & 0.25 & 0.207 & 0.511 & 0.674 & 0.744 & 0.352 & 0.751 & 0.943 & 0.987 \\
                        &      & 0.5  & 0.235 & 0.544 & 0.680 & 0.749 & 0.384 & 0.785 & 0.949 & 0.987 \\
                        &      & 1    & 0.380 & 0.647 & 0.756 & 0.791 & 0.518 & 0.845 & 0.957 & 0.991 \\
\midrule
250                     & 100  & 0.1  & 0.112 & 0.351 & 0.677 & 0.799 & 0.196 & 0.486 & 0.790 & 0.920 \\
                        &      & 0.25 & 0.133 & 0.362 & 0.648 & 0.770 & 0.184 & 0.473 & 0.818 & 0.946 \\
                        &      & 0.5  & 0.151 & 0.408 & 0.690 & 0.817 & 0.196 & 0.492 & 0.837 & 0.961 \\
                        &      & 1    & 0.215 & 0.524 & 0.774 & 0.873 & 0.247 & 0.647 & 0.888 & 0.980 \\
                        & 250  & 0.1  & 0.226 & 0.574 & 0.752 & 0.805 & 0.304 & 0.720 & 0.937 & 0.983 \\
                        &      & 0.25 & 0.233 & 0.594 & 0.779 & 0.828 & 0.307 & 0.754 & 0.952 & 0.991 \\
                        &      & 0.5  & 0.265 & 0.606 & 0.790 & 0.839 & 0.341 & 0.761 & 0.959 & 0.993 \\
                        &      & 1    & 0.399 & 0.746 & 0.887 & 0.918 & 0.463 & 0.828 & 0.971 & 0.993 \\
                        & 500  & 0.1  & 0.421 & 0.715 & 0.811 & 0.842 & 0.549 & 0.877 & 0.973 & 0.991 \\
                        &      & 0.25 & 0.438 & 0.723 & 0.807 & 0.846 & 0.565 & 0.881 & 0.974 & 0.998 \\
                        &      & 0.5  & 0.513 & 0.767 & 0.874 & 0.893 & 0.622 & 0.889 & 0.982 & 1.000 \\
                        &      & 1    & 0.620 & 0.818 & 0.883 & 0.895 & 0.742 & 0.942 & 0.995 & 1.000 \\
                        & 1000 & 0.1  & 0.533 & 0.741 & 0.806 & 0.820 & 0.728 & 0.945 & 0.993 & 1.000 \\
                        &      & 0.25 & 0.556 & 0.729 & 0.808 & 0.817 & 0.737 & 0.954 & 0.998 & 1.000 \\
                        &      & 0.5  & 0.619 & 0.800 & 0.857 & 0.866 & 0.756 & 0.957 & 0.998 & 1.000 \\
                        &      & 1    & 0.702 & 0.843 & 0.886 & 0.894 & 0.813 & 0.968 & 0.995 & 1.000 \\
\midrule
500                     & 500  & 0.1  & 0.466 & 0.778 & 0.854 & 0.884 & 0.574 & 0.889 & 0.973 & 0.996 \\
                        &      & 0.25 & 0.499 & 0.769 & 0.843 & 0.873 & 0.577 & 0.900 & 0.985 & 0.998 \\
                        &      & 0.5  & 0.514 & 0.786 & 0.873 & 0.890 & 0.610 & 0.906 & 0.987 & 0.998 \\
                        &      & 1    & 0.664 & 0.858 & 0.925 & 0.936 & 0.741 & 0.933 & 0.993 & 0.998 \\
                        & 1000 & 0.1  & 0.642 & 0.839 & 0.881 & 0.886 & 0.776 & 0.966 & 0.998 & 1.000 \\
                        &      & 0.25 & 0.646 & 0.832 & 0.868 & 0.876 & 0.785 & 0.970 & 1.000 & 1.000 \\
                        &      & 0.5  & 0.672 & 0.834 & 0.868 & 0.874 & 0.806 & 0.972 & 1.000 & 1.000 \\
                        &      & 1    & 0.784 & 0.908 & 0.929 & 0.933 & 0.881 & 0.982 & 1.000 & 1.000 \\

				\bottomrule
		\end{tabular}}
		\bnotetab{This table shows the proportion of test statistics of the treatment effect that reject the null hypotheses $\mathcal{H}_0: C^\treat_{it} - C^\control_{it} = 0$ or $\mathcal{H}_0:  \beta^\treat_i - \beta^\control_i = 0$, where $ \beta^\treat_i = \frac{1}{\Ttreat} \sum_{\Tcontrol+1}^T C^\treat_{it}$ and  $ \beta^\control_i = \frac{1}{\Ttreat} \sum_{\Tcontrol+1}^T C^\control_{it}$. We consider a 95\% confidence level (the test statistics are within $[-1.96, 1.96]$) over 500 Monte Carlo simulations . The test statistics normalize $\tilde{C}^\treat_{it}- \tilde{C}^\control_{it} $ and $\tilde{\beta}^\treat_i - \tilde{\beta}^\control_i$ by their estimated standard deviation from Equations \eqref{eqn:ite-asy-normal} and \eqref{eqn:ate-asy-normal}. {\bf The estimated standard deviations are estimated without imposing the null hypothesis of $\Lambda^\treat_i - \Lambda^\control_i = 0$}. The observation pattern follows the simultaneous treatment adoption pattern, but has fewer control observations: The observed state variable is defined as $S_i = \mathbbm{1}(\Lambda_i \geq 0 )$. For the units with $S_i = 1$, $50\% $ randomly selected units adopt the treatment from time $0.5 \cdot T$ and the remaining $50\% $ units stay in the control group until the end. For the units with $S_i= 0$, $75\%$ randomly selected units adopt the treatment from time $0.25\cdot T$ and the remaining $25\%$  units stay in the control group until the end. The proportion of acceptance decreases with $N, T, \mu_F$ and $\tilde{\beta}^\treat_i - \tilde{\beta}^\control_i$, implying that the statistical power increases with the data dimensionality, the scale of the treatment effect and the proportion of observed entries in the data. The power decreases slightly if the standard deviations are estimated without imposing the null hypothesis of $\Lambda^\treat_i - \Lambda^\control_i = 0$.}
		\label{tab:treatment-power-fewer-obs}
	\end{table}

	\begin{landscape}
		
		\subsection{Estimation under Misspecification}

		\begin{table}[H]
			\tcaptab{Benchmark and Propensity-Weighted Estimator for Omitted Factors (Estimate $k=1$ Factor)}
			\centering
			\setlength{\tabcolsep}{6pt} 
			\renewcommand{\arraystretch}{1.2} 
			{\small
				\begin{tabular}{l|rr|rr|rr||rr|rr|rr}
					\toprule
					$[\mu_{F,1}, \mu_{F,2}]$ & \multicolumn{6}{c||}{[1,1]} & \multicolumn{6}{c}{[5,0.5]}  \\ 
					$[\sigma_{F,1}, \sigma_{F,2}]$& \multicolumn{6}{c||}{[1,1]} & \multicolumn{6}{c}{[5,0.5]}   \\ 
					\midrule
					$(\mu_\Lambda, \sigma_\Lambda )$  & \multicolumn{2}{c|}{[1,1]} & \multicolumn{2}{c|}{[0,1]} & \multicolumn{2}{c||}{[-1,1]} & \multicolumn{2}{c|}{[1,1]} & \multicolumn{2}{c|}{[0,1]} & \multicolumn{2}{c}{[-1,1]} \\
					method &            XP & XP prop &    XP & XP prop &     XP & XP prop &                XP & XP prop &    XP & XP prop &     XP & XP prop \\
					\midrule
					obs $C^{(0)}_{it}$                  &         0.223 &   0.246 & 0.225 &   0.249 &  0.223 &   0.246 &             0.011 &   0.011 & 0.011 &   0.011 &  0.011 &   0.011 \\
					miss $C^{(0)}_{it}$                 &         0.475 &   0.277 & 0.480 &   0.290 &  0.473 &   0.278 &             0.006 &   0.005 & 0.007 &   0.007 &  0.006 &   0.005 \\
					all $C^{(0)}_{it}$                  &         0.311 &   0.257 & 0.314 &   0.263 &  0.310 &   0.257 &             0.009 &   0.008 & 0.009 &   0.009 &  0.009 &   0.008 \\
					obs $C^{(0)}_{it}$($S=1$)           &         0.182 &   0.252 & 0.183 &   0.253 &  0.180 &   0.249 &             0.746 &   0.752 & 0.753 &   0.759 &  0.746 &   0.752 \\
					miss $C^{(0)}_{it}$($S=1$)          &         0.041 &   0.255 & 0.047 &   0.259 &  0.039 &   0.252 &             0.744 &   0.762 & 0.751 &   0.770 &  0.743 &   0.762 \\
					obs $C^{(0)}_{it}$($S=0$)           &         0.294 &   0.258 & 0.301 &   0.265 &  0.298 &   0.261 &             0.000 &   0.000 & 0.001 &   0.000 &  0.000 &   0.000 \\
					miss $C^{(0)}_{it}$($S=0$)          &         0.717 &   0.294 & 0.724 &   0.311 &  0.720 &   0.297 &             0.002 &   0.001 & 0.003 &   0.002 &  0.002 &   0.001 \\
					obs $C^{(1)}_{it}$                  &         0.402 &   0.270 & 0.407 &   0.281 &  0.401 &   0.271 &             0.006 &   0.005 & 0.007 &   0.006 &  0.006 &   0.005 \\
					$C^{(1)}_{it} - C^{(0)}_{it}$       &         4.224 &   0.829 & 2.438 &   0.684 &  4.216 &   0.836 &             0.025 &   0.024 & 0.023 &   0.022 &  0.025 &   0.024 \\
					$\beta^{(1)}_{i} - \beta^{(0)}_{i}$ &         5.231 &   0.770 & 2.432 &   0.533 &  3.801 &   0.716 &             0.021 &   0.020 & 0.019 &   0.018 &  0.022 &   0.020 \\
					ATE                                 &         1.482 &   0.063 & 0.058 &   0.007 &  0.736 &   0.040 &             0.002 &   0.003 & 0.000 &   0.001 &  0.002 &   0.002 \\
					\bottomrule
				\end{tabular}
			}
			\bnotetab{This table compares the percentage errors for various estimates with the benchmark estimator (XP) and the propensity weighted estimator $\text{XP}_{\text{PROP}}$ for omitted and weak factors. {\bf We estimate only $k=1$ factor while the data is simulated with a two-factor model} and a simultaneous treatment adoption for different means and variances of the latent factors. {\bf For $[\sigma_{F,1}, \sigma_{F,2}]=[5, 0.5]$ the second factor is weak}. In more detail: $Y_{it}^\control = \Lambda_{i,1}^\control F_{t,1}  + \Lambda_{i,2}^\control F_{t,2}  + e_{it}^\control$ and $Y_{it}^\treat=\Lambda_{i,1}^\treat F_{t,1}  + \Lambda_{i,2}^\treat F_{t,2}  + e_{it}^\treat$. The first half of the cross-section depends on the first factor, while the second half depends on the second factor: For $i = 1, \cdots, N/2$, $\Lambda^\control_{i,1} \sim \calN(\mu_\Lambda, \sigma_\Lambda^2)$, $\Lambda^\treat_{i,1} = \Lambda_{i,1}^\control + 0.2$   and $\Lambda^\treat_{i,2} = \Lambda^\control_{i,2} = 0$, and for $i = N/2+1, \cdots, N$, $\Lambda^\treat_{i,1} = \Lambda^\control_{i,1} = 0$, $\Lambda^\control_{i,2} \sim \calN(\mu_\Lambda, \sigma_\Lambda^2)$ and $\Lambda^\treat_{i,2} = \Lambda^\control_{i,2} + 0.2$. {\bf The difference between the treated and control loadings is a constant.} Let $N = 250$, $T = 250$ and $e_{it} \stackrel{\text{i.i.d.}}{\sim}  \calN(0,1)$. The observation pattern depends on an observed state variable defined as $S_i = \mathbbm{1}(\Lambda_{i,2}^\control \neq 0 )$ which only depends on the loadings of the second factor. Once a unit adopts treatment, it stays treated afterwards. For the units with $S_i = 1$, $50\%$ randomly selected units adopt the treatment from time $0.5\cdot T$ and the remaining $50\%$  units stay in the control group until the end.  all units are in the control group. For the units with $S_i= 0$, $90\%$ randomly selected units adopt the treatment from time $0.5\cdot T$ and the remaining $10\%$  units stay in the control group until the end. We report the relative MSE for common components for observed and unobserved treated and control common components. We also report the results conditional on the characteristic $S_i$ and the relative MSE of $\beta^{(1)}_{i} - \beta^{(0)}_{i}$ capturing the average treatment effect over time for each unit and ATE which is the relative MSE of the overall average treatment effect $\sum_{(i,t): W_{it} = 0} \big( \hat{C}^\treat_{it} - \hat{C}^\control_{it} \big)$. The results are generated from 1,000 Monte Carlo simulations. The results show that $\text{XP}_{\text{PROP}}$ can be a more robust estimator for missing observations under misspecification (omitted or weak factors).
			}
		\end{table}

		\begin{table}[H]
			\tcaptab{Benchmark and Propensity-Weighted Estimator for Weak Factors (Estimate $k=2$ Factors)}
			\centering
			{\small
				\begin{tabular}{l|rr|rr|rr||rr|rr|rr}
					\toprule
					$[\mu_{F,1}, \mu_{F,2}]$ & \multicolumn{6}{c||}{[1,1]} & \multicolumn{6}{c}{[5,0.5]}  \\ 
					$[\sigma_{F,1}, \sigma_{F,2}]$& \multicolumn{6}{c||}{[1,1]} & \multicolumn{6}{c}{[5,0.5]}   \\ 
					\midrule
					$(\mu_\Lambda, \sigma_\Lambda )$  & \multicolumn{2}{c|}{[1,1]} & \multicolumn{2}{c|}{[0,1]} & \multicolumn{2}{c||}{[-1,1]} & \multicolumn{2}{c|}{[1,1]} & \multicolumn{2}{c|}{[0,1]} & \multicolumn{2}{c}{[-1,1]} \\
					method &            XP & XP prop &    XP & XP prop &     XP & XP prop &                XP & XP prop &    XP & XP prop &     XP & XP prop \\
					\midrule
					obs $C^{(0)}_{it}$                  &         0.007 &   0.007 & 0.014 &   0.014 &  0.007 &   0.007 &             0.002 &   0.002 & 0.002 &   0.002 &  0.002 &   0.002 \\
					miss $C^{(0)}_{it}$                 &         0.022 &   0.022 & 0.045 &   0.045 &  0.022 &   0.022 &             0.034 &   0.026 & 0.023 &   0.020 &  0.036 &   0.028 \\
					all $C^{(0)}_{it}$                  &         0.012 &   0.012 & 0.025 &   0.025 &  0.012 &   0.012 &             0.018 &   0.014 & 0.013 &   0.011 &  0.019 &   0.015 \\
					obs $C^{(0)}_{it}$($S=1$)           &         0.007 &   0.007 & 0.013 &   0.013 &  0.007 &   0.007 &             0.086 &   0.087 & 0.114 &   0.117 &  0.091 &   0.092 \\
					miss $C^{(0)}_{it}$($S=1$)          &         0.011 &   0.011 & 0.019 &   0.019 &  0.010 &   0.010 &             0.079 &   0.084 & 0.114 &   0.123 &  0.081 &   0.087 \\
					obs $C^{(0)}_{it}$($S=0$)           &         0.008 &   0.008 & 0.016 &   0.016 &  0.008 &   0.008 &             0.000 &   0.000 & 0.001 &   0.001 &  0.001 &   0.001 \\
					miss $C^{(0)}_{it}$($S=0$)          &         0.029 &   0.029 & 0.060 &   0.060 &  0.029 &   0.029 &             0.034 &   0.026 & 0.023 &   0.019 &  0.036 &   0.027 \\
					obs $C^{(1)}_{it}$                  &         0.020 &   0.020 & 0.042 &   0.041 &  0.021 &   0.021 &             0.001 &   0.001 & 0.003 &   0.003 &  0.001 &   0.001 \\
					$C^{(1)}_{it} - C^{(0)}_{it}$       &         0.522 &   0.533 & 0.514 &   0.541 &  0.516 &   0.528 &             1.654 &   1.266 & 0.538 &   0.453 &  1.719 &   1.312 \\
					$\beta^{(1)}_{i} - \beta^{(0)}_{i}$ &         0.317 &   0.321 & 0.301 &   0.309 &  0.316 &   0.319 &             0.492 &   0.405 & 0.177 &   0.164 &  0.496 &   0.413 \\
					ATE                                 &         0.020 &   0.021 & 0.002 &   0.002 &  0.018 &   0.018 &             0.239 &   0.196 & 0.001 &   0.001 &  0.237 &   0.197 \\
					\bottomrule
				\end{tabular}
			}
			\bnotetab{This table compares the percentage errors for various estimates with the benchmark estimator (XP) and the propensity weighted estimator $\text{XP}_{\text{PROP}}$ for omitted and weak factors. {\bf We estimate $k=2$ factors and the data is simulated with a two-factor model} and a simultaneous treatment adoption for different means and variances of the latent factors. {\bf For $[\sigma_{F,1}, \sigma_{F,2}]=[5, 0.5]$ the second factor is weak}. In more detail: $Y_{it}^\control = \Lambda_{i,1}^\control F_{t,1}  + \Lambda_{i,2}^\control F_{t,2}  + e_{it}^\control$ and $Y_{it}^\treat=\Lambda_{i,1}^\treat F_{t,1}  + \Lambda_{i,2}^\treat F_{t,2}  + e_{it}^\treat$. The first half of the cross-section depends on the first factor, while the second half depends on the second factor: For $i = 1, \cdots, N/2$, $\Lambda^\control_{i,1} \sim \calN(\mu_\Lambda, \sigma_\Lambda^2)$, $\Lambda^\treat_{i,1} = \Lambda_{i,1}^\control + 0.2$   and $\Lambda^\treat_{i,2} = \Lambda^\control_{i,2} = 0$, and for $i = N/2+1, \cdots, N$, $\Lambda^\treat_{i,1} = \Lambda^\control_{i,1} = 0$, $\Lambda^\control_{i,2} \sim \calN(\mu_\Lambda, \sigma_\Lambda^2)$ and $\Lambda^\treat_{i,2} = \Lambda^\control_{i,2} + 0.2$. {\bf The difference between the treated and control loadings is a constant.} Let $N = 250$, $T = 250$ and $e_{it} \stackrel{\text{i.i.d.}}{\sim}  \calN(0,1)$. The observation pattern depends on an observed state variable defined as $S_i = \mathbbm{1}(\Lambda_{i,2}^\control \neq 0 )$ which only depends on the loadings of the second factor. Once a unit adopts treatment, it stays treated afterwards. For the units with $S_i = 1$, $50\%$ randomly selected units adopt the treatment from time $0.5\cdot T$ and the remaining $50\%$  units stay in the control group until the end.  all units are in the control group. For the units with $S_i= 0$, $90\%$ randomly selected units adopt the treatment from time $0.5\cdot T$ and the remaining $10\%$  units stay in the control group until the end. We report the relative MSE for common components for observed and unobserved treated and control common components. We also report the results conditional on the characteristic $S_i$ and the relative MSE of $\beta^{(1)}_{i} - \beta^{(0)}_{i}$ capturing the average treatment effect over time for each unit and ATE which is the relative MSE of the overall average treatment effect $\sum_{(i,t): W_{it} = 0} \big( \hat{C}^\treat_{it} - \hat{C}^\control_{it} \big)$. The results are generated from 1,000 Monte Carlo simulations. The results show that $\text{XP}_{\text{PROP}}$ can be a more robust estimator for missing observations under misspecification (weak factors).
			}
		\end{table}
	\end{landscape}

		\begin{table}[h!]
			\tcaptab{Benchmark and Propensity-Weighted Estimator under Model Misspecification}
			\centering
			\setlength{\tabcolsep}{3pt} 
			\renewcommand{\arraystretch}{1.2} 
			{\small \begin{tabular}{l|rr|rr|rr|rr|rr}
					\toprule
					$k$ estimated factors & \multicolumn{2}{c|}{1} & \multicolumn{2}{c|}{2} & \multicolumn{2}{c|}{3} & \multicolumn{2}{c|}{4} & \multicolumn{2}{c}{5} \\ 
					\midrule
					Method &            XP & $\text{XP}_{\text{PROP}}$ &                XP & $\text{XP}_{\text{PROP}}$ &            XP & $\text{XP}_{\text{PROP}}$ &                XP & $\text{XP}_{\text{PROP}}$ &                XP & $\text{XP}_{\text{PROP}}$ \\
					\midrule
					obs $C^{(0)}_{it}$                  & {\bf 0.310} &   0.327 & {\bf 0.070} &   0.072 & {\bf 0.024} &   0.026 & {\bf 0.025} &   0.026 & {\bf 0.030} &   0.032 \\
					miss $C^{(0)}_{it}$                 & 1.163 &   {\bf 0.824} & 0.571 &   {\bf 0.441} & 0.302 &   {\bf 0.295} & {\bf 0.314} &   0.389 & {\bf 0.339} &   0.428 \\
					all $C^{(0)}_{it}$                  & 0.450 &   {\bf 0.391} & 0.149 &   {\bf 0.129} & {\bf 0.077} &   0.078 & {\bf 0.095} &   0.107 & {\bf 0.107} &   0.124 \\
					obs $C^{(0)}_{it}$($S=1$)           & {\bf 0.316} &   0.372 & {\bf 0.070} &   0.079 & {\bf 0.026} &   0.028 & {\bf 0.026} &   0.029 & {\bf 0.032} &   0.036 \\
					miss $C^{(0)}_{it}$($S=1$)          & {\bf 0.287} &   0.459 & {\bf 0.112} &   0.152 & {\bf 0.088} &   { 0.101} & {\bf 0.129} &   0.152 & {\bf 0.144} &   0.172 \\
					obs $C^{(0)}_{it}$($S=0$)           & 0.406 &   {\bf 0.380} & 0.086 &   {\bf 0.080} & {0.027} &   {0.027} & {0.026} &   {0.026} & {\bf 0.031} &   0.032 \\
					miss $C^{(0)}_{it}$($S=0$)          & 1.621 &   {\bf 0.997} & 0.808 &   {\bf 0.574} & 0.392 &   {\bf 0.376} & {\bf 0.372} &   0.467 & {\bf 0.403} &   0.517 \\
					obs $C^{(1)}_{it}$                  & {\bf 0.580} &   0.617 & {\bf 0.283} &   0.286 & {\bf 0.142} &   0.149 & {\bf 0.134} &   0.139 & {\bf 0.131} &   0.135 \\
					$C^{(1)}_{it} - C^{(0)}_{it}$       & 1.160 &   {\bf 1.063} & 0.652 &   {\bf 0.598} & {\bf 0.337} &   0.342 & {\bf 0.324} &   0.387 & {\bf 0.332} &   0.398 \\
					$\beta^{(1)}_{i} - \beta^{(0)}_{i}$ & 6.105 &   {\bf 3.891} & 1.373 &   {\bf 1.026} & {\bf 0.094} &   0.105 & 0.108 &   {\bf 0.105} & 0.121 &   {\bf 0.120} \\
					ATE                                 & 1.379 &   {\bf 1.006} & 0.300 &   {\bf 0.264} & 0.029 &   {\bf 0.027} & 0.222 &   {\bf 0.204} & 0.363 &   {\bf 0.341} \\
					\bottomrule
			\end{tabular}}
			\label{tab:exp}
			\bnotetab{This table compares the percentage errors for various estimates with the benchmark estimator (XP) and the propensity weighted estimator $\text{XP}_{\text{PROP}}$ for a misspecified model. The data is simulated with {\bf non-linear one-factor model} and a simultaneous treatment adoption. The control and treated panel follow $Y_{it}^\control = \exp(\Lambda_i^\control F_t)  + e_{it}^\control$ and $Y_{it}^\treat= \exp(\Lambda_i^\treat F_t) + e_{it}^\treat$, where $F_t \sim \calN(0,1)$, $\Lambda_i \sim \calN(0,0.25) $ $\Lambda^\treat_i = \Lambda_i^\control + \calN(0.2,0.25)$ and $e_{it} \stackrel{\iid}{\sim} \calN(0,1)$. We set $N = 250$, $T = 250$. The observation pattern depends on an observed state variable defined as $S_i = \mathbbm{1}(\Lambda_i^\control \geq 0 )$.  Once a unit adopts treatment, it stays treated afterwards. For the units with $S_i = 1$, $50\%$ randomly selected units adopt the treatment from time $0.5\cdot T$ and the remaining $50\%$  units stay in the control group until the end.  all units are in the control group. For the units with $S_i= 0$, $90\%$ randomly selected units adopt the treatment from time $0.5\cdot T$ and the remaining $10\%$  units stay in the control group until the end.  We report the relative MSE for common components for observed and unobserved treated and control common components for different numbers of estimated factors. We also report the results conditional on the characteristic $S_i$ and the relative MSE of $\beta^{(1)}_{i} - \beta^{(0)}_{i}$ capturing the average treatment effect over time for each unit and ATE which is the relative MSE of the overall average treatment effect $\sum_{(i,t): W_{it} = 0} \big( \hat{C}^\treat_{it} - \hat{C}^\control_{it} \big)$. The results are generated from 1,000 Monte Carlo simulations. The results show that $\text{XP}_{\text{PROP}}$ can be a more robust estimator for missing observations under misspecification (non-linear functional form).
			}
			\setlength{\tabcolsep}{6pt} 
			\renewcommand{\arraystretch}{1.0} 
		\end{table}

		\begin{table}[H]
			\tcaptab{Benchmark and Propensity-Weighted Estimator under Misspecification}
			\centering
			\setlength{\tabcolsep}{3pt} 
			\renewcommand{\arraystretch}{1.2} 
			{\small
				\begin{tabular}{l|rr|rr|rr|rr|rr}
					\toprule
					$k$ estimated factors& \multicolumn{2}{c|}{1} & \multicolumn{2}{c|}{2} & \multicolumn{2}{c|}{3} & \multicolumn{2}{c|}{4} & \multicolumn{2}{c}{5} \\ 
					\midrule
					Method &            XP & $\text{XP}_{\text{PROP}}$ &                XP & $\text{XP}_{\text{PROP}}$ &            XP & $\text{XP}_{\text{PROP}}$ &                XP & $\text{XP}_{\text{PROP}}$ &                XP & $\text{XP}_{\text{PROP}}$ \\
					\midrule
					
					obs $C^{(0)}_{it}$                  & 0.051 &   0.051 & 0.011 &   0.011 & 0.005 &   0.005 & 0.003 &   0.003 & 0.003 &   0.003 \\
					miss $C^{(0)}_{it}$                 & 1.254 &   1.237 & 1.081 &   1.042 & 1.048 &   1.084 & 0.924 &   0.914 & 0.863 &   0.880 \\
					all $C^{(0)}_{it}$                  & 0.138 &   0.137 & 0.198 &   0.192 & 0.192 &   0.221 & 0.164 &   0.169 & 0.148 &   0.154 \\
					obs $C^{(0)}_{it}$($S=1$)           & 0.087 &   0.087 & 0.009 &   0.009 & 0.005 &   0.005 & 0.003 &   0.003 & 0.003 &   0.003 \\
					miss $C^{(0)}_{it}$($S=1$)          & 0.677 &   0.677 & 0.521 &   0.521 & 0.524 &   0.576 & 0.472 &   0.482 & 0.422 &   0.440 \\
					obs $C^{(0)}_{it}$($S=0$)           & 0.928 &   0.928 & 0.566 &   0.559 & 0.181 &   0.172 & 0.041 &   0.039 & 0.021 &   0.021 \\
					miss $C^{(0)}_{it}$($S=0$)          & 5.466 &   3.477 & 4.286 &   3.442 & 3.595 &   3.164 & 2.579 &   2.235 & 2.354 &   2.384 \\
					obs $C^{(1)}_{it}$                  & 0.204 &   0.208 & 0.095 &   0.100 & 0.051 &   0.056 & 0.032 &   0.037 & 0.027 &   0.031 \\
					$C^{(1)}_{it} - C^{(0)}_{it}$       & 0.230 &   0.234 & 0.112 &   0.117 & 0.069 &   0.074 & 0.049 &   0.054 & 0.043 &   0.047 \\
					$\beta^{(1)}_{i} - \beta^{(0)}_{i}$ & 0.915 &   0.887 & 0.486 &   0.534 & 0.181 &   0.198 & 0.099 &   0.110 & 0.068 &   0.074 \\
					ATE                                 & 0.817 &   0.792 & 0.417 &   0.454 & 0.169 &   0.179 & 0.138 &   0.149 & 0.108 &   0.115 \\
					\bottomrule
			\end{tabular}}
			\bnotetab{This table compares the percentage errors for various estimates with the benchmark estimator (XP) and the propensity weighted estimator $\text{XP}_{\text{PROP}}$ for a misspecified model. The data is simulated with {\bf non-linear one-factor model} and a simultaneous treatment adoption. The control and treated panel follow $Y_{it}^\control = \exp(\Lambda_i^\control F_t)  + e_{it}^\control$ and $Y_{it}^\treat= \exp(\Lambda_i^\treat F_t) + e_{it}^\treat$, where $F_t \sim \calN(0,1)$, $\Lambda_i \sim \calN(0,0.25) $ $\Lambda^\treat_i = \Lambda_i^\control + \calN(0.2,0.25)$ and $e_{it} \stackrel{\text{i.i.d.}}{\sim} \calN(0,1)$. We set $N = 250$, $T = 250$. The observation pattern depends on an observed state variable defined as $S_i = \mathbbm{1}(\Lambda_i^\control \geq \mu_\Lambda )$ \textbf{(here loadings have a negative mean, $\mu_\Lambda = -0.5$.} When $\Lambda_i^\control \geq \mu_\Lambda$,  $S_i = 1$). Once a unit adopts treatment, it stays treated afterwards. For the units with $S_i = 1$, $50\%$ randomly selected units adopt the treatment from time $0.5\cdot T$ and the remaining $50\%$  units stay in the control group until the end.  all units are in the control group. For the units with $S_i= 0$, $90\%$ randomly selected units adopt the treatment from time $0.5\cdot T$ and the remaining $10\%$  units stay in the control group until the end.  We report the relative MSE for common components for observed and unobserved treated and control common components for different numbers of estimated factors. We also report the results conditional on the characteristic $S_i$ and the relative MSE of $\beta^{(1)}_{i} - \beta^{(0)}_{i}$ capturing the average treatment effect over time for each unit and ATE which is the relative MSE of the overall average treatment effect $\sum_{(i,t): W_{it} = 0} \big( \hat{C}^\treat_{it} - \hat{C}^\control_{it} \big)$. The results are generated from 1,000 Monte Carlo simulations. The results show that $\text{XP}_{\text{PROP}}$ can be a more robust estimator for missing observations under misspecification (non-linear functional form).
			}
		\end{table}
		
		\begin{table}[H]
			\tcaptab{Benchmark and Propensity-Weighted Estimator under Misspecification}
			\centering
			\setlength{\tabcolsep}{3pt} 
			\renewcommand{\arraystretch}{1.2} 
			{\small
				\begin{tabular}{l|rr|rr|rr|rr|rr}
					\toprule
					$k$ estimated factors & \multicolumn{2}{c|}{1} & \multicolumn{2}{c|}{2} & \multicolumn{2}{c|}{3} & \multicolumn{2}{c|}{4} & \multicolumn{2}{c}{5} \\ 
					\midrule
					Method &            XP & $\text{XP}_{\text{PROP}}$ &                XP & $\text{XP}_{\text{PROP}}$ &            XP & $\text{XP}_{\text{PROP}}$ &                XP & $\text{XP}_{\text{PROP}}$ &                XP & $\text{XP}_{\text{PROP}}$ \\
					\midrule
					obs $C^{(0)}_{it}$                  & 0.026 &   0.026 & 0.008 &   0.008 & 0.004 &   0.004 & 0.004 &   0.004 & 0.003 &   0.003 \\
					miss $C^{(0)}_{it}$                 & 0.345 &   0.321 & 0.509 &   0.457 & 0.377 &   0.395 & 0.324 &   0.342 & 0.312 &   0.326 \\
					all $C^{(0)}_{it}$                  & 0.076 &   0.072 & 0.192 &   0.174 & 0.125 &   0.135 & 0.109 &   0.115 & 0.104 &   0.110 \\
					obs $C^{(0)}_{it}$($S=1$)           & 0.025 &   0.025 & 0.007 &   0.007 & 0.004 &   0.004 & 0.003 &   0.003 & 0.003 &   0.003 \\
					miss $C^{(0)}_{it}$($S=1$)          & 0.305 &   0.305 & 0.463 &   0.436 & 0.366 &   0.387 & 0.319 &   0.337 & 0.306 &   0.319 \\
					obs $C^{(0)}_{it}$($S=0$)           & 0.959 &   0.959 & 0.758 &   0.755 & 0.369 &   0.361 & 0.118 &   0.115 & 0.072 &   0.073 \\
					miss $C^{(0)}_{it}$($S=0$)          & 3.968 &   1.443 & 2.275 &   1.372 & 0.910 &   0.798 & 0.550 &   0.531 & 0.496 &   0.535 \\
					obs $C^{(1)}_{it}$                  & 0.107 &   0.108 & 0.034 &   0.039 & 0.019 &   0.023 & 0.014 &   0.016 & 0.012 &   0.014 \\
					$C^{(1)}_{it} - C^{(0)}_{it}$       & 0.136 &   0.138 & 0.070 &   0.072 & 0.048 &   0.052 & 0.039 &   0.042 & 0.037 &   0.039 \\
					$\beta^{(1)}_{i} - \beta^{(0)}_{i}$ & 0.670 &   0.664 & 0.241 &   0.226 & 0.109 &   0.117 & 0.071 &   0.074 & 0.067 &   0.068 \\
					ATE                                 & 0.446 &   0.444 & 0.166 &   0.155 & 0.081 &   0.083 & 0.058 &   0.060 & 0.057 &   0.058 \\
					\bottomrule
			\end{tabular}}
			\bnotetab{This table compares the percentage errors for various estimates with the benchmark estimator (XP) and the propensity weighted estimator $\text{XP}_{\text{PROP}}$ for a misspecified model. The data is simulated with {\bf non-linear one-factor model} and a simultaneous treatment adoption. The control and treated panel follow $Y_{it}^\control = \exp(\Lambda_i^\control F_t)  + e_{it}^\control$ and $Y_{it}^\treat= \exp(\Lambda_i^\treat F_t) + e_{it}^\treat$, where $F_t \sim \calN(0,1)$, $\Lambda_i \sim \calN(0,0.25) $ $\Lambda^\treat_i = \Lambda_i^\control + \calN(0.2,0.25)$ and $e_{it} \stackrel{\text{i.i.d.}}{\sim} \calN(0,1)$. We set $N = 250$, $T = 250$. The observation pattern depends on an observed state variable defined as $S_i = \mathbbm{1}(\Lambda_i^\control \geq \mu_\Lambda )$ \textbf{(here loadings have a positive mean, $\mu_\Lambda = 0.25$.} When $\Lambda_i^\control \geq \mu_\Lambda$,  $S_i = 1$). Once a unit adopts treatment, it stays treated afterwards. For the units with $S_i = 1$, $50\%$ randomly selected units adopt the treatment from time $0.5\cdot T$ and the remaining $50\%$  units stay in the control group until the end.  all units are in the control group. For the units with $S_i= 0$, $90\%$ randomly selected units adopt the treatment from time $0.5\cdot T$ and the remaining $10\%$  units stay in the control group until the end.  We report the relative MSE for common components for observed and unobserved treated and control common components for different numbers of estimated factors. We also report the results conditional on the characteristic $S_i$ and the relative MSE of $\beta^{(1)}_{i} - \beta^{(0)}_{i}$ capturing the average treatment effect over time for each unit and ATE which is the relative MSE of the overall average treatment effect $\sum_{(i,t): W_{it} = 0} \big( \hat{C}^\treat_{it} - \hat{C}^\control_{it} \big)$. The results are generated from 1,000 Monte Carlo simulations. The results show that $\text{XP}_{\text{PROP}}$ can be a more robust estimator for missing observations under misspecification (non-linear functional form).
			}		
		\end{table}	
		
		\newpage
		\section{Proofs}
		
		\subsection{Proof of Proposition \ref{prop:compare-su}}

		In this note, we assume every entry is randomly missing with observed probability $q$. For a direct comparison, we follow the order of estimation in \cite{jin2020factor} and switch the role of factors and loadings in our all-purpose estimator: We first estimate the factors from the time-series sample covariance matrix, and then estimate the loadings from a time-series regression of the observed outcomes on the estimated factors.

		Let the time-series sample covariance matrix be $\Sigma^\dagger$, where
		\begin{equation}\label{eqn:est-sample-covariance}
			\tilde{\Sigma}^\dagger_{st} = \frac{1}{|\mathcal{O}_{st}|} \sum_{i \in \mathcal{O}_{st}} Y_{is} Y_{it}.
		\end{equation}
		The estimated factors $\tilde{F}$ are $\sqrt{T}$ times the eigenvectors of the $r$ largest eigenvalues of the sample covariance matrix, that is,
		\begin{equation}\label{eqn:pca-estimate-factors}
			\frac{1}{T} \tilde{\Sigma}^\dagger \tilde{F} = \tilde{F} \tilde{V}, 
		\end{equation}
		where $\tilde{V}$ is a diagonal matrix. Then for every unit $i$, we regress the observed $Y_{it}$ on $\tilde{F}$ to estimate the loadings
		\begin{equation}\label{eqn:reg-estimate-loadings}
			\tilde{\Lambda}_i = \left(\sum_{t = 1}^N W_{it} \tilde F_t \tilde F_t^\T  \right)^\I  \left( \sum_{t = 1}^N W_{it} \tilde F_t Y_{it}    \right).
		\end{equation}

		\subsubsection{Proof of Proposition \ref{prop:compare-su}.1}
		
		Based on the proof of Theorem \ref{theorem:asy-normal-equal-weight}.1, we have the following expansion of the estimated factors $\tilde{F}_t$,
		\begin{align*}
			\nonumber \sqrt{N}(\tilde{F}_t - H^\T F_t )   =&   \tilde D^{-1}  \frac{\sqrt{N}}{T} \sum_{s=1}^T  H^\T F_s F_s^\T \frac{1}{|\tlo_{st}|}  \sum_{i = 1}^N W_{it} \Lambda_i e_{it} \\
			& + \tilde D^\I \cdot \frac{\sqrt{N}}{T} \sum_{s=1}^T H^\T F_s F_s^\T \Lp \frac{1}{|\tlo_{st}|}  \sum_{i = 1}^N W_{it} \Lambda_i \Lambda_i^\T - \frac{1}{N} \sum_{i = 1}^N \Lambda_i \Lambda_i^\T \Rp \cdot F_t + o_P(1).
		\end{align*}
		
		Our approach to estimate the factors differs from \cite{jin2020factor} in that  we adjust each entry in the sample covariance by the number of units that are observed in both time periods (i.e., the denominator $\mathcal{O}_{st}$ in Equation \ref{eqn:est-sample-covariance}), while \cite{jin2020factor} adjust for the overall observed proportion (i.e., the denominator in Equation \ref{eqn:est-sample-covariance} is replaced by $N \hat{q}$ where $\hat{q} = \frac{1}{NT} \sum_{i=1}^N \sum_{t=1}^T W_{it}$).
		
		\paragraph{First term: $\frac{\sqrt{N}}{T} \sum_{s=1}^T   F_s F_s^\T \frac{1}{|\tlo_{st}|}  \sum_{i = 1}^N W_{it} \Lambda_i e_{it} $.}
		\texttt{}\\
		The corresponding term in the initial estimator in \cite{jin2020factor} is $\big(\frac{1}{T} \sum_{t=1}^T  F_t F_t^\T \big) \cdot \big( \frac{1}{\sqrt{N} q} \sum_{i = 1}^N W_{is} \Lambda_i e_{is} \big)$. 
		
		When every entry is missing at random, we have $\frac{\sqrt{N}}{|\tlo_{st}|} = \frac{1}{\sqrt{N} \cdot (|\tlo_{st}|/N)} \rightarrow \frac{1}{N \sqrt{q}}$.  \cite{jin2020factor} differ from our term in that \cite{jin2020factor} disentangle the average over the time dimension from the average over the unit dimension. Under the simplified factor model, we can disentangle two averages and show that our term is asymptotically equivalent to the corresponding term in \cite{jin2020factor}. 
		
		Under the simplified factor model, and using a similar proof as the third step in the proof of Proposition \ref{prop:simple-assump-imply-general-assump}.2 on p23 in the online appendix, we can show that this term is asymptotically normal with the asymptotic variance
		\begin{align*}
			\AVar \Lp\frac{\sqrt{N}}{T} \sum_{t=1}^T  F_t F_t^\T \frac{1}{|\tlo_{st}|}  \sum_{i = 1}^N W_{is} \Lambda_i e_{is} \Rp =  \Lp \lim_{T \rightarrow \infty}\frac{1}{T^2} \sum_{s= 1}^T \sum_{u = 1}^T   \frac{q_{st,ut}}{q_{st}q_{ut}}  \Rp   \Sigma_F  \Sigma_\Lambda  \Sigma_F = \omega_{tt} \Sigma_F  \Sigma_\Lambda   \Sigma_F \sigma_e^2, 
		\end{align*}
		where $\omega_{tt} = \frac{1}{q}$. Therefore, the asymptotic variance coincides with the asymptotic variance of the initial estimator of factors in \cite{jin2020factor}. 
		
		\paragraph{Second term: $\frac{\sqrt{N}}{T} \sum_{s=1}^T F_s F_s^\T \Lp \frac{1}{|\tlo_{st}|}  \sum_{i = 1}^N W_{it} \Lambda_i \Lambda_i^\T - \frac{1}{N} \sum_{i = 1}^N \Lambda_i \Lambda_i^\T \Rp \cdot F_t $.} 
		\texttt{}\\
		The corresponding term in the initial estimator in \cite{jin2020factor} is \\ $\Lp \frac{1}{T} \sum_{s=1}^T F_s F_s^\T \Rp \cdot \Lp \frac{1}{\sqrt{N} q}  \sum_{i = 1}^N W_{it} \Lambda_i \Lambda_i^\T - \frac{1}{\sqrt{N}} \sum_{i = 1}^N \Lambda_i \Lambda_i^\T \Rp \cdot F_t $. 
		
		Under the simplified factor model, we can disentangle the average over the time dimension from the average over the unit dimension, and show that our term is asymptotically equivalent to the corresponding term in \cite{jin2020factor}. Using a similar proof as the fourth step in the proof of Proposition \ref{prop:simple-assump-imply-general-assump}.1(b), we can show that this term has a stable limiting distribution with the asymptotic variance
		\begin{align*}
			& \AVar \Lp \frac{\sqrt{N}}{T} \sum_{s=1}^T F_s F_s^\T \Lp \frac{1}{|\tlo_{st}|}  \sum_{i = 1}^N W_{it} \Lambda_i \Lambda_i^\T - \frac{1}{N} \sum_{i = 1}^N \Lambda_i \Lambda_i^\T \Rp \cdot F_t  \Rp \\
			=& \Lp \lim_{T \rightarrow \infty}\frac{1}{T^2} \sum_{s= 1}^T \sum_{u = 1}^T   \frac{q_{st,ut}}{q_{st}q_{ut}}  - 1  \Rp (F_t^\T  \otimes  \Sigma_F )  \Xi_\Lambda   (F_t \otimes  \Sigma_F) = (\omega_{tt} -1)(F_t^\T \otimes  \Sigma_F )  \Xi_\Lambda   (F_t \otimes  \Sigma_F) 
		\end{align*}
		and $\omega_{tt} - 1 = \frac{1 - q}{q}$ when every entry is missing at random. Therefore, the asymptotic variance coincides with the asymptotic variance of the initial estimator in \cite{jin2020factor}.
		
		\paragraph{Combining two terms.} 
		\texttt{}\\
		Since both terms have the same asymptotic variance as those in the initial estimator in \cite{jin2020factor}, our estimated factors are asymptotically the same as the initial estimates of factors in \cite{jin2020factor}.
		Since the iterated estimator is more efficient than the initial estimator in \cite{jin2020factor}, our estimated factors are asymptotically less efficient than the iterated estimates of factors in \cite{jin2020factor}.
		
		\subsubsection{Proof of Proposition \ref{prop:compare-su}.2}
		Based on the proof of Theorem \ref{theorem:asy-normal-equal-weight}.2, we have the following expansion of the estimated loadings $\tilde \Lambda_i$,
		\begin{align*}
			\sqrt{\delta_{NT}}  (\tilde \Lambda_i - H^\I \Lambda_i) =& \bigg( \frac{1}{T} \sum_{t=1}^T W_{it} H^\T F_t F_t^\T H \bigg)^{-1}  \bigg(    \frac{\sqrt{\delta_{NT}}}{T} \sum_{t=1}^T W_{it} H^\T F_t e_{it} \bigg) \\
			& + \bigg( \frac{1}{T} \sum_{t=1}^T W_{it} H^\T F_t F_t^\T H \bigg)^{-1}  \bigg( \sqrt{\delta_{NT}} H^\T \*X_i^\T H \tilde{D}^\I H^\I \Lambda_i \bigg), 
		\end{align*}
		where $\*X_i = \frac{1}{T^2} \sum_{s = 1}^T \sum_{t=1}^T F_s F_s^\T \Lp \frac{1}{|\tlo_{st}|}  \sum_{i = 1}^N W_{it} \Lambda_i \Lambda_i^\T - \frac{1}{N} \sum_{i = 1}^N \Lambda_i \Lambda_i^\T \Rp W_{it} F_t F_t^\T$.
		
		Our approach to estimate loadings differs from \cite{jin2020factor} in that we only regress on observed time periods, while \cite{jin2020factor} regress on all time periods using $\tilde{X}_{it} = X_{it} W_{it}$ and adjust for the observed proportion $\hat{q}$ in the regression. 
		
		\paragraph{First term: $\bigg( \frac{1}{T} \sum_{t=1}^T W_{it} H^\T F_t F_t^\T H \bigg)^{-1}  \bigg(  \frac{\sqrt{\delta_{NT}}}{T} \sum_{t=1}^T W_{it} H^\T F_t e_{it} \bigg)$.}
		\texttt{}\\
		The corresponding term in the initial estimator in \cite{jin2020factor} is $ \frac{H^\T}{\sqrt{T} q }  \cdot \sum_{t=1}^T W_{it} F_t e_{it}$. Note that \cite{jin2020factor} regress on all time periods and then the term $\Big( \frac{1}{T} \sum_{t=1}^T W_{it} H^\T F_t F_t^\T H \Big)^{-1} $ in their case equals to an identity matrix given $\tilde{F}^\T \tilde{F}/T = I_r$. From from Lemma 7 on p41 in the online appendix, we can show $H H^\T = \Sigma_F^\I + O_P\big(\frac{1}{\sqrt{\delta_{NT}}} \big)$, and then the first term equals 
		\begin{align*}
			& \bigg( \frac{1}{T} \sum_{t=1}^T W_{it} H^\T F_t F_t^\T H \bigg)^{-1}  \bigg(  \frac{\sqrt{\delta_{NT}}}{T} \sum_{t=1}^T W_{it} H^\T F_t e_{it} \bigg) \\ 
			=& H^\T \underbrace{(H H^\T)^\I}_{\Sigma_F + O_P\big(\frac{1}{\sqrt{\delta_{NT}}} \big)} \underbrace{\bigg( \frac{1}{T} \sum_{t=1}^T W_{it} F_t F_t^\T \bigg)^{-1}}_{\frac{1}{q} \Sigma_F^\I + O_P\big(\frac{1}{\sqrt{T}} \big)}   \bigg(  \frac{\sqrt{\delta_{NT}}}{T} \sum_{t=1}^T W_{it} F_t e_{it} \bigg) \\
			=& H^\T   \bigg(  \frac{\sqrt{\delta_{NT}}}{T q} \sum_{t=1}^T W_{it} F_t e_{it} \bigg) + O_P\bigg(\frac{1}{\sqrt{\delta_{NT}}}\bigg),
		\end{align*}
		where $\frac{1}{T} \sum_{t=1}^T W_{it} F_t F_t^\T = q \Sigma_F + O_P\big(\frac{1}{\sqrt{T}}\big)$ holds under the assumption that the observation pattern is exogenous and does not depend on the value of $F_t$. Therefore this term is asymptotically the same as the corresponding term in the initial estimator in \cite{jin2020factor}. 
		
		\paragraph{Second term: $\bigg(\frac{1}{T} \sum_{t=1}^T W_{it} H^\T F_t F_t^\T H \bigg)^{-1}  \bigg( \sqrt{\delta_{NT}} H^\T \*X_i^\T H \tilde{D}^\I H^\I \Lambda_i \bigg)$.}
		\texttt{}\\
		The corresponding term in the initial estimator in  \cite{jin2020factor} is $ \frac{H^\T}{\sqrt{T} q }  \cdot \sum_{t=1}^T W_{it} F_t F_t^\T \Lambda_i (W_{it} - 1) $. 
		
		Under the simplified factor model, and using a similar proof as the fifth step in the proof of Proposition \ref{prop:simple-assump-imply-general-assump}.1(b), we can show that $H^\T \*X_i^\T H \tilde{D}^\I H^\I \Lambda_i$ has a stable limiting distribution with the asymptotic variance
		\begin{align*}
			&	\AVar(\sqrt{N} H^\T \*X_i^\T H \tilde{D}^\I H^\I \Lambda_i) \\ =& \Lp \lim_{T \rightarrow \infty}\frac{1}{T^4} \sum_{s= 1}^T \sum_{t= 1}^T \sum_{u = 1}^T \sum_{v = 1}^T   \frac{q_{st,uv}}{q_{st}q_{uv}}  - 1  \Rp \cdot q^2 \cdot \big( I_r \otimes (\Lambda_i^{\top}  \Sigma_{\Lambda}^{-1} \Sigma_F^{-1}) \big)  ( \Sigma_{F} \otimes  \Sigma_F )  \Xi_\Lambda   ( \Sigma_{F} \otimes  \Sigma_F )  \big( I_r \otimes ( \Sigma_F^{-1}  \Sigma_{\Lambda}^{-1} \Lambda_i ) \big) .
		\end{align*}
		Note that $\lim_{T \rightarrow \infty}\frac{1}{T^4} \sum_{s= 1}^T \sum_{t= 1}^T \sum_{u = 1}^T \sum_{v = 1}^T   \frac{q_{st,uv}}{q_{st}q_{uv}} = 1$ when every entry is missing at random, and then $\AVar(\*X_i) = 0$. Therefore, this term vanishes when every entry is missing at random. As a comparison, the corresponding term in the initial estimator in \cite{jin2020factor} is non-negligible. 
		
		\paragraph{Combining two terms.} Since the first term has the same asymptotic variance, and the second term has a smaller asymptotic variance as those in the initial estimator in \cite{jin2020factor}, our estimated factors are asymptotically more efficient than the initial estimates of factors in \cite{jin2020factor}. For the iterated estimator in \cite{jin2020factor}, the asymptotic variance is $\lim_{(\ell, N, T) \rightarrow \infty} \AVar(\hat{\Lambda}^{(\ell)}) = (Q^\T)^\I \Phi_{1g,i}(q) Q^\I$  and $\Phi_{1g,i}(q) = \AVar\big( \frac{1}{\sqrt{T} q }  \cdot \sum_{t=1}^T W_{it} F_t e_{it} \big)$, so the second term vanishes and the first term is asymptotically the same (our definition and \cite{jin2020factor}'s definition of $Q$ are identical under the simplified factor model). Therefore, our estimator is asymptotically the same as the iterated estimator in \cite{jin2020factor}.

		\subsection{Proof of Proposition \ref{prop:simple-assump-imply-general-assump}: Simplified Model}

		In this section, we prove Proposition \ref{prop:simple-assump-imply-general-assump}.1,  \ref{prop:simple-assump-imply-general-assump}.2 and \ref{prop:simple-assump-imply-general-assump}.3 step by step.
		
		\subsubsection{Proof of Proposition \ref{prop:simple-assump-imply-general-assump}.1(a)}
		In this proof, we suppose Assumption \ref{ass:obs-equal-weight} holds without further statement. We show  Assumption \ref{ass:factor-model}.\ref{ass:factor} holds under Assumption \ref{ass:simple-factor-model}.1, \ref{ass:factor-model}.\ref{ass:loading} holds under \ref{ass:simple-factor-model}.2,  \ref{ass:factor-model}.\ref{ass:error} holds under \ref{ass:simple-factor-model}.3, and \ref{ass:factor-model}.\ref{ass:factor-error} holds under \ref{ass:simple-factor-model}.4. 
		\paragraph{Step 1: Show that Assumption \ref{ass:factor-model}.\ref{ass:factor} holds under Assumption \ref{ass:simple-factor-model}.1.} \texttt{} \\
		\textit{Step 1.1: Show that $\frac{1}{T} \sum_{t=1}^T F_t F_t^\T \xrightarrow{P} \Sigma_F$ holds under Assumption \ref{ass:simple-factor-model}.1.}

		From LLN, under Assumption \ref{ass:simple-factor-model}.1, $\frac{1}{T} \sum_{t=1}^T F_t F_t^\T \xrightarrow{P} \Sigma_F$.  \\ \\
		\textit{Step 1.2: Show that $\+E\norm{\sqrt{T} \Lp \frac{1}{T} \sum_{t =1}^T F_t F_t^\T - \Sigma_F \Rp}^2  \leq M$ holds under Assumption \ref{ass:simple-factor-model}.1.} 
		
		Let $F_{t,j}$ be the $j$-th entry of $F_t$ and $\Sigma_{F, jk}$ be the $(j,k)$-th entry of $\Sigma_F$.
		\begin{eqnarray*}
			\Scale[1]{\+E \norm{\sqrt{T} \Lp \frac{1}{T} \sum_{t =1}^T F_t F_t^\T - \Sigma_F \Rp}^2} &=& \Scale[1]{T \sum_{j,k} \+E[\frac{1}{T} \sum_{t =1}^T F_{t,j} F_{t,k} - \Sigma_{F,jk}  ]^2}  \\
			&=&  \Scale[1]{  \sum_{j,k}  \+E[F_{t,j} F_{t,k} - \Sigma_{F,jk}  ]^2 } \leq M.
		\end{eqnarray*} 
		\textit{Step 1.3: Show that $\frac{1}{|\tlq_{ij}|} \sum_{t \in \tlq_{ij}} F_t F_t^\T \xrightarrow{P} \Sigma_F$ and \\ $\+E\norm{\sqrt{|\tlq_{ij}|} \Lp \frac{1}{|\tlq_{ij}|} \sum_{t \in \tlq_{ij}} F_t F_t^\T - \Sigma_F \Rp}^2  \leq M$ hold under Assumption \ref{ass:simple-factor-model}.1.}  
		
		Since $W_{it}$ is independent of $F_s$ for all $i, t, s$ from Assumption \ref{ass:obs-equal-weight}, following the same proof, under Assumption \ref{ass:simple-factor-model}.1, we have  for any $\tlq_{ij}$, $\frac{1}{|\tlq_{ij}|} \sum_{t \in \tlq_{ij}} F_t F_t^\T \xrightarrow{P} \Sigma_F$ and \\ $\+E\norm{\sqrt{|\tlq_{ij}|} \Lp \frac{1}{|\tlq_{ij}|} \sum_{t \in \tlq_{ij}} F_t F_t^\T - \Sigma_F \Rp}^2  \leq M$. 
		\paragraph{Step 2: Show that Assumption \ref{ass:factor-model}.\ref{ass:loading} holds under  Assumption \ref{ass:simple-factor-model}.2.} \texttt{} \\
		\textit{Step 2.1: Show that $\frac{1}{N} \sum_{i = 1}^N \Lambda_i \Lambda_i^\T \xrightarrow{P} \Sigma_{\Lambda}  $ and
			$\+E \norm{\sqrt{N} \Lp  \frac{1}{N} \sum_{i=1}^N \Lambda_i \Lambda_i^\T  - \Sigma_\Lambda \Rp} \leq M $ hold under Assumption \ref{ass:simple-factor-model}.2.}  \\ \\
		Under Assumption \ref{ass:simple-factor-model}.2, similar to the proof for the factors,   we have $\frac{1}{N} \sum_{i = 1}^N \Lambda_i \Lambda_i^\T \xrightarrow{P} \Sigma_{\Lambda}  $ and 
		$\+E \norm{\sqrt{N} \Lp  \frac{1}{N} \sum_{i=1}^N \Lambda_i \Lambda_i^\T  - \Sigma_\Lambda \Rp} \leq M $. 
		\paragraph{Step 3: Show that Assumption \ref{ass:factor-model}.\ref{ass:error} holds under Assumption \ref{ass:simple-factor-model}.3}  \texttt{} \\
		\textit{Step 3.1: Show that $\+E[e_{is}e_{it}] = \gamma_{st,i}$ with $|\gamma_{st,i}| \leq \gamma_{st}$ for some $\gamma_{st}$ and all $i$. For all $t$, $\sum_{s=1}^T \gamma_{st} \leq M$ holds under Assumption \ref{ass:simple-factor-model}.3.} \\ \\
		For $s \neq t$, $\+E[e_{is}e_{it}]  = 0$ and $\gamma_{st} = 0$; for $s = t$, $\+E[e_{is}e_{it}] = \+E[e_{is}^2] = \sigma_e^2$ and $\gamma_{ss} = \sigma_e^2$. Then $\sum_{s=1}^T \gamma_{st} = \sigma_e^2 \leq M$.  \\ \\
		\textit{Step 3.2: Show that $\+E[e_{it}e_{jt}] = \tau_{ij,t}$ with $|\tau_{ij,t}| \leq \tau_{ij}$  for some $\tau_{ij}$ and all $t$. For all $i$, $\sum_{j = 1}^N \tau_{ij} \leq M$ holds under Assumption \ref{ass:simple-factor-model}.3.} \\ \\
		For $i \neq j$, $\+E[e_{it}e_{jt}] = 0$ and $\tau_{ij} = 0$; for $i = j$, $\+E[e_{it}e_{jt}] = \+E[e_{it}^2] = \sigma_e^2$ and $\tau_{ii} = \sigma_e^2$. Then $\sum_{i=1}^N \tau_{ij} = \sigma_e^2 \leq M$. \\ \\
		\textit{Step 3.3: Show that $\+E[e_{it}e_{js}] = \tau_{ij,ts}$ and $\sum_{j=1}^N \sum_{s = 1}^T |\tau_{ij,ts}| \leq M$ for all $i$ and $t$ holds under Assumption \ref{ass:simple-factor-model}.3.} \\ \\
		If $i \neq j$ or $s \neq t$, $\+E[e_{it} e_{js}] = 0$ and $\tau_{ij,ts}$; if $i = j$ and $s = t$, $\+E[e_{it} e_{js}] = \sigma_e^2$ and $\tau_{ii,tt} = \sigma_e^2$. Then $\sum_{j = 1}^N \sum_{s = 1}^T |\tau_{ij,ts}| = \sigma_e^2 \leq M$. \\ \\
		\textit{Step 3.4: Show that for all $i$ and $j$, $\+E \left\vert \frac{1}{|\tlq_{ij}|^{1/2}} \sum_{t \in \tlq_{ij}} \Lp e_{it}e_{jt} -  \+E[e_{it}e_{jt}] \Rp \right\vert^4 \leq M$ holds under  Assumption \ref{ass:simple-factor-model}.3.} \\ \\
		Denote $v_t = e_{it} e_{jt} - \+E[e_{it} e_{jt}]$. Then for any set $\tls$, 
		\begin{eqnarray*}
			\+E\Ls \frac{1}{|\tls|} \sum_{t \in \tls} v_t \Rs^4 &=& \frac{1}{|\tls|^2} \+E \Ls \sum_{t,s,u,w \in \tls} v_t v_s v_u v_w \Rs  = \frac{1}{|\tls|^2} \Lp \sum_{t \in \tls} \+E[v_t^4] + 3 \sum_{t \neq s} \+E[v_t^2 v_s^2]  \Rp \\
			&=& \frac{1}{|\tls|^2} \Lp |\tls| \+E[v_t^4] + 3 |\tls| (|\tls| - 1) \Lp \+E[v_t^2] \Rp^2 \Rp \leq M
		\end{eqnarray*}
		following the boundedness of $\+E[v_t^2]$ and $\+E[v_t^4]$. 
		\paragraph{Step 4: Show that Assumption \ref{ass:factor-model}.\ref{ass:factor-error} holds under  Assumption \ref{ass:simple-factor-model}.4}  \texttt{} \\
		Since $F_t$ and $e_{it}$ are both i.i.d. and $W_{js}$ is independent of $F_t$ and $e_{it}$ ,we can reshuffle the index $t$ for every $(i,j)$, such as that $\tlq_{ij} = \{1, 2, \cdots, |\tlq_{ij}|\}$. Following the weak dependence assumption between the factors and errors, we have $\+E \norm{\frac{1}{\sqrt{|\tlq_{ij}|}} \sum_{i \in \tlq_{ij}} F_t e_{it} }^2 \leq M. $
		
		
		\subsubsection{Proof of Proposition \ref{prop:simple-assump-imply-general-assump}.1(b)}
			In this proof suppose Assumptions \ref{ass:obs-equal-weight} and \ref{ass:simple-factor-model} hold without further statement. We show that each part in Assumption \ref{ass:mom-clt} holds under Assumption \ref{ass:simple-moment}. For notation simplicity, denote $q_{ij} = q_{ij,ij}$.  
			\paragraph{Step 1: Show that Assumption \ref{ass:mom-clt}.1 holds under Assumption \ref{ass:simple-factor-model}} \texttt{} \\
			Denote $v_{ij,s} = e_{is}e_{js} - \+E[e_{is}e_{js}]$.  For $\phi_{i,st} = W_{it} F_s, \Lambda_{i}, W_{it} \Lambda_{i}$ , since $e$ is independent of $F$, $\Lambda$ and $W$, then $\phi_{i,st}$ is independent of $e$ for all $i,s,t$, $\+E[\phi_{i,st}^2] \leq M$ for some generic $M$, and we have 
			\begin{eqnarray*}
				\Scale[1]{\+E\norm{\frac{1}{\sqrt{N}} \sum_{i=1}^N \frac{1}{\sqrt{|\tlq_{ij}|}} \sum_{s \in \tlq_{ij}} \phi_{i,st}  v_{ij,s} }^2 } = \Scale[1]{\frac{1}{N} \sum_{i=1}^N \sum_{l=1}^N \frac{1}{\sqrt{|\tlq_{ij}| \cdot |\tlq_{lj}|}} \+E \Ls(\sum_{s\in \tlq_{ij}} \phi_{i,st} v_{ij,s}) (\sum_{s\in \tlq_{lj}} \phi_{l,st} v_{lj,s}) \Rs}. 
			\end{eqnarray*}
			When $i = l = j$, 
			\begin{align*}
				&\+E \Big[ (\sum_{s\in \tlq_{ij}} \phi_{i,st}  v_{ij,s}) (\sum_{s\in \tlq_{lj}} \phi_{l,st}  v_{lj,s}) \Big] = \+E[(\sum_{s\in \tlq_{jj}} \phi_{j,st} (e_{js}^2 - \sigma_e^2))^2] = \sum_{s \in \tlq_{jj}} \+E[ \phi_{j,st}^2] \+E[(e_{js}^2 -\sigma_e^2)^2 ]  \\
				\leq& M |\tlq_{jj}| (\+E[e_{js}^4] - \sigma_e^4).    
			\end{align*}
			When $i \neq j$ and $l = j$ (when $i = j$ and $l \neq j$, it is similar), $\+E[e_{it}e_{jt}] = 0$ and
			\[\Scale[1]{\+E \Ls(\sum_{s \in \tlq_{ij}} \phi_{i,st}  v_{ij,s}) (\sum_{s \in \tlq_{lj}} \phi_{l,st}  v_{lj,s}) \Rs =  \sum_{s \in \tlq_{ij}} \sum_{s \in \tlq_{jj}} \+E[\phi_{i,st}  \phi_{j,st}  ] \+E[(e_{it}e_{jt})(e_{jt}^2 - \sigma_e^2)] = 0.     }    \]
			When $i \neq j$ and $l \neq j$, 
			\[\Scale[1]{ \+E \Ls(\sum_{s \in \tlq_{ij}}\phi_{i,st}  v_{ij,s}) (\sum_{s \in \tlq_{lj}} \phi_{l,st}  v_{lj,s}) \Rs =  \sum_{s \in \tlq_{ij}}  \sum_{s \in \tlq_{lj}} \+E[\phi_{i,st}  \phi_{l,st}  ]  \+E[e_{is}e_{ls}e_{js}^2]  = 0.   }\]
			Then we have 
			\begin{align*}
				& \Scale[1]{\frac{1}{N} \sum_{i=1}^N \sum_{j=1}^N \frac{1}{\sqrt{|\tlq_{ij}| \cdot |\tlq_{lj}|}} \+E \Ls(\sum_{s \in \tlq_{ij}} \phi_{i,st} v_{ij,s}) (\sum_{t\in \tlq_{lj}} \phi_{l,st} v_{lj,s}) \Rs  } \\
				\leq& \Scale[1]{ \frac{M}{N} \frac{1}{|\tlq_{ll}|}  \cdot |\tlq_{ll}| (\+E[e_{lt}^4] - \sigma_e^4)   = M \cdot \frac{\+E[e_{lt}^4] - \sigma_e^4}{N} \leq M. } 
			\end{align*}
			\paragraph{Step 2: Show that Assumption \ref{ass:mom-clt}.2 holds under Assumptions \ref{ass:simple-factor-model}} \texttt{} \\ 
			Since $F$, $\Lambda$ and $e$ are independent, and $W$ is independent of $F$ and $e$, then $\phi_{it}$ is independent of $F$ and $e$ for any $i$ and $t$ and we have 
			\begin{eqnarray*}
				&& \Scale[1]{\+E \Ls \norm{\frac{1}{\sqrt{N}} \sum_{i=1}^N\frac{1}{\sqrt{|\tlq_{ij}|}} \sum_{s \in \tlq_{ij}} \phi_{it} F_s^\T e_{is} }^2  \Rs } \\
				&=& \Scale[1]{\sum_{k,m} \+E \Ls \Lp \frac{1}{\sqrt{N}} \sum_{i=1}^N\frac{1}{\sqrt{|\tlq_{ij}|}} \sum_{s \in \tlq_{ij}} \phi_{it,k} F_{s,m} e_{is} \Rp^2 \Rs } \\ 
				&=& \Scale[1]{ \sum_{k,m} \frac{1}{N} \sum_{i=1}^N \sum_{l=1}^N\frac{1}{\sqrt{|\tlq_{ij}| |\tlq_{lj}|}} \+E \Ls \Lp  \sum_{s \in \tlq_{ij}} \phi_{it,k} F_{s,m} e_{it} \Rp  \Lp  \sum_{s \in \tlq_{lj}} \phi_{lt,k} F_{s,m} e_{lt} \Rp \Rs  } \\
				&=& \Scale[1]{ \sum_{k,m} \frac{1}{N} \sum_{i=1}^N \sum_{l=1}^N\frac{1}{\sqrt{|\tlq_{ij}| |\tlq_{lj}|}}  \sum_{s \in \tlq_{ij}} \sum_{s^\prime \in \tlq_{lj}} \+E [\Lambda_{i,k} \Lambda_{l,k}]  \+E[ F_{s,m} F_{s^\prime,m}]  \+E[ e_{is} e_{ls^\prime}]    } \\
				&=& \Scale[1]{ \sum_{k,m} \frac{1}{N} \sum_{i=1}^N \frac{1}{|\tlq_{ij}|}  \sum_{s \in \tlq_{ij}}  \+E [\Lambda_{i,k}^2 ] \+E[ F_{s,m}^2]  \+E[ e_{is}^2]    } \leq M, 
			\end{eqnarray*}
			where the last equality follows from $\+E[e_{it}e_{ls}] = 0$ for $(i,t) \neq (l,s)$ from Assumption \ref{ass:simple-factor-model}.3. 
			\paragraph{Step 3: Show that Assumption \ref{ass:mom-clt}.3 holds under Assumption \ref{ass:simple-factor-model} and Assumption \ref{ass:simple-moment}} \texttt{} \\
			\textit{Step 3.1: Show that $\lim_{N \rightarrow \infty} \frac{1}{N^2} \sum_{i=1}^N \sum_{l=1}^N  \frac{q_{ij, lj}}{ q_{ij} q_{lj}}  \Lambda_i \Lambda_i^\T  \Sigma_F \Lambda_l \Lambda_l^\T $ exists under Assumption \ref{ass:simple-moment}} \\ \\
			For notation simplicity, denote $x_{il} \stackrel{\Delta}{=} \tvec \Big( \Lambda_i \Lambda_i^\T  \Sigma_F \Lambda_l \Lambda_l^\T \Big) $ and $x_{il,m}$ is the $m$-th entry in $x_{il}$. $x_{ii}$ is i.i.d, and $x_{il}$ is i.i.d. for $i \neq l$. from Assumption \ref{ass:simple-factor-model}.2. From the definition of $q_{ij}$ and $q_{ij,lj}$,  we have $0 \leq  \frac{q_{ij,lj}}{q_{ij}q_{lj}} - 1 \leq \frac{1}{q_{lj}} - 1\leq \frac{1}{\underline{q}} - 1$. As $q_{ij}$ and $q_{ij,kl}$ are independent of $\Lambda_m \Lambda_m^\T$ for all $i, j, k, l, m$, we have
			\begin{eqnarray*}
				&& \+E \Lp \frac{1}{N^2} \sum_{i= 1}^N \sum_{l = 1}^N    \frac{q_{ij,lj}}{q_{ij}q_{lj}}  (x_{il,m} - \+E[ x_{il,m}]) \Rp^2 \\
				&=&  \frac{1}{N^4} \sum_{i=1}^N \sum_{l = 1}^N \Lp  \frac{q_{ij,lj}}{q_{ij}q_{lj}}   \Rp^2 \cdot \+E [ x_{il,m} - \+E[ x_{il,m}] ] ^2 \\
				&& + \frac{1}{N^4} \sum_{i=1}^N \sum_{l \neq k}   \frac{q_{ij,lj}}{q_{ij}q_{lj}} \cdot \frac{q_{ij,kj}}{q_{ij}q_{kj}} \cdot \+E[ ( x_{il,m}  - \+E[ x_{il,m}]) (x_{ik,m} - \+E[ x_{ik,m}]) ] \\
				&& +  \frac{1}{N^4} \sum_{i \neq k} \sum_{l = 1}^N   \frac{q_{ij,kj}}{q_{ij}q_{kj}} \cdot \frac{q_{lj,kj}}{q_{lj}q_{kj}} \cdot \+E [  (x_{il,m} - \+E[ x_{il,m}])  ( x_{kl,m} - \+E[ x_{kl,m}] )] \\
				&& + \frac{1}{N^4} \sum_{i \neq k} \sum_{l \neq n}  \frac{q_{ij,lj}}{q_{ij}q_{lj}} \cdot \frac{q_{kj,nj}}{q_{kj}q_{nj}} \cdot \+E[ (x_{il, m} - \+E[ x_{il,m}]) ( x_{kn, m} - \+E[ x_{kn,m}] ) ]\\
				&=& O\Lp \frac{1}{N} \Rp
			\end{eqnarray*}
			as the last term is 0 ($\+E[ (x_{il, m} - \+E[ x_{il,m}]) ( x_{kn, m} - \+E[ x_{kn,m}] ) ] = \+E[ x_{il, m} - \+E[ x_{il,m}]] \+E[ x_{kn, m} - \+E[ x_{kn,m}]  ] = 0$ following that $x_{il}$ is i.i.d. for $i \neq l$). Hence, from Chebyshev's inequality and Assumption \ref{ass:simple-moment}.2 we conclude
			\begin{eqnarray*}
				&& \lim_{N \rightarrow \infty}\frac{1}{N^2} \sum_{i= 1}^N \sum_{l = 1}^N   \frac{q_{ij,lj}}{q_{ij}q_{lj}}  \Lambda_i \Lambda_i^\T  \Sigma_F \Lambda_l \Lambda_l^\T  \\
				&\xrightarrow{p}&   \Lp \lim_{N \rightarrow \infty}\frac{1}{N^2} \sum_{i= 1}^N \sum_{l = 1}^N   \frac{q_{ij,lj}}{q_{ij}q_{lj}}  \Rp   \Sigma_\Lambda  \Xi_F   \Sigma_\Lambda = \omega_{jj} \Sigma_\Lambda  \Xi_F   \Sigma_\Lambda.
			\end{eqnarray*}
			
			\noindent \textit{Step 3.2: Show that $\frac{\sqrt{T}}{N} \sum_{i = 1}^N  \Lambda_i \Lambda_i^\T  \frac{1}{|\tlq_{ij}|} \sum_{t \in \tlq_{ij}} F_t e_{jt} $ is asymptotically normal under  Assumption \ref{ass:simple-moment}.1.} \\ \\
			Assumption \ref{ass:simple-moment}.1 and Assumption \ref{ass:simple-factor-model} together with the CLT imply
			\[\Scale[1]{\frac{1}{\sqrt{|\tlq_{ij}|}} \sum_{t \in \tlq_{ij}} F_t e_{jt} \xrightarrow{d} N(0, \Sigma_F \sigma_e^2) } \]
			and following Assumption \ref{ass:obs-equal-weight}.2, we have 
			\[\Scale[1]{\frac{\sqrt{T}}{|\tlq_{ij}|} \sum_{t \in \tlq_{ij}} F_t e_{jt} \xrightarrow{d} N(0, \frac{1}{q_{ij}} \Sigma_F \sigma_e^2).  } \]
			The CLT implies that $\begin{bmatrix}
				\frac{\sqrt{T}}{|\tlq_{1j}|} \sum_{t \in \tlq_{1j}} F_t^\T e_{jt} & \frac{\sqrt{T}}{|\tlq_{2j}|} \sum_{t \in \tlq_{2j}} F_t e_{jt} & \cdots & \frac{\sqrt{T}}{|\tlq_{Nj}|} \sum_{t \in \tlq_{Nj}} F_t e_{jt} 
			\end{bmatrix}^\T$ is jointly asymptotic normal, and for $i \neq l$, 
			\[\Scale[1]{\ACov \Lp \frac{\sqrt{T}}{|\tlq_{ij}|} \sum_{t \in \tlq_{ij}} F_t e_{jt}, \frac{\sqrt{T}}{|\tlq_{lj}|} \sum_{t \in \tlq_{lj}} F_t e_{jt}  \Rp = \lim_{T\rightarrow \infty} \frac{T}{|\tlq_{ij}| |\tlq_{lj}|} \sum_{t \in \tlq_{ij}} \sum_{s \in \tlq_{ij}} \+E \Big[ F_t e_{jt} e_{js} F_s^\T \Big] = \frac{q_{ij, lj}}{ q_{ij} q_{lj}} } \Sigma_F\sigma_e^2   \]
			since $\+E[e_{jt} e_{js}] = 0$. Because of Step 3.1 ($\lim_{N \rightarrow \infty} \frac{1}{N^2} \sum_{i=1}^N \sum_{l=1}^N  \frac{q_{ij, lj}}{ q_{ij} q_{lj}}  \Lambda_i \Lambda_i^\T  \Sigma_F \Lambda_l \Lambda_l^\T $ exists),
			we have 
			\[\frac{\sqrt{T}}{N} \sum_{i = 1}^N  \Lambda_i \Lambda_i^\T  \frac{1}{|\tlq_{ij}|} \sum_{t \in \tlq_{ij}} F_t e_{jt} \xrightarrow{d} N(0, \covI_{\Lambda,j}),   \]
			where 
			\begin{eqnarray*}
				\covI_{\Lambda,j} &=&  \AVar\Lp \frac{\sqrt{T}}{N} \sum_{i = 1}^N  \Lambda_i \Lambda_i^\T \frac{1}{|\tlq_{ij}|} \sum_{t \in \tlq_{ij}} F_t e_{jt}\Rp =  \lim_{N \rightarrow \infty} \frac{\sigma_e^2}{N^2} \sum_{i=1}^N \sum_{l=1}^N \frac{q_{ij, lj}}{ q_{ij} q_{lj}}  \Lambda_i \Lambda_i^\T  \Sigma_F \Lambda_l \Lambda_l^\T   \\
				&=& \omega_{jj}  \cdot \sigma_e^2  \Sigma_\Lambda  \Sigma_F  \Sigma_\Lambda , 
			\end{eqnarray*}
			which follows from the results in Step 3.1. 
			\paragraph{Step 4: Show that Assumption \ref{ass:mom-clt}.4 holds under Assumptions \ref{ass:simple-factor-model} and \ref{ass:simple-moment}} \texttt{} \\
			$\Lambda_i e_{it}$ is independent across $i$. The CLT and the independence of $\Lambda$ and $e$ yield 
			\begin{align*}
				\frac{1}{\sqrt{N}} \sum_{i \in \tlo_t}  \Lambda_i  e_{it}  \xrightarrow{d} \calN(0, \Sigma_{\Lambda,t} \sigma_e^2).
			\end{align*}
			\paragraph{Step 5: Show that Assumption \ref{ass:mom-clt}.5 holds under Assumptions \ref{ass:simple-factor-model} and \ref{ass:simple-moment}} \texttt{} \\ 
			\textit{Step 5.1: Show that $\lim_{N \rightarrow \infty} \frac{1}{N^2} \sum_{i= 1}^N \sum_{k = 1}^N  \Lp  \frac{q_{ij,kj}}{q_{ij}q_{kj}} - 1  \Rp  (I_r \otimes  \Lambda_i \Lambda_i^\T)  \Xi_F   (I_r \otimes  \Lambda_k \Lambda_k^\T) $ exists under Assumption \ref{ass:simple-moment}.5} \\ \\
			For notation simplicity, denote $x_{ik} \stackrel{\Delta}{=}  \tvec(  (I_r \otimes  \Lambda_i \Lambda_i^\T)  \Xi_F   (I_r \otimes  \Lambda_k \Lambda_k^\T) ) $ and $x_{ik,m}$ is the $m$-th entry in $x_{ik}$. $x_{ii}$ is i.i.d and $x_{ik}$ for $i \neq k$ is i.i.d. from Assumption \ref{ass:simple-factor-model}.2. From the definition of $q_{ij}$ and $q_{ij,kj}$,  we have $0 \leq  \frac{q_{ij,kj}}{q_{ij}q_{kj}} - 1 \leq \frac{1}{q_{kj}} - 1\leq \frac{1}{\underline{q}} - 1$. Then, we have
			\begin{eqnarray*}
				&& \+E \Lp \frac{1}{N^2} \sum_{i= 1}^N \sum_{k = 1}^N  \Lp  \frac{q_{ij,kj}}{q_{ij}q_{kj}} - 1  \Rp (x_{ik,m} - \+E[ x_{ik,m}]) \Rp^2 \\
				&=&  \frac{1}{N^4} \sum_{i=1}^N \sum_{k = 1}^N \Lp  \frac{q_{ij,kj}}{q_{ij}q_{kj}} - 1  \Rp^2\+E [ x_{ik,m} - \+E[ x_{ik,m}] ] ^2 \\
				&& + \frac{1}{N^4} \sum_{i=1}^N \sum_{k \neq l} \Lp  \frac{q_{ij,kj}}{q_{ij}q_{kj}} - 1  \Rp \Lp  \frac{q_{ij,lj}}{q_{ij}q_{lj}} - 1  \Rp \+E[(x_{ik,m} - \+E[ x_{ik,m}]) ( x_{il,m}  - \+E[ x_{il,m}]) ] \\
				&& +  \frac{1}{N^4} \sum_{i \neq l} \sum_{k = 1}^N \Lp  \frac{q_{ij,kj}}{q_{ij}q_{kj}} - 1  \Rp  \Lp  \frac{q_{lj,kj}}{q_{lj}q_{kj}} - 1  \Rp \+E [(x_{ik,m} - \+E[ x_{ik,m}]) ( x_{lk,m} - \+E[ x_{lk,m}] ) ] \\
				&& + \frac{1}{N^4} \sum_{i \neq l} \sum_{k \neq n} \Lp  \frac{q_{ij,kj}}{q_{ij}q_{kj}} - 1  \Rp  \Lp  \frac{q_{lj,nj}}{q_{lj}q_{nj}} - 1  \Rp \+E[ (x_{ik, m} - \+E[ x_{ik,m}]) ( x_{ln, m} - \+E[ x_{ln,m}] ) ]\\
				&=& O\Lp \frac{1}{N} \Rp.
			\end{eqnarray*}
			as the last term is 0. Hence, by Chebyshev's inequality
			\begin{eqnarray*}
				&& \lim_{N \rightarrow \infty}\frac{1}{N^2} \sum_{i= 1}^N \sum_{k = 1}^N  \Lp  \frac{q_{ij,kj}}{q_{ij}q_{kj}} - 1  \Rp (I_r \otimes  \Lambda_i \Lambda_i^\T)  \Xi_F   (I_r \otimes  \Lambda_k \Lambda_k^\T) \\
				&\xrightarrow{p}&   \Lp \lim_{N \rightarrow \infty}\frac{1}{N^2} \sum_{i= 1}^N \sum_{k = 1}^N   \frac{q_{ij,kj}}{q_{ij}q_{kj}} - 1  \Rp (I_r \otimes  \Sigma_\Lambda )  \Xi_F   (I_r \otimes  \Sigma_\Lambda) = (\omega_{jj} -1)(I_r \otimes  \Sigma_\Lambda )  \Xi_F   (I_r \otimes  \Sigma_\Lambda) 
			\end{eqnarray*}
			where the last equality follows from Assumption \ref{ass:simple-moment}.2.\\ 
			
			\noindent \textit{Step 5.2: Show that $ \frac{\sqrt{T}}{N} \sum_{i = 1}^N \Lambda_i \Lambda_i^\T \Lp \frac{1}{|\tlq_{ij}|} \sum_{t \in \tlq_{ij}} F_t F_t^\T - \frac{1}{T} \sum_{t = 1}^T F_t F_t^\T \Rp $ is asymptotically normal and $ \frac{\sqrt{T}}{N} \sum_{i = 1}^N \Lambda_i \Lambda_i^\T \Lp \frac{1}{|\tlq_{ij}|} \sum_{t \in \tlq_{ij}} F_t F_t^\T - \frac{1}{T} \sum_{t = 1}^T F_t F_t^\T \Rp  u_i$ converges stably in law for $u_i = \Lambda_{i}$} \\ \\
			Denote $\+E[\tvec(F_t F_t^\T - \Sigma_{F})  \tvec(F_t F_t^\T - \Sigma_{F})^\T] = \Xi_F$. From Assumptions \ref{ass:obs}.3 and \ref{ass:simple-factor-model}.1 we have
			\begin{eqnarray*}
				v^{(i,j)} &\stackrel{\Delta}{=}& \frac{\sqrt{T}}{|\tlq_{ij}|} \sum_{t \in \tlq_{ij}}   \tvec(F_t F_t^\T) - \frac{1}{\sqrt{T}} \sum_{t = 1}^T  \tvec(F_t F_t^\T) \\
				&=& \frac{1}{\sqrt{T}}  \Lp \frac{1}{\tilde q_{ij}} - 1 \Rp \sum_{t \in \tlq_{ij}}  \tvec( F_t F_t^\T - \Sigma_{F}) - \frac{1}{\sqrt{T}} \sum_{t \not\in \tlq_{ij}}  \tvec( F_t F_t^\T - \Sigma_{F})  \\
				&\xrightarrow{d}& N \Lp 0, \frac{1-q_{ij}}{q_{ij}}  \Xi_F \Rp. 
			\end{eqnarray*}
			For any $(i,j) \neq (k,l)$, we define 
				\[A = \sum_{t \in \tlq_{ij} \cap \tlq_{kl}} F_t F_t^\T,\quad  B = \sum_{t \in \tlq_{ij} \cap \tlq_{kl}^\mathsf{C}} F_t F_t^\T, \quad C = \sum_{t \in \tlq_{ij}^\mathsf{C} \cap \tlq_{kl}} F_t F_t^\T,\quad D = \sum_{t \in \tlq_{ij}^\mathsf{C} \cap \tlq_{kl}^\mathsf{C}} F_t F_t^\T.   \]
				Assumption \ref{ass:obs-equal-weight}.2 implies  $\frac{|\tlq_{ij} \cap \tlq_{kl}|}{T} \stackrel{\Delta}{=} \tilde q_{ij,kl} \rightarrow q_{ij,kl}$, $\frac{|\tlq_{ij} \cap \tlq_{kl}^\mathsf{C}|}{T} = \tilde q_{ij} - \tilde q_{ij,kl} \rightarrow q_{ij} -  q_{ij,kl}$, $\frac{|\tlq_{ij}^\mathsf{C} \cap \tlq_{kl}|}{T} = \tilde q_{kl} - \tilde q_{ij,kl} \rightarrow q_{kl} -  q_{ij,kl}$, and $\frac{|\tlq_{ij}^\mathsf{C} \cap \tlq_{kl}^\mathsf{C}|}{T} = 1 - \tilde q_{ij} - \tilde q_{kl} + \tilde q_{ij,kl} \rightarrow 1 -  q_{ij} - q_{kl} +  q_{ij,kl}$.
				Hence, we can show that $v^{(i,j)}$ and $v^{(k,l)}$ are jointly asymptotically normal with 
			\begin{eqnarray*}
				\Cov(v^{(i,j)}, v^{(k,l)}) &=& \Scale[1]{\Cov\Lp \frac{1}{\sqrt{T}}  \Lp \frac{1}{\tilde q_{ij}} - 1 \Rp ( A+B) - \frac{1}{\sqrt{T}} (C+D), \frac{1}{\sqrt{T}}  \Lp \frac{1}{\tilde q_{kl}} - 1 \Rp ( A+C) - \frac{1}{\sqrt{T}} (B+D)  \Rp}  \\
				&=& \Scale[1]{ \Lp \frac{1}{\tilde q_{ij}} - 1 \Rp  \Lp \frac{1}{\tilde q_{kl}} - 1 \Rp   \tilde q_{ij,kl}    \Xi_F - \Lp \frac{1}{\tilde q_{ij}} - 1 \Rp  (\tilde q_{ij} - \tilde q_{ij,kl}) \Xi_F}  \\
				&& \Scale[1]{- \Lp \frac{1}{\tilde q_{kl}} - 1 \Rp   (\tilde q_{kl} - \tilde q_{ij,kl} ) \Xi_F + (1  - \tilde q_{ij}  - \tilde q_{kl} + \tilde q_{ij,kl}) \Xi_F } \\
				&\rightarrow& \Scale[1]{\Lp  \frac{q_{ij,kl}}{q_{ij}q_{kl}} - 1  \Rp \Xi_F  }.
			\end{eqnarray*}
			
			The vectorized form of $\Lambda_i \Lambda_i^\T \Lp \frac{1}{|\tlq_{ij}|} \sum_{t \in \tlq_{ij}} F_t F_t^\T - \frac{1}{T} \sum_{t = 1}^T F_t F_t^\T \Rp$ is $ (I_r \otimes  \Lambda_i \Lambda_i^\T) v^{(i,j)}$. Then
			Assumption \ref{ass:mom-clt}.\ref{ass:asy-normal-add-term-thm-loading} holds with
			\begin{eqnarray*}
				\Phi_j &=& \lim_{N \rightarrow \infty} \frac{1}{N^2} \sum_{i= 1}^N \sum_{k = 1}^N  \Lp  \frac{q_{ij,kj}}{q_{ij}q_{kj}} - 1  \Rp  (I_r \otimes  \Lambda_i \Lambda_i^\T)  \Xi_F   (I_r \otimes  \Lambda_k \Lambda_k^\T)\\ 
				&=&  \Lp\omega_{jj}  - 1  \Rp (I_r \otimes  \Sigma_\Lambda )  \Xi_F   (I_r \otimes  \Sigma_\Lambda)
			\end{eqnarray*}
			following from the results in Step 5.1.

			We can show the stable convergence in law similar to \cite{jin2020factor}.  Note that
			\begin{align*}
				& \frac{\sqrt{T}}{N} \sum_{i = 1}^N \Lambda_i \Lambda_i^\T  \Big( \frac{1}{|\tlq_{ij}|} \sum_{t \in \tlq_{ij}} (F_t F_t^\T - \Sigma_F)  - \frac{1}{T} \sum_{t = 1}^T (F_t F_t^\T - \Sigma_F) \Big)  \Lambda_j \\
				=& \frac{1}{\sqrt{T}} \sum_{t=1}^T \Big( \frac{1}{N} \sum_{i=1}^N \Big\lbrace\frac{W_{it}W_{jt}}{\hat{q}_{ij}} -1 \Big\rbrace \Lambda_i \Lambda_i^\T  \Big) (F_t F_t^\T - \Sigma_F)  \Lambda_j, 
			\end{align*}
			where $\hat{q}_{ij} = \frac{|\tlq_{ij}|}{N}$.
			Define the sigma-field $\mathcal{G}_{Tt} = \sigma(\{W_{is}, s \leq t,  \text{all } i \}, \Lambda)$ that is generated from $\{W_{is}, s \leq t,  \text{all } i \}$ and $\Lambda$. Let $\mathcal{G} = \sigma(\cup_{t = 1}^T \mathcal{G}_{Tt} )$. Let $\omega \in R^r$ be a nonrandom vector with $\norm{\omega} = 1$. Let
			\begin{align*}
				\psi_{jt} = \omega^\T \cdot \frac{1}{\sqrt{T}} \Big( \frac{1}{N} \sum_{i=1}^N \Big\lbrace\frac{W_{it}W_{jt}}{\hat{q}_{ij}} -1 \Big\rbrace \Lambda_i \Lambda_i^\T  \Big) (F_t F_t^\T - \Sigma_F)  \Lambda_j 
			\end{align*}
			Since $F_t$ is i.i.d. from Assumption \ref{ass:simple-factor-model} and is independent of $W$ and $\Lambda$, we have $\+E[\psi_{jt}|\mathcal{G}_{T,t-1} ] = 0$ and $\sum_{t=1}^T  \+E[\psi_{jt}^2| \mathcal{G}_{T,t-1}] \xrightarrow{p} \omega^\T ( \Lambda_j^\T \otimes I_r) \Phi_j (\Lambda_j \otimes I_r) \omega$. Let  $\sum_{t=1}^T  \+E[\psi_{jt}^{2+\epsilon}| \mathcal{G}_{T,t-1}] \leq  \frac{\max_i\norm{\Lambda_i}^{6+3\epsilon}}{T^\epsilon} \Big( \frac{1}{\underline{q}}-1 \Big)^{2+\epsilon}  \frac{1}{T} \sum_{t = 1}^T \norm{F_t F_t^\T - \Sigma_F}^{2+\epsilon} \xrightarrow{p} 0$ from Assumption \ref{ass:simple-factor-model}. The conditional Lindeberg condition in \cite{hausler2015stable} holds and by the stable martingale central limit theorem (Theorem 6.1 in \cite{hausler2015stable}), 
			\[\sum_{t = 1}^T \psi_{jt} \xrightarrow{d} \mathcal{N} (0, h_j(\Lambda_j)) \quad \mathcal{G}\text{-stably as } T \rightarrow \infty, \]
			where $h_j(\Lambda_j) = ( \Lambda_j^\T \otimes I_r) \Phi_j (\Lambda_j \otimes I_r)$.
			\\ \\
			\textit{Step  5.3: Show that $ \lim_{N\rightarrow \infty}  \frac{1}{N^4} \sum_{i=1}^N \sum_{l = 1}^N \sum_{j = 1}^N \sum_{k = 1}^N  \Lp  \frac{q_{li,kj}}{q_{li}q_{kj}} - 1  \Rp W_{it} W_{jt}  (\Lambda_i \Lambda_i^\T \otimes I_r ) (I_r \otimes  \Lambda_l \Lambda_l^\T) \Xi_F$\\ $ (I_r \otimes  \Lambda_k \Lambda_k^\T) (\Lambda_j \Lambda_j^\T \otimes I_r )      $ exists. } \\ \\
			For notation simplicity, denote $x_{ilkj} \stackrel{\Delta}{=}  \tvec(W_{it} W_{jt}  (\Lambda_i \Lambda_i^\T \otimes I_r ) (I_r \otimes  \Lambda_l \Lambda_l^\T) \Xi_F (I_r \otimes  \Lambda_k \Lambda_k^\T) (\Lambda_j \Lambda_j^\T \otimes I_r )  ) $ and $x_{ilkj,m}$ is the $m$-th entry in $x_{ilkj}$. From the definition of $q_{ij}$ and $q_{ij,kj}$,  we have $0 \leq  \frac{q_{ij,kj}}{q_{ij}q_{kj}} - 1 \leq \frac{1}{q_{kj}} - 1\leq \frac{1}{\underline{q}} - 1$ and therefore
			\begin{eqnarray*} 
				\textstyle
				&& \+E \Lp \frac{1}{N^4} \sum_{i=1}^N \sum_{l = 1}^N \sum_{j = 1}^N \sum_{k = 1}^N  \Lp  \frac{q_{li,kj}}{q_{li}q_{kj}} - 1  \Rp  (x_{ilkj,m} - \+E[ x_{ilkj,m}]) \Rp^2 \\
				&=& \Scale[0.98]{\frac{1}{N^8} \sum_{\substack{\text{distinct } \\ i, i', j, j', \\ k, k', l, l'  }  } \Lp  \frac{q_{li,kj}}{q_{li}q_{kj}} - 1  \Rp  \Lp  \frac{q_{l'i',k'j'}}{q_{l'i'}q_{k'j'}} - 1  \Rp \+E[ ( (x_{ilkj,m} - \+E[ x_{ilkj,m}] ) (  (x_{i'l'k'j',m} - \+E[ x_{i'l'k'j',m}] ) ] + \text{other terms}} \\
				&=& O\Lp \frac{1}{N} \Rp, 
			\end{eqnarray*}
			as the number of terms in the other terms is $O(N^7)$ and $\+E[ ( (x_{ilkj,m} - \+E[ x_{ilkj,m}] ) (  (x_{i'l'k'j',m} - \+E[ x_{i'l'k'j',m}] ) ] = 0$ for distinct $ i, i', j, j', k, k', l, l'$.   Hence, by Chebyshev's inequality, it holds that
			\begin{eqnarray*}
				&& \Scale[0.95]{ \lim_{N\rightarrow \infty}  \frac{1}{N^4} \sum_{i=1}^N \sum_{l = 1}^N \sum_{j = 1}^N \sum_{k = 1}^N  \Lp  \frac{q_{li,kj}}{q_{li}q_{kj}} - 1  \Rp W_{it} W_{jt}   (\Lambda_i \Lambda_i^\T \otimes I_r ) (I_r \otimes  \Lambda_l \Lambda_l^\T) \Xi_F (I_r \otimes  \Lambda_k \Lambda_k^\T) (\Lambda_j \Lambda_j^\T \otimes I_r )      }  \\
				&\xrightarrow{p}&  \Scale[0.95]{  \Lp \lim_{N\rightarrow \infty} \frac{1}{N^4} \sum_{i=1}^N \sum_{l = 1}^N \sum_{j = 1}^N \sum_{k = 1}^N   \frac{q_{li,kj}}{q_{li}q_{kj}} - 1  \Rp ( \Sigma_{\Lambda,t} \otimes I_r) (I_r \otimes  \Sigma_\Lambda )  \Xi_F   (I_r \otimes  \Sigma_\Lambda)  ( \Sigma_{\Lambda,t} \otimes I_r)} \\
				&=& ( \omega - 1 )   ( \Sigma_{\Lambda,t} \otimes I_r) (I_r \otimes  \Sigma_\Lambda )  \Xi_F   (I_r \otimes  \Sigma_\Lambda)  ( \Sigma_{\Lambda,t} \otimes I_r),
			\end{eqnarray*} 
			where the last equality follows Assumption \ref{ass:simple-moment}.2. \\ \\
			\textit{Step 5.4: Show that $ \frac{\sqrt{T}}{N^2} \sum_{l = 1}^N \sum_{i = 1}^N \Lambda_l \Lambda_l^\T \Lp \frac{1}{|\tlq_{li}|} \sum_{t \in \tlq_{li}} F_t F_t^\T - \frac{1}{T} \sum_{t = 1}^T F_t F_t^\T \Rp   W_{it} \Lambda_i \Lambda_i^\T$ is asymptotic normal and $ \frac{\sqrt{T}}{N^2} \sum_{l = 1}^N \sum_{i = 1}^N \Lambda_l \Lambda_l^\T \Lp \frac{1}{|\tlq_{li}|} \sum_{t \in \tlq_{li}} F_t F_t^\T - \frac{1}{T} \sum_{t = 1}^T F_t F_t^\T \Rp   W_{it} \Lambda_i \Lambda_i^\T v_t$ converges stably in law for $v_t = H^\T \tilde D^\I (H^\T)^{-1} F_t$ } \\ \\
			The vectorized form of $\Lambda_l \Lambda_l^\T \Lp  \frac{1}{|\tlq_{li}|} \sum_{s \in \tlq_{li}} F_s F_s^\T - \frac{1}{T} \sum_{s=1}^T F_s F_s^\T \Rp W_{it} \Lambda_i \Lambda_{i} $ is $W_{it} (\Lambda_i \Lambda_i^\T \otimes I_r ) (I_r \otimes  \Lambda_l \Lambda_l^\T) v^{(l,i)}$. Then 
			\begin{align*}
				\mathbf{\Phi}_t =& \Scale[0.9]{\ACov_{\Lambda, W}  \Lp  \frac{\sqrt{T}}{N^2} \sum_{i = 1}^N \sum_{l = 1}^N W_{it} (\Lambda_i \Lambda_i^\T \otimes I_r ) (I_r \otimes  \Lambda_l \Lambda_l^\T) v^{(l,i)},  \frac{\sqrt{T}}{N^2} \sum_{i = 1}^N \sum_{l = 1}^N W_{it} (\Lambda_i \Lambda_i^\T \otimes I_r ) (I_r \otimes  \Lambda_l \Lambda_l^\T) v^{(l,i)} \Rp}  \\
				=& \Scale[0.9]{\lim_{N, T\rightarrow \infty} \frac{1}{N^4} \sum_{i=1}^N \sum_{l = 1}^N \sum_{j = 1}^N \sum_{k = 1}^N \Cov_{\Lambda, W} \Lp (W_{it} \Lambda_i \Lambda_i^\T \otimes I_r ) (I_r \otimes  \Lambda_l \Lambda_l^\T) v^{(l,i)}, W_{jt} (\Lambda_j \Lambda_j^\T \otimes I_r ) (I_r \otimes  \Lambda_k \Lambda_k^\T) v^{(k,j)}  \Rp    }  \\
				=& \Scale[0.9]{\lim_{N\rightarrow \infty}   \frac{1}{N^4} \sum_{i=1}^N \sum_{l = 1}^N \sum_{j = 1}^N \sum_{k = 1}^N  W_{it} W_{jt} (\Lambda_i \Lambda_i^\T \otimes I_r ) (I_r \otimes  \Lambda_l \Lambda_l^\T) \mathbf{\Phi}^{\mathbf{v}}_{g(l,i), g(k,j)} (I_r \otimes  \Lambda_k \Lambda_k^\T) (\Lambda_j \Lambda_j^\T \otimes I_r )      }  \\
				=& \Scale[0.9]{\lim_{N\rightarrow \infty}  \frac{1}{N^4} \sum_{i=1}^N \sum_{l = 1}^N \sum_{j = 1}^N \sum_{k = 1}^N  \Lp  \frac{q_{li,kj}}{q_{li}q_{kj}} - 1  \Rp W_{it} W_{jt}  (\Lambda_i \Lambda_i^\T \otimes I_r ) (I_r \otimes  \Lambda_l \Lambda_l^\T) \Xi_F (I_r \otimes  \Lambda_k \Lambda_k^\T) (\Lambda_j \Lambda_j^\T \otimes I_r )      }  \\
				=&  \textstyle \Lp \omega - 1  \Rp  ( \Sigma_{\Lambda,t} \otimes I_r) (I_r \otimes  \Sigma_\Lambda )  \Xi_F   (I_r \otimes  \Sigma_\Lambda)  ( \Sigma_{\Lambda,t} \otimes I_r)
			\end{align*}
			where the last equality follows from Step 5.3.
			
			We can rewrite $ \frac{\sqrt{T}}{N^2} \sum_{l = 1}^N \sum_{i = 1}^N \Lambda_l \Lambda_l^\T \Lp \frac{1}{|\tlq_{li}|} \sum_{t \in \tlq_{li}} F_t F_t^\T - \frac{1}{T} \sum_{t = 1}^T F_t F_t^\T \Rp   W_{it} \Lambda_i \Lambda_i^\T$ as 
			\begin{align*}
				\frac{1}{T} \sum_{s= 1}^T  \underbrace{\Big( \frac{1}{N^2} \sum_{i =1}^N \sum_{l = 1}^N \Big\lbrace\frac{W_{it}W_{lt}}{\hat{q}_{il}} -1 \Big\rbrace  \Lambda_l   \Lambda_l^\T  (F_s F_s^\T -  \Sigma_F)   W_{is}  \Lambda_i \Lambda_i^\T \Big)}_{\phi^\dagger_{st}},
			\end{align*}
			where $\hat{q}_{il} =  \frac{|\tlq_{il}|}{N}$.  Define the sigma-field $\mathcal{G}^t_{Ts} = \sigma(\{W_{iu}, u \leq s,  \text{all } i \}, \Lambda, F_t)$ that is generated from $\{W_{iu}, u \leq s,  \text{all } i \}$, $\Lambda$ and $F_t$. Let $\mathcal{G}^t = \sigma(\cup_{s= 1}^T \mathcal{G}^t_{Ts} )$. Let $\omega \in R^r$ be a nonrandom vector with $\norm{\omega} = 1$. Let
			\begin{align*}
				\phi_{st} =& \frac{\omega^\T}{\sqrt{T}} \cdot \Big( \frac{1}{N^2} \sum_{i =1}^N \sum_{l = 1}^N \Big\lbrace\frac{W_{it}W_{lt}}{\hat{q}_{il}} -1 \Big\rbrace  \Lambda_l   \Lambda_l^\T  (F_s F_s^\T -  \Sigma_F)   W_{is}  \Lambda_i \Lambda_i^\T \Big) H^\T \tilde D^\I (H^\T)^{-1} F_t \\ =& \frac{\omega^\T}{\sqrt{T}} \cdot  \Big(I_r \otimes (H^\T \tilde D^\I (H^\T)^{-1} F_t )^\T \Big) \tvec(\phi^\dagger_{st})
			\end{align*}
			Since $F_t$ is i.i.d. from Assumption \ref{ass:simple-factor-model} and is independent of $W$ and $\Lambda$, we have for $s \neq t$, $\+E[\phi_{st}|\mathcal{G}^t_{T,s-1} ] = 0$ and $\sum_{s \neq t}  \+E[\phi_{st}^2| \mathcal{G}_{T,s-1}^t] \xrightarrow{p} \omega^\T ( \Lambda_j^\T \otimes I_r) \Phi_j (\Lambda_j \otimes I_r) \omega$, given that $\+E[\phi_{tt}^2| \mathcal{G}^t_{T,t-1}] \xrightarrow{p} 0 $ and removing it from the summation does not change the limit. Let  $\sum_{s\neq t}  \+E[\phi_{st}^{2+\epsilon}| \mathcal{G}^t_{T,s-1}] \leq  \frac{\max_i\norm{\Lambda_i}^{8+4\epsilon}}{T^\epsilon} \Big( \frac{1}{\underline{q}}-1 \Big)^{2+\epsilon}  \frac{1}{T} \sum_{t = 1}^T \norm{F_t F_t^\T - \Sigma_F}^{2+\epsilon} \xrightarrow{p} 0$ from Assumption \ref{ass:simple-factor-model}. The conditional Lindeberg condition in \cite{hausler2015stable} holds and by the stable martingale central limit theorem (Theorem 6.1 in \cite{hausler2015stable}), 
			\[ \sum_{s = 1}^T \phi_{st} \xrightarrow{d} \mathcal{N} (0, g_t(v_t)) \quad \mathcal{G}^t \text{-stably as } T \rightarrow \infty, \]
			where $g_t(v_t) =  \mathbf{M}_{F,t} \mathbf{\Phi}_t  \mathbf{M}_{F,t}^\T$.
		\\ \\
		\textit{Step 5.5: Show that $ \frac{\sqrt{T}}{N} \sum_{i = 1}^N \Lambda_i \Lambda_i^\T \Lp \frac{1}{|\tlq_{ij}|} \sum_{t \in \tlq_{ij}} F_t F_t^\T - \frac{1}{T} \sum_{t = 1}^T F_t F_t^\T \Rp $ and \\  $ \frac{\sqrt{T}}{N^2} \sum_{l = 1}^N \sum_{i = 1}^N \Lambda_l \Lambda_l^\T \Lp \frac{1}{|\tlq_{li}|} \sum_{t \in \tlq_{li}} F_t F_t^\T - \frac{1}{T} \sum_{t = 1}^T F_t F_t^\T \Rp   W_{it} \Lambda_i \Lambda_i^\T$ are jointly asymptotically normal and their asymptotic covariance converges.} \\ \\
		The randomness of these two terms come both from $v^{(l,i)}$, which is asymptotically normal. These two terms are weighted average of $v^{(l,i)}$ and therefore they are jointly asymptotically normal as well. Next, we show that their aymptotic covariance converges and we provide the limit.
		\begin{align*}
			\mathbf{\Phi}_{j,t}^\cov =& \Scale[0.9]{\ACov_{\Lambda, W}  \Lp  \frac{\sqrt{T}}{N^2} \sum_{i = 1}^N \sum_{l = 1}^N W_{it} (\Lambda_i \Lambda_i^\T \otimes I_r ) (I_r \otimes  \Lambda_l \Lambda_l^\T) v^{(l,i)},  \frac{\sqrt{T}}{N} \sum_{i = 1}^N   (I_r \otimes  \Lambda_i \Lambda_i^\T) v^{(i,j)}\Rp}  \\
			=& \Scale[0.9]{\lim_{N, T\rightarrow \infty} \frac{1}{N^3} \sum_{i=1}^N \sum_{l = 1}^N  \sum_{k = 1}^N Cov_{\Lambda, W} \Lp (W_{it} \Lambda_i \Lambda_i^\T \otimes I_r ) (I_r \otimes  \Lambda_l \Lambda_l^\T) v^{(l,i)}, (I_r \otimes  \Lambda_k \Lambda_k^\T) v^{(k,j)}  \Rp    }  \\
			=& \Scale[0.9]{\lim_{N\rightarrow \infty}   \frac{1}{N^3} \sum_{i=1}^N \sum_{l = 1}^N  \sum_{k = 1}^N  W_{it} (\Lambda_i \Lambda_i^\T \otimes I_r ) (I_r \otimes  \Lambda_l \Lambda_l^\T) \mathbf{\Phi}^{\mathbf{v}}_{g(l,i), g(k,j)} (I_r \otimes  \Lambda_k \Lambda_k^\T)   }  \\
			=& \Scale[0.9]{\lim_{N\rightarrow \infty}  \frac{1}{N^3} \sum_{i=1}^N \sum_{l = 1}^N  \sum_{k = 1}^N  \Lp  \frac{q_{li,kj}}{q_{li}q_{kj}} - 1  \Rp W_{it}  (\Lambda_i \Lambda_i^\T \otimes I_r ) (I_r \otimes  \Lambda_l \Lambda_l^\T) \Xi_F (I_r \otimes  \Lambda_k \Lambda_k^\T)   }  \\
			=&  \textstyle \Lp \lim_{N\rightarrow \infty} \frac{1}{N^3} \sum_{i=1}^N \sum_{l = 1}^N  \sum_{k = 1}^N   \frac{q_{li,kj}}{q_{li}q_{kj}} - 1  \Rp  ( \Sigma_{\Lambda,t} \otimes I_r) (I_r \otimes  \Sigma_\Lambda )  \Xi_F   (I_r \otimes  \Sigma_\Lambda) \\
			=&  \textstyle \Lp \omega_j - 1  \Rp  ( \Sigma_{\Lambda,t} \otimes I_r) (I_r \otimes  \Sigma_\Lambda )  \Xi_F   (I_r \otimes  \Sigma_\Lambda) ,
		\end{align*}
		where the last equality follows from a similar argument as Step 5.1 and 5.3 and we can show that  $\lim_{N\rightarrow \infty}  \frac{1}{N^3} \sum_{i=1}^N \sum_{l = 1}^N  \sum_{k = 1}^N  \Lp  \frac{q_{li,kj}}{q_{li}q_{kj}} - 1  \Rp W_{it}  (\Lambda_i \Lambda_i^\T \otimes I_r ) (I_r \otimes  \Lambda_l \Lambda_l^\T) \Xi_F (I_r \otimes  \Lambda_k \Lambda_k^\T)$ converges to $\Lp \lim_{N\rightarrow \infty} \frac{1}{N^3} \sum_{i=1}^N \sum_{l = 1}^N  \sum_{k = 1}^N   \frac{q_{li,kj}}{q_{li}q_{kj}} - 1  \Rp  ( \Sigma_{\Lambda,t} \otimes I_r) (I_r \otimes  \Sigma_\Lambda )  \Xi_F   (I_r \otimes  \Sigma_\Lambda) $. 
		
		We use similar arguments as in Steps 5.2 and 5.4 to show that \\ $ \frac{\sqrt{T}}{N} \sum_{i = 1}^N \Lambda_i \Lambda_i^\T \Lp \frac{1}{|\tlq_{ij}|} \sum_{t \in \tlq_{ij}} F_t F_t^\T - \frac{1}{T} \sum_{t = 1}^T F_t F_t^\T \Rp u_i$ and \\  $ \frac{\sqrt{T}}{N^2} \sum_{l = 1}^N \sum_{i = 1}^N \Lambda_l \Lambda_l^\T \Lp \frac{1}{|\tlq_{li}|} \sum_{t \in \tlq_{li}} F_t F_t^\T - \frac{1}{T} \sum_{t = 1}^T F_t F_t^\T \Rp   W_{it} \Lambda_i \Lambda_i^\T v_t$ converge jointly and stably in law. 
		\paragraph{Step 6: Show that Assumption \ref{ass:mom-clt}.6 holds under Assumptions \ref{ass:simple-factor-model} and \ref{ass:simple-moment}}  \texttt{} \\
		Denote $V_{li} = \frac{1}{|\tlq_{li}|} \sum_{s \in \tlq_{li}} F_s F_s^\T - \frac{1}{T} \sum_{s=1}^T F_s F_s^\T $. From Assumption \ref{ass:mom-clt}.\ref{ass:asy-normal-add-term-thm-loading}, it is asymptotic normal. We first calculate the variance of the term $\frac{1}{N} \sum_{i = 1}^N W_{it} V_{li} \Lambda_i e_{it}$: 
		\begin{eqnarray*}
			&& \Scale[1]{\Cov \Lp \sum_{i = 1}^NW_{it} V_{li} \Lambda_i e_{it} ,  \sum_{i = 1}^NW_{it}  V_{li} \Lambda_i e_{it}  \Rp} \\
			&=& \Scale[1]{\sum_{i = 1}^N \sum_{m = 1}^N \+E \Ls W_{it} W_{mt} V_{li} \Lambda_i e_{it} e_{mt} \Lambda_m^\T V_{mi}^\T  \Rs  }  \\
			&=& \Scale[1]{ \sum_{i = 1}^N \sum_{m = 1}^N  \+E[ V_{li} \+E[W_{it} W_{mt} \Lambda_i  \+E[e_{it} e_{mt}]  \Lambda_m^\T ] V_{mi}^\T ]  } \\
			&=&  \Scale[1]{ \sum_{i = 1}^N   \+E[ V_{li} \+E[W_{it} \Lambda_i  \+E[e_{it}^2]  \Lambda_i^\T |S ] V_{li}^\T ] = O\Lp \frac{N}{T} \Rp } 
		\end{eqnarray*}
		since $\+E[e_{it} e_{mt}] = 0$ for $i \neq m$. Then Assumption \ref{ass:mom-clt}.6 holds.

		
		\subsubsection{Proof of Proposition \ref{prop:simple-assump-imply-general-assump}.2(a)}
		\begin{proof}[Proof of Proposition \ref{prop:simple-assump-imply-general-assump}.2(a)]  \texttt{}
			\paragraph{Step 1: Show that Assumption \ref{ass:factor-model-conditional} holds under Assumption \ref{ass:simple-factor-model-conditional}} 
			\texttt{} \\
			\textit{Step 1.1: Show that $\frac{1}{N} \sum_{i = 1}^N \frac{W_{it}}{\ps} \Lambda_i \Lambda_i^\T \xrightarrow{P} \Sigma_{\Lambda}  $ holds under Assumption \ref{ass:simple-factor-model-conditional}} \\ \\
			It is equivalent to show $\frac{1}{N} \sum_{i = 1}^N (\frac{W_{it}}{\ps} - 1)  \tvec(\Lambda_i \Lambda_i^\T)  \xrightarrow{P} 0$. Denote $y_i = (\frac{W_{it}}{\ps} - 1)  \tvec(\Lambda_i \Lambda_i^\T)$. Then
			$\+E[y_i] = \+E[\+E[y_i|S_i]] = \+E[\+E[\frac{W_{it}}{\ps} - 1|S_i] \+E[\tvec(\Lambda_i \Lambda_i^\T)|S_i]] = 0$ and 
			\begin{eqnarray*}
				Var(y_i) &=& \+E[ \Var(y_i|S_i)] + \Var(\+E[y_i|S_i])  \\
				&=&\+E \Bigg[ \Var \Big(\frac{W_{it}}{\ps} - 1|S_i \Big) \Var(\tvec(\Lambda_i \Lambda_i^\T)|S_i) \\
				&& + \Var \Big(\frac{W_{it}}{\ps} - 1|S_i \Big)\+E[\tvec(\Lambda_i \Lambda_i^\T)|S_i]\+E[\tvec(\Lambda_i \Lambda_i^\T)|S_i]^\T \\
				&& + \+E \Big[\frac{W_{it}}{\ps} - 1|S_i \Big]^2 \Var(\tvec(\Lambda_i \Lambda_i^\T)|S_i) \Bigg] \\
				&=& \+E\bigg[ \bigg(\frac{1}{\spsi}-1 \bigg) \bigg(\Var(\tvec(\Lambda_i \Lambda_i^\T)|S_i) + \+E[\tvec(\Lambda_i \Lambda_i^\T)|S_i] \+E[\tvec(\Lambda_i \Lambda_i^\T)|S_i]^\T\bigg)  \bigg] \\
				&\leq& (\frac{1}{\underline{p}}-1) \Xi_\Lambda,
			\end{eqnarray*}
			where $\Xi_\Lambda = \+E[\tvec(\Lambda_i \Lambda_i^\T) \tvec(\Lambda_i \Lambda_i^\T)^\T]$. Since $\Lambda_i$ is independent of $\Lambda_j$ conditional on $S_i$ and $S_j$, together with the fact that $W_{it}$ is independent of $W_{jt}$ conditional on $S$ from Assumption \ref{ass:obs}.2,  by Chebyshev’s Inequality,  we have
			\[ P(|\frac{1}{N}\sum_{i = 1}^N y_i| > t) \leq \frac{\Var( \frac{1}{N}\sum_{i = 1}^N y_i)}{t^2} = \frac{1}{N^2 t^2} \sum_{i = 1}^N (\frac{1}{\underline{p}}-1) \Xi_\Lambda = \frac{1}{N t^2} (\frac{1}{\underline{p} } - 1) \Xi_\Lambda. \]
			Thus, we conclude that $\frac{1}{N} \sum_{i = 1}^N \frac{W_{it}}{\ps} \Lambda_i \Lambda_i^\T \xrightarrow{P} \Sigma_{\Lambda}  $.  \\ \\
			\textit{Step 1.2: Show that $\+E \norm{\sqrt{N}\Lp \frac{1}{N}  \sum_{i = 1}^N \frac{W_{it}}{\spsi} \Lambda_i \Lambda_i^\T - \Sigma_{\Lambda} \Rp }  \leq M$ holds under Assumption \ref{ass:simple-factor-model-conditional}.}
			\begin{eqnarray*}
				\Scale[1]{ \Lp \+E\norm{\sqrt{N} \Lp \frac{1}{N} \sum_{i =1}^N \frac{W_{it}}{\spsi} \Lambda_i \Lambda_i^\T - \Sigma_\Lambda \Rp}\Rp^2} &\leq& \Scale[1]{\+E \norm{\sqrt{N} \Lp \frac{1}{N} \sum_{i =1}^N \frac{W_{it}}{\spsi} \Lambda_i \Lambda_i^\T - \Sigma_\Lambda \Rp}^2} \\
				&=& \Scale[1]{N \sum_{j,k} \+E[\frac{1}{N} \sum_{i =1}^N \frac{W_{it}}{\spsi} \Lambda_{i,j} \Lambda_{i,k} - \Sigma_{\Lambda,jk}  ]^2}  \leq M.
			\end{eqnarray*}
			Denote $y_i = \frac{W_{it}}{\spsi} \Lambda_{i,j} \Lambda_{i,k} - \Sigma_{\Lambda,jk}$. Then, 
			\[\+E(\frac{1}{N} \sum_{i = 1}^N y_i)^2 = \frac{1}{N^2} \Lp \sum_{i = 1}^N \+E[y_i^2] + \sum_{i \neq l} \+E[y_i y_l] \Rp. \]
			Note that $\+E[y_i^2] = \+E\Ls \frac{W_{it}}{\spsi} \Lambda_{i,j} \Lambda_{i,k}   \Rs^2 - \Sigma_{\Lambda,jk}^2 = \+E\Ls \frac{1}{\spsi}\+E[\Lambda_{i,j}^2 \Lambda_{i,k}^2|S_i] \Rs - \Sigma_{\Lambda,jk}^2 \leq \frac{1}{\underline{p}}\+E[\Lambda_{i,j}^2 \Lambda_{i,k}^2] - \Sigma_{\Lambda,jk}^2   $ is bounded and  $\+E[y_i y_l] = \+E[\+E[y_i y_l |S_i]] = \+E[\+E[y_i|S_i] \+E[y_l|S_i]] = 0$. Hence, we get
			\begin{eqnarray*}
				\Scale[1]{N \sum_{j,k} \+E[\frac{1}{N} \sum_{i =1}^N \frac{W_{it}}{\spsi} \Lambda_{i,j} \Lambda_{i,k} - \Sigma_{\Lambda,jk}  ]^2} &\leq& \Scale[1]{\frac{1}{N} \sum_{j, k} \sum_{i = 1}^N \Lp\frac{1}{\underline{p}}\+E[\Lambda_{i,j}^2 \Lambda_{i,k}^2] - \Sigma_{\Lambda,jk}^2 \Rp } \leq M 
			\end{eqnarray*}
			and therefore $\+E \norm{\sqrt{N}\Lp \frac{1}{N}  \sum_{i = 1}^N \frac{W_{it}}{\spsi} \Lambda_i \Lambda_i^\T - \Sigma_{\Lambda} \Rp }  \leq M$. 
		\end{proof}
		
		\subsubsection{Proof of Proposition \ref{prop:simple-assump-imply-general-assump}.2(b)}
		\begin{proof}[Proof of Proposition \ref{prop:simple-assump-imply-general-assump}.2(b)]
				In this proof, suppose Assumptions \ref{ass:obs-equal-weight}, \ref{ass:obs}, \ref{ass:simple-factor-model} and \ref{ass:simple-factor-model-conditional}  hold without further statement. We show each part in Assumption \ref{ass:mom-clt-conditional} holds under Assumption \ref{ass:simple-moment-conditional}. For notation simplicity, denote $q_{ij} = q_{ij,ij}$.  
				\paragraph{Step 1: Show that Assumption \ref{ass:mom-clt-conditional}.1 holds under Assumption \ref{ass:simple-factor-model}}  \texttt{} \\  
				Denote $v_{ij,s} = e_{is}e_{js} - \+E[e_{is}e_{js}]$.  For $\phi_{i,st} = \frac{W_{it} F_s}{\ps}, \Lambda_{i}, \frac{W_{it}}{\ps} \Lambda_{i}$ , since $e$ is independent of $F$, $\Lambda$, $W$ and $S$, then $\phi_{i,st}$ is independent of $e$ for all $i,s,t$. We can use the same steps as in Step 1 of the proof of Proposition \ref{prop:simple-assump-imply-general-assump}.2 to show that Assumption \ref{ass:mom-clt-conditional}.1 holds. 
				\paragraph{Step 2: Show that Assumption \ref{ass:mom-clt-conditional}.2 holds under Assumptions \ref{ass:simple-factor-model}}  \texttt{} \\ 
				Since $F$, $\Lambda$ and $e$ are independent, $W$ is independent of $F$ and $e$, and $S$ is independent of $F$ and $e$, then $\phi_{it} = \Lambda_{i}$ and $\frac{W_{it}}{\ps} \Lambda_{i}$  is independent of $F$ and $e$ for any $i$ and $t$. We can use the same steps as in Step 2 of the proof of Proposition  \ref{prop:simple-assump-imply-general-assump}.2 to show that Assumption \ref{ass:mom-clt-conditional}.2 holds. 
				\paragraph{Step 3: Show that Assumption \ref{ass:mom-clt-conditional}.3 holds under Assumption \ref{ass:simple-factor-model} and Assumption \ref{ass:simple-moment}.2} \texttt{} \\
				Assumption \ref{ass:mom-clt-conditional}.3 is identical to Assumption \ref{ass:mom-clt}.3. We can use the same steps as in Step 3 of the proof of Proposition  \ref{prop:simple-assump-imply-general-assump}.2 to show that Assumption \ref{ass:mom-clt-conditional}.3 holds. 
				\paragraph{Step 4: Show that Assumption \ref{ass:mom-clt-conditional}.4 holds under Assumptions \ref{ass:simple-factor-model}, \ref{ass:simple-factor-model-conditional} and \ref{ass:simple-moment-conditional}.} \texttt{} \\ 
				Conditional on $S_i$ and $S_j$, $\frac{W_{it}}{\spsi} \Lambda_i  e_{it}$ and $\frac{W_{jt}}{p_{jt}} \Lambda_j  e_{jt}$ are independent for $i \neq j$. Let $y_i = \frac{W_{it}}{\spsi} \Lambda_i  e_{it}$. Since $S$ is independent of $e$, we have $\+E[y_i|S_i] = 0$. For any $i$, it holds that
				\begin{eqnarray*}
					\+E [ \norm{y_i}^4 |S_i] &=& \sum_k \+E\Ls \frac{W_{it}}{({\spsi})^4} \Lambda_{i,k}^4 e_{it}^4 |S_i \Rs =  \sum_k \+E\Ls \frac{W_{it}}{({\spsi})^4} \Lambda_{i,k}^4 e_{it}^4 |S_i \Rs \\
					&=& \sum_k \frac{1}{(\spsi)^3} \+E[\Lambda_{i,k}^4|S_i] \+E[e_{it}^4 ] \leq M,
				\end{eqnarray*}
				following from Assumptions \ref{ass:simple-moment}.1 and \ref{ass:simple-moment}.2. From Theorem 6.5 in Hansen (2020), $\frac{1}{\sqrt{N}} \sum_{i \in \tlo_t} \frac{1}{\ps} \Lambda_i  e_{it}$ is asymptotically normal and 
				\[\Scale[1]{\covI_{F,t} = \AVar_{\Lambda,e} \Lp \frac{1}{\sqrt{N}} \sum_{i \in \tlo_t} \frac{1}{\ps} \Lambda_i  e_{it} \Rp  = \lim_{N \rightarrow \infty} \frac{1}{N} \sum_{i = 1}^N \frac{1}{\spsi} \+E[\Lambda_{i} \Lambda_i^\T|S_i] \sigma_e^2   },  \]
				as $\lim_{N \rightarrow \infty} \frac{1}{N} \sum_{i = 1}^N \frac{1}{\spsi} \+E[\Lambda_{i} \Lambda_i^\T|S_i] $ exists based on Assumption \ref{ass:simple-moment-conditional}.3. 
				\paragraph{Step 5: Show that Assumption \ref{ass:mom-clt-conditional}.5 holds under Assumptions \ref{ass:simple-factor-model}, \ref{ass:simple-moment}.2 and \ref{ass:simple-moment-conditional}.} \texttt{} \\
				\textit{Step 5.1: Show that \\ $\Scale[0.9]{ \lim_{N\rightarrow \infty}  \frac{1}{N^4} \sum_{i=1}^N \sum_{l = 1}^N \sum_{j = 1}^N \sum_{k = 1}^N  \Lp  \frac{q_{li,kj}}{q_{li}q_{kj}} - 1  \Rp  \frac{W_{it}}{\spsi} \frac{W_{jt}}{\spsj} (\Lambda_i \Lambda_i^\T \otimes I_r ) (I_r \otimes  \Lambda_l \Lambda_l^\T) \Xi_F (I_r \otimes  \Lambda_k \Lambda_k^\T) (\Lambda_j \Lambda_j^\T \otimes I_r )      }$ exists } \\ \\
				For notation simplicity, denote $y_{ilkj} \coloneqq  \tvec(  (\Lambda_i \Lambda_i^\T \otimes I_r ) (I_r \otimes  \Lambda_l \Lambda_l^\T) \Xi_F (I_r \otimes  \Lambda_k \Lambda_k^\T) (\Lambda_j \Lambda_j^\T \otimes I_r )  ) $, $x_{ilkj} \coloneqq \frac{W_{it}}{\spsi} \frac{W_{jt}}{\spsj}  y_{ilkj} $ and $x_{ilkj,m}$ is the $m$-th entry in $x_{ilkj}$. From the definition of $q_{ij}$ and $q_{ij,kj}$,  we have $0 \leq  \frac{q_{ij,kj}}{q_{ij}q_{kj}} - 1 \leq \frac{1}{q_{kj}} - 1\leq \frac{1}{\underline{q}} - 1$ and therefore 
				\begin{align*} 
					\textstyle
					& \+E \Lp \frac{1}{N^4} \sum_{i=1}^N \sum_{l = 1}^N \sum_{j = 1}^N \sum_{k = 1}^N  \Lp  \frac{q_{li,kj}}{q_{li}q_{kj}} - 1  \Rp  (x_{ilkj,m} - \+E[ x_{ilkj,m}|S]) \Rp^2 \\
					=& \Scale[1]{\frac{1}{N^8} \sum_{\substack{\text{distinct } \\ i, i', j, j', \\ k, k', l, l'  }  } \Lp  \frac{q_{li,kj}}{q_{li}q_{kj}} - 1  \Rp  \Lp  \frac{q_{l'i',k'j'}}{q_{l'i'}q_{k'j'}} - 1  \Rp \+E[ ( x_{ilkj,m} - \+E[ x_{ilkj,m}|S] ) (  x_{i'l'k'j',m} - \+E[ x_{i'l'k'j',m}|S] ) ] } \\
					& \quad + \text{other terms} \\
					=& O\Lp \frac{1}{N} \Rp ,
				\end{align*}
				since the number of terms in the other terms is of order $O(N^7)$ and
				\begin{align*}
					&\+E \Bigg[ (x_{ilkj,m} - \+E[ x_{ilkj,m}|S] ) (  x_{i'l'k'j',m} - \+E[ x_{i'l'k'j',m}|S] )  \Bigg]  \\=& \+E\Bigg[ \bigg( \frac{W_{it}}{\spsi} \frac{W_{jt}}{\spsj} y_{ilkj,m} - \+E \Big[ \frac{W_{it}}{\spsi} \frac{W_{jt}}{\spsj}  y_{ilkj,m} \Big|S \Big] \bigg) \bigg( \frac{W_{i't}}{p_{i't}^{S_{i'}}} \frac{W_{j't}}{p_{j't}^{S_{j'}} }  y_{i'l'k'j',m} - \+E \Big[ \frac{W_{i't}}{p_{i't}^{S_{i'}}} \frac{W_{j't}}{p_{j't}^{S_{j'}} }  y_{i'l'k'j',m} \Big|S \Big] \bigg)   \Bigg] = 0
				\end{align*}
				for distinct $ i, i', j, j', k, k', l, l'$ by Assumption \ref{ass:obs}.2 and \ref{ass:simple-factor-model-conditional}. Similarly, we can show that
				\begin{align*} 
					\textstyle
					& \+E \Lp \frac{1}{N^4} \sum_{i=1}^N \sum_{l = 1}^N \sum_{j = 1}^N \sum_{k = 1}^N  \Lp  \frac{q_{li,kj}}{q_{li}q_{kj}} - 1  \Rp  (\+E[ x_{ilkj,m}|S]  - \+E[ x_{ilkj,m}]) \Rp^2 
					= O\Lp \frac{1}{N} \Rp 
				\end{align*}
				based on the fact that $\Lambda_i$ is iid.  Hence, the Chebyshev's inequality implies
				\begin{align*}
					& \Scale[1]{ \lim_{N\rightarrow \infty}  \frac{1}{N^4} \sum_{i=1}^N \sum_{l = 1}^N \sum_{j = 1}^N \sum_{k = 1}^N  \Lp  \frac{q_{li,kj}}{q_{li}q_{kj}} - 1  \Rp  (\Lambda_i \Lambda_i^\T \otimes I_r ) (I_r \otimes  \Lambda_l \Lambda_l^\T) \Xi_F (I_r \otimes  \Lambda_k \Lambda_k^\T) (\Lambda_j \Lambda_j^\T \otimes I_r )      }  \\
					=&  \Scale[1]{ \lim_{N\rightarrow \infty} \frac{1}{N^4} \sum_{i=1}^N \sum_{l = 1}^N \sum_{j = 1}^N \sum_{k = 1}^N   \Lp  \frac{q_{li,kj}}{q_{li}q_{kj}} - 1  \Rp   \Big(  \+E[\Lambda_i \Lambda_i^\T]  \otimes  \+E[\Lambda_l \Lambda_l^\T]  \Big)  \Xi_F  \Big(  \+E[\Lambda_j \Lambda_j^\T]  \otimes  \+E[\Lambda_k \Lambda_k^\T]  \Big) } \\
					=&  \Scale[1]{  \Lp \lim_{N\rightarrow \infty} \frac{1}{N^4} \sum_{i=1}^N \sum_{l = 1}^N \sum_{j = 1}^N \sum_{k = 1}^N   \frac{q_{li,kj}}{q_{li}q_{kj}} - 1  \Rp  \Big( \Sigma_{\Lambda} \otimes \Sigma_{\Lambda}  \Big)  \Xi_F  \Big(  \Sigma_{\Lambda}  \otimes \Sigma_{\Lambda}  \Big)  } \\
					=& (\omega - 1)  \Big( \Sigma_{\Lambda} \otimes \Sigma_{\Lambda}  \Big)  \Xi_F  \Big(  \Sigma_{\Lambda}  \otimes \Sigma_{\Lambda}  \Big),
				\end{align*} 
				where the last equality follows from Assumption \ref{ass:simple-moment}.2. \\ \\
				\textit{Step 5.2: Show that $\frac{1}{N^2} \sum_{i=1}^N \sum_{l = 1}^N \Lambda_l \Lambda_l^\T \Lp  \frac{1}{|\tlq_{li}|} \sum_{s \in \tlq_{li}} F_s F_s^\T - \frac{1}{T} \sum_{s=1}^T F_s F_s^\T \Rp \frac{W_{it}}{\spsi} \Lambda_i \Lambda_{i} $ is asymptotically normal and $\frac{1}{N^2} \sum_{i=1}^N \sum_{l = 1}^N \Lambda_l \Lambda_l^\T \Lp  \frac{1}{|\tlq_{li}|} \sum_{s \in \tlq_{li}} F_s F_s^\T - \frac{1}{T} \sum_{s=1}^T F_s F_s^\T \Rp \frac{W_{it}}{\spsi} \Lambda_i \Lambda_{i} v_t$ converges stably in law for $v_t = H^\T \tilde D^\I (H^\T)^{-1} F_t$ } \\ 
				The vectorized form of $\Lambda_l \Lambda_l^\T \Lp  \frac{1}{|\tlq_{li}|} \sum_{s \in \tlq_{li}} F_s F_s^\T - \frac{1}{T} \sum_{s=1}^T F_s F_s^\T \Rp \Lambda_i \Lambda_{i} $ is $ (\Lambda_i \Lambda_i^\T \otimes I_r ) (I_r \otimes  \Lambda_l \Lambda_l^\T) v^{(l,i)}$. Then, it holds that 
				\begin{align*}
					\mathbf{\Phi} =& \Scale[0.95]{\ACov  \Lp  \frac{\sqrt{T}}{N^2} \sum_{i = 1}^N \sum_{l = 1}^N  \frac{W_{it}}{\spsi} (\Lambda_i \Lambda_i^\T \otimes I_r ) (I_r \otimes  \Lambda_l \Lambda_l^\T) v^{(l,i)},  \frac{\sqrt{T}}{N^2} \sum_{i = 1}^N \sum_{l = 1}^N  \frac{W_{it}}{\spsi} (\Lambda_i \Lambda_i^\T \otimes I_r ) (I_r \otimes  \Lambda_l \Lambda_l^\T) v^{(l,i)} \Rp}  \\
					=& \Scale[0.93]{\lim_{N, T\rightarrow \infty} \frac{1}{N^4} \sum_{i=1}^N \sum_{l = 1}^N \sum_{j = 1}^N \sum_{k = 1}^N \Cov \Lp \frac{W_{it}}{\spsi} (\Lambda_i \Lambda_i^\T \otimes I_r ) (I_r \otimes  \Lambda_l \Lambda_l^\T) v^{(l,i)}, \frac{W_{jt}}{\spsj} (\Lambda_j \Lambda_j^\T \otimes I_r ) (I_r \otimes  \Lambda_k \Lambda_k^\T) v^{(k,j)}  \Rp    }  \\
					=& \Scale[0.95]{\lim_{N\rightarrow \infty}   \frac{1}{N^4} \sum_{i=1}^N \sum_{l = 1}^N \sum_{j = 1}^N \sum_{k = 1}^N \frac{W_{it}}{\spsi} \frac{W_{jt}}{\spsj} (\Lambda_i \Lambda_i^\T \otimes I_r ) (I_r \otimes  \Lambda_l \Lambda_l^\T) \mathbf{\Phi}^{\mathbf{v}}_{g(l,i), g(k,j)} (I_r \otimes  \Lambda_k \Lambda_k^\T) (\Lambda_j \Lambda_j^\T \otimes I_r )      }  \\
					=& \Scale[0.95]{\lim_{N\rightarrow \infty}  \frac{1}{N^4} \sum_{i=1}^N \sum_{l = 1}^N \sum_{j = 1}^N \sum_{k = 1}^N  \Lp  \frac{q_{li,kj}}{q_{li}q_{kj}} - 1  \Rp  \frac{W_{it}}{\spsi} \frac{W_{jt}}{\spsj}  (\Lambda_i \Lambda_i^\T \otimes I_r ) (I_r \otimes  \Lambda_l \Lambda_l^\T) \Xi_F (I_r \otimes  \Lambda_k \Lambda_k^\T) (\Lambda_j \Lambda_j^\T \otimes I_r )      }  \\
					=&  \textstyle \Lp \omega - 1  \Rp  ( \Sigma_\Lambda  \otimes  \Sigma_\Lambda )  \Xi_F   (\Sigma_\Lambda  \otimes  \Sigma_\Lambda) .
				\end{align*}
				
				The stable convergence follows from a similar argument as in Step 5.4 in the proof of Proposition \ref{prop:simple-assump-imply-general-assump}.1(b)
				
				\textit{Step 5.3: Show that $ \frac{\sqrt{T}}{N} \sum_{i = 1}^N  \Lambda_i \Lambda_i^\T \Lp \frac{1}{|\tlq_{ij}|} \sum_{t \in \tlq_{ij}} F_t F_t^\T - \frac{1}{T} \sum_{t = 1}^T F_t F_t^\T \Rp $ and  \\ $ \frac{\sqrt{T}}{N^2} \sum_{l = 1}^N \sum_{i = 1}^N \Lambda_l \Lambda_l^\T \Lp \frac{1}{|\tlq_{li}|} \sum_{t \in \tlq_{li}} F_t F_t^\T - \frac{1}{T} \sum_{t = 1}^T F_t F_t^\T \Rp   \frac{W_{it}}{\spsi} \Lambda_i \Lambda_i^\T$ are jointly asymptotically normal and their asymptotic covariance converges.} \\ 
				The randomness of these two terms both come from $v^{(l,i)}$, which is asymptotic normal. These two terms are weighted average of $v^{(l,i)}$ and therefore they are jointly asymptotic normal. Next we show their aymptotic covariance converges and we provide the limit:
				\begin{align*}
					\mathbf{\Phi}_t^\cov =& \Scale[0.9]{\ACov \Lp  \frac{\sqrt{T}}{N^2} \sum_{i = 1}^N \sum_{l = 1}^N \frac{W_{it}}{\spsi}   (\Lambda_i \Lambda_i^\T \otimes I_r ) (I_r \otimes  \Lambda_l \Lambda_l^\T) v^{(l,i)},  \frac{\sqrt{T}}{N} \sum_{i = 1}^N   (I_r \otimes  \Lambda_i \Lambda_i^\T) v^{(i,j)}\Rp}  \\
					=& \Scale[0.9]{\lim_{N, T\rightarrow \infty} \frac{1}{N^3} \sum_{i=1}^N \sum_{l = 1}^N  \sum_{k = 1}^N \Cov\Lp \frac{W_{it}}{\spsi}  (\Lambda_i \Lambda_i^\T \otimes I_r ) (I_r \otimes  \Lambda_l \Lambda_l^\T) v^{(l,i)}, (I_r \otimes  \Lambda_k \Lambda_k^\T) v^{(k,j)}  \Rp    }  \\
					=& \Scale[0.9]{\lim_{N\rightarrow \infty}   \frac{1}{N^3} \sum_{i=1}^N \sum_{l = 1}^N  \sum_{k = 1}^N \frac{W_{it}}{\spsi}   (\Lambda_i \Lambda_i^\T \otimes I_r ) (I_r \otimes  \Lambda_l \Lambda_l^\T) \mathbf{\Phi}^{\mathbf{v}}_{g(l,i), g(k,j)} (I_r \otimes  \Lambda_k \Lambda_k^\T)   }  \\
					=& \Scale[0.9]{\lim_{N\rightarrow \infty}  \frac{1}{N^3} \sum_{i=1}^N \sum_{l = 1}^N  \sum_{k = 1}^N  \Lp  \frac{q_{li,kj}}{q_{li}q_{kj}} - 1  \Rp  (\Lambda_i \Lambda_i^\T \otimes I_r ) (I_r \otimes  \Lambda_l \Lambda_l^\T) \Xi_F (I_r \otimes  \Lambda_k \Lambda_k^\T)   }  \\
					=&  \textstyle \Lp \lim_{N\rightarrow \infty} \frac{1}{N^3} \sum_{i=1}^N \sum_{l = 1}^N  \sum_{k = 1}^N   \frac{q_{li,kj}}{q_{li}q_{kj}} - 1  \Rp  ( \Sigma_{\Lambda} \otimes I_r) (I_r \otimes  \Sigma_\Lambda )  \Xi_F   (I_r \otimes  \Sigma_\Lambda) \\
					=&  \textstyle \Lp \lim_{N\rightarrow \infty} \frac{1}{N^3} \sum_{i=1}^N \sum_{l = 1}^N  \sum_{k = 1}^N   \frac{q_{li,kj}}{q_{li}q_{kj}} - 1  \Rp  ( \Sigma_{\Lambda} \otimes  \Sigma_\Lambda )  \Xi_F   (I_r \otimes  \Sigma_\Lambda) \\
					=&  \textstyle \Lp\omega_j - 1  \Rp  ( \Sigma_{\Lambda} \otimes  \Sigma_\Lambda )  \Xi_F   (I_r \otimes  \Sigma_\Lambda) ,
				\end{align*}
				where the second to last equality follow from a similar argument as in Step 5.1 in the proof of Proposition \ref{prop:simple-assump-imply-general-assump}.2.
				Similar as in Step 5.1, we can show\\ $\lim_{N\rightarrow \infty}  \frac{1}{N^3} \sum_{i=1}^N \sum_{l = 1}^N  \sum_{k = 1}^N  \Lp  \frac{q_{li,kj}}{q_{li}q_{kj}} - 1  \Rp \frac{W_{it}}{\spsi}  (\Lambda_i \Lambda_i^\T \otimes I_r ) (I_r \otimes  \Lambda_l \Lambda_l^\T) \Xi_F (I_r \otimes  \Lambda_k \Lambda_k^\T)$ converges to $\Lp \lim_{N\rightarrow \infty} \frac{1}{N^3} \sum_{i=1}^N \sum_{l = 1}^N  \sum_{k = 1}^N   \frac{q_{li,kj}}{q_{li}q_{kj}} - 1  \Rp  ( \Sigma_{\Lambda}\otimes  \Sigma_\Lambda )  \Xi_F   (I_r \otimes  \Sigma_\Lambda) $. 
				
				The joint stable convergence between $ \frac{\sqrt{T}}{N} \sum_{i = 1}^N  \Lambda_i \Lambda_i^\T \Lp \frac{1}{|\tlq_{ij}|} \sum_{t \in \tlq_{ij}} F_t F_t^\T - \frac{1}{T} \sum_{t = 1}^T F_t F_t^\T \Rp u_i$ and $ \frac{\sqrt{T}}{N^2} \sum_{l = 1}^N \sum_{i = 1}^N \Lambda_l \Lambda_l^\T \Lp \frac{1}{|\tlq_{li}|} \sum_{t \in \tlq_{li}} F_t F_t^\T - \frac{1}{T} \sum_{t = 1}^T F_t F_t^\T \Rp   \frac{W_{it}}{\spsi} \Lambda_i \Lambda_i^\T v_t$  follows from a argument as in Step 5 in the proof of Proposition \ref{prop:simple-assump-imply-general-assump}.1(b).
				\paragraph{Step 6: Show that Assumption \ref{ass:mom-clt-conditional}.6 holds under Assumptions \ref{ass:simple-factor-model}, \ref{ass:simple-factor-model-conditional}, and \ref{ass:simple-moment-conditional}.} \texttt{} \\ 
				Denote $V_{li} = \frac{1}{|\tlq_{li}|} \sum_{s \in \tlq_{li}} F_s F_s^\T - \frac{1}{T} \sum_{s=1}^T F_s F_s^\T $, which is asymptotically normal by Assumption \ref{ass:mom-clt-conditional}.\ref{ass:asy-normal-add-term-thm-loading}. We calculate the variance of the term $\frac{1}{N} \sum_{i = 1}^N \frac{W_{it}}{\spsi} V_{li} \Lambda_i e_{it}$. Since $e$ is independent of $F$, $\Lambda$, $W$ and $S$, we have 
				\begin{eqnarray*}
					&& \Scale[1]{\Cov \Lp \sum_{i = 1}^N \frac{W_{it}}{\spsi} V_{li} \Lambda_i e_{it} ,  \sum_{i = 1}^N \frac{W_{it}}{\spsi} V_{li} \Lambda_i e_{it} \Rp} \\
					&=& \Scale[1]{\sum_{i = 1}^N \sum_{m = 1}^N \+E \Ls \frac{W_{it} W_{mt}}{\spsi p_{mt}^{S_m} } V_{li} \Lambda_i e_{it} e_{mt} \Lambda_m^\T V_{mi}^\T \Rs  }  \\
					&=& \Scale[1]{ \sum_{i = 1}^N \sum_{m = 1}^N  \+E[ V_{li} \+E[\frac{W_{it} W_{mt}}{\spsi p_{mt}^{S_m} }  \Lambda_i  \+E[e_{it} e_{mt}]  \Lambda_m^\T  ] V_{mi}^\T ]  } \\
					&=&  \Scale[1]{ \sum_{i = 1}^N   \+E[ V_{li} \+E[\frac{W_{it}}{(\spsi)^2 }  \Lambda_i  \+E[e_{it}^2]  \Lambda_i^\T |S ] V_{li}^\T ] = O\Lp \frac{N}{T} \Rp } ,
				\end{eqnarray*}
				since $\+E[e_{it} e_{mt}] = 0$ for $i \neq m$.  Then Assumption \ref{ass:mom-clt-conditional}.6 holds.
			\end{proof}

			\subsubsection{Proof of Proposition \ref{prop:simple-assump-imply-general-assump}.3: Treatment Tests for Simplified Model} 
			\paragraph{Step 1: Show that Assumption \ref{ass:add-factor}.1 holds} \texttt{} \\ 
			Since $F_t$ is i.i.d. by Assumption \ref{ass:simple-factor-model}.1, $e_{it}$ is i.i.d. by Assumption \ref{ass:simple-factor-model}.3, and $F_t$ is independent of $e_{it}$, we can apply the CLT resulting in $\frac{1}{\sqrt{\Ttr}} \sum_{T-\Ttr+1}^{T} F_t e_{it} \xrightarrow{d} N(0, \Sigma_{F,e_i}), $
			where $\Sigma_{F,e_i} = \sigma_e^2 \Sigma_F $.  
			\paragraph{Step 2:  Show that Assumption \ref{ass:add-factor}.2 holds} \texttt{}  \\ 
			We calculate the covariance of $\sum_{t=T-\Ttr+1}^T  \sum_{j = 1}^N W_{jt}  \Lambda_j  e_{jt}$. Since $e$ is independent of $W$, and $\Lambda$ and $e_{it}$ is i.i.d., we have
			\begin{align*}
				& \Scale[1]{ \Cov\Lp \sum_{t=T-\Ttr+1}^T  \sum_{j = 1}^N W_{jt}  \Lambda_j  e_{jt}, \sum_{t=T-\Ttr+1}^T  \sum_{j = 1}^N W_{jt} \Lambda_j  e_{jt}  \Rp} \\
				=& \Scale[1]{\sum_{t=T-\Ttr+1}^T \sum_{s=T-\Ttr+1}^T  \sum_{i = 1}^N \sum_{j = 1}^N  \Cov \Lp W_{it}  \Lambda_i  e_{it}, W_{js} \Lambda_j  e_{js}  \Rp   } \\
				=& \Scale[1]{ \sum_{t=T-\Ttr+1}^T \sum_{s=T-\Ttr+1}^T  \sum_{i = 1}^N \sum_{j = 1}^N (\+E\Ls W_{it} W_{jt}  \Lambda_i  \Lambda_j^\T   e_{it} e_{js} \Rs - \+E\Ls W_{it}  \Lambda_i  e_{it}\Rs \+E\Ls W_{jt}  \Lambda_j^\T  e_{js} \Rs)  } \\
				=& \Scale[1]{ \sum_{t=T-\Ttr+1}^T   \sum_{i = 1}^N \sigma_e^2 \+E\Ls W_{it}  \Lambda_i  \Lambda_i^\T  \Rs  = O(N \Ttr )  }.
			\end{align*}
			
			Next we calculate the covariance of $ \sum_{t = \Tcontrol+1}^T \sum_{j = 1}^N Z_t F_t^\T W_{jt} \Lambda_j  e_{jt}$. Since $e$ is independent of $F$, $W$, and $\Lambda$,  and $\norm{Z_t} \leq M $, we obtain
			\begin{align*}
				& \Scale[1]{ \Cov\Lp \sum_{t=T-\Ttr+1}^T  \sum_{j = 1}^N Z_t F_t^\T  W_{jt}  \Lambda_j  e_{jt}, \sum_{t=T-\Ttr+1}^T  \sum_{j = 1}^N Z_t F_t^\T  W_{jt} \Lambda_j  e_{jt}  \Rp} \\
				=& \Scale[1]{\sum_{t=T-\Ttr+1}^T \sum_{s=T-\Ttr+1}^T  \sum_{i = 1}^N \sum_{j = 1}^N  \Cov \Lp Z_t F_t^\T  W_{it}  \Lambda_i  e_{it}, Z_s F_s^\T  W_{js} \Lambda_j  e_{js}  \Rp   } \\
				=& \Scale[1]{ \sum_{t=T-\Ttr+1}^T \sum_{s=T-\Ttr+1}^T  \sum_{i = 1}^N \sum_{j = 1}^N Z_t ( \+E\Ls W_{it} W_{jt} F_t^\T \Lambda_i  \Lambda_j^\T F_s  e_{it} e_{js} \Rs - \+E\Ls W_{it} F_t^\T  \Lambda_i  e_{it}\Rs \+E\Ls W_{jt}  \Lambda_j^\T F_s  e_{js} \Rs)  Z_s^\T } \\
				=& \Scale[1]{ \sum_{t=T-\Ttr+1}^T   \sum_{i = 1}^N \sigma_e^2 Z_t   \+E\Ls W_{it} F_t^\T \Lambda_i  \Lambda_i^\T F_t  \Rs Z_t^\T  = O(N \Ttr )  } .
			\end{align*}
			Hence, \ref{ass:add-factor}.2 holds. 
			\paragraph{Step 3:  Show that Assumption \ref{ass:add-factor}.3 holds} \texttt{}  \\ 
			\ref{ass:add-factor}.3 can be shown with similar arguments as in Step 5 in the proof of Proposition \ref{prop:simple-assump-imply-general-assump}.1(b). \\
			
			Assumption \ref{ass:add-factor-conditional}.1 is the same as Assumption \ref{ass:add-factor}.1. Assumptions \ref{ass:add-factor-conditional}.2 and \ref{ass:add-factor-conditional}.3 can be shown similarly as Assumptions \ref{ass:add-factor}.2 and \ref{ass:add-factor}.3.

			\subsection{Proof of Theorem \ref{thm:consistency-same-H}: Consistency of Loadings}
			Denote by $W_i \in \+R^{T \times 1}$ the $i$-th row in $W$, $e_i \in \+R^{T \times 1}$ the $i$-th row in $e$ and $\tilde q_{ij} = \frac{|\tlq_{ij}|}{T}$. Plugging $\tilde Y = (\Lambda F^\T) \odot W + e \odot W$ into
			\[\Lp \frac{1}{N} ( \tilde Y \tilde Y^\T) \odot \Big[ \frac{1}{|\tlq_{ij}|} \Big]  \Rp \tilde \Lambda = \tilde \Lambda \tilde D,\]
			and multiplying with $\tilde D^\I$ on the right, we obtain 
			\[\frac{1}{NT} \Ls\Lp (W \odot (\Lambda F^\T) + W \odot e) ( (F \Lambda^\T) \odot W^\T + e^\T \odot W^\T  )\Rp  \odot \Big[ \frac{1}{\tilde q_{ij}} \Big]  \Rs  \tilde \Lambda \tilde D^\I = \tilde \Lambda.\]
			Note that the $(i,j)$-th entries in $(W \odot (\Lambda F^\T))((F \Lambda^\T) \odot W^\T)$, $(W \odot (\Lambda F^\T)) (e^\T \odot W^\T) $, $(W \odot e) ((F \Lambda^\T) \odot W^\T) $ and $ (W \odot e) (e^\T \odot W^\T)$ take the following form:
			\begin{eqnarray*}
				\Lp(W \odot (\Lambda F^\T)) ((F \Lambda^\T) \odot W^\T) \Rp_{ij} &=& \Lambda_i^\T F^\T\text{diag}(W_i \odot W_j) F  \Lambda_j \\
				\Lp (W \odot (\Lambda F^\T)) (e^\T \odot W^\T) \Rp_{ij} &=&  e_i^\T\text{diag}(W_i \odot W_j) F \Lambda_j  \\
				\Lp (W \odot e) ((F \Lambda^\T) \odot W^\T) \Rp_{ij} &=&  \Lambda_i^\T F^\T\text{diag}(W_i \odot W_j) e_j \\
				\Lp (W \odot e) (e^\T \odot W^\T) \Rp_{ij} &=&  e_i^\T\text{diag}(W_i \odot W_j) e_j. 
			\end{eqnarray*}
			Then, we have 
			\begin{eqnarray}\label{eqn:decomp-f-tilde}
				\nonumber \tilde \Lambda_j = \frac{1}{NT} \tilde D^\I \!\!\!\!\!\!\!\!\! && \left[  \sum_{i=1}^N \tilde \Lambda_i \Lambda_i^\T F^\T\text{diag}(W_i \odot W_j) F \Lambda_j/\tilde q_{ij}  + \sum_{i=1}^N \tilde \Lambda_i e_i^\T\text{diag}(W_i \odot W_j) F \Lambda_j/\tilde q_{ij} \right. \\
				&& \left. + \sum_{i=1}^N \tilde \Lambda_i \Lambda_i^\T F^\T\text{diag}(W_i \odot W_j) e_j/\tilde q_{ij} + \sum_{i=1}^N \tilde \Lambda_i e_i^\T\text{diag}(W_i \odot W_j) F e_j/\tilde q_{ij} \right].
			\end{eqnarray}
			Denote $H_j = \frac{1}{NT} \tilde D^\I \sum_{i=1}^N \tilde \Lambda_i \Lambda_i^\T F^\T\text{diag}(W_i \odot W_j) F/\tilde q_{ij}$. From \eqref{eqn:decomp-f-tilde}, we have 
			\[\tilde \Lambda_j - H_j \Lambda_j = \tilde D^\I \Lp \frac{1}{N} \sum_{i=1}^N \tilde \Lambda_i \gamma(i,j) + \frac{1}{N} \sum_{i=1}^N \tilde \Lambda_i \zeta_{ij} + \frac{1}{N} \sum_{i=1}^N \tilde \Lambda_i \eta_{ij} + \frac{1}{N} \sum_{i=1}^N \tilde \Lambda_i \xi_{ij} \Rp,\]
			where 
			\begin{eqnarray*}
				\gamma(i,j) &=& \frac{1}{|\tlq_{ij}|} \sum_{t \in \tlq_{ij}} \+E [e_{it}e_{jt}]   \\
				\zeta_{ij} &=& \frac{1}{|\tlq_{ij}|} \sum_{t \in \tlq_{ij}} e_{it}e_{jt} - \gamma(i,j) \\
				\eta_{ij} &=& \frac{1}{|\tlq_{ij}|} \sum_{t \in \tlq_{ij}} \Lambda_i^\T F_t e_{jt}  \\
				\xi_{ij} &=& \frac{1}{|\tlq_{ij}|} \sum_{t \in \tlq_{ij}} \Lambda_j^\T F_t e_{it}.
			\end{eqnarray*}

			We first provide a bound for $\gamma(i,j)$, $\eta_{ij}$ and $\xi_{ij}$, which we need to show the consistency of $\tilde \Lambda_j$.
			
			\begin{lemma}\label{lemma:prep-consistency}
				Under Assumptions \ref{ass:obs} and \ref{ass:factor-model}, we have for some $M < \infty$, and for all $N$ and $T$, 
				\begin{enumerate}
					\item $\sum_{i=1}^N \gamma(i,j)^2 \leq M$ and $\frac{1}{N}\sum_{i=1}^N \sum_{i=1}^N \gamma(i,j)^2 \leq M$, where $\gamma(i,j) = \+E\Ls\frac{1}{|\tlq_{ij}|} \sum_{t \in \tlq_{ij}} e_{it}e_{jt}  \Rs$.
					\item $\+E\Ls \Lp\frac{1}{\sqrt{|\tlq_{ij}|}} \Lambda_i^\T \sum_{t \in \tlq_{ij}} F_t e_{jt} \Rp^2 \Rs \leq M$, for all $j$.
				\end{enumerate}
			\end{lemma}
			
			\begin{proof}[Proof of Lemma \ref{lemma:prep-consistency}]
				\begin{enumerate}
					\item Let $\rho(i,j) = \gamma(i,j)/ \Ls \Lp \frac{1}{|\tlq_{ij}|} \sum_{t \in \tlq_{ij}}  \+E [e_{it}^2 ]  \Rp  \Lp \frac{1}{|\tlq_{ij}|} \sum_{t \in \tlq_{ij}}  \+E [e_{jt}^2]  \Rp  \Rs^{1/2}$ \\
					$ = \gamma(i,j)/ \Ls\gamma(i,i)  \gamma(j,j) \Rs^{1/2}$. Then $|\rho(i,j)| \leq 1$ and $\rho(i,j)^2 \leq |\rho(i,j)|$. Assumption \ref{ass:factor-model}.\ref{ass:error}.3 implies that $|\gamma(i,i) | \leq M$ and $|\gamma(j,j) | \leq M$. We then have for all $i$ and $j$,\\ $\gamma(i,j)^2 = \gamma(i,i) \gamma(j,j) \rho(i,j)^2 \leq M |\gamma(i,i) \gamma(j,j) |^{1/2}  |\rho(i,j)| = M |\gamma(i,j)| $. This allows us to bound $\sum_{i=1}^N \gamma(i,j)^2 $ as follows
					\begin{align*}
						\sum_{i = 1}^N \gamma(i,j)^2  &=  M \sum_{i = 1}^N  |\gamma(i,i) \gamma(j,j)|^{1/2}  |\rho(i,j)| \leq M \underbrace{\sum_{i = 1}^N  |\gamma(i,j)|}_{\substack{\leq M \text{ from} \\ \text{ Assumption \ref{ass:factor-model}.3.(c)} } }  \leq M^2. 
					\end{align*}
					For $\frac{1}{N} \sum_{i = 1}^N \sum_{j=1}^N \gamma(i,j)^2$, we have 
					\begin{eqnarray*}
						\frac{1}{N} \sum_{i = 1}^N \sum_{j=1}^N \gamma(i,j)^2  &\leq&  \frac{M}{N} \sum_{i = 1}^N \sum_{j=1}^N |\gamma(i,i) \gamma(j,j)|^{1/2}  |\rho(i,j)| \\
						&\leq& \frac{M}{N} \sum_{i = 1}^N \sum_{j=1}^N |\gamma(i,j)| \leq M^2,
					\end{eqnarray*}
					where the last inequality follows from Assumption \ref{ass:factor-model}.\ref{ass:error}.(c).
					\item
					\begin{eqnarray*}
						\+E\Ls \Lp\frac{1}{\sqrt{|\tlq_{ij}|}} \Lambda_i^\T \sum_{t \in \tlq_{ij}} F_t e_{jt} \Rp^2 \Rs \leq  \+E[\norm{\Lambda_i}^2] \cdot  \+E \norm{\frac{1}{\sqrt{|\tlq_{ij}|}}\sum_{t \in \tlq_{ij}} F_t e_{jt}}^2  \leq \bar \Lambda^2 M
					\end{eqnarray*}
					following from Assumption \ref{ass:factor-model}.4 and the independence of $\Lambda$ with $F$ and $e$. 
				\end{enumerate}
			\end{proof}
			
			\begin{lemma}\label{lemma:consistency}
				Under Assumptions \ref{ass:obs} and \ref{ass:factor-model},  let $\delta = \min (N, T)$, we have
				\begin{eqnarray}
					\delta \Lp \frac{1}{N} \sum_{j=1}^N \norm{\tilde \Lambda_j - H_j \Lambda_j}^2 \Rp = O_P(1),
				\end{eqnarray}
				where $H_j = \frac{1}{NT} \tilde D^\I \sum_{i=1}^N \tilde \Lambda_i \Lambda_i^\T F^\T\text{diag}(W_i \odot W_j) F/\tilde q_{ij}$. 
			\end{lemma}
			
			\begin{proof}[Proof of Lemma \ref{lemma:consistency}]
				From the Cauchy-Schwartz inequality, we have 
				$\norm{\tilde \Lambda_j - H_j \Lambda_j }^2 \leq 4 \norm{\tilde D^\I}^2 (a_j + b_j + c_j + d_j)$, where 
				\begin{align*}
					a_j = \frac{1}{N^2} \norm{\sum_{i=1}^N \tilde \Lambda_i \gamma(i,j)}^2, &\quad \quad b_j = \frac{1}{N^2} \norm{\sum_{i=1}^N \tilde \Lambda_i \zeta_{ij}}^2 \\
					c_j = \frac{1}{N^2} \norm{\sum_{i=1}^N \tilde \Lambda_i \eta_{ij}}^2, & \quad \quad d_j = \frac{1}{N^2} \norm{\sum_{i=1}^N \tilde \Lambda_i \xi_{ij}}^2.
				\end{align*}
				Let us first consider $\frac{1}{N}\sum_{j=1}^N a_j \leq \frac{1}{N}$.  From $\frac{1}{N^2} \norm{\sum_{i=1}^N \tilde \Lambda_i \gamma(i,j)}^2 \leq \Lp \frac{1}{N} \sum_{i=1}^N \norm{\tilde \Lambda_i}^2  \Rp \Lp \frac{1}{N} \sum_{i=1}^N \gamma(i,j)^2 \Rp$, we conclude 
				\begin{eqnarray*}
					\frac{1}{N}\sum_{j=1}^N a_j \leq \frac{1}{N} \underbrace{\Lp  \frac{1}{N}  \sum_{i=1}^N \norm{\tilde \Lambda_i}^2 \Rp}_{O_P(1)}  \underbrace{\Lp \frac{1}{N}\sum_{i=1}^N \sum_{j=1}^N \gamma(i,j)^2  \Rp }_{O_P(1)} = O_P \Lp \frac{1}{N} \Rp,
				\end{eqnarray*}
				since $\frac{1}{N} \sum_{i=1}^N \norm{\tilde \Lambda_i}^2 = O_P(1)$ which follows from $\frac{1}{N} \tilde \Lambda^\T \tilde \Lambda = I_r$,  Assumption \ref{ass:factor-model}.2 and Lemma \ref{lemma:prep-consistency}.1.
				
				Next, let us consider $\frac{1}{N} \sum_{j=1}^N b_j$. Similar to the proof of Theorem 1 in \cite{bai2002determining}, it holds that
				\begin{eqnarray*}
					\frac{1}{N} \sum_{j=1}^N b_j \leq \Lp \frac{1}{N} \sum_{i=1}^N \norm{\tilde \Lambda_i}^2 \Rp \Lp \frac{1}{N^2} \sum_{i=1}^N \sum_{l=1}^N \Lp \sum_{j=1}^N \zeta_{ij} \zeta_{lj} \Rp^2  \Rp^{1/2},
				\end{eqnarray*}
				$\+E \Ls \sum_{j=1}^N \zeta_{ij} \zeta_{lj} \Rs^2 \leq N^2 \max_{i,j} \+E |\zeta_{ij}|^4$ and 
				\[\+E |\zeta_{ij}|^4 = \frac{1}{|\tlq_{ij}|^2} \underbrace{\+E \left\vert \frac{1}{|\tlq_{ij}|^{1/2}} \sum_{t \in \tlq_{ij}} \Lp e_{it}e_{jt} -  \+E[e_{it}e_{jt}] \Rp \right\vert^4 }_{\leq M}  \leq \frac{ M}{|\tlq_{ij}|^2} = O_P \Lp \frac{1}{T^2} \Rp\]
				because of Assumption \ref{ass:factor-model}.\ref{ass:error}.(e). Thus, $\frac{1}{N} \sum_{j=1}^N b_j = O_P \Lp \frac{1}{T} \Rp$.
				
				Next, we consider $\frac{1}{N} \sum_{j=1}^N c_j$. For any $c_j$, it holds that
				\begin{eqnarray*}
					c_j &=& \frac{1}{N^2}\norm{\sum_{i=1}^N \tilde \Lambda_i \eta_{ij}}^2  = \frac{1}{N^2} \norm{ \sum_{i=1}^N \tilde \Lambda_i \Lambda_i^\T \cdot  \Lp \frac{1}{|\tlq_{ij}|} \sum_{t \in \tlq_{ij}} F_t e_{jt} \Rp}^2 \\
					&\leq& \underbrace{\Lp \frac{1}{N} \sum_{i=1}^N \norm{\tilde \Lambda_i}^2 \Rp}_{O_P(1)} \underbrace{\Lp \frac{1}{N} \sum_{i=1}^N  \frac{1}{|\tlq_{ij}|} \Lp \frac{1}{\sqrt{|\tlq_{ij}|}} \Lambda_i^\T \sum_{i \in \tlq_{ij}} F_t e_{jt}\Rp^2 \Rp}_{\leq \max_i \frac{1}{|\tlq_{ij}|} \cdot \frac{1}{N}  \sum_{i=1}^N \Lp \frac{1}{\sqrt{|\tlq_{ij}|}} \Lambda_i^\T \sum_{t \in \tlq_{ij}} F_t e_{jt}\Rp^2 = O_P(\frac{1}{T}) \cdot O_P(1)  }  = O_P \Lp \frac{1}{T} \Rp
				\end{eqnarray*}
				since $\+E\Big[\frac{1}{N}  \sum_{i=1}^N \Lp \frac{1}{\sqrt{|\tlq_{ij}|}} \Lambda_i^\T \sum_{i \in \tlq_{ij}} F_t e_{jt}\Rp^2 \Big] \leq \bar \Lambda^2 M$ following from Lemma \ref{lemma:prep-consistency}.2 and therefore \\ $\frac{1}{N}  \sum_{i=1}^N \Lp \frac{1}{\sqrt{|\tlq_{ij}|}} \Lambda_i^\T \sum_{i \in \tlq_{ij}} F_t e_{jt}\Rp^2 = O_P(1)$. Thus $\frac{1}{N} \sum_{j=1}^N c_j =O_P \Lp \frac{1}{T} \Rp$. Similarly, we can show $\frac{1}{N} \sum_{j=1}^N d_j = O_P \Lp \frac{1}{T} \Rp$. In the following Lemma \ref{lemma:eigenvalue-converge}, we show $\norm{\tilde D^\I} = O_P(1)$. Hence,  
				\[\frac{1}{N} \sum_{j=1}^N \norm{\tilde \Lambda_j - H_j \Lambda_j }^2 \leq 4\norm{\tilde D^\I}^2 \frac{1}{N} \sum_{j=1}^N (a_j + b_j + c_j + d_j) = O_P \Lp \frac{1}{T} \Rp + O_P \Lp \frac{1}{N} \Rp = O_P \Lp \frac{1}{\delta_{NT}} \Rp.  \]
			\end{proof}

			\begin{lemma}\label{lemma:HF-H-diff}
				Under Assumptions \ref{ass:obs} and \ref{ass:factor-model}, let $H = \frac{1}{NT} \tilde D^\I \tilde \Lambda^\T \Lambda F^\T F$  and $\delta_{NT} = \min (N, T)$, then we have $H_j = O_P(1)$, $H = O_P(1)$ and for all $j$
				\[H_j - H = O_P \Lp 1/\dnt \Rp. \]
			\end{lemma}
			\begin{proof}[Proof of Lemma \ref{lemma:HF-H-diff}]
				Let us first show $H_j - H = O_P \Lp 1/\dnt \Rp$. From the definition of $H_j$ and $H$, it holds that
				\begin{align*}
					\norm{H_j - H} &= \norm{\frac{1}{N} \tilde D^\I \sum_{i=1}^N \tilde \Lambda_i \Lambda_i^\T \Big( \frac{1}{|\tlq_{ij}|} \sum_{t \in \tlq_{ij}} F_t F_t^\T - \frac{1}{T} \sum_{t=1}^T F_t F_t^\T \Big) } \\
					&\leq  \norm{\tilde D^\I} \Ls \frac{1}{N}  \sum_{i=1}^N\norm{ \tilde \Lambda_i} \norm{\Lambda_i} \norm{ \frac{1}{|\tlq_{ij}|} \sum_{t \in \tlq_{ij}} F_t F_t^\T - \frac{1}{T} \sum_{t=1}^T F_t F_t^\T } \Rs \\
					&\leq \underbrace{\norm{\tilde D^\I}}_{O_P(1)} \underbrace{\Ls \frac{1}{N}  \sum_{i=1}^N \norm{ \tilde \Lambda_i}^2 \Rs^{1/2}}_{O_P(1)} \Bigg[ \underbrace{\frac{1}{N}  \sum_{i=1}^N  \norm{\Lambda_i}^2 \norm{ \frac{1}{|\tlq_{ij}|} \sum_{t \in \tlq_{ij}} F_t F_t^\T - \frac{1}{T} \sum_{t=1}^T F_t F_t^\T }^2}_{:=\Delta}  \Bigg]^{1/2}.
				\end{align*}
				As $\Lambda_i$ is independent of $F_t$ we have 
				\begin{align*}
					\+E [\Delta] &= \frac{1}{N}  \sum_{i=1}^N  \+E  \norm{\Lambda_i}^2 \underbrace{\+E  \norm{ \frac{1}{|\tlq_{ij}|} \sum_{t \in \tlq_{ij}} F_t F_t^\T - \frac{1}{T} \sum_{t=1}^T F_t F_t^\T }^2 }_{\leq \frac{M}{|\tlq_{ij}|} \text{ from Assumption \ref{ass:factor-model}.1} } \leq \Big(\max_i \frac{1}{|\tlq_{ij}|} \Big)  \cdot \frac{M}{N} \cdot  \sum_{i=1}^N  \+E  \norm{\Lambda_i}^2  = O \Big( \frac{1}{T}  \Big).
				\end{align*}
				Hence, we conclude that $\Delta = O_P \Big( \frac{1}{T}  \Big)$ and $H_j - H = O_P \Big( \frac{1}{T}  \Big) = O_P \Lp \frac{1}{\dnt} \Rp $. Finally, let us show \\ $H_j = \frac{1}{NT} \tilde D^\I \sum_{i=1}^N \tilde \Lambda_i \Lambda_i^\T F^\T\text{diag}(W_i \odot W_j) F/\tilde q_{ij} = O_P(1)$:
				\begin{align*}
					\norm{H_j}^2 \leq \underbrace{\norm{\tilde D^{-1}}^2}_{O_P(1)} \underbrace{\Lp \frac{1}{N} \sum_{i=1}^N \norm{ \tilde \Lambda_i}^2  \Rp}_{O_P(1)}  \underbrace{\Lp \frac{1}{N} \sum_{i=1}^N \norm{ \Lambda_i^\T \frac{1}{|\tlq_{ij}|} \sum_{t \in \tlq_{ij}} F_t F_t^\T  }^2 \Rp}_{:=\Delta_1}  
				\end{align*}
				Note that since $\Lambda$ is independent of $F$, we have 
				\begin{align*}
					\+E \Ls \Delta_1 \Rs \leq  \+E \Ls \frac{1}{N} \sum_{i=1}^N \norm{\Lambda_i }^2 \norm{ \frac{1}{|\tlq_{ij}|} \sum_{t \in \tlq_{ij}} F_t F_t^\T }^2 \Rs =  \frac{1}{N} \sum_{i=1}^N \underbrace{\+E [\norm{\Lambda_i }^2]}_{\leq \bar \Lambda} \underbrace{ \+E\norm{ \frac{1}{|\tlq_{ij}|} \sum_{t \in \tlq_{ij}} F_t F_t^\T }^2 }_{\leq \bar F} \leq  \bar \Lambda \cdot \bar F.
				\end{align*}
				Hence, $\Delta_1 = O_P(1)$ and therefore $H_j = O_P(1)$. Moreover, $H = H_j - O_P \Big( \frac{1}{T}  \Big) = O_P(1)$.
			\end{proof}
			
			\begin{proof}[Proof of Theorem  \ref{thm:consistency-same-H}]
				\[\frac{1}{N} \sum_{j=1}^N \norm{\tilde \Lambda_j - H \Lambda_j }^2  \leq  \frac{1}{N} \sum_{j=1}^N \norm{\tilde \Lambda_j - H_j \Lambda_j }^2 + \frac{1}{N} \sum_{j=1}^N \norm{ ( H_j - H) \Lambda_j }^2. \]
				The first term satisfies $\frac{1}{N} \sum_{j=1}^N \norm{\tilde \Lambda_j - H_j \Lambda_j }^2 = O_P \Lp1/\delta \Rp$ as shown in Lemma \ref{lemma:consistency}. The second term is bounded by 
				\begin{align*}
					\frac{1}{N} \sum_{j=1}^N \norm{( H_j - H) \Lambda_j }^2 &\leq \Lp \frac{1}{N}\sum_{j=1}^N \norm{ H_j - H }^2\Rp  \underbrace{\Lp \frac{1}{N}\sum_{j=1}^N \norm{\Lambda_j }^2\Rp   }_{O_P(1)}. 
				\end{align*}
				Note that 
				\begin{align*}
					& \quad \frac{1}{N}\sum_{j=1}^N \norm{ (H_j - H) \Lambda_j }^2 \\
					&= \frac{1}{N}\sum_{j=1}^N  \norm{\frac{1}{N} \tilde D^\I \sum_{i=1}^N \tilde \Lambda_i \Lambda_i^\T \Big( \frac{1}{|\tlq_{ij}|} \sum_{t \in \tlq_{ij}} F_t F_t^\T - \frac{1}{T} \sum_{t=1}^T F_t F_t^\T \Big) \Lambda_j  }^2 \\
					&\leq  \norm{\tilde D^\I} \frac{1}{N}\sum_{j=1}^N  \Ls \frac{1}{N}  \sum_{i=1}^N\norm{ \tilde \Lambda_i} \norm{\Lambda_i} \norm{ \frac{1}{|\tlq_{ij}|} \sum_{t \in \tlq_{ij}} F_t F_t^\T - \frac{1}{T} \sum_{t=1}^T F_t F_t^\T } \norm{\Lambda_j } \Rs^2 \\
					&\leq \underbrace{\norm{\tilde D^\I}}_{O_P(1)} \underbrace{\Ls \frac{1}{N}  \sum_{i=1}^N \norm{ \tilde \Lambda_i}^2 \Rs}_{O_P(1)} \underbrace{\Ls \frac{1}{N^2}  \sum_{i=1}^N \sum_{j=1}^N \norm{\Lambda_i}^2 \norm{\Lambda_j}^2 \norm{ \frac{1}{|\tlq_{ij}|} \sum_{t \in \tlq_{ij}} F_t F_t^\T - \frac{1}{T} \sum_{t=1}^T F_t F_t^\T }^2  \Rs}_{:=\Delta_2}   .
				\end{align*}
				The expectation of $\Delta_2$ is
				\begin{align*}
					\+E [\Delta_2] &= \frac{1}{N^2}  \sum_{i=1}^N \sum_{j=1}^N \+E [ \norm{\Lambda_i}^2 \+E  \norm{\Lambda_j}^2] \underbrace{\+E  \norm{ \frac{1}{|\tlq_{ij}|} \sum_{t \in \tlq_{ij}} F_t F_t^\T - \frac{1}{T} \sum_{t=1}^T F_t F_t^\T }^2 }_{\leq \frac{M}{|\tlq_{ij}|} \text{ from Assumption \ref{ass:factor-model}.1} } \\
					&\leq \Big(\max_{i,j} \frac{1}{|\tlq_{ij}|} \Big)  \cdot \frac{M}{N} \cdot  \sum_{i=1}^N \sum_{j=1}^N  (\+E  \norm{\Lambda_i}^4 \+E  \norm{\Lambda_j}^4)^{1/2}  = O \Big( \frac{1}{T}  \Big).
				\end{align*}
				Hence $\Delta_2 = O_P \Big( \frac{1}{T}  \Big)$ and therefore  
				\begin{align}
					\frac{1}{N} \sum_{j=1}^N \norm{( H_j - H) \Lambda_j }^2 = O_P \Big( \frac{1}{T}  \Big) \label{eqn:squared-difference-Hj-H}    .
				\end{align} 
				Combinging this result with the bound for the first term $\frac{1}{N} \sum_{j=1}^N \norm{\tilde \Lambda_j - H_j \Lambda_j }^2 = O_P \Lp1/\delta_{NT} \Rp$ based on Lemma \ref{lemma:consistency}, we conclude that
				\[\frac{1}{N} \sum_{j=1}^N \norm{\tilde \Lambda_j - H \Lambda_j }^2 = O_P \Lp \frac{1}{\delta_{NT}} \Rp.\]
			\end{proof}

			\subsection{Proof of Theorem \ref{theorem:asy-normal-equal-weight}: Asymptotic Distribution}

			\begin{lemma}\label{lemma:eigenvalue-converge}\label{lemma:eigenvalue}
				Assume Assumptions \ref{ass:obs} and \ref{ass:factor-model} hold. As $T, N \rightarrow \infty$, it holds that
				\begin{enumerate}
					\item $\frac{1}{N} \tilde \Lambda^\T \Lp \frac{1}{N} ( \tilde Y \tilde Y^\T) \odot \Big[ \frac{1}{|\tlq_{ij}|} \Big]  \Rp \tilde \Lambda  \xrightarrow{p} D$,
					\item $\frac{1}{N^2}  \tilde \Lambda^\T \Lp  \Lp (\Lambda F^\T) \odot W \Rp\Lp(F \Lambda^\T) \odot W^\T \Rp   \odot \Big[ \frac{1}{|\tlq_{ij}|} \Big]  \Rp \tilde F \xrightarrow{p} D$
					\item $\frac{1}{N^2 T}  \tilde \Lambda^\T \Lp  \Lambda F^\T F \Lambda^\T \Rp \tilde \Lambda \xrightarrow{p} D$,
				\end{enumerate}
				where $D = \textnormal{diag}(d_1, d_2, \cdots, d_r)$ are the eigenvalues of $\Sigma_\Lambda \Sigma_F$. 
			\end{lemma}
			\begin{proof}[Proof of Lemma \ref{lemma:eigenvalue-converge}]
				This proof is based on the proof of (R12) on page 1175 in \cite{stock2002forecasting} which shows that the eigenvalues converge on the fully observed panel. Let $\gamma$ denote $N \times 1$ vector and let $\Gamma = \{\gamma| \gamma^\T \gamma/N = 1\}$, $R(\gamma) = \frac{1}{N^2} \gamma^\T \Lp ( \tilde Y \tilde Y^\T) \odot \Big[ \frac{1}{|\tlq_{ij}|} \Big]  \Rp \gamma$, \\ $\tilde R(\gamma) = \frac{1}{N^2} \gamma^\T \Lp \Lp(\Lambda F^\T) \odot W \Rp\Lp(F \Lambda^\T) \odot W^\T \Rp   \odot \Big[ \frac{1}{|\tlq_{ij}|} \Big]  \Rp \gamma $ and $R^\ast(\gamma) = \frac{1}{N^2T} \gamma^\T    \Lambda F^\T F \Lambda^\T \gamma$. We follow similar steps as \cite{stock2002forecasting} and can sequentially show that
				\begin{itemize}
					\item [(R2)] $\sup_{\gamma \in \Gamma} \frac{1}{N^2} \gamma^\T \Lp \Lp (W \odot e) (e^\T \odot W^\T)  \Rp \odot \Big[ \frac{1}{|\tlq_{ij}|} \Big]  \Rp \gamma \xrightarrow{p} 0$
					\item [(R5)] $\sup_{\gamma \in \Gamma} \frac{1}{N^2} \Big|\gamma^\T \Lp \Lp ( (W \odot e)  (F \Lambda^\T) \odot W^\T) \Rp \odot \Big[ \frac{1}{|\tlq_{ij}|} \Big]  \Rp  \gamma \Big| \xrightarrow{p} 0  $
					\item [(R6)] $\sup_{\gamma \in \Gamma} |R(\gamma) - \tilde R(\gamma)| \xrightarrow{p} 0  $ and $\sup_{\gamma \in \Gamma} |R(\gamma) - R^\ast(\gamma)| \xrightarrow{p} 0  $
					
				\end{itemize}
				\begin{proof}[Proof of (R6)]
					We have the decomposition
					\[R(\gamma) - R^\ast(\gamma) = R(\gamma) - \tilde R(\gamma) + \tilde  R(\gamma) - R^\ast(\gamma). \]
					For $R(\gamma) - \tilde R(\gamma)$, we have
					\begin{align*}
						R(\gamma) - \tilde R(\gamma) &= \frac{1}{N^2} \gamma^\T \Lp \Lp  (W \odot e) (e^\T \odot W^\T)  \Rp \odot\Big[ \frac{1}{|\tlq_{ij}|} \Big] \Rp \gamma \\
						& \quad + \frac{2}{N^2} \gamma^\T \Lp \Lp (  (W \odot e) (F \Lambda^\T) \odot W^\T) \Rp \odot \Big[ \frac{1}{|\tlq_{ij}|} \Big] \Rp  \gamma    
					\end{align*}
					and 
					\begin{eqnarray*}
						\sup_{\gamma \in \Gamma} |R(\gamma) - \tilde R(\gamma)| &\leq& \sup_{\gamma \in \Gamma}  \frac{1}{N^2} | \gamma^\T \Lp \Lp  (W \odot e) (e^\T \odot W^\T) \Rp \odot \Big[ \frac{1}{|\tlq_{ij}|} \Big] \Rp \gamma|   \\
						&&+ \sup_{\gamma \in \Gamma} \frac{2}{N^2} |\gamma^\T \Lp \Lp  (W \odot e)( (F \Lambda^\T) \odot W^\T) \Rp \odot \Big[ \frac{1}{|\tlq_{ij}|} \Big] \Rp  \gamma| \rightarrow 0. 
					\end{eqnarray*}
					For $\tilde  R(\gamma) - R^\ast(\gamma) $, we have for any $\gamma \in \Gamma$
					\begin{eqnarray*}
						\tilde  R(\gamma) - R^\ast(\gamma)  &=& \frac{1}{N^2} \sum_{i = 1}^N \sum_{j = 1}^N \gamma_i \gamma_j  \Lambda_i^\T \underbrace{\Lp \frac{1}{|\tlq_{ij}|} \sum_{t \in \tlq_{ij}} F_t F_t^\T - \frac{1}{T} \sum_{t = 1}^T F_t F_t^\T  \Rp }_{\Delta_{F, ij}}   \Lambda_j \\ 
						&\leq& \underbrace{\Lp\frac{1}{N^2} \sum_{i = 1}^N \sum_{j=1}^N \gamma_i^2 \gamma_j^2  \Rp^{1/2}}_{=1 \text{ from } \gamma^\T \gamma/N = 1 }  \Lp\frac{1}{N^2} \sum_{i = 1}^N \sum_{j=1}^N  \Lp \Lambda_i^\T \Delta_{F, ij} \Lambda_j \Rp^2 \Rp^{1/2} \\
						&=& \Lp\frac{1}{N^2} \sum_{i = 1}^N \sum_{j=1}^N  \Lp \Lambda_i^\T \Delta_{F, ij} \Lambda_j \Rp^2 \Rp^{1/2}. 
					\end{eqnarray*}
					
					Note that 
					\begin{align*}
						& \+E \left[\frac{1}{N^2} \sum_{i = 1}^N \sum_{j=1}^N  \Lp \Lambda_i^\T \Delta_{F, ij} \Lambda_j \Rp^2 \right] \leq \frac{1}{N^2} \sum_{i = 1}^N \sum_{j=1}^N  \+E \Lp \Lambda_i^\T \Delta_{F, ij} \Lambda_j \Rp^2 \\ 
						\leq& \frac{1}{N^2} \sum_{i = 1}^N \sum_{j=1}^N  \underbrace{\+E [ \norm{\Lambda_{i}}^2 \norm{\Lambda_{j}}^2 ] }_{\leq M \text{ from Assumption \ref{ass:factor-model}.\ref{ass:loading}}} \+E \norm{\Delta_{F, ij}}^2 \\
						\leq &  \frac{M}{N^2 T} \sum_{i = 1}^N \sum_{j=1}^N \Bigg[ \underbrace{\+E  \norm{\frac{\sqrt{T}}{|\tlq_{ij}|} \sum_{t \in \tlq_{ij}} (F_t F_t^\T - \Sigma_{F})   }^2}_{\leq M \text{ from Assumptions  \ref{ass:factor-model}.\ref{ass:factor} and \ref{ass:obs-equal-weight}.\ref{ass:add-obs}}} + \underbrace{ \+E \norm{\frac{1}{\sqrt{T}} \sum_{t = 1}^T (F_t F_t^\T - \Sigma_{F})  }^2}_{\leq M \text{ from Assumptions  \ref{ass:factor-model}.\ref{ass:factor}} } \Bigg]   = O\left(\frac{1}{T} \right),
					\end{align*}
					where the second inequality follows from the independence between factors and loadings. Then by applying the Markov inequality, we have 
					\[\sup_{\gamma \in \Gamma} |\tilde  R(\gamma) - R^\ast(\gamma)| \xrightarrow{p} 0  \]
					and 
					\[\sup_{\gamma \in \Gamma} | R(\gamma) - R^\ast(\gamma)| \xrightarrow{p} 0. \]

				\end{proof}
				\begin{itemize}
					\item [(R7)] $|\sup_{\gamma \in \Gamma} R(\gamma) - \sup_{\gamma \in \Gamma}\tilde R(\gamma)| \xrightarrow{p} 0  $ and $|\sup_{\gamma \in \Gamma} R(\gamma) - \sup_{\gamma \in \Gamma} R^\ast(\gamma)| \xrightarrow{p} 0  $
					\item [(R8)] $\sup_{\gamma \in \Gamma} R^\ast(\gamma) \xrightarrow{p} d_1$, where $d_1$ is the largest eigenvalue of $\Sigma_{F} \Sigma_{\Lambda}$
					\item [(R9)] $\sup_{\gamma \in \Gamma} R(\gamma) \xrightarrow{p} d_1$
					\item [(R10)] Let $\gamma^\ast = \arg\,\sup_{\gamma \in \Gamma} R(\gamma)$. We have $\tilde R (\gamma^\ast) \xrightarrow{p} d_1$ and  $R^\ast (\gamma^\ast) \xrightarrow{p} d_1$
					
					\item [(R11)] Let $\underline{\tilde \Lambda}_1$ denote the first column of $\tilde \Lambda$ and let $d_{\underline{\Lambda}_1} = \text{sign}(\underline{\tilde \Lambda}_1, \underline{\Lambda}_1)$, meaning $d_{\underline{\Lambda}_1} = 1$ if $\underline{\tilde \Lambda}^\T_1 \underline{\Lambda}_1 \geq 0$ and $d_{\underline{\Lambda}_1} = -1$ if $\underline{\tilde \Lambda}^\T_1 \underline{\Lambda}_1 < 0$. Then $d_{\underline{\Lambda}_1} \underline{\tilde \Lambda}^\T_1 \underline{\Lambda}_1 (\Lambda^\T \Lambda/N)^{-1/2} \xrightarrow{p} l_1^\T$, where $l_1 = (1, 0, \cdots, 0)^\T$. 
					\item [(R12)] Suppose that the $N \times r$ matrix $\tilde \Lambda$ is formed as the $r$ ordered eigenvectors of $ (Y \odot W)(Y^\T \odot W^\T) \odot \left[\frac{1}{|\tlq_{ij}|} \right]$ normalized as $\tilde \Lambda^\T \tilde \Lambda/N = I_r$. Let  $D_{\Lambda} = \text{diag}(\text{sign}(\tilde \Lambda^\T \Lambda))$. Then  $D_{\Lambda} \tilde \Lambda^\T \Lambda (\Lambda^\T \Lambda/N)^{-1/2} \xrightarrow{p} I_r$.  Let $\underline{\tilde \Lambda}_j$ be the $j$-th column in $\tilde{\Lambda}$. By definition, $\underline{\tilde \Lambda}_1 = \sqrt{N} \arg\,\sup_{\gamma \in \Gamma} R(\gamma)$.
					\item [(R13)] For $j = 1, 2, \cdots, r$, $R(\underline{\tilde \Lambda}_j) \xrightarrow{p} d_j$, $\tilde R(\underline{\tilde \Lambda}_j) \xrightarrow{p} d_j$ and $R^\ast(\underline{\tilde \Lambda}_j) \xrightarrow{p} d_j$. 
				\end{itemize}
				\begin{proof}[Proof of (R13)]
					The result for $R(\underline{\tilde \Lambda}_1) \xrightarrow{p} d_1$, $\tilde R(\underline{\tilde \Lambda}_1) \xrightarrow{p} d_1$ and $R^\ast(\underline{\tilde \Lambda}_1) \xrightarrow{p} d_1$ is given in (R9) and (R10). The results for the other columns mimic the steps in (R8)-(R10), for the other principal components, that is, by maximizing $R(\cdot)$ and $R^\ast(\cdot)$ sequentially using orthonormal subspaces of $\Gamma$.
				\end{proof}
				Note that Lemma \ref{lemma:eigenvalue-converge}.1 is a consequence of
				\[ \frac{1}{NT^2} \tilde \Lambda^\T \Lp ( \tilde Y \tilde Y^\T) \odot \Big[ \frac{1}{|\tlq_{ij}|} \Big] \Rp \tilde \Lambda = \text{diag}(R(\underline{\tilde \Lambda}_1), \cdots, R(\underline{\tilde \Lambda}_r)) \xrightarrow{p} \text{diag}(d_1, \cdots d_r)\]
				following from (R13); Lemma \ref{lemma:eigenvalue-converge}.2 follows from
				\begin{align*}
					& \frac{1}{NT^2}  \tilde \Lambda^\T \Lp \Lp \Lp W \odot (\Lambda F^\T) \Rp   \Lp(F \Lambda^\T) \odot W^\T \Rp \Rp \odot\Big[ \frac{1}{|\tlq_{ij}|} \Big] \Rp \tilde \Lambda \\
					= \quad& \text{diag}(\tilde R(\underline{\tilde \Lambda}_1), \cdots, \tilde R(\underline{\tilde \Lambda}_r)) \xrightarrow{p} \text{diag}(d_1, \cdots d_r)    
				\end{align*}
				based on (R13); Lemma \ref{lemma:eigenvalue-converge}.3 holds because of
				\[\frac{1}{NT^2}  \tilde \Lambda^\T \Lp \Lambda F^\T F \Lambda^\T  \Rp \tilde \Lambda = \text{diag}(R^\ast (\underline{\tilde \Lambda}_1), \cdots, R^\ast (\underline{\tilde \Lambda}_r))   \xrightarrow{p} \text{diag}(d_1, \cdots d_r),\]
				which follows from (R13). 
			\end{proof}

			\begin{lemma}\label{lemma:def-q}
				Under Assumptions \ref{ass:obs} and \ref{ass:factor-model}, it holds that
				\begin{enumerate}
					\item $\frac{1}{N} \tilde \Lambda^\T \Lambda \xrightarrow{p} Q $, where $Q$ is invertible, $Q = D^{1/2} \Upsilon \Sigma_F^{-1/2}$, the diagonal entries of \\ $D = \diag(d_1, d_2, \cdots, d_r)$ are the eigenvalues of $\Sigma_F^{1/2} \Sigma_\Lambda \Sigma_F^{1/2}$, and $\Upsilon$ is the corresponding eigenvector matrix such that $\Upsilon^\T \Upsilon = I$.
					\item $H^\I \xrightarrow{p} Q^\T$, where $H = \frac{1}{NT} \tilde D^\I \tilde \Lambda^\T \Lambda F^\T F$.
				\end{enumerate}
			\end{lemma}
			
			\begin{proof}[Proof of Lemma \ref{lemma:def-q}]
				\texttt{}
				\begin{enumerate}
					\item The proof is similar to the proof of Proposition 1 in \cite{bai2003inferential}. Left multiplying $\tilde \Sigma \tilde \Lambda  = \tilde \Lambda \tilde D$ by $\frac{1}{N} \Lp \frac{F^\T F}{T} \Rp^{1/2} \Lambda$, we obtain 
					\[ \frac{1}{N} \Lp \frac{F^\T F}{T} \Rp^{1/2} \Lambda^\T \tilde \Sigma \tilde \Lambda =    \Lp \frac{F^\T F}{T} \Rp^{1/2} \frac{\Lambda^\T \tilde \Lambda }{N}  \tilde D \]
					and then 
					\[ \Lp \frac{F^\T F}{T} \Rp^{1/2} \frac{\Lambda^\T \Lambda}{N}\Lp \frac{F^\T F}{T} \Rp \frac{\Lambda^\T  \tilde \Lambda}{N}  + d_{NT} =    \Lp \frac{F^\T F}{T} \Rp^{1/2} \frac{\Lambda^\T \tilde \Lambda }{N}  \tilde D  ,\]
					where $d_{NT} =  \frac{1}{N} \Lp \frac{F^\T F}{T} \Rp^{1/2} \Lambda^\T \tilde d_{NT} \tilde \Lambda$ and $\tilde d_{NT}$ has 
					\begin{eqnarray*}
						\tilde d_{NT,ij} &=& \Lambda_i^\T \Lp \frac{1}{|\tlq_{ij}|}  F^\T\text{diag}(W_i \odot W_j) F  -  \frac{1}{T} F^\T F \Rp  \Lambda_j +  \frac{1}{|\tlq_{ij}|} e_i^\T\text{diag}(W_i \odot W_j) F \Lambda_j \\
						&& + \frac{1}{|\tlq_{ij}|} \Lambda_i^\T F^\T\text{diag}(W_i \odot W_j) e_j +  \frac{1}{|\tlq_{ij}|} e_i^\T\text{diag}(W_i \odot W_j) e_j .
					\end{eqnarray*}
					From Assumption \ref{ass:factor-model}.\ref{ass:factor}, $ \frac{1}{|\tlq_{ij}|}  F^\T\text{diag}(W_i \odot W_j) F  -  \frac{1}{T} F^\T F = O_P \Lp \frac{1}{\sqrt{T}} \Rp$ and then it holds that 
					\[\frac{1}{N} \Lambda^\T \tilde d_{NT} = O_P \Lp \frac{1}{\dnt} \Rp \]
					following from Lemma \ref{lemma:prep-asy-F}. The remaining steps to show $\frac{1}{N} \tilde \Lambda^\T \Lambda \xrightarrow{p} Q $ are exactly the same as those in Proposition 1 in \cite{bai2003inferential}. 
					\item Note that 
					\[H =  \frac{1}{NT} \tilde D^\I \tilde \Lambda^\T \Lambda F^\T F \xrightarrow{p} D^\I Q \Sigma_F = D^\I D^{1/2} \Upsilon \Sigma_F^{-1/2} \Sigma_F =  D^{-1/2} \Upsilon \Sigma_F^{1/2} = (Q^\T)^\I. \]
				\end{enumerate}
			\end{proof}
			
			\subsubsection{Proof of Theorem \ref{theorem:asy-normal-equal-weight}.1}
			
			\begin{lemma}\label{lemma:prep-asy-F}
				Suppose Assumptions \ref{ass:obs}, \ref{ass:factor-model} and \ref{ass:mom-clt} hold. Conditional on $S$, we have
				\begin{enumerate}
					\item $\frac{1}{N} \sum_{i=1}^N \tilde \Lambda_i \gamma(i,j)  = O_P \Lp \frac{1}{\sqrt{N\delta_{NT}}} \Rp$
					\item $\frac{1}{N} \sum_{i=1}^N \tilde \Lambda_i \zeta_{ij} = O_P \Lp \frac{1}{\sqrt{T\delta_{NT}}} \Rp$
					\item $\frac{1}{N} \sum_{i=1}^N \tilde \Lambda_i \eta_{ij} = O_P \Lp \frac{1}{\sqrt{T}}\Rp$
					\item $\frac{1}{N} \sum_{i=1}^N \tilde \Lambda_i \xi_{ij} = O_P \Lp \frac{1}{\sqrt{T\delta_{NT}}} \Rp$.
				\end{enumerate}
			\end{lemma}
			
			\begin{proof}
				\begin{enumerate}
					\item \textit{Show that $\frac{1}{N} \sum_{i=1}^N \tilde \Lambda_i \gamma(i,j)  = O_P \Lp \frac{1}{\sqrt{N\delta_{NT}}} \Rp$}.\\
					First, we decompose $\frac{1}{N} \sum_{i=1}^N \tilde \Lambda_i \gamma(i,j) $ into
					\begin{eqnarray*}
						\frac{1}{N} \sum_{i=1}^N \tilde \Lambda_i \gamma(i,j) = \frac{1}{N} \sum_{i=1}^N (\tilde \Lambda_i - H \Lambda_i)\gamma(i,j)  + \frac{1}{N} \sum_{i=1}^N H \Lambda_i \gamma(i,j).
					\end{eqnarray*}
					For the second term $\frac{1}{N} \sum_{i=1}^N H \Lambda_i \gamma(i,j)$, we have 
					\[\+E \Ls \norm{\sum_{i=1}^N \Lambda_i \gamma(i,j)}\Rs \leq  \sum_{i=1}^N  \underbrace{\+E[\norm{\Lambda_i} ] }_{\leq \bar \Lambda} \cdot |\gamma(i,j)| \leq \bar \Lambda \sum_{i=1}^N  |\gamma(i,j)|  \leq \bar \Lambda \cdot M \]
					following from the independence of $\Lambda$ and $e$ and  Assumptions \ref{ass:factor-model}.\ref{ass:loading} and \ref{ass:factor-model}.\ref{ass:error}.(c). 
					Together with $H = O_P(1)$, we have $\frac{1}{N}  \sum_{i=1}^N H \Lambda_i \gamma(i,j) =  O_P \Lp \frac{1}{N} \Rp$. Next we consider the first term $\frac{1}{N} \sum_{i=1}^N (\tilde \Lambda_i - H \Lambda_i)\gamma(i,j) $. We conclude that 
					\begin{align*}
						\norm{\frac{1}{N} \sum_{i=1}^N (\tilde \Lambda_i - H \Lambda_i) \gamma(i,j)} &\leq \underbrace{\Lp \frac{1}{N}\sum_{i=1}^N \norm{\tilde \Lambda_i - H \Lambda_i}^2 \Rp^{1/2} }_{O_P \Lp \frac{1}{\dnt} \Rp \text{ from Theorem \ref{thm:consistency-same-H}}}  \frac{1}{\sqrt{N}} \underbrace{\Lp \sum_{i=1}^N \gamma(i,j)^2 \Rp^{1/2}}_{\leq M \text{ from Lemma \ref{lemma:prep-consistency}}} \\  &= O_P \Lp \frac{1}{\sqrt{N\delta_{NT}}} \Rp.
					\end{align*} 
					Hence, 
					\begin{eqnarray*}
						\frac{1}{N} \sum_{i=1}^N \tilde \Lambda_i \gamma(i,j) = \underbrace{\frac{1}{N} \sum_{i=1}^N (\tilde \Lambda_i - H \Lambda_i)\gamma(i,j)}_{ O_P \Lp \frac{1}{\sqrt{N\delta_{NT}}} \Rp}   + \underbrace{\frac{1}{N} \sum_{i=1}^N H \Lambda_i \gamma(i,j) }_{O_P \Lp \frac{1}{N} \Rp}  = O_P \Lp \frac{1}{\sqrt{N\delta_{NT}}} \Rp.
					\end{eqnarray*}
					\item \textit{Show that $\frac{1}{N} \sum_{i=1}^N \tilde \Lambda_i \zeta_{ij} = O_P \Lp \frac{1}{\sqrt{T\delta_{NT}}} \Rp$}. \\ Let us decompose $\frac{1}{N} \sum_{i=1}^N \tilde \Lambda_i \zeta_{ij} $.
					\begin{eqnarray*}
						\frac{1}{N} \sum_{i=1}^N \tilde \Lambda_i \zeta_{ij}  = \frac{1}{N} \sum_{i=1}^N (\tilde \Lambda_i - H \Lambda_i)\zeta_{ij}   + \frac{1}{N} \sum_{i=1}^N H \Lambda_i \zeta_{ij} 
					\end{eqnarray*}
					First, we consider $\frac{1}{N} \sum_{i=1}^N \zeta_{ij}^2$. 
					\begin{eqnarray*}
						\+E \Ls \frac{1}{N} \sum_{i=1}^N \zeta_{ij}^2 \Rs &\leq&  \frac{1}{N} \sum_{i=1}^N \frac{1}{|\tlq_{ij}|} \underbrace{\+E\Ls \frac{1}{\sqrt{|\tlq_{ij}| }  } \sum_{t \in \tlq_{ij}} \Lp e_{it}e_{jt} - \+E[e_{it} e_{jt}] \Rp  \Rs^2 }_{\leq M \text{ from Assumption \ref{ass:factor-model}.\ref{ass:error}.(e)} } = \max_i \frac{1}{|\tlq_{ij}|} \cdot M =  O\Lp \frac{1}{T } \Rp.
					\end{eqnarray*}
					Hence, $\frac{1}{N} \sum_{i=1}^N \zeta_{ij}^2 = O_P \Lp \frac{1}{T } \Rp$. For the first term $\frac{1}{N} \sum_{i=1}^N (\tilde \Lambda_i - H \Lambda_i)\zeta_{ij}$, we obtain 
					\[\norm{\frac{1}{N}  \sum_{i=1}^N (\tilde \Lambda_i - H \Lambda_i)\zeta_{ij} } \leq \underbrace{\Lp \frac{1}{N}\sum_{i=1}^N \norm{\tilde \Lambda_i - H \Lambda_i}^2 \Rp^{1/2}}_{O_P \Lp \frac{1}{\dnt} \Rp \text{ from Theorem \ref{thm:consistency-same-H}} } \underbrace{ \Lp \frac{1}{N} \sum_{i=1}^N \zeta_{ij}^2 \Rp^{1/2} }_{O_P \Lp \frac{1}{\sqrt{T} } \Rp} = O_P \Lp \frac{1}{\sqrt{T\delta_{NT}}} \Rp. \]
					For the second term $\frac{1}{N} \sum_{i=1}^N H \Lambda_i \zeta_{ij} $, let us consider $\frac{1}{N} \sum_{i=1}^N \Lambda_i \zeta_{ij}$. 
					\begin{align*}
						\+E \Ls \norm{\frac{1}{N} \sum_{i=1}^N \Lambda_i \zeta_{ij}}^2 \Rs  =&   \+E \Ls \norm{ \frac{1}{N} \sum_{i=1}^N \Lambda_i \frac{1}{|\tlq_{ij}|} \sum_{t \in \tlq_{ij}} \Lp e_{it}e_{jt} - \+E[e_{it}e_{jt} ] \Rp }^2  \Rs 
						= O\Lp \frac{1}{NT} \Rp,
					\end{align*}
					which follows from Assumption \ref{ass:mom-clt}.1.
					Hence  $\frac{1}{N} \sum_{i=1}^N H \Lambda_i \zeta_{ij} =  O_P \Lp \frac{1}{\sqrt{NT}} \Rp$ and 
					\begin{eqnarray*}
						\frac{1}{N} \sum_{i=1}^N \tilde \Lambda_i \zeta_{ij}  = \underbrace{\frac{1}{N} \sum_{i=1}^N (\tilde \Lambda_i - H \Lambda_i)\zeta_{ij} }_{O_P \Lp \frac{1}{\sqrt{T\delta_{NT}}} \Rp}   + \underbrace{\frac{1}{N} \sum_{i=1}^N H \Lambda_i \zeta_{ij} }_{O_P \Lp \frac{1}{\sqrt{NT}} \Rp} = O_P \Lp \frac{1}{\sqrt{T\delta_{NT}}}\Rp. 
					\end{eqnarray*}
					\item \textit{Show that} $\frac{1}{N} \sum_{i=1}^N \tilde \Lambda_i \eta_{ij} = O_P \Lp \frac{1}{\sqrt{T}}\Rp$. \\ We decompose $\frac{1}{N} \sum_{i=1}^N \tilde \Lambda_i \eta_{ij} $ into two terms:
					\begin{eqnarray*}
						\frac{1}{N} \sum_{i=1}^N \tilde \Lambda_i \eta_{ij}  = \frac{1}{N} \sum_{i=1}^N (\tilde \Lambda_i - H \Lambda_i)\eta_{ij}  + \frac{1}{N} \sum_{i=1}^N H \Lambda_i \eta_{ij}
					\end{eqnarray*}
					Let us first consider  $\frac{1}{N} \sum_{i=1}^N \Lambda_i \eta_{ij}$ in the second term: 
					\begin{align*}
						\norm{\frac{1}{N} \sum_{i=1}^N \Lambda_i \eta_{ij}}^2  &= \norm{ \frac{1}{N} \sum_{i=1}^N \Lambda_i \Lambda_i^\T \frac{1}{|\tlq_{ij}|} \sum_{t \in \tlq_{ij}} F_t e_{jt} }^2 \\
						&\leq \underbrace{\Lp \frac{1}{N} \sum_{i=1}^N  \norm{\Lambda_i}^4 \Rp}_{O_P(1)} \cdot \underbrace{\Lp   \frac{1}{N} \sum_{i=1}^N  \norm{\frac{1}{|\tlq_{ij}|} \sum_{t \in \tlq_{ij}} F_t e_{jt} }^2\Rp}_{O_P \Lp \frac{1}{T} \Rp }   = O_P \Lp \frac{1}{T} \Rp,
					\end{align*}
					where $\frac{1}{N} \sum_{i=1}^N  \norm{\Lambda_i}^4 = O_P(1)$ follows from
					\[\+E \Ls \frac{1}{N} \sum_{i=1}^N  \norm{\Lambda_i}^4 \Rs = \frac{1}{N} \sum_{i=1}^N \+E \Ls  \norm{\Lambda_i}^4 \Rs \leq M\]
					and $ \frac{1}{N} \sum_{i=1}^N  \norm{\frac{1}{|\tlq_{ij}|} \sum_{t \in \tlq_{ij}} F_t e_{jt} }^2 = O_P \Lp \frac{1}{T} \Rp$ holds because of 
					\[\+E \Ls \frac{1}{N} \sum_{i=1}^N  \norm{\frac{1}{|\tlq_{ij}|} \sum_{t \in \tlq_{ij}} F_t e_{jt} }^2  \Rs =  \frac{1}{N} \sum_{i=1}^N \frac{1}{|\tlq_{ij}|} \underbrace{\+E \Ls  \norm{\frac{1}{\sqrt{|\tlq_{ij}| } } \sum_{t \in \tlq_{ij}} F_t e_{jt} }^2 \Rs}_{\leq M} =  O_P \Lp \frac{1}{T} \Rp.  \]
					Since $H =  O_P(1)$, the second term satisfies $\frac{1}{N} \sum_{i=1}^N H \Lambda_i \eta_{ij} = O_P \Lp \frac{1}{\sqrt{T}} \Rp$. 
					Next, we consider the first term 
					\begin{align*}
						\norm{\frac{1}{N}  \sum_{i=1}^N (\tilde \Lambda_i  - H_i \Lambda_i)\eta_{ij} }^2 \leq \underbrace{\Lp \frac{1}{N}\sum_{i=1}^N \norm{\tilde \Lambda_i - H_i \Lambda_i}^2 \Rp}_{O_P \Lp \frac{1}{ \delta } \Rp} \underbrace{\Lp \frac{1}{N} \sum_{i=1}^N \eta_{ij}^2 \Rp}_{ O_P \Lp \frac{1}{T } \Rp }  = O_P \Lp \frac{1}{\sqrt{T\delta_{NT}} } \Rp,
					\end{align*}
					where $\frac{1}{N} \sum_{i=1}^N \eta_{ij}^2 =  O_P \Lp \frac{1}{T } \Rp$ follows from 
					\begin{align*}
						\+E \Ls \frac{1}{N} \sum_{i=1}^N \eta_{ij}^2 \Rs =  \frac{1}{N} \sum_{i=1}^N \frac{1}{|\tlq_{ij}|} \underbrace{\+E\Ls \Lp\frac{1}{\sqrt{|\tlq_{ij}|}} \Lambda_i^\T \sum_{t \in \tlq_{ij}} F_t e_{jt} \Rp^2 \Rs }_{\leq M \text{ from Lemma \ref{lemma:prep-consistency}.2 }}  = O\Lp \frac{1}{T} \Rp.
					\end{align*}
					Putting this together, we conclude
					\begin{eqnarray*}
						\frac{1}{N} \sum_{i=1}^N \tilde \Lambda_i \eta_{ij}  = \underbrace{\frac{1}{N} \sum_{i=1}^N (\tilde \Lambda_i - H \Lambda_i)\eta_{ij} }_{O_P \Lp \frac{1}{\sqrt{T\delta_{NT}} } \Rp }  + \underbrace{\frac{1}{N} \sum_{i=1}^N H \Lambda_i \eta_{ij} }_{O_P \Lp \frac{1}{\sqrt{T}} \Rp } = O_P \Lp \frac{1}{\sqrt{T}} \Rp.
					\end{eqnarray*}
					\item \textit{Show that} $\frac{1}{N} \sum_{i=1}^N \tilde \Lambda_i \xi_{ij} = O_P \Lp \frac{1}{\sqrt{T\delta_{NT}}} \Rp$. \\ We start by decomposing  $\frac{1}{N} \sum_{i=1}^N \tilde \Lambda_i \xi_{ij} $:
					\begin{align*}
						\frac{1}{N} \sum_{i=1}^N \tilde \Lambda_i \xi_{ij}  
						&= \frac{1}{NT} \sum_{i=1}^N \tilde \Lambda_i e_i^\T\text{diag}(W_i \odot W_j)  F \Lambda_j/\tilde q_{ij} \\
						&= \frac{1}{NT} \sum_{i=1}^N (\tilde \Lambda_i - H \Lambda_i)e_i^\T\text{diag}(W_i \odot W_j)  F \Lambda_j /\tilde q_{ij}  \\
						& \quad + \frac{1}{NT} \sum_{i=1}^N H \Lambda_i e_i^\T\text{diag}(W_i \odot W_j) F \Lambda_j/\tilde q_{ij}.
					\end{align*}
					Let us first consider the first term $\frac{1}{NT} \sum_{i=1}^N (\tilde \Lambda_i - H \Lambda_i)e_i^\T\text{diag}(W_i \odot W_j)  F \Lambda_j /\tilde q_{ij}$:  
					\begin{eqnarray*}
						&& \norm{ \frac{1}{NT} \sum_{i=1}^N (\tilde  \Lambda_i - H_i  \Lambda_i )e_i^\T\text{diag}(W_i \odot W_j)  F \Lambda_j/\tilde q_{ij} } \\ &\leq&  \underbrace{\Lp\max_i \frac{1}{\sqrt{|\tlq_{ij}|} } \Rp }_{O_P \Lp \frac{1}{\sqrt{T} }\Rp } \underbrace{ \Lp \frac{1}{N} \sum_{i=1}^N \norm{(\tilde  \Lambda_i  - H  \Lambda_i)}^2 \Rp^{1/2}}_{ O_P \Lp \frac{1}{ \dnt } \Rp  }  \underbrace{\Lp\frac{1}{N} \sum_{i=1}^N \norm{\frac{1}{\sqrt{|\tlq_{ij}|}} \sum_{t \in \tlq_{ij}} F_t e_{it} }^2  \Rp^{1/2}}_{O_P(1)} \underbrace{\norm{\Lambda_j}}_{O_P(1)}   \\
						&=& O_P \Lp \frac{1}{\sqrt{T\delta_{NT}} } \Rp, 
					\end{eqnarray*}
					where $\frac{1}{N} \sum_{i=1}^N \norm{\frac{1}{\sqrt{|\tlq_{ij}|}} \sum_{t \in \tlq_{ij}} F_t e_{it} }^2 = O_P(1)$ follows from 
					\[\+E \Ls \frac{1}{N} \sum_{i=1}^N \norm{\frac{1}{\sqrt{|\tlq_{ij}|}} \sum_{t \in \tlq_{ij}} F_t e_{it} }^2  \Rs =  \frac{1}{N} \sum_{i=1}^N  \underbrace{\+E \Ls  \norm{\frac{1}{\sqrt{|\tlq_{ij}| } } \sum_{t \in \tlq_{ij}} F_t e_{it} }^2 \Rs}_{\leq M} =  O_P \Lp 1\Rp,  \]
					where the first equality holds since $S$ is independent of $F$ and $e$. Let us first consider the second term $\frac{1}{NT} \sum_{i=1}^N H \Lambda_i e_i^\T\text{diag}(W_i \odot W_j) F \Lambda_j/\tilde q_{ij}$:  
					\begin{eqnarray*}
						&& \norm{\frac{1}{NT} \sum_{i=1}^N H  \Lambda_i  e_i^\T\text{diag}(W_i \odot W_j)  F \Lambda_j/\tilde q_{ij}} \\ 
						&\leq& \norm{H} \Lp \max_i \frac{1 }{ \sqrt{N|\tlq_{ij}|}} \Rp \underbrace{ \norm{\frac{1}{\sqrt{N}} \sum_{i=1}^N\frac{1}{\sqrt{|\tlq_{ij}|}} \sum_{t \in \tlq_{ij}} \Lambda_i F_t^\T e_{it} }  }_{O_P(1) \text{ from Assumption \ref{ass:mom-clt}.2}} \norm{\Lambda_j} = O_P \Lp \frac{1}{\sqrt{NT} } \Rp.
					\end{eqnarray*}
					Hence, it holds that
					\begin{align*}
						\frac{1}{N} \sum_{i=1}^N \tilde \Lambda_i \xi_{ij}  
						&= \underbrace{\frac{1}{NT} \sum_{i=1}^N (\tilde \Lambda_i - H \Lambda_i)e_i^\T\text{diag}(W_i \odot W_j)  F \Lambda_j /\tilde q_{ij}   }_{ O_P \Lp \frac{1}{\sqrt{T\delta_{NT}} } \Rp  } \\
						& \quad + \underbrace{\frac{1}{NT} \sum_{i=1}^N H \Lambda_i e_i^\T\text{diag}(W_i \odot W_j) F \Lambda_j/\tilde q_{ij} }_{ O_P \Lp \frac{1}{\sqrt{NT} } \Rp}  =  O_P \Lp \frac{1}{\sqrt{T\delta_{NT}} } \Rp. 
					\end{align*}
					
				\end{enumerate}
			\end{proof}

			\begin{lemma}\label{lemma:HDinvHTinv}
				Under Assumptions \ref{ass:obs} and \ref{ass:factor-model}, it holds that
				\begin{align}
					H^\T \tilde D^\I (H^\T)^{-1}  =  \Big(\frac{\Lambda^\T \Lambda}{N} \Big)^\I \Big(\frac{F^\T F}{T} \Big)^\I + O_P \Big( \frac{1}{\dnt} \Big). \label{eqn:HTH-equi-form}
				\end{align}
			\end{lemma}
			
			\begin{proof}[Proof of Lemma \ref{lemma:HDinvHTinv}]
				From the definition of $H = \tilde D^\I \frac{\tilde \Lambda^\T \Lambda}{N} \frac{F^\T F}{T} $ and $H = (Q^\I)^\T + O_P \Big( \frac{1}{\dnt} \Big) = \Big( \frac{\Lambda^\T \tilde \Lambda}{N} \Big)^\I + O_P \Big( \frac{1}{\dnt} \Big) $, we derive that 
				\begin{align} 
					\nonumber H^\T \tilde D^\I (H^\T)^{-1}  =& H^\T \tilde D^\I \tilde D \Big(\frac{\Lambda^\T \tilde \Lambda}{N} \Big)^\I \Big(\frac{F^\T F}{T} \Big)^\I = H^\T  \Big(\frac{\Lambda^\T \tilde \Lambda}{N} \Big)^\I \Big(\frac{F^\T F}{T} \Big)^\I \\
					\nonumber =& H^\T H \Big(\frac{F^\T F}{T} \Big)^\I + O_P \Big( \frac{1}{\dnt} \Big)  \\
					=&  \Big(\frac{\Lambda^\T \Lambda}{N} \Big)^\I \Big(\frac{F^\T F}{T} \Big)^\I + O_P \Big( \frac{1}{\dnt} \Big), 
				\end{align}
				where the last equality follows from $H^\T H =  \Big(\frac{\Lambda^\T \Lambda}{N} \Big)^\I + O_P \Big( \frac{1}{\dnt} \Big)$. This last statement is a consequence of
				\begin{align*}
					H^\T H  =&   \Big(\frac{\tilde \Lambda^\T  \Lambda}{N} \Big)^\I \Big(\frac{\Lambda^\T \tilde \Lambda}{N} \Big)^\I + O_P \Big( \frac{1}{\dnt} \Big)  \\
					=& \Big(\frac{\Lambda^\T  \Lambda}{N} \Big)^\I \Big(H^\T H  \Big)^\I  \Big(\frac{\Lambda^\T \Lambda}{N} \Big)^\I  + O_P \Big( \frac{1}{\dnt} \Big) .
				\end{align*}
				We multiply both side by $\frac{\Lambda^\T  \Lambda}{N} $
				\begin{align*}
					\Big(\frac{\Lambda^\T  \Lambda}{N} \Big)   H^\T H  =&  \Bigg( \Big(\frac{\Lambda^\T  \Lambda}{N} \Big)   H^\T H \Bigg)^\I  + O_P \Big( \frac{1}{\dnt} \Big)  \\    
				\end{align*}
				and therefore obtain
				\begin{align*}
					\Big(\frac{\Lambda^\T  \Lambda}{N} \Big)   H^\T H =& I_r + O_P \Big( \frac{1}{\dnt} \Big) \\
					\Rightarrow H^\T H =&  \Big(\frac{\Lambda^\T \Lambda}{N} \Big)^\I + O_P \Big( \frac{1}{\dnt} \Big) .
				\end{align*}
			\end{proof}

			\begin{proof}[Proof of Theorem \ref{theorem:asy-normal}.1]
				We have the following decomposition for $ \tilde \Lambda_j - H \Lambda_j$:
				\[\sqrt{T}( \tilde \Lambda_j - H \Lambda_j ) = \sqrt{T}( \tilde \Lambda_j - H_j \Lambda_j ) + \sqrt{T}( H_j - H  ) \Lambda_j.\]
				First, we consider $\sqrt{T}( \tilde \Lambda_j - H_j \Lambda_j )$.
				By Lemma \ref{lemma:prep-asy-F}, the decomposition of $\tilde \Lambda_j - H_j \Lambda_j$ is
				\[\tilde \Lambda_j - H_j \Lambda_j =\tilde D^\I \Big( \underbrace{ \frac{1}{N} \sum_{i=1}^N \tilde \Lambda_i \gamma(i,j)}_{O_P \Lp \frac{1}{\sqrt{N\delta_{NT}} } \Rp}  + \underbrace{\frac{1}{N} \sum_{i=1}^N \tilde \Lambda_i \zeta_{ij}}_{O_P \Lp \frac{1}{\sqrt{T\delta_{NT}} } \Rp}  + \underbrace{\frac{1}{N} \sum_{i=1}^N \tilde \Lambda_i \eta_{ij} }_{O_P \Lp \frac{1}{\sqrt{T}  } \Rp} + \underbrace{\frac{1}{N} \sum_{i=1}^N \tilde \Lambda_i \xi_{ij} }_{O_P \Lp \frac{1}{\sqrt{T\delta_{NT}} } \Rp} \Big).\]
				When $\sqrt{T}/N \rightarrow 0$, the limiting distribution is determined by $\frac{1}{N} \sum_{i=1}^N \tilde \Lambda_i \eta_{ij}$, i.e., 
				\begin{eqnarray*}
					\sqrt{T}( \tilde \Lambda_j - H_j \Lambda_j ) &=& \tilde D^{-1}  \frac{1}{N} \sum_{i=1}^N \sqrt{\frac{T}{|\tlq_{ij}|}} H \Lambda_i \Lambda_i^\T  \frac{1}{\sqrt{|\tlq_{ij}|}} \sum_{t \in \tlq_{ij}} F_t e_{jt} + o_P(1). 
				\end{eqnarray*}
				Assumption \ref{ass:mom-clt}.\ref{ass:asy-normal-main-term-thm-loading} yields
				\[\frac{1}{N} \sum_{i=1}^N \sqrt{\frac{T}{|\tlq_{ij}|}} \Lambda_i \Lambda_i^\T  \frac{1}{\sqrt{|\tlq_{ij}|}} \sum_{t \in \tlq_{ij}} F_t e_{jt} \xrightarrow{d} \calN(0, \covI_{\Lambda,j}).\]
				Lemma \ref{lemma:def-q} implies $H \xrightarrow{p} (Q^\I)^\T$ and Lemma \ref{lemma:eigenvalue-converge} implies $\tilde D^\I \xrightarrow{p} D^\I$. Combined with Slutsky's Theorem, we conclude that
				\begin{align}
					\tilde D^{-1}  \frac{\sqrt{T}}{N} H \sum_{i=1}^N  \Lambda_i  \Lambda_i^\T  \frac{1}{|\tlq_{ij}|} \sum_{t \in \tlq_{ij}} F_t e_{jt} \xrightarrow{d} \calN(0, D^\I (Q^\I)^\T \covI_{\Lambda,j} Q^\I  D^\I).    
				\end{align}
				We provide a consistent estimate for the asymptotic variance $D^\I (Q^\I)^\T \covI_{\Lambda,j} Q^\I  D^\I$ in Lemma \ref{lemma:Lambda-asy-var-main-estimator}. 
				
				Next, we consider $\sqrt{T}( H_j - H  ) \Lambda_j$. Lemma \ref{lemma:HF-H-diff} implies, $H_j - H = O_P \Lp \frac{1}{\sqrt{T}} \Rp$, and therefore $\sqrt{T}( H_j - H  ) \Lambda_j  = O_P(1)$. This term  contributes to the asymptotic distribution of $\tilde \Lambda_j$. Recall the definition, $H = \frac{1}{NT} \tilde D^\I  \tilde \Lambda^\T \Lambda F^\T  F$ and $H_j  = \frac{1}{N} \tilde D^\I \sum_{i = 1}^N \tilde \Lambda_i \Lambda_i^\T \frac{1}{|\tlq_{ij}|} \sum_{t \in \tlq_{ij}} F_t F_t^\T$. We have 
				\begin{align*}
					H_j - H &= \tilde D^\I \cdot \frac{1}{N} \sum_{i = 1}^N  \tilde \Lambda_i \Lambda_i^\T \underbrace{\Lp \frac{1}{|\tlq_{ij}|} \sum_{t \in \tlq_{ij}} F_t F_t^\T - \frac{1}{T} \sum_{t = 1}^T F_t F_t^\T  \Rp }_{\Delta_{F, ij}}   \\
					&= \tilde D^\I \cdot \underbrace{\frac{1}{N} \sum_{i = 1}^N  H \Lambda_i \Lambda_i^\T \Delta_{F, ij}}_{\Delta_{H,1}}  + \tilde D^\I \cdot \underbrace{\frac{1}{N} \sum_{i = 1}^N  (\tilde \Lambda_i - H \Lambda_i) \Lambda_i^\T \Delta_{F, ij}}_{\Delta_{H,2}} .
				\end{align*}
				For the term $\Delta_{H,2}$, we have
				\begin{align*}
					\norm{\Delta_{H,2}} &\leq \underbrace{\Lp \frac{1}{N} \sum_{i=1}^N \norm{\tilde \Lambda_i - H \Lambda_i}^2 \Rp^{1/2}}_{O_P \Lp \frac{1}{\dnt} \Rp }  \Lp \frac{1}{N} \sum_{i=1}^N \norm{\Lambda_i}^2 \norm{\Delta_{F, ij}}^2  \Rp^{1/2},
				\end{align*}
				and the second term $\frac{1}{N} \sum_{i=1}^N \norm{\Lambda_i}^2 \norm{\Delta_{F, ij}}^2 $ satisfies
				\begin{align*}
					\+E \Big[ \frac{1}{N} \sum_{i=1}^N \norm{\Lambda_i}^2 \norm{\Delta_{F, ij}}^2 \Big] &= \frac{1}{N} \sum_{i=1}^N \+E \norm{\Lambda_i}^2  \cdot \+E  \norm{\Delta_{F, ij}}^2 = O \Lp \frac{1}{T} \Rp    .
				\end{align*}
				Hence, it holds that $\norm{\Delta_{H,2}} = O_P \Lp \frac{1}{\sqrt{T\delta }} \Rp$. From Assumption \ref{ass:mom-clt}.\ref{ass:asy-normal-add-term-thm-loading} and Slutsky's theorem, we have 
				\begin{align}
					\nonumber &\tilde D^\I  H  \cdot \frac{\sqrt{T}}{N} \sum_{i = 1}^N \Lambda_i \Lambda_i^\T  \Big( \frac{1}{|\tlq_{ij}|} \sum_{t \in \tlq_{ij}} F_t F_t^\T - \frac{1}{T} \sum_{t = 1}^T F_t F_t^\T \Big) \Lambda_j \\
					\rightarrow& \calN\Lp  0, D^\I (Q^\I)^\T \covII_{\Lambda,j} Q^\I D^\I \Rp   
					\quad \mathcal{G}^t-\text{stably},
				\end{align}
				where $\covII_{\Lambda,j} = h_j(\Lambda_j)$. We provide a consistent estimate for the asymptotic variance $D^\I (Q^\I)^\T \covII_{\Lambda,j} Q^\I D^\I$ in Lemma \ref{lemma:Lambda-var-corrector}.

				Furthermore, $\sqrt{T} (\tilde \Lambda_j - H_j \Lambda_j)$ and $\sqrt{T} (H_j - H) \Lambda_j$ are asymptotically independent because the randomness of $\sqrt{T} (\tilde \Lambda_j - H_j \Lambda_j)$  comes from $ F_t e_{jt}$ while the randomness of $\sqrt{T} (H_j - H) \Lambda_j$ comes from $\frac{1}{|\tlq_{ij}|} \sum_{t \in \tlq_{ij}} F_t F_t^\T - \frac{1}{T} \sum_{t = 1}^T F_t F_t^\T $. Then, we have
				\[\sqrt{T} (\tilde \Lambda_j - H \Lambda_j) \rightarrow  \calN \Big( 0, D^\I (Q^\I)^\T \big[ \covI_{\Lambda,j} + \covII_{\Lambda,j} \big]Q^\I  D^\I \Big) 	\quad \mathcal{G}^t-\text{stably} . \]
				If we multiply $\tilde \Lambda_j - H \Lambda_j$ by $H^\I$ from the left and use the results that $H \xrightarrow{p} (Q^\I)^\T$ from Lemma \ref{lemma:def-q} and $H^\T \tilde D^\I (H^\T)^{-1}  =  \Big(\frac{\Lambda^\T \Lambda}{N} \Big)^\I \Big(\frac{F^\T F}{T} \Big)^\I + O_P \Big( \frac{1}{\dnt} \Big)$ from Lemma \ref{lemma:HDinvHTinv}, we conclude that
				\[\sqrt{T} (H^\I \tilde \Lambda_j - \Lambda_j) \rightarrow \calN \Big( 0, \Sigma_{F}^\I \Sigma_{\Lambda}^\I   \big[ \covI_{\Lambda,j} + \covII_{\Lambda,j} \big]\Sigma_{\Lambda}^\I  \Sigma_{F}^\I  \Big) 	\quad \mathcal{G}^t-\text{stably}, \]
				or equivalently, 
				\[ \sqrt{T}  \Sigma_{\Lambda,j}^{-1/2}  (H^\I \tilde \Lambda_j - \Lambda_j) \xrightarrow{d} \calN \Big( 0,I_r  \Big) \]
				for $ \Sigma_{\Lambda,j} =\Sigma_{F}^\I \Sigma_{\Lambda}^\I   \big[ \covI_{\Lambda,j} + \covII_{\Lambda,j} \big]\Sigma_{\Lambda}^\I  \Sigma_{F}^\I $.
				

			\end{proof}

			\subsubsection{Proof of Theorem \ref{theorem:asy-normal-equal-weight}.2}
			\begin{lemma}\label{lemma:f-est-error-times-f-and-e-adj}
				Under Assumptions \ref{ass:obs}, \ref{ass:factor-model} and \ref{ass:mom-clt}, we have
				\begin{enumerate}
					\item $\frac{1}{N} \sum_{i = 1}^N W_{it}  \Lp \tilde \Lambda_i -  H_i \Lambda_i \Rp  e_{it}  = O_P \Lp \frac{1}{\delta_{NT}} \Rp$
					\item $\frac{1}{N} \sum_{i = 1}^N W_{it}  \Lp \tilde \Lambda_i -  H \Lambda_i \Rp  e_{it}  = O_P \Lp \frac{1}{\delta_{NT}} \Rp$
					\item $\frac{1}{N} \sum_{i = 1}^N \Lp \tilde \Lambda_i - H_i \Lambda_i  \Rp  \Lambda_i^\T = O_P \Lp \frac{1}{\delta_{NT}} \Rp$
					\item $\frac{1}{N} \sum_{i = 1}^N W_{it} \Lp \tilde \Lambda_i - H_i \Lambda_i  \Rp  \Lambda_i^\T = O_P \Lp \frac{1}{\delta_{NT}} \Rp$.
				\end{enumerate}
			\end{lemma}
			
			\begin{proof}[Proof of Lemma \ref{lemma:f-est-error-times-f-and-e-adj}]
				\begin{enumerate}
					\item $\frac{1}{N} \sum_{i = 1}^N W_{it} \Lp \tilde \Lambda_i -  H_i \Lambda_i \Rp  e_{it}$ has the decomposition 
					\begin{eqnarray*}
						\frac{1}{N} \sum_{i = 1}^N W_{it} \Lp \tilde \Lambda_i - H_i \Lambda_i \Rp e_{it} &=& \tilde D^\I \Big[ \underbrace{ \frac{1}{N^2} \sum_{i = 1}^N \sum_{l = 1}^N W_{it} \tilde \Lambda_l \gamma(l,i) e_{it} }_{\text{\RNum{1}} }  + \underbrace{\frac{1}{N^2} \sum_{i = 1}^N \sum_{l = 1}^N W_{it}  \tilde \Lambda_l \zeta_{li} e_{it}}_{\text{\RNum{2}}  }  \\
						&& + \underbrace{\frac{1}{N^2} \sum_{i = 1}^N  \sum_{l = 1}^N W_{it}  \tilde \Lambda_l \eta_{li} e_{it} }_{\text{\RNum{3}}}  + \underbrace{\frac{1}{N^2} \sum_{i = 1}^N  \sum_{l = 1}^N  W_{it} \tilde \Lambda_l \xi_{li} e_{it} }_{\text{\RNum{4}}}   \Big]
					\end{eqnarray*}
					We decompose the term $\text{\RNum{1}} $ further into 
					\begin{align*}
						\text{\RNum{1}} &= \underbrace{\frac{1}{N^2} \sum_{i = 1}^N \sum_{l = 1}^N  W_{it}(\tilde \Lambda_l - H \Lambda_l) \gamma(l,i) e_{it}}_{\text{\RNum{1}}_1}  + H \cdot \underbrace{\frac{1}{N^2} \sum_{i = 1}^N \sum_{l = 1}^N  W_{it} \Lambda_l \gamma(l,i) e_{it}}_{\text{\RNum{1}}_2}.
					\end{align*}
					The first term $\text{\RNum{1}}_1$ satisfies
					\begin{align*}
						\norm{\text{\RNum{1}}_1} &\leq \frac{1}{\sqrt{N}} \underbrace{\Big( \frac{1}{N} \sum_{l = 1}^N \norm{\tilde \Lambda_l - H \Lambda_l}^2 \Big)^{1/2}}_{O_P \Lp \frac{1}{\dnt} \Rp}     \Big( \underbrace{\frac{1}{N}\sum_{i = 1}^N \sum_{l = 1}^N W_{it} |\gamma(l,i)|^2}_{\leq M\text{ from Lemma \ref{lemma:prep-consistency}.1}}  \cdot \underbrace{\frac{1}{N}\sum_{i = 1}^N  e_{it}^2}_{O_P(1)}  \Big)^{1/2} \\
						&= O_P \Lp \frac{1}{\sqrt{N\delta_{NT}}} \Rp
					\end{align*}
					The moment of the second term $\text{\RNum{1}}_2$ has the following bound
					\begin{align*}
						\+E[\norm{\text{\RNum{1}}_2}] &\leq  \frac{1}{N^2} \sum_{i = 1}^N \sum_{l = 1}^N |\gamma(l,i)| \underbrace{\+E \Big[W_{it} \norm{\Lambda_l }  \Big]}_{\leq \bar \Lambda}  \underbrace{\+E[| e_{it}|]}_{\leq M}   = O\Lp \frac{1}{N } \Rp.
					\end{align*}
					Hence, it holds that $\text{\RNum{1}} = O_P \Lp \frac{1}{\sqrt{N\delta_{NT}}} \Rp + O_P \Lp \frac{1}{N } \Rp= O_P \Lp \frac{1}{\sqrt{N\delta_{NT}}} \Rp$. We also decompose the term $\text{\RNum{2}} $ into two further terms:
					\begin{align*}
						\text{\RNum{2}} &= \underbrace{\frac{1}{N^2} \sum_{i = 1}^N \sum_{l = 1}^N  W_{it} (\tilde \Lambda_l - H \Lambda_l) \zeta_{li} e_{it}}_{\text{\RNum{2}}_1}  + H \cdot \underbrace{\frac{1}{N^2} \sum_{i = 1}^N \sum_{l = 1}^N  W_{it} \Lambda_l \zeta_{li} e_{it}}_{\text{\RNum{2}}_2}  .    
					\end{align*}
					For the second term $\text{\RNum{2}}_2$, we have 
					\begin{align*}
						\text{\RNum{2}}_2 &=  \frac{1}{N}  \sum_{i = 1}^N \underbrace{\Ls \frac{1}{N} \sum_{l = 1}^N  W_{it} \Lambda_l \Big[ \frac{1}{|\tlq_{il}|} \sum_{s \in \tlq_{il}} e_{is} e_{ls} - \+E[e_{is} e_{ls}] \Big] \Rs}_{z_i}   e_{it}.   
					\end{align*}
					Assumption \ref{ass:mom-clt}.1 implies $\+E \norm{z_i}^2= O\Lp \frac{1}{NT} \Rp$. Then, it holds 
					\[\+E \big[ \norm{\text{\RNum{2}}_2}  \big] \leq \frac{1}{N} \sum_{i = 1}^N \+E \norm{z_i e_{it}} \leq \frac{1}{N} \sum_{i = 1}^N (\+E \big[ \norm{z_i}^2  \big] \+E e_{it}^2)^{1/2}  = O\Lp \frac{1}{\sqrt{NT}} \Rp. \]
					Hence, we obtain $\text{\RNum{2}}_2 = O_P \Lp \frac{1}{\sqrt{NT}} \Rp$. For the first term $\text{\RNum{2}}_1$, we have
					\begin{align*}
						\norm{\text{\RNum{2}}_1} &\leq \underbrace{\Lp \frac{1}{N} \sum_{l = 1}^N \norm{\tilde \Lambda_l - H \Lambda_l}^2  \Rp^{1/2}}_{O_P \Lp \frac{1}{ \dnt} \Rp }  \underbrace{\Lp \frac{1}{N} \sum_{l = 1}^N \Big( \frac{1}{N} \sum_{i = 1}^N W_{it} \zeta_{li} e_{it}  \Big)^2  \Rp^{1/2}}_{O_P \Lp \frac{1}{ \sqrt{T}} \Rp} , 
					\end{align*}
					where the second term is $O_P \Lp \frac{1}{ \sqrt{T}} \Rp$ following from
					\begin{align*}
						\frac{1}{N} \sum_{l = 1}^N \Big( \frac{1}{N} \sum_{i = 1}^N W_{it} \zeta_{li} e_{it}  \Big)^2  \leq& \frac{1}{N} \sum_{l = 1}^N \Big( \frac{1}{N} \sum_{i = 1}^N W_{it}  \zeta_{li}^2    \Big) \cdot \Big( \frac{1}{N} \sum_{i = 1}^N  e_{it}^2  \Big) \\
						\leq& \underbrace{ \frac{1}{N} \sum_{l = 1}^N \Big( \frac{1}{N} \sum_{i = 1}^N \zeta_{li}^2    \Big) \cdot \Big( \frac{1}{N} \sum_{i = 1}^N  e_{it}^2  \Big)}_{\substack{ =O_P \Lp \frac{1}{ T} \Rp \text{ follows from } \+E\Ls \frac{1}{N} \sum_{l = 1}^N \Big( \frac{1}{N} \sum_{i = 1}^N \zeta_{li}^2    \Big) \cdot \Big( \frac{1}{N} \sum_{i = 1}^N  e_{it}^2  \Big)\Rs  \\ =  \frac{1}{N^3} \sum_{l = 1}^N  \sum_{i = 1}^N   \sum_{j = 1}^N \+E[\zeta_{li}^2 e_{jt}^2 ] \\ \leq  \frac{1}{N^3} \sum_{l = 1}^N  \sum_{i = 1}^N   \sum_{j = 1}^N (\+E[\zeta_{li}^4] \+E[ e_{jt}^4 ])^{1/2}
								=  O\Lp \frac{1}{ T} \Rp }  } 
					\end{align*}
					Hence, we conclude $\text{\RNum{2}} = O_P \Lp \frac{1}{\sqrt{T\delta_{NT}}} \Rp$. For the third term $\text{\RNum{3}}$, we have the decomposition 
					\begin{align*}
						\text{\RNum{3}} &= \underbrace{ \frac{1}{N^2} \sum_{i = 1}^N  \sum_{l = 1}^N W_{it}  (\tilde \Lambda_l - H \Lambda_l)\eta_{li} e_{it}}_{\text{\RNum{3}}_1  }  + H \cdot \underbrace{\frac{1}{N^2} \sum_{i = 1}^N  \sum_{l = 1}^N W_{it}  \Lambda_l \eta_{li} e_{it}}_{\text{\RNum{3}}_2}     .
					\end{align*}
					For the first term $\text{\RNum{3}}_1 $, we have 
					\begin{align*}
						\norm{\text{\RNum{3}}_1} &\leq     \underbrace{\Lp \frac{1}{N} \sum_{l = 1}^N \norm{\tilde \Lambda_l - H \Lambda_l}^2  \Rp^{1/2}}_{O_P \Lp \frac{1}{ \dnt} \Rp } \Lp \frac{1}{N} \sum_{l = 1}^N \Big( \frac{1}{N} \sum_{i = 1}^N W_{it} \eta_{li} e_{it}  \Big)^2  \Rp^{1/2}, 
					\end{align*}
					and the second term $\frac{1}{N} \sum_{l = 1}^N \Big( \frac{1}{N} \sum_{i = 1}^N W_{it} \eta_{li} e_{it}  \Big)^2 $ satisfies
					\begin{align*}
						\frac{1}{N} \sum_{l = 1}^N \Big( \frac{1}{N} \sum_{i = 1}^N W_{it} \eta_{li} e_{it}  \Big)^2 \leq&\frac{1}{N} \sum_{l = 1}^N  \Big( \frac{1}{N} \sum_{i = 1}^N \eta_{li}^2   \Big) \cdot   \Big( \frac{1}{N} \sum_{i = 1}^N W_{it}   \Big)\\
						\leq& \frac{1}{N} \sum_{l = 1}^N \Big(\frac{\norm{\Lambda_l}^2 }{N} \sum_{i = 1}^N \norm{\frac{1}{|\tlq_{li}|} \sum_{s \in \tlq_{li}}  F_s e_{is}   }^2  \Big) \Big( \frac{1}{N} \sum_{i = 1}^N W_{it}   \Big)  =  O_P \Lp \frac{1}{ T} \Rp
					\end{align*}
					following from
					\begin{align*}
						& \+E \Ls \frac{1}{N} \sum_{l = 1}^N \Big(\frac{\norm{\Lambda_l}^2 }{N} \sum_{i = 1}^N \norm{\frac{1}{|\tlq_{li}|} \sum_{s \in \tlq_{li}}  F_s e_{is}   }^2  \Big) \Big( \frac{1}{N} \sum_{i = 1}^N W_{it} e_{it}^2   \Big)   \Rs   \\
						=&  \frac{1}{N} \sum_{l = 1}^N  \+E[\norm{\Lambda_l}^2  ] \frac{1}{N^2} \sum_{i = 1}^N \sum_{j=1}^N \+E \Ls  \norm{\frac{1}{|\tlq_{li}|} \sum_{s \in \tlq_{li}}  F_s e_{is}   }^2 W_{jt} e_{jt}^2    \Rs \\
						\leq& \frac{1}{N} \sum_{l = 1}^N  \+E[\norm{\Lambda_l}^2] \frac{1}{N^2} \sum_{i = 1}^N \sum_{j=1}^N  \bigg( \underbrace{\+E \Ls  \norm{\frac{1}{|\tlq_{li}|} \sum_{s \in \tlq_{li}}  F_s e_{is}   }^4  \Rs }_{\leq \frac{M}{|\tlq_{li}|^2}}   \+E \Ls   e_{jt}^4  \Rs \bigg)^{1/2} \\
						\leq& \cdot \max_{li} \frac{1}{|\tlq_{li}|} \cdot \bar{\Lambda} = O\Lp \frac{1}{ T} \Rp.
					\end{align*}
					Hence, we obtain the rate convergence rate $\text{\RNum{3}}_1 = O_P \Lp \frac{1}{\sqrt{T\delta_{NT}} } \Rp$. Next, let us consider $\text{\RNum{3}}_2$: 
					\begin{align*}
						\text{\RNum{3}}_2 &= \frac{1}{N^2}  \sum_{l = 1}^N \Lambda_l \Lambda_l^\T  \sum_{i = 1}^N W_{it}   \frac{1}{|\tlq_{li}|} \sum_{s \in \tlq_{li}}  F_s e_{is} e_{it} \\
						&= \underbrace{\frac{1}{N^2}  \sum_{l = 1}^N \Lambda_l \Lambda_l^\T  \sum_{i = 1}^N W_{it}   \frac{1}{|\tlq_{li}|} \sum_{s \in \tlq_{li}}  F_s \+E[e_{is} e_{it}] }_{\text{\RNum{3}}_{2,1}} \\
						&\quad + \underbrace{\frac{1}{N^2}  \sum_{l = 1}^N \Lambda_l \Lambda_l^\T  \sum_{i = 1}^N W_{it}   \frac{1}{|\tlq_{li}|} \sum_{s \in \tlq_{li}}  F_s \big( e_{is} e_{it} - \+E[e_{is} e_{it}]  \big)}_{\text{\RNum{3}}_{2,2}}\\
					\end{align*}
					The first term $\text{\RNum{3}}_{2,1}$ in the above decomposition satisfies 
					\begin{align*}
						\norm{\text{\RNum{3}}_{2,1}}^2 =& \Bigg(\frac{1}{N} \sum_{l = 1}^N \norm{\Lambda_l}^4 \Bigg) \Bigg(\frac{1}{N} \sum_{l = 1}^N \norm{\frac{1}{N} \sum_{i = 1}^N    \frac{W_{it}}{|\tlq_{li}|} \sum_{s \in \tlq_{li}}  F_s \+E[e_{is} e_{it}]  }^2\Bigg) \\
						\leq& \underbrace{\Bigg(\frac{1}{N} \sum_{l = 1}^N \norm{\Lambda_l}^4 \Bigg)  }_{O_P(1)} \Bigg(\frac{1}{N} \sum_{l = 1}^N \bigg(\frac{1}{NT} \sum_{i = 1}^N \sum_{s \in \tlq_{li}}  \norm{ \frac{T W_{it}}{|\tlq_{li}|} F_s}^2    \bigg) \bigg( \frac{1}{NT} \sum_{i = 1}^N \sum_{s \in \tlq_{li}}  (\+E[e_{is} e_{it}] )^2\bigg) \Bigg) .
					\end{align*}
					Using Assumption \ref{ass:factor-model}.3(d) we conclude that 
					\begin{align*}
						\sum_{i = 1}^N \sum_{s \in \tlq_{li}}  (\+E[e_{is} e_{it}] )^2 \leq \sum_{i = 1}^N \sum_{s \in \tlq_{li}}  |\+E[e_{is} e_{it}]| \leq  \sum_{i = 1}^N \sum_{t = 1}^T  |\+E[e_{is} e_{it}]| \leq M.
					\end{align*}
					Moreover, it holds that
					\begin{align*}
						\+E \Bigg[\frac{1}{N} \sum_{l = 1}^N \bigg(\frac{1}{NT} \sum_{i = 1}^N \sum_{s \in \tlq_{li}}  \norm{ \frac{T W_{it}}{|\tlq_{li}|} F_s}^2    \bigg)   \Bigg] =& \frac{T}{N^2 } \sum_{l = 1}^N \sum_{i = 1}^N \frac{1}{|\tlq_{li}|^2} \sum_{s \in \tlq_{li}}  \+E [W_{it}] \+E  \norm{ F_s}^2  \leq M.
					\end{align*}
					Then, it holds that $\norm{\text{\RNum{3}}_{2,1}} = O_P \Big(\frac{1}{\sqrt{NT}}\big)$. For $\text{\RNum{3}}_{2,2}$, we have 
					\begin{align*}
						\norm{\text{\RNum{3}}_{2,2}}^2 =& \Bigg(\frac{1}{N} \sum_{l = 1}^N \norm{\Lambda_l}^4 \Bigg) \Bigg(\frac{1}{N} \sum_{l = 1}^N \norm{\frac{1}{N} \sum_{i = 1}^N    \frac{W_{it}}{|\tlq_{li}|} \sum_{s \in \tlq_{li}}  F_s (e_{is} e_{it} - \+E[e_{is} e_{it}])  }^2\Bigg).
					\end{align*}
					Assumption \ref{ass:mom-clt}.1 yields that
					\begin{align*}
						& \+E \Bigg[ \frac{1}{N} \sum_{l = 1}^N \norm{\frac{1}{N} \sum_{i = 1}^N    \frac{W_{it}}{|\tlq_{li}|} \sum_{s \in \tlq_{li}}  F_s (e_{is} e_{it} - \+E[e_{is} e_{it}])  }^2 \Bigg] \\
						=&  \frac{1}{N} \sum_{l = 1}^N  \+E \Bigg[ \norm{\frac{1}{N} \sum_{i = 1}^N    \frac{W_{it}}{|\tlq_{li}|} \sum_{s \in \tlq_{li}}  F_s (e_{is} e_{it} - \+E[e_{is} e_{it}])  }^2\Bigg] \leq \frac{M}{NT}.
					\end{align*}
					
					Hence, we obtain the overall rate for the third term $\text{\RNum{3}} = O_P \Lp \frac{1}{\sqrt{T\delta_{NT}}} \Rp$. The rate for the fourth term $\text{\RNum{4}} = O_P \Lp \frac{1}{\sqrt{T\delta_{NT}}} \Rp$ can be shown  similarly.

					\item $\frac{1}{N} \sum_{i = 1}^N W_{it}  \Lp \tilde \Lambda_i -  H \Lambda_i \Rp  e_{it}  $ has the decomposition
					\begin{eqnarray*}
						\frac{1}{N} \sum_{i = 1}^N W_{it}   \Lp \tilde \Lambda_i -  H \Lambda_i \Rp  e_{it} &=& \underbrace{\frac{1}{N} \sum_{i = 1}^N W_{it}   \Lp \tilde \Lambda_i -  H_i \Lambda_i \Rp  e_{it} }_{=O_P \Lp \frac{1}{\delta_{NT}} \Rp \text{ from Lemma \ref{lemma:f-est-error-times-f-and-e-adj}.1}} \\
						&& + \underbrace{\frac{1}{N} \sum_{i = 1}^N W_{it}   (H_i - H) \Lambda_i e_{it}}_{\Delta}.
					\end{eqnarray*}
					The second term $\Delta$ satisfies 
					\begin{align*}
						\norm{\Delta}^2 &=  \norm{ \frac{1}{N^2}  \sum_{l = 1}^N \Lambda_l  \Lambda_l^\T  \Bigg[ \sum_{i =1}^N  W_{it}  \Lp \frac{1}{T} \sum_{s=1}^T F_s F_s^\T - \frac{1}{|\tlq_{li}|} \sum_{s \in \tlq_{li}} F_s F_s^\T \Rp \Lambda_i e_{it} \Bigg] }^2 \\
						&\leq \Bigg( \underbrace{\frac{1}{N} \sum_{l = 1}^N \norm{\Lambda_l }^4}_{O_P(1)} \Bigg) \Bigg( \frac{1}{N} \sum_{l = 1}^N  \Big\lVert\underbrace{\frac{1}{N} \sum_{i =1}^N  W_{it}  \Lp \frac{1}{T} \sum_{s=1}^T F_s F_s^\T - \frac{1}{|\tlq_{li}|} \sum_{s \in \tlq_{li}} F_s F_s^\T \Rp \Lambda_i e_{it}}_{z_l} \Big\rVert^2 \Bigg).
					\end{align*}
					Assumption \ref{ass:mom-clt}.\ref{ass:add-mom-three-sum-lam-err} implies $ \+E[\norm{\sqrt{NT} z_l}^2] \leq M$ and thus
					\[\+E \Big[ \frac{1}{N} \sum_{l = 1}^N \norm{z_l}^2 \Big] = \frac{1}{N}\sum_{l = 1}^N \+E[\norm{z_l}^2] \leq O \Big(\frac{1}{NT} \Big) \]
					yielding $\Delta = O_P \Big(\frac{1}{\sqrt{NT}} \Big)$. 
					Hence, we conclude that 
					\begin{align*}
						\frac{1}{N} \sum_{i = 1}^N W_{it}   \Lp \tilde \Lambda_i -  H \Lambda_i \Rp  e_{it} &=    O_P \Lp \frac{1}{\delta_{NT}} \Rp +  O_P \Lp \frac{1}{\sqrt{NT}}  \Rp = O_P \Lp \frac{1}{\delta_{NT}} \Rp.
					\end{align*}
					\item $\frac{1}{N} \sum_{i = 1}^N \Lp \tilde \Lambda_i - H_i \Lambda_i  \Rp  \Lambda_i^\T $ has the decomposition
					\begin{eqnarray*}
						\frac{1}{N} \sum_{i = 1}^N \Lp \tilde \Lambda_i - H_i \Lambda_i \Rp \Lambda_i^\T &=& \tilde D^\I \Big[ \underbrace{ \frac{1}{N^2} \sum_{i = 1}^N \sum_{l = 1}^N \tilde \Lambda_l \Lambda_i^\T \gamma(l,i)  }_{\text{\RNum{1}} }  + \underbrace{\frac{1}{N^2} \sum_{i = 1}^N \sum_{l = 1}^N  \tilde \Lambda_l \Lambda_i^\T \zeta_{li}  }_{\text{\RNum{2}}  }  \\
						&& + \underbrace{\frac{1}{N^2} \sum_{i = 1}^N  \sum_{l = 1}^N  \tilde \Lambda_l \Lambda_i^\T \eta_{li}  }_{\text{\RNum{3}}}  + \underbrace{\frac{1}{N^2} \sum_{i = 1}^N  \sum_{l = 1}^N   \tilde \Lambda_l \Lambda_i^\T \xi_{li} }_{\text{\RNum{4}}}   \Big]
					\end{eqnarray*}
					We decompose the first term $\text{\RNum{1}} $ further into two parts 
					\begin{align*}
						\text{\RNum{1}} &= \underbrace{\frac{1}{N^2} \sum_{i = 1}^N \sum_{l = 1}^N  (\tilde \Lambda_l - H \Lambda_l)  \Lambda_i^\T \gamma(l,i) }_{\text{\RNum{1}}_1}  + H \cdot \underbrace{\frac{1}{N^2} \sum_{i = 1}^N \sum_{l = 1}^N   \Lambda_l  \Lambda_i^\T \gamma(l,i) }_{\text{\RNum{1}}_2}.
					\end{align*}
					The first term $\text{\RNum{1}}_1$ satisfies
					\begin{align*}
						\norm{\text{\RNum{1}}_1} &\leq \frac{1}{\sqrt{N}} \underbrace{\Big( \frac{1}{N} \sum_{l = 1}^N \norm{\tilde \Lambda_l - H \Lambda_l}^2 \Big)^{1/2}}_{O_P \Lp \frac{1}{\dnt} \Rp}     \Big( \underbrace{\frac{1}{N}\sum_{i = 1}^N \sum_{l = 1}^N  |\gamma(l,i)|^2}_{\leq M \text{ from Lemma \ref{lemma:prep-consistency}.1}}  \cdot \underbrace{\frac{1}{N}\sum_{i = 1}^N  \norm{ \Lambda_i}^2 }_{O_P(1)}  \Big)^{1/2} \\
						&= O_P \Lp \frac{1}{\sqrt{N\delta_{NT}}} \Rp.
					\end{align*}
					The second term $\text{\RNum{1}}_2$ satisfies 
					\begin{align*}
						\+E[\norm{\text{\RNum{1}}_2}] &\leq  \frac{1}{N^2} \sum_{i = 1}^N \sum_{l = 1}^N |\gamma(l,i)| \underbrace{\+E \Big[ \norm{\Lambda_l } \norm{\Lambda_i }  \Big]}_{\leq \bar \Lambda}  \underbrace{\+E[| e_{it}|]}_{\leq M}   = O\Lp \frac{1}{N } \Rp.
					\end{align*}
					Hence, we conclude that $\text{\RNum{1}} = O_P \Lp \frac{1}{\sqrt{N\delta_{NT}}} \Rp + O_P \Lp \frac{1}{N } \Rp= O_P \Lp \frac{1}{\sqrt{N\delta_{NT}}} \Rp$. 
					
					For the term $\text{\RNum{2}} $, we have the following decomposition:
					\begin{align*}
						\text{\RNum{2}} &= \underbrace{\frac{1}{N^2} \sum_{i = 1}^N \sum_{l = 1}^N  (\tilde \Lambda_l - H \Lambda_l) \Lambda_i ^\T \zeta_{li} }_{\text{\RNum{2}}_1}  + H \cdot \underbrace{\frac{1}{N^2} \sum_{i = 1}^N \sum_{l = 1}^N  \Lambda_l \Lambda_i^\T  \zeta_{li} }_{\text{\RNum{2}}_2}      
					\end{align*}
					The second term $\text{\RNum{2}}_2$ satisfies 
					\begin{align*}
						\text{\RNum{2}}_2 &=  \frac{1}{N}  \sum_{i = 1}^N \underbrace{\Ls \frac{1}{N} \sum_{l = 1}^N   \Lambda_l \Big[ \frac{1}{|\tlq_{il}|} \sum_{s \in \tlq_{il}} e_{is} e_{ls} - \+E[e_{is} e_{ls}] \Big] \Rs}_{z_i}  \Lambda_i^\T  .
					\end{align*}
					Assumption \ref{ass:mom-clt}.1 implies $\+E \big[\norm{z_i}^2  \big] = O\Lp \frac{1}{NT} \Rp$, which leads to 
					\[\+E \big[ \norm{\text{\RNum{2}}_2} \big] \leq \frac{1}{N} \sum_{i = 1}^N  \+E \Big[\norm{z_i}\norm{\Lambda_i}   \Big] \leq  \frac{1}{N} \sum_{i = 1}^N \Big( \+E \Big[\norm{\Lambda_i}^2   \Big] \+E \Big[\norm{z_i}^2  \Big] \Big)^{1/2} = O\Lp \frac{1}{\sqrt{NT}} \Rp. \]
					Hence, the second term has the rate $\text{\RNum{2}}_2 = O_P \Lp \frac{1}{\sqrt{NT}} \Rp$. For the first term $\text{\RNum{2}}_1$, we have
					\begin{align*}
						\norm{\text{\RNum{2}}_1} &\leq \underbrace{\Lp \frac{1}{N} \sum_{l = 1}^N \norm{\tilde \Lambda_l - H \Lambda_l}^2  \Rp^{1/2}}_{O_P \Lp \frac{1}{ \dnt} \Rp } \Lp \frac{1}{N} \sum_{l = 1}^N  \norm{\frac{1}{N} \sum_{i = 1}^N  \Lambda_i \zeta_{li} }^2  \Rp^{1/2}  \\
						& \leq O_P \Lp \frac{1}{ \dnt} \Rp \cdot \bigg( \underbrace{\Big( \frac{1}{N} \sum_{i = 1}^N  \norm{\Lambda_i}^2 \Big) }_{O_P(1)}  \underbrace{\Big( \frac{1}{N^2}  \sum_{l = 1}^N  \sum_{i = 1}^N  \zeta_{li}^2 \Big)}_{O_P \Lp \frac{1}{ T} \Rp}  \bigg)^{1/2} = O \Lp \frac{1}{ \sqrt{T \delta_{NT} } } \Rp.
					\end{align*}
					Hence, we obtain the rate $\text{\RNum{2}} = O_P \Lp \frac{1}{\sqrt{T\delta_{NT}}} \Rp$. 
					
					We decompose third term $\text{\RNum{3}}$ further into two parts: 
					\begin{align*}
						\text{\RNum{3}} &= \underbrace{ \frac{1}{N^2} \sum_{i = 1}^N  \sum_{l = 1}^N  (\tilde \Lambda_l - H \Lambda_l)\Lambda_i^\T \eta_{li} }_{\text{\RNum{3}}_1  }  + H \cdot \underbrace{\frac{1}{N^2} \sum_{i = 1}^N  \sum_{l = 1}^N  \Lambda_l \Lambda_i^\T \eta_{li}}_{\text{\RNum{3}}_2}     
					\end{align*}
					For the first term $\text{\RNum{3}}_1 $, we have 
					\begin{align*}
						\norm{\text{\RNum{3}}_1} &\leq     \underbrace{\Lp \frac{1}{N} \sum_{l = 1}^N \norm{\tilde \Lambda_l - H \Lambda_l}^2  \Rp^{1/2}}_{O_P \Lp \frac{1}{ \dnt} \Rp } \Lp \frac{1}{N} \sum_{l = 1}^N \norm{\frac{1}{N} \sum_{i = 1}^N \Lambda_i^\T \eta_{li}}^2  \Rp^{1/2}, 
					\end{align*}
					and the second term $\frac{1}{N} \sum_{l = 1}^N \norm{\frac{1}{N} \sum_{i = 1}^N \Lambda_i^\T \eta_{li}}^2  $ satisfies
					\begin{align*}
						& \+E \Ls \frac{1}{N} \sum_{l = 1}^N \norm{\frac{1}{N} \sum_{i = 1}^N \Lambda_i^\T \eta_{li}}^2 \Rs =\+E \Ls \frac{1}{N} \sum_{l = 1}^N \norm{\frac{1}{N} \sum_{i = 1}^N \Lambda_i^\T   \frac{1}{|\tlq_{li}|} \sum_{s \in \tlq_{li}}\Lambda_l^\T   F_s e_{is}}^2 \Rs \\
						=&\frac{1}{N} \sum_{l = 1}^N  \+E \Ls \norm{\Lambda_l^\T \frac{1}{N} \sum_{i = 1}^N    \frac{1}{|\tlq_{li}|} \sum_{s \in \tlq_{li}}   F_s \Lambda_i^\T e_{is}}^2 \Rs \\
						\leq&\frac{1}{N} \sum_{l = 1}^N     \+E \Ls \Lp \frac{1}{N} \sum_{i = 1}^N \norm{\Lambda_l^\T  \Lambda_i} \norm{ \frac{1}{|\tlq_{li}|} \sum_{s \in \tlq_{li}}   F_s e_{is}} \Rp^2 \Rs \\
						\leq&\frac{1}{N^3} \sum_{l = 1}^N  \sum_{i = 1}^N \sum_{j = 1}^N \underbrace{\+E [\norm{\Lambda_l^\T  \Lambda_i} \norm{\Lambda_l^\T  \Lambda_j}] }_{\leq \bar{\Lambda}}  \Bigg(\underbrace{\+E\norm{ \frac{1}{|\tlq_{li}|} \sum_{s \in \tlq_{li}}   F_s e_{is}}^2}_{O\Lp \frac{1}{ T} \Rp} \underbrace{\+E \norm{ \frac{1}{|\tlq_{lj}|} \sum_{s \in \tlq_{lj}}   F_s e_{js}}^2}_{O\Lp \frac{1}{ T} \Rp}    \Bigg)^{1/2}  \\ 
						=&  O\Lp \frac{1}{ T} \Rp.
					\end{align*}
					Hence, we obtain $\text{\RNum{3}}_1 = O_P \Lp \frac{1}{\sqrt{T\delta_{NT}} } \Rp$. Next, we consider $\text{\RNum{3}}_2$: 
					\begin{align*}
						\norm{\text{\RNum{3}}_2}^2 &= \norm{\frac{1}{N^2}  \sum_{l = 1}^N \Lambda_l \Lambda_l^\T  \sum_{i = 1}^N   \frac{1}{|\tlq_{li}|} \sum_{s \in \tlq_{li}}  F_s \Lambda_i^\T  e_{is} }   \\
						&\leq \underbrace{\Lp \frac{1}{N}  \sum_{l = 1}^N  \norm{\Lambda_l}^4 \Rp}_{O_P(1)} \underbrace{\Lp \frac{1}{N}  \sum_{l = 1}^N \norm{\frac{1}{N} \sum_{i = 1}^N   \frac{1}{|\tlq_{li}|} \sum_{s \in \tlq_{li}}   F_s \Lambda_i^\T e_{is} }^2    \Rp}_{\substack{= O\Lp \frac{1}{ NT} \Rp \text{ from } \+E \Ls \frac{1}{N}  \sum_{l = 1}^N \norm{\frac{1}{N} \sum_{i = 1}^N   \frac{1}{|\tlq_{li}|} \sum_{s \in \tlq_{li}}   F_s \Lambda_i^\T e_{is} }^2   \Rs \\  =\frac{1}{N}  \sum_{l = 1}^N \+E \Ls \norm{\frac{1}{N} \sum_{i = 1}^N   \frac{1}{|\tlq_{li}|} \sum_{s \in \tlq_{li}}  F_s \Lambda_i^\T  e_{is} }^2   \Rs= O\Lp \frac{1}{ NT} \Rp\\ \text{from Assumption \ref{ass:mom-clt}.2} }  } 
					\end{align*}
					Hence, we conclude that $\text{\RNum{3}} = O_P \Lp \frac{1}{\sqrt{T\delta_{NT}}} \Rp$. The rate for the last term $\text{\RNum{4}} = O_P \Lp \frac{1}{\sqrt{T\delta_{NT}}} \Rp$ can be shown  similarly. 
					\item $\frac{1}{N} \sum_{i = 1}^N W_{it} \Lp \tilde \Lambda_i - H_i \Lambda_i  \Rp  \Lambda_i^\T $ has the decomposition
					\begin{eqnarray*}
						&& \frac{1}{N} \sum_{i = 1}^N \Lp W_{it} \tilde \Lambda_i - H_i \Lambda_i \Rp \Lambda_i^\T  \\
						&=& \tilde D^\I \Big[ \underbrace{ \frac{1}{N^2} \sum_{i = 1}^N \sum_{l = 1}^N W_{it} \tilde \Lambda_l \Lambda_i^\T \gamma(l,i)  }_{\text{\RNum{1}} }  + \underbrace{\frac{1}{N^2} \sum_{i = 1}^N \sum_{l = 1}^N W_{it} \tilde \Lambda_l \Lambda_i^\T \zeta_{li}  }_{\text{\RNum{2}}  }  \\
						&& + \underbrace{\frac{1}{N^2} \sum_{i = 1}^N  \sum_{l = 1}^N W_{it}  \tilde \Lambda_l \Lambda_i^\T \eta_{li}  }_{\text{\RNum{3}}}  + \underbrace{\frac{1}{N^2} \sum_{i = 1}^N  \sum_{l = 1}^N W_{it}  \tilde \Lambda_l \Lambda_i^\T \xi_{li} }_{\text{\RNum{4}}}   \Big].
					\end{eqnarray*}
					We decompose the term $\text{\RNum{1}} $ further into two parts 
					\begin{align*}
						\text{\RNum{1}} &= \underbrace{\frac{1}{N^2} \sum_{i = 1}^N \sum_{l = 1}^N  W_{it} (\tilde \Lambda_l - H \Lambda_l)  \Lambda_i^\T \gamma(l,i) }_{\text{\RNum{1}}_1}  + H \cdot \underbrace{\frac{1}{N^2} \sum_{i = 1}^N \sum_{l = 1}^N  W_{it} \Lambda_l  \Lambda_i^\T \gamma(l,i) }_{\text{\RNum{1}}_2}.
					\end{align*}
					The first term $\text{\RNum{1}}_1$ has
					\begin{align*}
						\norm{\text{\RNum{1}}_1} &\leq \frac{1}{\sqrt{N}} \underbrace{\Big( \frac{1}{N} \sum_{l = 1}^N \norm{\tilde \Lambda_l - H \Lambda_l}^2 \Big)^{1/2}}_{O_P \Lp \frac{1}{\dnt} \Rp}     \Big( \underbrace{\frac{1}{N}\sum_{i = 1}^N \sum_{l = 1}^N  |\gamma(l,i)|^2}_{\leq M \text{ from Lemma \ref{lemma:prep-consistency}.1}}  \cdot \underbrace{\frac{1}{N}\sum_{i = 1}^N W_{it} \norm{ \Lambda_i}^2 }_{\leq \frac{1}{N } \sum_{i = 1}^N \norm{ \Lambda_i}^2 = O_P(1) }  \Big)^{1/2} \\
						&= O_P \Lp \frac{1}{\sqrt{N\delta_{NT}}} \Rp.
					\end{align*}
					The second term $\text{\RNum{1}}_2$ satisfies 
					\begin{align*}
						\+E[\norm{\text{\RNum{1}}_2}  ] &\leq  \frac{1}{N^2} \sum_{i = 1}^N \sum_{l = 1}^N |\gamma(l,i)| \underbrace{\+E \Big[W_{it}  \Big]}_{\leq 1 }  \underbrace{\+E \Big[ \norm{\Lambda_l } \norm{\Lambda_i }  \Big]}_{\leq \bar \Lambda}  \underbrace{\+E[| e_{it}|]}_{\leq M}   = O\Lp \frac{1}{N } \Rp.
					\end{align*}
					Hence, we conclude $\text{\RNum{1}} = O_P \Lp \frac{1}{\sqrt{N\delta_{NT}}} \Rp + O_P \Lp \frac{1}{N } \Rp= O_P \Lp \frac{1}{\sqrt{N\delta_{NT}}} \Rp$. 
					
					The term $\text{\RNum{2}} $ can be decomposed into
					\begin{align*}
						\text{\RNum{2}} &= \underbrace{\frac{1}{N^2} \sum_{i = 1}^N \sum_{l = 1}^N W_{it} (\tilde \Lambda_l - H \Lambda_l) \Lambda_i ^\T \zeta_{li} }_{\text{\RNum{2}}_1}  + H \cdot \underbrace{\frac{1}{N^2} \sum_{i = 1}^N \sum_{l = 1}^N  W_{it} \Lambda_l \Lambda_i^\T  \zeta_{li} }_{\text{\RNum{2}}_2}      .
					\end{align*}
					For the second term $\text{\RNum{2}}_2$, we have 
					\begin{align*}
						\text{\RNum{2}}_2 &=  \frac{1}{N}  \sum_{l = 1}^N  \Lambda_l  \underbrace{\Ls \frac{1}{N}  \sum_{i = 1}^N W_{it} \Lambda_i^\T  \Big[ \frac{1}{|\tlq_{il}|} \sum_{s \in \tlq_{il}} e_{is} e_{ls} - \+E[e_{is} e_{ls}] \Big] \Rs}_{z_i}  
					\end{align*}
					From Assumption \ref{ass:mom-clt}.1 we infer that $\+E \big[ \norm{z_i}^2  \big] = O\Lp \frac{1}{NT} \Rp$. Then, it holds that 
					\[\+E \big[ \norm{\text{\RNum{2}}_2} \big] \leq \frac{1}{N} \sum_{l = 1}^N  \+E \Big[\norm{\Lambda_l} \norm{z_l}  \Big] \leq  \frac{1}{N} \sum_{l = 1}^N \Big( \+E \Big[\norm{\Lambda_l}^2   \Big] \+E \Big[\norm{z_l}^2  \Big] \Big)^{1/2} = O\Lp \frac{1}{\sqrt{NT}} \Rp. \]
					Hence, we obtain $\text{\RNum{2}}_2 = O_P \Lp \frac{1}{\sqrt{NT}} \Rp$. The first term $\text{\RNum{2}}_1$ satisfies
					\begin{align*}
						\norm{\text{\RNum{2}}_1} &\leq \underbrace{\Lp \frac{1}{N} \sum_{l = 1}^N \norm{\tilde \Lambda_l - H \Lambda_l}^2  \Rp^{1/2}}_{O_P \Lp \frac{1}{ \dnt} \Rp } \Lp \frac{1}{N} \sum_{l = 1}^N  \norm{\frac{1}{N} \sum_{i = 1}^N W_{it} \Lambda_i \zeta_{li} }^2  \Rp^{1/2}  \\
						& \leq O_P \Lp \frac{1}{ \dnt} \Rp \cdot \bigg(  \underbrace{\Big( \frac{1}{N} \sum_{i = 1}^N  \norm{\Lambda_i}^2 \Big) }_{O_P(1)}  \underbrace{\Big( \frac{1}{N^2}  \sum_{l = 1}^N  \sum_{i = 1}^N   \zeta_{li}^2 \Big)}_{O_P \Lp \frac{1}{ T} \Rp}  \bigg)^{1/2} \\
						&=O_P \Lp \frac{1}{ \sqrt{T \delta_{NT} }} \Rp.
					\end{align*}
					As a result we conclude $\text{\RNum{2}} = O_P \Lp \frac{1}{\sqrt{T\delta_{NT}}} \Rp$. 
					
					For the third term $\text{\RNum{3}}$, we have the decomposition 
					\begin{align*}
						\text{\RNum{3}} &= \underbrace{ \frac{1}{N^2} \sum_{i = 1}^N  \sum_{l = 1}^N W_{it}   (\tilde \Lambda_l - H \Lambda_l)\Lambda_i^\T \eta_{li} }_{\text{\RNum{3}}_1  }  + H \cdot \underbrace{\frac{1}{N^2} \sum_{i = 1}^N  \sum_{l = 1}^N W_{it}  \Lambda_l \Lambda_i^\T \eta_{li}}_{\text{\RNum{3}}_2}  .   
					\end{align*}
					For the first term $\text{\RNum{3}}_1 $ we have 
					\begin{align*}
						\norm{\text{\RNum{3}}_1} &\leq     \underbrace{\Lp \frac{1}{N} \sum_{l = 1}^N \norm{\tilde \Lambda_l - H \Lambda_l}^2  \Rp^{1/2}}_{O_P \Lp \frac{1}{ \dnt} \Rp } \Lp \frac{1}{N} \sum_{l = 1}^N \norm{\frac{1}{N} \sum_{i = 1}^N W_{it}  \Lambda_i^\T \eta_{li}}^2  \Rp^{1/2}, 
					\end{align*}
					and the second term $\frac{1}{N} \sum_{l = 1}^N \norm{\frac{1}{N} \sum_{i = 1}^N W_{it}  \Lambda_i^\T \eta_{li}}^2  $ satisfies
					\begin{align*}
						& \+E \Ls \frac{1}{N} \sum_{l = 1}^N \norm{\frac{1}{N} \sum_{i = 1}^N W_{it}  \Lambda_i^\T \eta_{li}}^2  \Rs \\ =& \+E \Ls \frac{1}{N} \sum_{l = 1}^N \norm{\frac{1}{N} \sum_{i = 1}^N W_{it}  \Lambda_i^\T   \frac{1}{|\tlq_{li}|} \sum_{s \in \tlq_{li}}\Lambda_l^\T   F_s e_{is}}^2 \Rs \\
						=&\frac{1}{N} \sum_{l = 1}^N  \+E \Ls \norm{\Lambda_l^\T \frac{1}{N} \sum_{i = 1}^N  W_{it} \cdot  \frac{1}{|\tlq_{li}|} \sum_{s \in \tlq_{li}}   F_s \Lambda_i^\T e_{is}}^2  \Rs \\
						\leq&\frac{1}{N} \sum_{l = 1}^N  \cdot \+E \Ls \Lp \frac{1}{N} \sum_{i = 1}^N \norm{\Lambda_l^\T  \Lambda_i} \norm{ \frac{1}{|\tlq_{li}|} \sum_{s \in \tlq_{li}}   F_s e_{is}} \Rp^2  \Rs \\
						\leq&\frac{1}{N^3 } \sum_{l = 1}^N  \sum_{i = 1}^N \sum_{j = 1}^N \underbrace{\+E [\norm{\Lambda_l^\T  \Lambda_i} \norm{\Lambda_l^\T  \Lambda_j}|S] }_{\leq \bar{\Lambda}}  \Bigg(\underbrace{\+E\norm{ \frac{1}{|\tlq_{li}|} \sum_{s \in \tlq_{li}}   F_s e_{is}}^2}_{O\Lp \frac{1}{ T} \Rp} \underbrace{\+E \norm{ \frac{1}{|\tlq_{lj}|} \sum_{s \in \tlq_{lj}}   F_s e_{js}}^2}_{O\Lp \frac{1}{ T} \Rp}    \Bigg)^{1/2}  \\ 
						=&  O\Lp \frac{1}{ T} \Rp.
					\end{align*}
					This results in $\text{\RNum{3}}_1 = O_P \Lp \frac{1}{\sqrt{T\delta_{NT}} } \Rp$. Next, we consider $\text{\RNum{3}}_2$: 
					\begin{align*}
						\norm{\text{\RNum{3}}_2}^2 &= \norm{\frac{1}{N^2}  \sum_{l = 1}^N \Lambda_l \Lambda_l^\T  \sum_{i = 1}^N  W_{it} \cdot \frac{1}{|\tlq_{li}|} \sum_{s \in \tlq_{li}}  F_s \Lambda_i^\T  e_{is} }   \\
						&\leq \underbrace{\Lp \frac{1}{N}  \sum_{l = 1}^N  \norm{\Lambda_l}^4 \Rp}_{O_P(1)} \underbrace{\Lp \frac{1}{N}  \sum_{l = 1}^N \norm{\frac{1}{N} \sum_{i = 1}^N W_{it} \cdot  \frac{1}{|\tlq_{li}|} \sum_{s \in \tlq_{li}}   F_s \Lambda_i^\T e_{is} }^2    \Rp}_{\substack{= O\Lp \frac{1}{ NT} \Rp \text{ from } \+E \Ls \frac{1}{N}  \sum_{l = 1}^N \norm{\frac{1}{N} \sum_{i = 1}^N  W_{it}  \frac{1}{|\tlq_{li}|} \sum_{s \in \tlq_{li}}   F_s \Lambda_i^\T e_{is} }^2  |S \Rs \\  =\frac{1}{N}  \sum_{l = 1}^N \+E \Ls \norm{\frac{1}{N} \sum_{i = 1}^N  W_{it}  \frac{1}{|\tlq_{li}|} \sum_{s \in \tlq_{li}}  F_s \Lambda_i^\T  e_{is} }^2  |S \Rs= O\Lp \frac{1}{ NT} \Rp\\ \text{from Assumption \ref{ass:mom-clt}.2} }  } 
					\end{align*}
					In conclusion, we obtain the rate $\text{\RNum{3}} = O_P \Lp \frac{1}{\sqrt{T\delta_{NT}}} \Rp$. The rate for the last term $\text{\RNum{4}} = O_P \Lp \frac{1}{\sqrt{T\delta_{NT}}} \Rp$ follows from similar arguments. 
					
				\end{enumerate}
			\end{proof}

			\begin{proof}[Proof of Theorem \ref{theorem:asy-normal-equal-weight}.2]
				We regress $Y_{it}$ on $\tilde \Lambda_i$ using the observed units at time $t$ (where $W_{it} = 1$)
				\begin{align*}
					\tilde F_t &= \Big( \sum_{i = 1}^N W_{it} \tilde \Lambda_i \tilde \Lambda_i^\T  \Big)^\I  \Big( \sum_{i = 1}^N W_{it} \tilde \Lambda_i Y_{it}    \Big) .
				\end{align*}
				We first analyze 
				\begin{align*}
					\tilde F_t^\dagger  &= \Big( \sum_{i = 1}^N W_{it} H \Lambda_i  \Lambda_i^\T H^\T  \Big)^\I  \Big( \sum_{i = 1}^N W_{it} \tilde \Lambda_i Y_{it}    \Big)  .
				\end{align*}
				We have the following decomposition for $\tilde F_t$
				\begin{align*}
					\tilde F_t^\dagger  =& \Big(\frac{1}{N}  \sum_{i = 1}^N W_{it} H \Lambda_i  \Lambda_i^\T H^\T  \Big)^\I  \Big(\frac{1}{N}  \sum_{i = 1}^N W_{it} \tilde \Lambda_i (\Lambda_i^\T F_t + e_{it})    \Big)    \\
					=& (H^\I)^\T F_t + \underbrace{(H^\I)^\T \Big(\frac{1}{N}  \sum_{i = 1}^N W_{it} \Lambda_i  \Lambda_i^\T  \Big)^\I  \Big(\frac{1}{N}  \sum_{i = 1}^N W_{it} \Lambda_i  e_{it}    \Big) }_{\Delta_1}   \\
					& + \underbrace{\Big(\frac{1}{N}  \sum_{i = 1}^N W_{it} H \Lambda_i  \Lambda_i^\T H^\T  \Big)^\I  \Big(\frac{1}{N}  \sum_{i = 1}^N W_{it} (\tilde \Lambda_i - H \Lambda_i) \Lambda_i^\T F_t   \Big) }_{\Delta_2}   \\
					&+ \Big(\frac{1}{N}  \sum_{i = 1}^N W_{it} H \Lambda_i  \Lambda_i^\T H^\T  \Big)^\I \underbrace{\Big(\frac{1}{N}  \sum_{i = 1}^N W_{it} (\tilde \Lambda_i - H \Lambda_i)  e_{it}   \Big) }_{{O_P \Lp \frac{1}{\delta_{NT} } \Rp \text{ from Lemma \ref{lemma:f-est-error-times-f-and-e-adj}.2 }} } .
				\end{align*}
				For $\Delta_1$, $\frac{1}{\sqrt{N}} \sum_{i = 1}^N W_{it} \Lambda_i  e_{it} \xrightarrow{d} \calN(0, \covI_{F,t})$ from Assumption \ref{ass:mom-clt}.\ref{ass:asy-normal-main-term-thm-factor}  and $\frac{1}{N} \sum_{i = 1}^N W_{it} \Lambda_i \Lambda_i  \xrightarrow{p} \Sigma_{\Lambda,t}$. Slutsky's theorem and Lemma \ref{lemma:def-q} ($H^\I \xrightarrow{p} Q^\T$) yield
				\begin{equation}\label{eqn:F-asy-term1-equal}
					\sqrt{N} \underbrace{(H^\I)^\T \Big(\frac{1}{N}  \sum_{i = 1}^N W_{it} \Lambda_i  \Lambda_i^\T  \Big)^\I  \Big(\frac{1}{N}  \sum_{i = 1}^N W_{it} \Lambda_i  e_{it}    \Big) }_{\bm{\varepsilon}_{F,t,1}}  \xrightarrow{d} \calN(0, Q \Sigma_{\Lambda,t}^\I \covI_{F,t} \Sigma_{\Lambda,t}^\I Q^\T).    
				\end{equation}
				Next, we decompose $\Delta_2$ into two parts
				\begin{align*}
					\Delta_2 =& \Big(\frac{1}{N}  \sum_{i = 1}^N W_{it} H \Lambda_i  \Lambda_i^\T H^\T  \Big)^\I \underbrace{ \Big(\frac{1}{N}  \sum_{i = 1}^N W_{it} (\tilde \Lambda_i - H_i \Lambda_i) \Lambda_i^\T F_t   \Big)}_{{O_P \Lp \frac{1}{\delta_{NT}} \Rp \text{ from Lemma \ref{lemma:f-est-error-times-f-and-e-adj}.4 }} }   \\
					&+  \Big(\frac{1}{N}  \sum_{i = 1}^N W_{it} H \Lambda_i  \Lambda_i^\T H^\T  \Big)^\I  \Big(\frac{1}{N}  \sum_{i = 1}^N W_{it} (H_i - H ) \Lambda_i \Lambda_i^\T F_t   \Big)  .
				\end{align*}
				For $\sum_{i = 1}^N W_{it} (H_i - H ) \Lambda_i \Lambda_i^\T $ in the second term, we obtain
				\begin{align*}
					\frac{1}{N} \sum_{i = 1}^N W_{it} (H_i - H ) \Lambda_i \Lambda_i^\T  =&\tilde D^{-1} \cdot \frac{1}{N^2} \sum_{i =1}^N \sum_{l = 1}^N \tilde \Lambda_l   \Lambda_l^\T  \Lp \frac{1}{|\tlq_{li}|} \sum_{s \in \tlq_{li}} F_s F_s^\T - \frac{1}{T} \sum_{s=1}^T F_s F_s^\T   \Rp W_{it}  \Lambda_i \Lambda_i^\T  \\
					=&  \tilde D^{-1} \cdot \underbrace{\frac{1}{N^2} \sum_{i =1}^N \sum_{l = 1}^N (\tilde \Lambda_l - H \Lambda_l )   \Lambda_l^\T  \Lp \frac{1}{|\tlq_{li}|} \sum_{s \in \tlq_{li}} F_s F_s^\T - \frac{1}{T} \sum_{s=1}^T F_s F_s^\T   \Rp W_{it}  \Lambda_i \Lambda_i^\T }_{\text{\RNum{1}} }  \\
					&+ \tilde D^{-1} H \cdot \underbrace{\frac{1}{N^2} \sum_{i =1}^N \sum_{l = 1}^N  \Lambda_l   \Lambda_l^\T  \Lp \frac{1}{|\tlq_{li}|} \sum_{s \in \tlq_{li}} F_s F_s^\T - \frac{1}{T} \sum_{s=1}^T F_s F_s^\T   \Rp W_{it}  \Lambda_i \Lambda_i^\T }_{\mathbf{X}_t }  .
				\end{align*}
				The first term \RNum{1} satisfies 
				\begin{align*}
					\norm{\text{\RNum{1}}}^2 &\leq \underbrace{\Lp \frac{1}{N} \sum_{l = 1}^N \norm{\tilde  \Lambda_l - H \Lambda_l}^2 \Rp}_{O_P \Lp \frac{1}{\delta_{NT}} \Rp}  \Lp\frac{1}{N} \sum_{l = 1}^N  \norm{ \Lambda_l }^2 \norm{\frac{1}{N} \sum_{i =1}^N W_{it} \Lambda_i \Lambda_i^\T \Lp \frac{1}{T} \sum_{s=1}^T F_s F_s^\T - \frac{1}{|\tlq_{li}|} \sum_{s \in \tlq_{li}} F_s F_s^\T \Rp}^2 \Rp    \\
					&\leq O_P \Lp \frac{1}{\delta_{NT}} \Rp \cdot \underbrace{\Lp\frac{1}{N} \sum_{l = 1}^N   \norm{ \Lambda_l }^2 \Lp \frac{1}{N} \sum_{i =1}^N \norm{\Lambda_i }^2 \norm{\frac{1}{T} \sum_{s=1}^T F_s F_s^\T - \frac{1}{|\tlq_{li}|} \sum_{s \in \tlq_{li}} F_s F_s^\T } \Rp^2  \Rp }_{\text{\RNum{1}}_1} ,
				\end{align*}
				where the first moment of $\text{\RNum{1}}_1$ has the following bound
				\begin{align*}
					\+E[ \text{\RNum{1}}_1 ] =& \frac{1}{N^3} \sum_{l = 1}^N \sum_{i =1}^N \sum_{j =1}^N  \underbrace{\+E\Big[ \norm{ \Lambda_l }^2  \norm{\Lambda_i }^2  \norm{\Lambda_j }^2 |S\Big]}_{\bar{\Lambda}}  \\
					& \cdot \underbrace{\+E\Bigg[ \norm{\frac{1}{T} \sum_{s=1}^T F_s F_s^\T - \frac{1}{|\tlq_{li}|} \sum_{s \in \tlq_{li}} F_s F_s^\T } \norm{\frac{1}{T} \sum_{s=1}^T F_s F_s^\T - \frac{1}{|\tlq_{lj}|} \sum_{s \in \tlq_{lj}} F_s F_s^\T } \bigg| \Bigg]}_{\leq \frac{M}{T} \text{ from } \+E[ab] \leq (\+E[a^2]\+E[b^2])^{1/2}  }. 
				\end{align*}
				Hence, we conclude that $\text{\RNum{1}} = O_P \Lp \frac{1}{\sqrt{T\delta_{NT}}} \Rp$. The second term $\mathbf{X}_t$ is asymptotically normal based on Assumption \ref{ass:mom-clt}.\ref{ass:asy-normal-add-term-thm-loading} and its convergence rate is $\sqrt{T}$. Hence in $\Delta_2$, the leading term is 
				\begin{align}\label{eqn:F-asy-term2-equal}
					\bm{\varepsilon}_{F,t, 2} =    \Big(\frac{1}{N}  \sum_{i = 1}^N W_{it} H \Lambda_i  \Lambda_i^\T H^\T  \Big)^\I  \Big(\tilde D^{-1} H \mathbf{X}_t  F_t   \Big),  
				\end{align}
				where $\mathbf{X}_t = \frac{1}{N^2} \sum_{i =1}^N \sum_{l = 1}^N  \Lambda_l   \Lambda_l^\T  \Lp \frac{1}{|\tlq_{li}|} \sum_{s \in \tlq_{li}} F_s F_s^\T - \frac{1}{T} \sum_{s=1}^T F_s F_s^\T   \Rp W_{it}  \Lambda_i \Lambda_i^\T$.

				
				Next, we consider the difference between $\tilde F_t^\dagger$ and $\tilde F_t$. The leading term is
				\begin{align*}
					&\tilde F_t -  \tilde F_t^\dagger \\
					=&  \Bigg[ \Big( \sum_{i = 1}^N W_{it} \tilde \Lambda_i \tilde \Lambda_i^\T  \Big)^\I - \Big( \sum_{i = 1}^N W_{it} H \Lambda_i  \Lambda_i^\T H^\T  \Big)^\I  \Bigg]  \Big( \sum_{i = 1}^N W_{it} \tilde \Lambda_i Y_{it}    \Big) \\
					=&\Big( \sum_{i = 1}^N W_{it} \tilde \Lambda_i \tilde \Lambda_i^\T  \Big)^\I   \Bigg[ \sum_{i = 1}^N W_{it} H \Lambda_i  \Lambda_i^\T H^\T - \sum_{i = 1}^N W_{it} \tilde \Lambda_i \tilde \Lambda_i^\T    \Bigg] \Big( \sum_{i = 1}^N W_{it} H \Lambda_i  \Lambda_i^\T H^\T  \Big)^\I \Big( \sum_{i = 1}^N W_{it} \tilde \Lambda_i Y_{it}    \Big) \\
					=& \Big( \frac{1}{N} \sum_{i = 1}^N W_{it} \tilde \Lambda_i \tilde \Lambda_i^\T  \Big)^\I \Bigg[ \frac{1}{N} \sum_{i = 1}^N W_{it} H \Lambda_i  \Lambda_i^\T H^\T - \frac{1}{N} \sum_{i = 1}^N W_{it} \tilde \Lambda_i \tilde \Lambda_i^\T   \Bigg]  \Big( \frac{1}{N} \sum_{i = 1}^N W_{it} H \Lambda_i  \Lambda_i^\T H^\T  \Big)^\I  \\
					& \cdot \Big(\frac{1}{N} \sum_{i = 1}^N W_{it} H \Lambda_i  \Lambda_i^\T  F_t +  O_P \Big( \frac{1}{\sqrt{N}} \Big) \Big).
				\end{align*}
				Note that 
				\begin{align*}
					& \norm{\frac{1}{N} \sum_{i = 1}^N W_{it} \tilde \Lambda_i \tilde \Lambda_i^\T  - \frac{1}{N} \sum_{i = 1}^N W_{it} \tilde \Lambda_i \tilde \Lambda_i^\T  }  \leq \frac{1}{N} \sum_{i = 1}^N \norm{\tilde \Lambda_i \tilde \Lambda_i^\T - H \Lambda_i  \Lambda_i^\T H^\T} \\
					\leq& \frac{1}{N} \sum_{i = 1}^N  \norm{\tilde \Lambda_i} \norm{\tilde \Lambda_i - H \Lambda_i } + \frac{1}{N} \sum_{i = 1}^N \norm{\Lambda_i} \norm{\tilde \Lambda_i - H \Lambda_i } \\
					\leq&  \bigg( \frac{1}{N} \sum_{i = 1}^N  \norm{\tilde \Lambda_i}^2  \bigg)^{1/2}  \bigg( \frac{1}{N} \sum_{i = 1}^N \norm{\tilde \Lambda_i - H \Lambda_i }^2\bigg)^{1/2} +  \bigg(  \frac{1}{N} \sum_{i = 1}^N \norm{\Lambda_i}^2 \bigg)^{1/2} \bigg( \frac{1}{N} \sum_{i = 1}^N  \norm{\tilde \Lambda_i - H \Lambda_i }^2\bigg)^{1/2} = O_P \bigg( \frac{1}{\dnt} \bigg).
				\end{align*}
				following from Theorem \ref{thm:consistency-same-H},  $\frac{1}{N} \tilde{\Lambda}^\T \tilde{\Lambda} = I_r$ and Assumption \ref{ass:factor-model}.2. Hence, we have 
				\[\frac{1}{N} \sum_{i = 1}^N W_{it} \tilde \Lambda_i \tilde \Lambda_i^\T  \xrightarrow{p} \frac{1}{N} \sum_{i = 1}^N W_{it} H \Lambda_i  \Lambda_i^\T H^\T.  \]
				This is also equivalent to 
				\[\Big(\frac{1}{N} \sum_{i = 1}^N W_{it} \tilde \Lambda_i \tilde \Lambda_i^\T  \Big)^\I \Big(\frac{1}{N} \sum_{i = 1}^N W_{it} H \Lambda_i  \Lambda_i^\T H^\T  \Big) \xrightarrow{p} I_k.  \]
				For the term $ \sum_{i = 1}^N W_{it} H \Lambda_i  \Lambda_i^\T H^\T - \sum_{i = 1}^N W_{it} \tilde \Lambda_i \tilde \Lambda_i^\T   $, we have the decomposition
				\begin{align*}
					&\frac{1}{N} \sum_{i = 1}^N W_{it} \tilde \Lambda_i \tilde \Lambda_i^\T - \frac{1}{N}\sum_{i = 1}^N W_{it} H \Lambda_i  \Lambda_i^\T H^\T \\
					=& \frac{1}{N}\sum_{i = 1}^N W_{it}  (\tilde \Lambda_i -  H \Lambda_i) \tilde \Lambda_i^\T +  \frac{1}{N}\sum_{i = 1}^N W_{it} H \Lambda_i   (\tilde \Lambda_i -  H \Lambda_i)^\T \\
					=& \frac{1}{N}\sum_{i = 1}^N W_{it}  (\tilde \Lambda_i -  H \Lambda_i)  (H \Lambda_i)^\T +  \frac{1}{N}\sum_{i = 1}^N W_{it} H \Lambda_i   (\tilde \Lambda_i -  H \Lambda_i)^\T   + \frac{1}{N}\sum_{i = 1}^N W_{it}  (\tilde \Lambda_i -  H \Lambda_i) (\tilde \Lambda_i - H \Lambda_i)^\T \\
					=& \underbrace{\frac{1}{N}\sum_{i = 1}^N W_{it} (H_i - H)  \Lambda_i \Lambda_i^\T}_{\tilde D^{-1} H\mathbf{X}_t}  \cdot H^\T + H \cdot \underbrace{\frac{1}{N}\sum_{i = 1}^N W_{it}  \Lambda_i \Lambda_i^\T (H_i - H)^\T }_{(\tilde D^{-1} H\mathbf{X}_t)^\T}   \\
					&+ \underbrace{\frac{1}{N}\sum_{i = 1}^N W_{it}  (\tilde \Lambda_i -  H_i \Lambda_i)  \Lambda_i^\T}_{{O_P \Lp \frac{1}{\delta_{NT}} \Rp \text{ from Lemma \ref{lemma:f-est-error-times-f-and-e-adj}.4 }} }  \cdot H^\T + H \cdot \underbrace{\frac{1}{N}\sum_{i = 1}^N W_{it}\Lambda_i   (\tilde \Lambda_i -  H_i \Lambda_i)^\T}_{{O_P \Lp \frac{1}{\delta_{NT}} \Rp \text{ from Lemma \ref{lemma:f-est-error-times-f-and-e-adj}.4 }} }   \\
					&+ \underbrace{\frac{1}{N}\sum_{i = 1}^N W_{it}  (\tilde \Lambda_i -  H \Lambda_i) (\tilde \Lambda_i - H \Lambda_i)^\T }_{\text{\RNum{1}}} .
				\end{align*}
				For the term \RNum{1}, we have 
				\begin{align*}
					\norm{\text{\RNum{1}} } &\leq   \frac{1}{N}\sum_{i = 1}^N W_{it} \norm{\tilde \Lambda_i -  H \Lambda_i}^2 =  O_P \Lp \frac{1}{\delta_{NT}} \Rp.
				\end{align*}
				The term $\mathbf{X}_t$ is asymptotically normal from Assumption \ref{ass:mom-clt}.\ref{ass:asy-normal-add-term-thm-loading} and its convergence rate is $\sqrt{T}$. Hence, the leading term in $\tilde F_t -  \tilde F_t^\dagger$ is 
				\begin{align}
					\nonumber\bm{\varepsilon}_{F,t, 3} =&  - \Big(\frac{1}{N} \sum_{i = 1}^N W_{it} H \Lambda_i  \Lambda_i^\T H^\T  \Big)^{-1} \Big(\tilde D^{-1} H  \mathbf{X}_t H^\T + H (\tilde D^{-1} H \mathbf{X}_t)^\T \Big) \\
					\nonumber & \cdot\Big( \frac{1}{N} \sum_{i = 1}^N W_{it} H \Lambda_i  \Lambda_i^\T H^\T  \Big)^\I  \Big(\frac{1}{N} \sum_{i = 1}^N W_{it} H \Lambda_i  \Lambda_i^\T  \Big)  F_t \\
					=&  \Big( \frac{1}{N} \sum_{i = 1}^N W_{it} H \Lambda_i  \Lambda_i^\T H^\T  \Big)^\I  \Big(\tilde D^{-1} H \mathbf{X}_t H^\T + H (\tilde D^{-1} H \mathbf{X}_t)^\T \Big) (H^\T)^{-1} F_t , \label{eqn:F-asy-term3-equal}
				\end{align}
				where $\mathbf{X}_t = \frac{1}{N^2} \sum_{i =1}^N \sum_{l = 1}^N  \Lambda_l   \Lambda_l^\T  \Lp \frac{1}{|\tlq_{li}|} \sum_{s \in \tlq_{li}} F_s F_s^\T - \frac{1}{T} \sum_{s=1}^T F_s F_s^\T   \Rp W_{it}  \Lambda_i \Lambda_i^\T$. 
				In summary, the asymptotic distribution of $\tilde F_t$ is determined by \eqref{eqn:F-asy-term1-equal}, \eqref{eqn:F-asy-term2-equal} and \eqref{eqn:F-asy-term3-equal}, that is, 
				\begin{align*}
					\dnt (\tilde F_t -  (H^{-1})^\T F_t) =& \dnt (\bm{\varepsilon}_{F,t, 1} + \bm{\varepsilon}_{F,t, 2} + \bm{\varepsilon}_{F,t, 3})   + o_P(1).
				\end{align*}
				Let us first consider the asymptotic distribution of $\bm{\varepsilon}_{F,t, 2} + \bm{\varepsilon}_{F,t, 3}$. Note by Assumption \ref{ass:mom-clt}.\ref{ass:asy-normal-add-term-thm-loading}, it holds that $\sqrt{T} \tvec(\mathbf{X}_t) \xrightarrow{d} \calN (0, \mathbf{\Phi}_t)$. Denote $\tilde \Sigma_{\Lambda, t} := \frac{1}{N} \sum_{i = 1}^N W_{it}  \Lambda_i  \Lambda_i^\T$. We can rewrite $\bm{\varepsilon}_{F,t, 2}$ as
				\begin{align*}
					\bm{\varepsilon}_{F,t, 2} &= (H^\T)^\I \tilde \Sigma_{\Lambda, t}^\I H^\I \tilde D^{-1} H \Big(  F_t^\T \otimes I_r \Big) \tvec(\mathbf{X}_t).
				\end{align*}
				Furthermore, we can rewrite $\bm{\varepsilon}_{F,t, 3}$ as
				\begin{align*}
					\bm{\varepsilon}_{F,t, 3} &=  - (H^\T)^{-1} \tilde \Sigma_{\Lambda, t}^{-1} H^{-1} \Bigg(\tilde D^{-1} H \mathbf{X}_t F_t  + H \mathbf{X}_t^\T H^\T \tilde D^{-1} (H^\T)^{-1}  F_t  \Bigg) \\
					&=  - (H^\T)^{-1} \tilde \Sigma_{\Lambda, t}^{-1} H^{-1} \Big(\tilde D^{-1} H \big( F_t ^\T \otimes I_r \big)  + H \big(I_r \otimes (H^\T \tilde D^{-1} (H^\T)^{-1}  F_t   )^\T \big)  \Big) \tvec(\mathbf{X}_t).
				\end{align*}
				Then, for $\bm{\varepsilon}_{F,t, 2} + \bm{\varepsilon}_{F,t, 3}  $, we obtain 
				\begin{align*}
					& \sqrt{T} \Big( \bm{\varepsilon}_{F,t, 2} + \bm{\varepsilon}_{F,t, 3} \Big)      \\
					=& \sqrt{T} (H^\T)^\I \tilde \Sigma_{\Lambda, t}^\I H^\I  \Bigg( \tilde D^{-1} H \Big(  F_t^\T \otimes I_r \Big) \\
					&- \Big(\tilde D^{-1} H  \big(  F_t^\T \otimes I_r \big)  + H \big(I_r \otimes (H^\T \tilde D^\I (H^\T)^{-1} F_t )^\T \big)  \Big) \Bigg) \tvec(\mathbf{X}_t)  \\
					=&  -\sqrt{T} (H^\T)^\I \tilde \Sigma_{\Lambda, t}^\I \Big(I_r \otimes (H^\T \tilde D^\I (H^\T)^{-1} F_t )^\T \Big) \tvec(\mathbf{X}_t) .
				\end{align*}
				Combining this with the result of Lemma \ref{lemma:HDinvHTinv},  $H^\T \tilde D^\I (H^\T)^{-1} =  \Big( \frac{\Lambda^\T  \Lambda}{N} \Big)^\I  \Big(\frac{F^\T F}{T} \Big)^\I + O_P \Big( \frac{1}{\sqrt{\delta_{NT}}} \Big) $, we conclude that
				\begin{align*}
					& \sqrt{T} \Big( \bm{\varepsilon}_{F,t, 2} + \bm{\varepsilon}_{F,t, 3} \Big)      \\
					=&  -\sqrt{T} (H^\T)^\I \tilde \Sigma_{\Lambda, t}^\I \Big(I_r \otimes (H^\T \tilde D^\I (H^\T)^{-1} F_t )^\T \Big) \tvec(\mathbf{X}_t) \\
					\rightarrow&   \calN \Bigg(0,  Q \Sigma_{\Lambda, t}^\I    \covII_{F,t}  \Sigma_{\Lambda, t}^\I  Q^\T \Bigg)
					\quad \mathcal{G}^t-\text{stably},
				\end{align*}
				where  $\covII_{F,t} = g_t(F_t)$  and the function $g_t(\cdot)$ is defined in Assumption \ref{ass:mom-clt}.\ref{ass:asy-normal-add-term-thm-loading}. 
				
				Note that $\bm{\varepsilon}_{F,t, 1}  $ and $\bm{\varepsilon}_{F,t, 2} + \bm{\varepsilon}_{F,t, 3}  $ are asymptotically independent because the randomness of $\bm{\varepsilon}_{F,t, 1}  $ comes from the cross-section average of $ W_{it} \Lambda_i e_{it}$, and the randomness of $\bm{\varepsilon}_{F,t, 2} + \bm{\varepsilon}_{F,t, 3}  $  comes from $\frac{1}{T} \sum_{s=1}^T F_s F_s^\T - \frac{1}{|\tlq_{li}|} \sum_{s \in \tlq_{li}} F_s F_s^\T $.  Then, we have
				\begin{align*}
					\sqrt{\delta_{NT}} (\tilde F_t -  (H^{-1})^\T F_t) \rightarrow \calN \Bigg(0,  Q \Sigma_{\Lambda,t}^\I  \Big[  \problim \Big(  \frac{\delta_{NT}}{N} \covI_{F,t} + \frac{\delta_{NT}}{T} \covII_{F,t} \Big) \Big] \Sigma_{\Lambda,t}^\I  Q^\T  \Bigg) 
					\quad \mathcal{G}^t-\text{stably}.
				\end{align*}
				If we left-multiply $\tilde F_t -  (H^{-1})^\T F_t$ by $H^\T$, the delta method implies that  
				\begin{align*}
					\sqrt{\delta_{NT}} (H^\T \tilde F_t -  F_t) \rightarrow\calN \Bigg(0,  \Sigma_{\Lambda,t}^\I  \Big[  \problim \Big(  \frac{\delta_{NT}}{N} \covI_{F,t} + \frac{\delta_{NT}}{T} \covII_{F,t} \Big) \Big] \Sigma_{\Lambda,t}^\I   \Bigg)
					\quad \mathcal{G}^t-\text{stably},
				\end{align*}
				or equivalently, 
				\begin{align*}
					\sqrt{\delta_{NT}} \Sigma_{F,t}^{-1/2} (H^\T \tilde F_t -  F_t) \xrightarrow{d} \calN \big(0,  I_r \big)
				\end{align*}
				for $\Sigma_{F,t} =  \Sigma_{\Lambda,t}^\I  \Big[ \Big(  \frac{\delta_{NT}}{N} \covI_{F,t} + \frac{\delta_{NT}}{T} \covII_{F,t} \Big) \Big] \Sigma_{\Lambda,t}^\I $.
				
				
			\end{proof}
			
			\subsubsection{Proof of Theorem \ref{theorem:asy-normal-equal-weight}.3}
			\begin{proof}[Proof of Theorem \ref{theorem:asy-normal-equal-weight}.3]
				From $\tilde C_{jt} = \tilde \Lambda_j^\T \tilde F_t$ and $C_{jt} = \Lambda_j^\T F_t$, we have
				\[\tilde C_{jt} - C_{jt} = \Lambda_j^\T H^\T  (\tilde F_t - (H^\T)^{-1} F_t) + (\tilde \Lambda_j - H \Lambda_j)^\T \tilde F_t + o_P(1/\sqrt{\delta_{NT}}). \]
				The second term can be written as 
				\begin{eqnarray*}
					(\tilde \Lambda_j - H \Lambda_j)^\T \tilde F_t &=&  (\tilde \Lambda_j - H \Lambda_j)^\T (H^\T)^\I  F_t +  (\tilde \Lambda_j - H \Lambda_j)^\T ( \tilde F_t - (H^\T)^\I  F_t)  \\
					&=& (\tilde \Lambda_j - H \Lambda_j)^\T (H^\T)^\I  F_t + o_P(1/\sqrt{\delta_{NT}}).
				\end{eqnarray*}
				Thus, 
				\[\tilde C_{jt} - C_{jt} = \Lambda_j^\T H^\T  (\tilde F_t - (H^\T)^{-1} F_t)  + (\tilde \Lambda_j - H \Lambda_j)^\T (H^\T)^\I  F_t + o_P(1/\sqrt{\delta_{NT}}). \]
				Following Theorem 3 in \cite{bai2003inferential}, we can show that $H^\T H = \Lp \frac{\Lambda^\T \Lambda}{N} \Rp^\I + O_P \Lp \frac{1}{\delta_{NT}} \Rp$. Then,
				\begin{align*}
					\sqrt{\delta_{NT}} (\tilde C_{jt} - C_{jt}) =& \sqrt{\delta_{NT}} \Lambda_j^\T H^\T  (\tilde F_t - (H^\T)^{-1} F_t) + \sqrt{\delta_{NT}} F_t^\T H^\I (\tilde \Lambda_j - H \Lambda_j) + O_P \Lp \frac{1}{\sqrt{\delta_{NT}}} \Rp  \\   
					=& \sqrt{\delta_{NT}} \Lambda_j^\T  \Big(\frac{1}{N}  \sum_{i = 1}^N W_{it} \Lambda_i  \Lambda_i^\T  \Big)^\I  \Big(\frac{1}{N}  \sum_{i = 1}^N W_{it} \Lambda_i  e_{it}    \Big)  \\
					&- \sqrt{\delta_{NT}} \Lambda_j^\T H^\T \cdot (H^\T)^\I \tilde \Sigma_{\Lambda, t}^\I \Big(I_r \otimes (H^\T \tilde D^\I (H^\T)^{-1} F_t )^\T \Big) \tvec(\mathbf{X}_t) \\
					&+\sqrt{\delta_{NT}}  F_t^\T H^\I \cdot \tilde D^{-1} H  \frac{1}{N} \sum_{i=1}^N  \Lambda_i  \Lambda_i^\T  \frac{1}{|\tlq_{ij}|} \sum_{t \in \tlq_{ij}} F_t e_{jt}\\ 
					& +\sqrt{\delta_{NT}}  F_t^\T H^\I \tilde D^\I  H  \cdot \big(\Lambda_j^\T \otimes I_r  \big) \tvec(X_j )  + O_P \Lp \frac{1}{\sqrt{\delta_{NT}}} \Rp \\
					=& \sqrt{\delta_{NT}} \Lambda_j^\T  \Big(\frac{1}{N}  \sum_{i = 1}^N W_{it} \Lambda_i  \Lambda_i^\T  \Big)^\I  \Big(\frac{1}{N}  \sum_{i = 1}^N W_{it} \Lambda_i  e_{it}    \Big)  \\
					&- \sqrt{\delta_{NT}} \Lambda_j^\T  \tilde \Sigma_{\Lambda, t}^\I \Bigg(I_r \otimes \Big(\Big(\frac{\Lambda^\T \Lambda}{N} \Big)^\I \Big( \frac{F^\T F}{T} \Big)^\I  F_t \Big)^\T \Bigg) \tvec(\mathbf{X}_t) \\
					&+\sqrt{\delta_{NT}}  F_t^\T \Big(\frac{F^\T F}{T} \Big)^\I  \Big(\frac{\Lambda^\T \Lambda}{N} \Big)^\I  \frac{1}{N} \sum_{i=1}^N  \Lambda_i  \Lambda_i^\T  \frac{1}{|\tlq_{ij}|} \sum_{t \in \tlq_{ij}} F_t e_{jt}\\ 
					& +\sqrt{\delta_{NT}}  F_t^\T \Big(\frac{F^\T F}{T} \Big)^\I  \Big(\frac{\Lambda^\T \Lambda}{N} \Big)^\I  \cdot \big(\Lambda_j^\T \otimes I_r  \big) \tvec(X_j )  + O_P \Lp \frac{1}{\sqrt{\delta_{NT}}} \Rp ,
				\end{align*}
				where the last equality follows Lemma \ref{lemma:HDinvHTinv} and $\tilde \Sigma_{\Lambda, t} = \frac{1}{N} \sum_{i = 1}^N W_{it}  \Lambda_i  \Lambda_i^\T$, \\ $X_j = \frac{1}{N} \sum_{l = 1}^N \Lambda_l \Lambda_l^\T  \Big( \frac{1}{|\tlq_{lj}|} \sum_{t \in \tlq_{lj}} F_t F_t^\T - \frac{1}{T} \sum_{t = 1}^T F_t F_t^\T \Big)$, and \\ $\mathbf{X}_t = \frac{1}{N^2} \sum_{i =1}^N \sum_{l = 1}^N  \Lambda_l   \Lambda_l^\T  \Lp \frac{1}{|\tlq_{li}|} \sum_{s \in \tlq_{li}} F_s F_s^\T - \frac{1}{T} \sum_{s=1}^T F_s F_s^\T   \Rp W_{it}  \Lambda_i \Lambda_i^\T$. Note that \\ $\mathbf{X}_t = \frac{1}{N} \sum_{i = 1}^N W_{it} X_i \Lambda_i \Lambda_i^\T$, and then $\mathbf{X}_t$ and $X_j$ are correlated. 
				The other leading terms in $\sqrt{\delta_{NT}} (\tilde C_{jt} - C_{jt}) $ are asymptotically independent. Combining these results with Assumption \ref{ass:mom-clt}.\ref{ass:asy-normal-add-term-thm-loading}, we conclude that 
				\begin{align*}
					\sqrt{\delta_{NT}} (\tilde C_{jt} - C_{jt}) \rightarrow& \calN \Big( 0,  \problim \Big(  \frac{\delta_{NT}}{N} \Lambda_j^\T \Sigma_{\Lambda,t}^\I \covI_{F,t} \Sigma_{\Lambda,t}^\I \Lambda_j + \frac{\delta_{NT}}{T} F_t^\T \Sigma_F^\I \Sigma_\Lambda^\I \covI_{\Lambda,j}\Sigma_\Lambda^\I \Sigma_F^\I F_t \\
					& \quad + \frac{\delta_{NT}}{T}  \Lambda_j^\T \Sigma_{\Lambda,t}^\I  \covII_{F,t}   \Sigma_{\Lambda,t}^\I \Lambda_j 
					+  \frac{\delta_{NT}}{T} F_t^\T  \Sigma_F^\I \Sigma_\Lambda^\I \covII_{\Lambda,j} \Sigma_\Lambda^\I \Sigma_F^\I F_t \\
					& \quad-  2 \cdot \frac{\delta_{NT}}{T}  \Lambda_j^\T \Sigma_{\Lambda,t}^\I   \covIII_{\Lambda, F, j, t} \Sigma_\Lambda^\I \Sigma_F^\I F_t \Big) \Big) 
					\quad \mathcal{G}^t-\text{stably}
				\end{align*}
				$ \covIII_{\Lambda, F, j, t} = g^{\cov}_{j,t}(\Lambda_j, F_t)$, and the function $ g^{\cov}_{j,t}(\cdot,\cdot)$ is defined in Assumption \ref{ass:mom-clt}.\ref{ass:asy-normal-add-term-thm-loading},
				or equivalently,
				\begin{align*}
					\sqrt{\delta_{NT}} \Sigma_{C,jt}^{-1/2} (\tilde C_{jt} - C_{jt}) \rightarrow& \calN ( 0, 1 ) 
				\end{align*}
				for 
				\begin{align*}
					\Sigma_{C,jt} =&  \frac{\delta_{NT}}{N} \Lambda_j^\T \Sigma_{\Lambda,t}^\I \covI_{F,t} \Sigma_{\Lambda,t}^\I \Lambda_j + \frac{\delta_{NT}}{T} F_t^\T \Sigma_F^\I \Sigma_\Lambda^\I \covI_{\Lambda,j}\Sigma_\Lambda^\I \Sigma_F^\I F_t \\
					& \quad + \frac{\delta_{NT}}{T}  \Lambda_j^\T \Sigma_{\Lambda,t}^\I  \covII_{F,t}   \Sigma_{\Lambda,t}^\I \Lambda_j  +  \frac{\delta_{NT}}{T} F_t^\T  \Sigma_F^\I \Sigma_\Lambda^\I \covII_{\Lambda,j} \Sigma_\Lambda^\I \Sigma_F^\I F_t -  2 \cdot \frac{\delta_{NT}}{T}  \Lambda_j^\T \Sigma_{\Lambda,t}^\I   \covIII_{\Lambda, F, j, t} \Sigma_\Lambda^\I \Sigma_F^\I F_t 
				\end{align*}
			\end{proof}
			
			\subsection{Proof of Theorem \ref{theorem:asy-normal}: Asymptotic Distribution of Probability Weighed Estimator}
			
			For notation convenience, we use the notation $\spsi = P(W_{it} = 1|S_i) $ throughout the proof of Theorem \ref{theorem:asy-normal}.
			\subsubsection{Proof of Theorem \ref{theorem:asy-normal}.1}
			\begin{lemma}\label{lemma:f-est-error-times-f-and-e}
				Under Assumptions \ref{ass:obs}, \ref{ass:factor-model}, \ref{ass:factor-model-conditional}, and \ref{ass:mom-clt-conditional}, we have
				\begin{enumerate}
					\item $\frac{1}{N} \sum_{i = 1}^N \frac{W_{it}}{\spsi}  \Lp \tilde \Lambda_i -  H_i \Lambda_i \Rp  e_{it}  = O_P \Lp \frac{1}{\delta_{NT}} \Rp$
					\item $\frac{1}{N} \sum_{i = 1}^N \frac{W_{it}}{\spsi}  \Lp \tilde \Lambda_i -  H \Lambda_i \Rp  e_{it}  = O_P \Lp \frac{1}{\delta_{NT}} \Rp$
					\item $\frac{1}{N} \sum_{i = 1}^N \frac{W_{it}}{\spsi}  \Lp \tilde \Lambda_i - H_i \Lambda_i  \Rp  \Lambda_i^\T = O_P \Lp \frac{1}{\delta_{NT}} \Rp$.
				\end{enumerate}
			\end{lemma}
			
			\begin{proof}[Proof of Lemma \ref{lemma:f-est-error-times-f-and-e}]
				\begin{enumerate}
					\item $\frac{1}{N} \sum_{i = 1}^N \frac{W_{it}}{\spsi}  \Lp \tilde \Lambda_i -  H_i \Lambda_i \Rp  e_{it}$ has the decomposition 
					\begin{eqnarray*}
						&& \frac{1}{N} \sum_{i = 1}^N \frac{W_{it}}{\spsi} \Lp \tilde \Lambda_i - H_i \Lambda_i \Rp e_{it} \\
						&=& \tilde D^\I \Big[ \underbrace{ \frac{1}{N^2} \sum_{i = 1}^N \sum_{l = 1}^N \frac{W_{it}}{\spsi}  \tilde \Lambda_l \gamma(l,i) e_{it} }_{\text{\RNum{1}} }  + \underbrace{\frac{1}{N^2} \sum_{i = 1}^N \sum_{l = 1}^N \frac{W_{it}}{\spsi}  \tilde \Lambda_l \zeta_{li} e_{it}}_{\text{\RNum{2}}  }  \\
						&& + \underbrace{\frac{1}{N^2} \sum_{i = 1}^N  \sum_{l = 1}^N \frac{W_{it}}{\spsi}  \tilde \Lambda_l \eta_{li} e_{it} }_{\text{\RNum{3}}}  + \underbrace{\frac{1}{N^2} \sum_{i = 1}^N  \sum_{l = 1}^N  \frac{W_{it}}{\spsi} \tilde \Lambda_l \xi_{li} e_{it} }_{\text{\RNum{4}}}   \Big].
					\end{eqnarray*}
					We decompose the term $\text{\RNum{1}} $ further into two parts: 
					\begin{align*}
						\text{\RNum{1}} &= \underbrace{\frac{1}{N^2} \sum_{i = 1}^N \sum_{l = 1}^N  \frac{W_{it}}{\spsi} (\tilde \Lambda_l - H \Lambda_l) \gamma(l,i) e_{it}}_{\text{\RNum{1}}_1}  + H \cdot \underbrace{\frac{1}{N^2} \sum_{i = 1}^N \sum_{l = 1}^N  \frac{W_{it}}{\spsi}  \Lambda_l \gamma(l,i) e_{it}}_{\text{\RNum{1}}_2}.
					\end{align*}
					The first term $\text{\RNum{1}}_1$ is bounded by
					\begin{align*}
						\norm{\text{\RNum{1}}_1} &\leq \frac{1}{\sqrt{N}} \underbrace{\Big( \frac{1}{N} \sum_{l = 1}^N \norm{\tilde \Lambda_l - H \Lambda_l}^2 \Big)^{1/2}}_{O_P \Lp \frac{1}{\sqrt{\delta_{NT}}} \Rp}     \Big( \underbrace{\frac{1}{N}\sum_{i = 1}^N \sum_{l = 1}^N \frac{W_{it}}{(\spsi)^2} |\gamma(l,i)|^2}_{\leq \frac{M}{\underline{p}^2} \text{ from Lemma \ref{lemma:prep-consistency}.1}}  \cdot \underbrace{\frac{1}{N}\sum_{i = 1}^N  e_{it}^2}_{O_P(1)}  \Big)^{1/2} \\
						&= O_P \Lp \frac{1}{\sqrt{N\delta_{NT}}} \Rp.
					\end{align*}
					The second term $\text{\RNum{1}}_2$ satisfies 
					\begin{align*}
						\+E[\norm{\text{\RNum{1}}_2}] &\leq  \frac{1}{N^2} \sum_{i = 1}^N \sum_{l = 1}^N |\gamma(l,i)| \underbrace{\+E \Big[\frac{W_{it}}{\spsi} \norm{\Lambda_l } \Big]}_{\leq \bar \Lambda}  \underbrace{\+E[| e_{it}|]}_{\leq M}   = O\Lp \frac{1}{N } \Rp.
					\end{align*}
					Hence, it holds that $\text{\RNum{1}} = O_P \Lp \frac{1}{\sqrt{N\delta_{NT}}} \Rp + O_P \Lp \frac{1}{N } \Rp= O_P \Lp \frac{1}{\sqrt{N\delta_{NT}}} \Rp$. For the term $\text{\RNum{2}} $, we have the decomposition
					\begin{align*}
						\text{\RNum{2}} &= \underbrace{\frac{1}{N^2} \sum_{i = 1}^N \sum_{l = 1}^N  \frac{W_{it}}{\spsi} (\tilde \Lambda_l - H \Lambda_l) \zeta_{li} e_{it}}_{\text{\RNum{2}}_1}  + H \cdot \underbrace{\frac{1}{N^2} \sum_{i = 1}^N \sum_{l = 1}^N  \frac{W_{it}}{\spsi} \Lambda_l \zeta_{li} e_{it}}_{\text{\RNum{2}}_2}      .
					\end{align*}
					For the second term $\text{\RNum{2}}_2$, we have 
					\begin{align*}
						\text{\RNum{2}}_2 &=  \frac{1}{N}  \sum_{i = 1}^N \underbrace{\Ls \frac{1}{N} \sum_{l = 1}^N  \frac{W_{it}}{\spsi} \Lambda_l \Big[ \frac{1}{|\tlq_{il}|} \sum_{s \in \tlq_{il}} e_{is} e_{ls} - \+E[e_{is} e_{ls}] \Big] \Rs}_{z_i}   e_{it}   .
					\end{align*}
					Assumption \ref{ass:mom-clt-conditional}.1 implies $\+E \norm{z_i}^2= O\Lp \frac{1}{NT} \Rp$. Hence, we obtain
					\[\+E \big[ \norm{\text{\RNum{2}}_2} \big] \leq \frac{1}{N} \sum_{i = 1}^N \+E \norm{z_i e_{it}} \leq \frac{1}{N} \sum_{i = 1}^N (\+E \big[ \norm{z_i}^2 \big] \+E e_{it}^2)^{1/2}  = O\Lp \frac{1}{\sqrt{NT}} \Rp. \]
					We conclude that $\text{\RNum{2}}_2 = O_P \Lp \frac{1}{\sqrt{NT}} \Rp$. For the first term $\text{\RNum{2}}_1$, we have the bound
					\begin{align*}
						\norm{\text{\RNum{2}}_1} &\leq \underbrace{\Lp \frac{1}{N} \sum_{l = 1}^N \norm{\tilde \Lambda_l - H \Lambda_l}^2  \Rp^{1/2}}_{O_P \Lp \frac{1}{ \sqrt{\delta_{NT}}} \Rp }  \underbrace{\Lp \frac{1}{N} \sum_{l = 1}^N \Big( \frac{1}{N} \sum_{i = 1}^N \frac{W_{it}}{\spsi} \zeta_{li} e_{it}  \Big)^2  \Rp^{1/2}}_{O_P \Lp \frac{1}{ \sqrt{T}} \Rp} , 
					\end{align*}
					where the second term is $O_P \Lp \frac{1}{ \sqrt{T}} \Rp$ following from
					\begin{align*}
						\frac{1}{N} \sum_{l = 1}^N \Big( \frac{1}{N} \sum_{i = 1}^N \frac{W_{it}}{\spsi} \zeta_{li} e_{it}  \Big)^2 
						\leq& \frac{1}{N} \sum_{l = 1}^N \Big( \frac{1}{N} \sum_{i = 1}^N \frac{W_{it}}{(\spsi)^2}  \zeta_{li}^2    \Big) \cdot \Big( \frac{1}{N} \sum_{i = 1}^N  e_{it}^2  \Big) \\
						\leq& \underbrace{\frac{1}{\underline{p}^2}  \frac{1}{N} \sum_{l = 1}^N \Big( \frac{1}{N} \sum_{i = 1}^N \zeta_{li}^2    \Big) \cdot \Big( \frac{1}{N} \sum_{i = 1}^N  e_{it}^2  \Big)}_{\substack{ =O_P \Lp \frac{1}{ T} \Rp \text{ follows from } \+E\Ls \frac{1}{N} \sum_{l = 1}^N \Big( \frac{1}{N} \sum_{i = 1}^N \zeta_{li}^2    \Big) \cdot \Big( \frac{1}{N} \sum_{i = 1}^N  e_{it}^2  \Big)\Rs  \\ =  \frac{1}{N^3} \sum_{l = 1}^N  \sum_{i = 1}^N   \sum_{j = 1}^N \+E[\zeta_{li}^2 e_{jt}^2 ] \\ \leq  \frac{1}{N^3} \sum_{l = 1}^N  \sum_{i = 1}^N   \sum_{j = 1}^N (\+E[\zeta_{li}^4] \+E[ e_{jt}^4 ])^{1/2}
								=  O\Lp \frac{1}{ T} \Rp }  } 
					\end{align*}
					Hence, we obtain the rate $\text{\RNum{2}} = O_P \Lp \frac{1}{\sqrt{T\delta_{NT}}} \Rp$. We aslo decompose the third term $\text{\RNum{3}}$ into two parts 
					\begin{align*}
						\text{\RNum{3}} &= \underbrace{ \frac{1}{N^2} \sum_{i = 1}^N  \sum_{l = 1}^N \frac{W_{it}}{\spsi}  (\tilde \Lambda_l - H \Lambda_l)\eta_{li} e_{it}}_{\text{\RNum{3}}_1  }  + H \cdot \underbrace{\frac{1}{N^2} \sum_{i = 1}^N  \sum_{l = 1}^N \frac{W_{it}}{\spsi}  \Lambda_l \eta_{li} e_{it}}_{\text{\RNum{3}}_2}    . 
					\end{align*}
					The first term $\text{\RNum{3}}_1 $ is bounded by
					\begin{align*}
						\norm{\text{\RNum{3}}_1} &\leq     \underbrace{\Lp \frac{1}{N} \sum_{l = 1}^N \norm{\tilde \Lambda_l - H \Lambda_l}^2  \Rp^{1/2}}_{O_P \Lp \frac{1}{ \sqrt{\delta_{NT}}} \Rp } \Lp \frac{1}{N} \sum_{l = 1}^N \Big( \frac{1}{N} \sum_{i = 1}^N \frac{W_{it}}{\spsi} \eta_{li} e_{it}  \Big)^2  \Rp^{1/2}, 
					\end{align*}
					and the second term $\frac{1}{N} \sum_{l = 1}^N \Big( \frac{1}{N} \sum_{i = 1}^N \frac{W_{it}}{\spsi} \eta_{li} e_{it}  \Big)^2 $ satisfies
					\begin{align*}
						& \frac{1}{N} \sum_{l = 1}^N \Big( \frac{1}{N} \sum_{i = 1}^N \frac{W_{it}}{\spsi} \eta_{li} e_{it}  \Big)^2 \\
						\leq&\frac{1}{N} \sum_{l = 1}^N  \Big( \frac{1}{N} \sum_{i = 1}^N \eta_{li}^2   \Big) \cdot   \Big( \frac{1}{N} \sum_{i = 1}^N \frac{W_{it}}{(\spsi)^2} e_{it}^2   \Big)\\
						\leq& \frac{1}{N} \sum_{l = 1}^N \Big(\frac{\norm{\Lambda_l}^2 }{N} \sum_{i = 1}^N \norm{\frac{1}{|\tlq_{li}|} \sum_{s \in \tlq_{li}}  F_s e_{is}   }^2  \Big) \Big( \frac{1}{N} \sum_{i = 1}^N \frac{W_{it}}{(\spsi)^2} e_{it}^2   \Big)  =  O_P \Lp \frac{1}{ T} \Rp,
					\end{align*}
					following from
					\begin{align*}
						& \+E \Ls \frac{1}{N} \sum_{l = 1}^N \Big(\frac{\norm{\Lambda_l}^2 }{N} \sum_{i = 1}^N \norm{\frac{1}{|\tlq_{li}|} \sum_{s \in \tlq_{li}}  F_s e_{is}   }^2  \Big) \Big( \frac{1}{N} \sum_{i = 1}^N \frac{W_{it}}{(\spsi)^2} e_{it}^2   \Big)   \Rs   \\
						=&  \frac{1}{N} \sum_{l = 1}^N  \+E[\norm{\Lambda_l}^2 ] \frac{1}{N^2} \sum_{i = 1}^N \sum_{j=1}^N \+E \Ls  \norm{\frac{1}{|\tlq_{li}|} \sum_{s \in \tlq_{li}}  F_s e_{is}   }^2 \frac{W_{jt}}{P(W_{jt} = 1|S_j)^2} e_{jt}^2  \Rs \\
						\leq& \frac{1}{N} \sum_{l = 1}^N  \+E[\norm{\Lambda_l}^2] \frac{1}{N^2} \sum_{i = 1}^N \sum_{j=1}^N \frac{1}{P(W_{jt} = 1|S_j)} \bigg( \underbrace{\+E \Ls  \norm{\frac{1}{|\tlq_{li}|} \sum_{s \in \tlq_{li}}  F_s e_{is}   }^4  \Rs }_{\leq \frac{M}{|\tlq_{li}|^2}}   \+E \Ls   e_{jt}^4  \Rs \bigg)^{1/2} \\
						\leq& \frac{1}{\underline{p} } \cdot \max_{li} \frac{1}{|\tlq_{li}|} \cdot \bar{\Lambda} = O\Lp \frac{1}{ T} \Rp.
					\end{align*}
					Hence, it holds that $\text{\RNum{3}}_1 = O_P \Lp \frac{1}{\sqrt{T\delta_{NT}} } \Rp$. Next let us consider $\text{\RNum{3}}_2$: 
					\begin{align*}
						\text{\RNum{3}}_2 &= \frac{1}{N^2}  \sum_{l = 1}^N \Lambda_l \Lambda_l^\T  \sum_{i = 1}^N \frac{W_{it}}{\spsi}   \frac{1}{|\tlq_{li}|} \sum_{s \in \tlq_{li}}  F_s e_{is} e_{it} \\
						&= \underbrace{\frac{1}{N^2}  \sum_{l = 1}^N \Lambda_l \Lambda_l^\T  \sum_{i = 1}^N \frac{W_{it}}{\spsi}   \frac{1}{|\tlq_{li}|} \sum_{s \in \tlq_{li}}  F_s \+E[e_{is} e_{it}] }_{\text{\RNum{3}}_{2,1}} \\
						&\quad + \underbrace{\frac{1}{N^2}  \sum_{l = 1}^N \Lambda_l \Lambda_l^\T  \sum_{i = 1}^N \frac{W_{it}}{\spsi}   \frac{1}{|\tlq_{li}|} \sum_{s \in \tlq_{li}}  F_s \big( e_{is} e_{it} - \+E[e_{is} e_{it}]  \big)}_{\text{\RNum{3}}_{2,2}}\\
					\end{align*}
					For $\text{\RNum{3}}_{2,1}$ we obtain 
					\begin{align*}
						\norm{\text{\RNum{3}}_{2,1}}^2 =& \Bigg(\frac{1}{N} \sum_{l = 1}^N \norm{\Lambda_l}^4 \Bigg) \Bigg(\frac{1}{N} \sum_{l = 1}^N \norm{\frac{1}{N} \sum_{i = 1}^N    \frac{W_{it}}{|\tlq_{li}|} \sum_{s \in \tlq_{li}}  F_s \+E[e_{is} e_{it}]  }^2\Bigg) \\
						\leq& \underbrace{\Bigg(\frac{1}{N} \sum_{l = 1}^N \norm{\Lambda_l}^4 \Bigg)  }_{O_P(1)} \Bigg(\frac{1}{N} \sum_{l = 1}^N \bigg(\frac{1}{NT} \sum_{i = 1}^N \sum_{s \in \tlq_{li}}  \norm{ \frac{T W_{it}}{\spsi |\tlq_{li}|} F_s}^2    \bigg) \bigg( \frac{1}{NT} \sum_{i = 1}^N \sum_{s \in \tlq_{li}}  (\+E[e_{is} e_{it}] )^2\bigg) \Bigg) .
					\end{align*}
					Assumption \ref{ass:factor-model}.3(d) implies
					\begin{align*}
						\sum_{i = 1}^N \sum_{s \in \tlq_{li}}  (\+E[e_{is} e_{it}] )^2 \leq \sum_{i = 1}^N \sum_{s \in \tlq_{li}}  |\+E[e_{is} e_{it}]| \leq  \sum_{i = 1}^N \sum_{t = 1}^T  |\+E[e_{is} e_{it}]| \leq M.
					\end{align*}
					Moreover, it holds that
					\begin{align*}
						\+E \Bigg[\frac{1}{N} \sum_{l = 1}^N \bigg(\frac{1}{NT} \sum_{i = 1}^N \sum_{s \in \tlq_{li}}  \norm{ \frac{T W_{it}}{\spsi |\tlq_{li}|} F_s}^2    \bigg)   \Bigg] =& \frac{T}{N^2 } \sum_{l = 1}^N \sum_{i = 1}^N \frac{1}{|\tlq_{li}|^2} \sum_{s \in \tlq_{li}} \frac{ \+E [W_{it}]}{\spsi}  \+E  \norm{ F_s}^2  \leq M
					\end{align*}
					Thus, we obtain the rate $\norm{\text{\RNum{3}}_{2,1}} = O_P \Big(\frac{1}{\sqrt{NT}}\big)$. For $\text{\RNum{3}}_{2,2}$, we have 
					\begin{align*}
						\norm{\text{\RNum{3}}_{2,2}}^2 =& \Bigg(\frac{1}{N} \sum_{l = 1}^N \norm{\Lambda_l}^4 \Bigg) \Bigg(\frac{1}{N} \sum_{l = 1}^N \norm{\frac{1}{N} \sum_{i = 1}^N    \frac{W_{it}}{\spsi |\tlq_{li}|} \sum_{s \in \tlq_{li}}  F_s (e_{is} e_{it} - \+E[e_{is} e_{it}])  }^2\Bigg) .
					\end{align*}
					Using Assumption \ref{ass:mom-clt-conditional}.1, we conclude
					\begin{align*}
						& \+E \Bigg[ \frac{1}{N} \sum_{l = 1}^N \norm{\frac{1}{N} \sum_{i = 1}^N    \frac{W_{it}}{\spsi |\tlq_{li}|} \sum_{s \in \tlq_{li}}  F_s (e_{is} e_{it} - \+E[e_{is} e_{it}])  }^2 \Bigg] \\
						=&  \frac{1}{N} \sum_{l = 1}^N  \+E \Bigg[ \norm{\frac{1}{N} \sum_{i = 1}^N    \frac{W_{it}}{\spsi |\tlq_{li}|} \sum_{s \in \tlq_{li}}  F_s (e_{is} e_{it} - \+E[e_{is} e_{it}])  }^2\Bigg] \leq \frac{M}{NT},
					\end{align*}
					and hence $\text{\RNum{3}} = O_P \Lp \frac{1}{\sqrt{T\delta_{NT}}} \Rp$. The last term has the rate $\text{\RNum{4}} = O_P \Lp \frac{1}{\sqrt{T\delta_{NT}}} \Rp$, which can be shown with similar arguments.

					\item 
					\begin{eqnarray*}
						\frac{1}{N} \sum_{i = 1}^N \frac{W_{it}}{\spsi}   \Lp \tilde \Lambda_i -  H \Lambda_i \Rp  e_{it} &=& \underbrace{\frac{1}{N} \sum_{i = 1}^N \frac{W_{it}}{\spsi}   \Lp \tilde \Lambda_i -  H_i \Lambda_i \Rp  e_{it} }_{=O_P \Lp \frac{1}{\delta_{NT}} \Rp \text{ from Lemma \ref{lemma:f-est-error-times-f-and-e}.1}} \\
						&& + \underbrace{\frac{1}{N} \sum_{i = 1}^N \frac{W_{it}}{\spsi}   (H_i - H) \Lambda_i e_{it}}_{\Delta}
					\end{eqnarray*}
					The term $\Delta$ has 
					\begin{align*}
						\norm{\Delta}^2 &=  \norm{ \frac{1}{N^2}  \sum_{l = 1}^N \Lambda_l  \Lambda_l^\T  \Bigg[ \sum_{i =1}^N  \frac{W_{it}}{\spsi}  \Lp \frac{1}{T} \sum_{s=1}^T F_s F_s^\T - \frac{1}{|\tlq_{li}|} \sum_{s \in \tlq_{li}} F_s F_s^\T \Rp \Lambda_i e_{it} \Bigg] }^2 \\
						&\leq \Bigg( \underbrace{\frac{1}{N} \sum_{l = 1}^N \norm{\Lambda_l }^4}_{O_P(1)} \Bigg) \Bigg( \frac{1}{N} \sum_{l = 1}^N  \Big\lVert\underbrace{\frac{1}{N} \sum_{i =1}^N  \frac{W_{it}}{\spsi}  \Lp \frac{1}{T} \sum_{s=1}^T F_s F_s^\T - \frac{1}{|\tlq_{li}|} \sum_{s \in \tlq_{li}} F_s F_s^\T \Rp \Lambda_i e_{it}}_{z_l} \Big\rVert^2 \Bigg).
					\end{align*}
					Assumption \ref{ass:mom-clt-conditional}.\ref{ass:add-mom-three-sum-lam-err} implies $ \+E[\norm{\sqrt{NT} z_l}^2] \leq M$ and thus
					\[\+E \Big[ \frac{1}{N} \sum_{l = 1}^N \norm{z_l}^2\Big] = \frac{1}{N}\sum_{l = 1}^N \+E[\norm{z_l}^2] \leq O \Big(\frac{1}{NT} \Big) \]
					and $\Delta = O_P \Big(\frac{1}{\sqrt{NT}} \Big)$. 
					Hence, we conclude that 
					\begin{align*}
						\frac{1}{N} \sum_{i = 1}^N \frac{W_{it}}{\spsi}   \Lp \tilde \Lambda_i -  H \Lambda_i \Rp  e_{it} &=    O_P \Lp \frac{1}{\delta_{NT}} \Rp +  O_P \Lp \frac{1}{\sqrt{NT}}  \Rp = O_P \Lp \frac{1}{\delta_{NT}} \Rp
					\end{align*}
					\item $\frac{1}{N} \sum_{i = 1}^N \frac{W_{it}}{\spsi} \Lp \tilde \Lambda_i - H_i \Lambda_i  \Rp  \Lambda_i^\T $ has the decomposition
					\begin{eqnarray*}
						\frac{1}{N} \sum_{i = 1}^N \Lp \frac{W_{it}}{\spsi} \tilde \Lambda_i - H_i \Lambda_i \Rp \Lambda_i^\T  
						&=& \tilde D^\I \Big[ \underbrace{ \frac{1}{N^2} \sum_{i = 1}^N \sum_{l = 1}^N \frac{W_{it}}{\spsi} \tilde \Lambda_l \Lambda_i^\T \gamma(l,i)  }_{\text{\RNum{1}} }  + \underbrace{\frac{1}{N^2} \sum_{i = 1}^N \sum_{l = 1}^N \frac{W_{it}}{\spsi} \tilde \Lambda_l \Lambda_i^\T \zeta_{li}  }_{\text{\RNum{2}}  }  \\
						&& + \underbrace{\frac{1}{N^2} \sum_{i = 1}^N  \sum_{l = 1}^N \frac{W_{it}}{\spsi}  \tilde \Lambda_l \Lambda_i^\T \eta_{li}  }_{\text{\RNum{3}}}  + \underbrace{\frac{1}{N^2} \sum_{i = 1}^N  \sum_{l = 1}^N \frac{W_{it}}{\spsi}  \tilde \Lambda_l \Lambda_i^\T \xi_{li} }_{\text{\RNum{4}}}   \Big]
					\end{eqnarray*}
					We decompose term $\text{\RNum{1}} $ further into two parts 
					\begin{align*}
						\text{\RNum{1}} &= \underbrace{\frac{1}{N^2} \sum_{i = 1}^N \sum_{l = 1}^N  \frac{W_{it}}{\spsi} (\tilde \Lambda_l - H \Lambda_l)  \Lambda_i^\T \gamma(l,i) }_{\text{\RNum{1}}_1}  + H \cdot \underbrace{\frac{1}{N^2} \sum_{i = 1}^N \sum_{l = 1}^N  \frac{W_{it}}{\spsi} \Lambda_l  \Lambda_i^\T \gamma(l,i) }_{\text{\RNum{1}}_2}.
					\end{align*}
					The first term $\text{\RNum{1}}_1$ is bounded by
					\begin{align*}
						\norm{\text{\RNum{1}}_1} &\leq \frac{1}{\sqrt{N}} \underbrace{\Big( \frac{1}{N} \sum_{l = 1}^N \norm{\tilde \Lambda_l - H \Lambda_l}^2 \Big)^{1/2}}_{O_P \Lp \frac{1}{\sqrt{\delta_{NT}}} \Rp}     \Big( \underbrace{\frac{1}{N}\sum_{i = 1}^N \sum_{l = 1}^N  |\gamma(l,i)|^2}_{\leq M \text{ from Lemma \ref{lemma:prep-consistency}.1}}  \cdot \underbrace{\frac{1}{N}\sum_{i = 1}^N \frac{W_{it}}{(\spsi)^2} \norm{ \Lambda_i}^2 }_{\leq \frac{1}{N \underline{p}^2} \sum_{i = 1}^N \norm{ \Lambda_i}^2 = O_P(1) }  \Big)^{1/2} \\
						&= O_P \Lp \frac{1}{\sqrt{N\delta_{NT}}} \Rp
					\end{align*}
					The second term $\text{\RNum{1}}_2$ satisfies 
					\begin{align*}
						\+E[\norm{\text{\RNum{1}}_2}  ] &\leq  \frac{1}{N^2} \sum_{i = 1}^N \sum_{l = 1}^N |\gamma(l,i)| \underbrace{\+E \Big[\frac{W_{it}}{\spsi}  \Big]}_{= 1}  \underbrace{\+E \Big[ \norm{\Lambda_l } \norm{\Lambda_i }  \Big]}_{\leq \bar \Lambda}  \underbrace{\+E[| e_{it}|]}_{\leq M}   = O\Lp \frac{1}{N } \Rp.
					\end{align*}
					Hence, we obtain $\text{\RNum{1}} = O_P \Lp \frac{1}{\sqrt{N\delta_{NT}}} \Rp + O_P \Lp \frac{1}{N } \Rp= O_P \Lp \frac{1}{\sqrt{N\delta_{NT}}} \Rp$. 
					
					For the term $\text{\RNum{2}} $, we have the decomposition
					\begin{align*}
						\text{\RNum{2}} &= \underbrace{\frac{1}{N^2} \sum_{i = 1}^N \sum_{l = 1}^N \frac{W_{it}}{\spsi} (\tilde \Lambda_l - H \Lambda_l) \Lambda_i ^\T \zeta_{li} }_{\text{\RNum{2}}_1}  + H \cdot \underbrace{\frac{1}{N^2} \sum_{i = 1}^N \sum_{l = 1}^N  \frac{W_{it}}{\spsi} \Lambda_l \Lambda_i^\T  \zeta_{li} }_{\text{\RNum{2}}_2}      .
					\end{align*}
					For the second term $\text{\RNum{2}}_2$, we obtain 
					\begin{align*}
						\text{\RNum{2}}_2 &=  \frac{1}{N}  \sum_{l = 1}^N  \Lambda_l  \underbrace{\Ls \frac{1}{N}  \sum_{i = 1}^N \frac{W_{it}}{\spsi} \Lambda_i^\T  \Big[ \frac{1}{|\tlq_{il}|} \sum_{s \in \tlq_{il}} e_{is} e_{ls} - \+E[e_{is} e_{ls}] \Big] \Rs}_{z_i}  
					\end{align*}
					Assumption \ref{ass:mom-clt-conditional}.1 implies $\+E \big[ \norm{z_i}^2  \big] = O\Lp \frac{1}{NT} \Rp$. This yields 
					\[\+E \big[ \norm{\text{\RNum{2}}_2}\big] \leq \frac{1}{N} \sum_{l = 1}^N  \+E \Big[\norm{\Lambda_l} \norm{z_l} \Big] \leq  \frac{1}{N} \sum_{l = 1}^N \Big( \+E \Big[\norm{\Lambda_l}^2   \Big] \+E \Big[\norm{z_l}^2  \Big] \Big)^{1/2} = O\Lp \frac{1}{\sqrt{NT}} \Rp. \]
					Hence, it holds that $\text{\RNum{2}}_2 = O_P \Lp \frac{1}{\sqrt{NT}} \Rp$. For the first term $\text{\RNum{2}}_1$, we have
					\begin{align*}
						\norm{\text{\RNum{2}}_1} &\leq \underbrace{\Lp \frac{1}{N} \sum_{l = 1}^N \norm{\tilde \Lambda_l - H \Lambda_l}^2  \Rp^{1/2}}_{O_P \Lp \frac{1}{ \sqrt{\delta_{NT}}} \Rp } \Lp \frac{1}{N} \sum_{l = 1}^N  \norm{\frac{1}{N} \sum_{i = 1}^N \frac{W_{it}}{\spsi} \Lambda_i \zeta_{li} }^2  \Rp^{1/2}  \\
						& \leq O_P \Lp \frac{1}{ \sqrt{\delta_{NT}}} \Rp \cdot \bigg( \underbrace{\max_i \frac{W_{it}}{(\spsi)^2}}_{\leq \frac{1}{\underline{p}^2}} \cdot \underbrace{\Big( \frac{1}{N} \sum_{i = 1}^N  \norm{\Lambda_i}^2 \Big) }_{O_P(1)}  \underbrace{\Big( \frac{1}{N^2}  \sum_{l = 1}^N  \sum_{i = 1}^N   \zeta_{li}^2 \Big)}_{O_P \Lp \frac{1}{ T} \Rp}  \bigg)^{1/2} \\
						&=O_P \Lp \frac{1}{ \sqrt{T\delta_{NT}}} \Rp.
					\end{align*}
					In conclusion, it holds that $\text{\RNum{2}} = O_P \Lp \frac{1}{\sqrt{T\delta_{NT}}} \Rp$. 
					
					For the third term $\text{\RNum{3}}$, we also have a decomposition into two parts
					\begin{align*}
						\text{\RNum{3}} &= \underbrace{ \frac{1}{N^2} \sum_{i = 1}^N  \sum_{l = 1}^N \frac{W_{it}}{\spsi}   (\tilde \Lambda_l - H \Lambda_l)\Lambda_i^\T \eta_{li} }_{\text{\RNum{3}}_1  }  + H \cdot \underbrace{\frac{1}{N^2} \sum_{i = 1}^N  \sum_{l = 1}^N \frac{W_{it}}{\spsi}  \Lambda_l \Lambda_i^\T \eta_{li}}_{\text{\RNum{3}}_2}    . 
					\end{align*}
					For the first term $\text{\RNum{3}}_1 $, we obtain the bound 
					\begin{align*}
						\norm{\text{\RNum{3}}_1} &\leq     \underbrace{\Lp \frac{1}{N} \sum_{l = 1}^N \norm{\tilde \Lambda_l - H \Lambda_l}^2  \Rp^{1/2}}_{O_P \Lp \frac{1}{ \sqrt{\delta_{NT}}} \Rp } \Lp \frac{1}{N} \sum_{l = 1}^N \norm{\frac{1}{N} \sum_{i = 1}^N \frac{W_{it}}{\spsi}  \Lambda_i^\T \eta_{li}}^2  \Rp^{1/2}, 
					\end{align*}
					and the second term $\frac{1}{N} \sum_{l = 1}^N \norm{\frac{1}{N} \sum_{i = 1}^N \frac{W_{it}}{\spsi}  \Lambda_i^\T \eta_{li}}^2  $ satisfies
					\begin{align*}
						& \+E \Ls \frac{1}{N} \sum_{l = 1}^N \norm{\frac{1}{N} \sum_{i = 1}^N \frac{W_{it}}{\spsi}  \Lambda_i^\T \eta_{li}}^2  \Rs \\ =& \+E \Ls \frac{1}{N} \sum_{l = 1}^N \norm{\frac{1}{N} \sum_{i = 1}^N \frac{W_{it}}{\spsi}  \Lambda_i^\T   \frac{1}{|\tlq_{li}|} \sum_{s \in \tlq_{li}}\Lambda_l^\T   F_s e_{is}}^2 \Rs \\
						=&\frac{1}{N} \sum_{l = 1}^N  \+E \Ls \norm{\Lambda_l^\T \frac{1}{N} \sum_{i = 1}^N  \frac{W_{it}}{\spsi} \cdot  \frac{1}{|\tlq_{li}|} \sum_{s \in \tlq_{li}}   F_s \Lambda_i^\T e_{is}}^2 \Rs \\
						\leq&\frac{1}{N} \sum_{l = 1}^N  \max_i \frac{1}{(\spsi)^2}  \cdot \+E \Ls \Lp \frac{1}{N} \sum_{i = 1}^N \norm{\Lambda_l^\T  \Lambda_i} \norm{ \frac{1}{|\tlq_{li}|} \sum_{s \in \tlq_{li}}   F_s e_{is}} \Rp^2  \Rs \\
						\leq&\frac{1}{N^3 \underline{p}^2} \sum_{l = 1}^N  \sum_{i = 1}^N \sum_{j = 1}^N \underbrace{\+E \big[\norm{\Lambda_l^\T  \Lambda_i} \norm{\Lambda_l^\T  \Lambda_j}\big] }_{\leq \bar{\Lambda}}  \Bigg(\underbrace{\+E\norm{ \frac{1}{|\tlq_{li}|} \sum_{s \in \tlq_{li}}   F_s e_{is}}^2}_{O\Lp \frac{1}{ T} \Rp} \underbrace{\+E \norm{ \frac{1}{|\tlq_{lj}|} \sum_{s \in \tlq_{lj}}   F_s e_{js}}^2}_{O\Lp \frac{1}{ T} \Rp}    \Bigg)^{1/2}  \\ 
						=&  O\Lp \frac{1}{ T} \Rp.
					\end{align*}
					Hence, we have the rate $\text{\RNum{3}}_1 = O_P \Lp \frac{1}{\sqrt{T\delta_{NT}} } \Rp$. Next let us consider $\text{\RNum{3}}_2$: 
					\begin{align*}
						\norm{\text{\RNum{3}}_2}^2 &= \norm{\frac{1}{N^2}  \sum_{l = 1}^N \Lambda_l \Lambda_l^\T  \sum_{i = 1}^N  \frac{W_{it}}{\spsi} \cdot \frac{1}{|\tlq_{li}|} \sum_{s \in \tlq_{li}}  F_s \Lambda_i^\T  e_{is} }   \\
						&\leq \underbrace{\Lp \frac{1}{N}  \sum_{l = 1}^N  \norm{\Lambda_l}^4 \Rp}_{O_P(1)} \underbrace{\Lp \frac{1}{N}  \sum_{l = 1}^N \norm{\frac{1}{N} \sum_{i = 1}^N \frac{W_{it}}{\spsi} \cdot  \frac{1}{|\tlq_{li}|} \sum_{s \in \tlq_{li}}   F_s \Lambda_i^\T e_{is} }^2    \Rp}_{\substack{= O\Lp \frac{1}{ NT} \Rp \text{ from } \+E \Ls \frac{1}{N}  \sum_{l = 1}^N \norm{\frac{1}{N} \sum_{i = 1}^N  \frac{W_{it}}{\spsi}  \frac{1}{|\tlq_{li}|} \sum_{s \in \tlq_{li}}   F_s \Lambda_i^\T e_{is} }^2  \Rs \\  =\frac{1}{N}  \sum_{l = 1}^N \+E \Ls \norm{\frac{1}{N} \sum_{i = 1}^N  \frac{W_{it}}{\spsi}  \frac{1}{|\tlq_{li}|} \sum_{s \in \tlq_{li}}  F_s \Lambda_i^\T  e_{is} }^2   \Rs= O\Lp \frac{1}{ NT} \Rp\\ \text{from Assumption \ref{ass:mom-clt-conditional}.2} }  } .
					\end{align*}
					Hence, we conclude that $\text{\RNum{3}} = O_P \Lp \frac{1}{\sqrt{T\delta_{NT}}} \Rp$. The last term satisfies $\text{\RNum{4}} = O_P \Lp \frac{1}{\sqrt{T\delta_{NT}}} \Rp$, which can be shown similarly. 
					
				\end{enumerate}
			\end{proof}

			\begin{proof}[Proof of Theorem \ref{theorem:asy-normal}.1]
				We regress $Y_{it}$ on $\tilde \Lambda_i$ using the observed units at time $t$ (where $W_{it} = 1$)
				\begin{align*}
					\tilde F^S_t &= \Big(\frac{1}{\spsi} \sum_{i = 1}^N W_{it} \tilde \Lambda_i \tilde \Lambda_i^\T  \Big)^\I  \Big(\frac{1}{\spsi} \sum_{i = 1}^N W_{it} \tilde \Lambda_i Y_{it}    \Big) .
				\end{align*}
				We first analyze 
				\begin{align*}
					\tilde{F}^{S\dagger}_t &= \Big(\frac{1}{\spsi} \sum_{i = 1}^N W_{it} H \Lambda_i  \Lambda_i^\T H^\T  \Big)^\I  \Big(\frac{1}{\spsi}  \sum_{i = 1}^N W_{it} \tilde \Lambda_i Y_{it}    \Big)  .
				\end{align*}
				We have the following decomposition for $\tilde F_t$
				\begin{align*}
					\tilde{F}^{S\dagger}_t  =& \Big(\frac{1}{N }  \sum_{i = 1}^N   \frac{W_{it}}{\spsi}  H \Lambda_i  \Lambda_i^\T H^\T  \Big)^\I  \Big(\frac{1}{N}  \sum_{i = 1}^N  \frac{W_{it}}{\spsi} \tilde \Lambda_i (\Lambda_i^\T F_t + e_{it})    \Big)    \\
					=& (H^\I)^\T F_t + \underbrace{(H^\I)^\T \Big(\frac{1}{N }  \sum_{i = 1}^N  \frac{W_{it}}{\spsi}  \Lambda_i  \Lambda_i^\T  \Big)^\I  \Big(\frac{1}{N }  \sum_{i = 1}^N  \frac{W_{it}}{\spsi}  \Lambda_i  e_{it}    \Big) }_{\Delta_1}   \\
					& + \underbrace{\Big(\frac{1}{N }  \sum_{i = 1}^N  \frac{W_{it}}{\spsi}  H \Lambda_i  \Lambda_i^\T H^\T  \Big)^\I  \Big(\frac{1}{N }  \sum_{i = 1}^N  \frac{W_{it}}{\spsi}  (\tilde \Lambda_i - H \Lambda_i) \Lambda_i^\T F_t   \Big) }_{\Delta_2}   \\
					&+ \Big(\frac{1}{N}  \sum_{i = 1}^N  \frac{W_{it}}{\spsi}  H \Lambda_i  \Lambda_i^\T H^\T  \Big)^\I \underbrace{\Big(\frac{1}{N }  \sum_{i = 1}^N  \frac{W_{it}}{\spsi}  (\tilde \Lambda_i - H \Lambda_i)  e_{it}   \Big) }_{{O_P \Lp \frac{1}{\delta } \Rp \text{ from Lemma \ref{lemma:f-est-error-times-f-and-e}.2 }} } 
				\end{align*}
				For $\Delta_1$, $\frac{1}{\sqrt{N}} \sum_{i = 1}^N \frac{W_{it}}{\spsi} \Lambda_i  e_{it} \xrightarrow{d} \calN(0, \covI_{F,t})$ from Assumption \ref{ass:mom-clt-conditional}.\ref{ass:asy-normal-main-term-thm-factor-conditional}  and $\frac{1}{N} \sum_{i = 1}^N \frac{W_{it}}{\spsi} \Lambda_i \Lambda_i  \xrightarrow{p} \Sigma_{\Lambda,t}$. From Slutsky's theorem and Lemma \ref{lemma:def-q} ($H^\I \xrightarrow{p} Q^\T$), we conclude
				\begin{equation}\label{eqn:F-asy-term1}
					\sqrt{N} \underbrace{(H^\I)^\T \Big(\frac{1}{N}  \sum_{i = 1}^N \frac{W_{it}}{\spsi} \Lambda_i  \Lambda_i^\T  \Big)^\I  \Big(\frac{1}{N}  \sum_{i = 1}^N \frac{W_{it}}{\spsi} \Lambda_i  e_{it}    \Big) }_{\bm{\varepsilon}_{F,t,1}}  \xrightarrow{d} \calN(0, Q \Sigma_{\Lambda}^\I \covI_{F,t} \Sigma_{\Lambda}^\I Q^\T).    
				\end{equation}
				For $\Delta_2$, we have the decomposition 
				\begin{align*}
					\Delta_2 =& \Big(\frac{1}{N}  \sum_{i = 1}^N \frac{W_{it}}{\spsi} H \Lambda_i  \Lambda_i^\T H^\T  \Big)^\I \underbrace{ \Big(\frac{1}{N}  \sum_{i = 1}^N \frac{W_{it}}{\spsi} (\tilde \Lambda_i - H_i \Lambda_i) \Lambda_i^\T F_t   \Big)}_{{O_P \Lp \frac{1}{\delta_{NT}} \Rp \text{ from Lemma \ref{lemma:f-est-error-times-f-and-e-adj}.4 }} }   \\
					&+  \Big(\frac{1}{N}  \sum_{i = 1}^N \frac{W_{it}}{\spsi} H \Lambda_i  \Lambda_i^\T H^\T  \Big)^\I  \Big(\frac{1}{N}  \sum_{i = 1}^N \frac{W_{it}}{\spsi} (H_i - H ) \Lambda_i \Lambda_i^\T F_t   \Big) . 
				\end{align*}
				For $\sum_{i = 1}^N \frac{W_{it}}{\spsi}  (H_i - H ) \Lambda_i \Lambda_i^\T $ in the second term, we have
				\begin{align*}
					\frac{1}{N} \sum_{i = 1}^N \frac{W_{it}}{\spsi}  (H_i - H ) \Lambda_i \Lambda_i^\T  =&\tilde D^{-1} \cdot \frac{1}{N^2} \sum_{i =1}^N \sum_{l = 1}^N \tilde \Lambda_l   \Lambda_l^\T  \Lp \frac{1}{|\tlq_{li}|} \sum_{s \in \tlq_{li}} F_s F_s^\T - \frac{1}{T} \sum_{s=1}^T F_s F_s^\T   \Rp \frac{W_{it}}{\spsi}   \Lambda_i \Lambda_i^\T  \\
					=&  \tilde D^{-1} \cdot \underbrace{\frac{1}{N^2} \sum_{i =1}^N \sum_{l = 1}^N (\tilde \Lambda_l - H \Lambda_l )   \Lambda_l^\T  \Lp \frac{1}{|\tlq_{li}|} \sum_{s \in \tlq_{li}} F_s F_s^\T - \frac{1}{T} \sum_{s=1}^T F_s F_s^\T   \Rp \frac{W_{it}}{\spsi}  \Lambda_i \Lambda_i^\T }_{\text{\RNum{1}} }  \\
					&+ \tilde D^{-1} H \cdot \underbrace{\frac{1}{N^2} \sum_{i =1}^N \sum_{l = 1}^N  \Lambda_l   \Lambda_l^\T  \Lp \frac{1}{|\tlq_{li}|} \sum_{s \in \tlq_{li}} F_s F_s^\T - \frac{1}{T} \sum_{s=1}^T F_s F_s^\T   \Rp \frac{W_{it}}{\spsi}   \Lambda_i \Lambda_i^\T }_{\mathbf{X}_t^S }  .
				\end{align*}
				The first term \RNum{1} is bounded by 
				\begin{align*}
					\norm{\text{\RNum{1}}}^2 &\leq \underbrace{\Lp \frac{1}{N} \sum_{l = 1}^N \norm{\tilde  \Lambda_l - H \Lambda_l}^2 \Rp}_{O_P \Lp \frac{1}{\delta_{NT}} \Rp}  \Lp\frac{1}{N} \sum_{l = 1}^N  \norm{ \Lambda_l }^2 \norm{\frac{1}{N} \sum_{i =1}^N \frac{W_{it}}{\spsi}  \Lambda_i \Lambda_i^\T \Lp \frac{1}{T} \sum_{s=1}^T F_s F_s^\T - \frac{1}{|\tlq_{li}|} \sum_{s \in \tlq_{li}} F_s F_s^\T \Rp}^2 \Rp    \\
					&\leq O_P \Lp \frac{1}{\delta_{NT}} \Rp \cdot \frac{1}{\underline{p}} \underbrace{\Lp\frac{1}{N} \sum_{l = 1}^N   \norm{ \Lambda_l }^2 \Lp \frac{1}{N} \sum_{i =1}^N \norm{\Lambda_i }^2 \norm{\frac{1}{T} \sum_{s=1}^T F_s F_s^\T - \frac{1}{|\tlq_{li}|} \sum_{s \in \tlq_{li}} F_s F_s^\T } \Rp^2  \Rp }_{\text{\RNum{1}}_1} ,
				\end{align*}
				where $\text{\RNum{1}}_1$ satiesfies 
				\begin{align*}
					\+E[ \text{\RNum{1}}_1 ] =& \frac{1}{N^3} \sum_{l = 1}^N \sum_{i =1}^N \sum_{j =1}^N  \underbrace{\+E\Big[ \norm{ \Lambda_l }^2  \norm{\Lambda_i }^2  \norm{\Lambda_j }^2 |S\Big]}_{\bar{\Lambda}}  \\
					& \cdot \underbrace{\+E\Bigg[ \norm{\frac{1}{T} \sum_{s=1}^T F_s F_s^\T - \frac{1}{|\tlq_{li}|} \sum_{s \in \tlq_{li}} F_s F_s^\T } \norm{\frac{1}{T} \sum_{s=1}^T F_s F_s^\T - \frac{1}{|\tlq_{lj}|} \sum_{s \in \tlq_{lj}} F_s F_s^\T } \bigg| \Bigg]}_{\leq \frac{M}{T} \text{ from } \+E[ab] \leq (\+E[a^2]\+E[b^2])^{1/2}  } 
				\end{align*}
				Hence, we conclude that $\text{\RNum{1}} = O_P \Lp \frac{1}{\sqrt{T\delta_{NT}}} \Rp$. The second term $\mathbf{X}_t^S$ is asymptotically normal from Assumption \ref{ass:mom-clt-conditional}.\ref{ass:asy-normal-add-term-thm-loading-conditional} and its convergence rate is $\sqrt{T}$. Hence, the leading term in $\Delta_2$ is 
				\begin{align}\label{eqn:F-asy-term2}
					\bm{\varepsilon}_{F,t, 2} =    \Big(\frac{1}{N}  \sum_{i = 1}^N \frac{W_{it}}{\spsi}  H \Lambda_i  \Lambda_i^\T H^\T  \Big)^\I  \Big(\tilde D^{-1} H \mathbf{X}_t^S  F_t   \Big),  
				\end{align}
				where $\mathbf{X}_t^S = \frac{1}{N^2} \sum_{i =1}^N \sum_{l = 1}^N  \Lambda_l   \Lambda_l^\T  \Lp \frac{1}{|\tlq_{li}|} \sum_{s \in \tlq_{li}} F_s F_s^\T - \frac{1}{T} \sum_{s=1}^T F_s F_s^\T   \Rp \frac{W_{it}}{\spsi}   \Lambda_i \Lambda_i^\T$.

				Next let us consider the difference between $\tilde{F}^{S\dagger}_t$ and $\tilde{F}^{S}_t$. The leading term is
				\begin{align*}
					&\tilde{F}^{S}_t -  \tilde{F}^{S\dagger}_t \\
					=&  \Bigg[ \Big(   \sum_{i = 1}^N \frac{W_{it}}{\spsi}\tilde \Lambda_i \tilde \Lambda_i^\T  \Big)^\I - \Big(  \sum_{i = 1}^N \frac{W_{it}}{\spsi} H \Lambda_i  \Lambda_i^\T H^\T  \Big)^\I  \Bigg]   \sum_{i = 1}^N \frac{W_{it}}{\spsi}\tilde \Lambda_i Y_{it}    \Big) \\
					=&\Big( \sum_{i = 1}^N \frac{W_{it}}{\spsi} \tilde \Lambda_i \tilde \Lambda_i^\T  \Big)^\I   \Bigg[ \sum_{i = 1}^N \frac{W_{it}}{\spsi} H \Lambda_i  \Lambda_i^\T H^\T - \sum_{i = 1}^N \frac{W_{it}}{\spsi}  \tilde \Lambda_i \tilde \Lambda_i^\T    \Bigg] \Big( \sum_{i = 1}^N \frac{W_{it}}{\spsi}  H \Lambda_i  \Lambda_i^\T H^\T  \Big)^\I \Big( \sum_{i = 1}^N \frac{W_{it}}{\spsi}  \tilde \Lambda_i Y_{it}    \Big) \\
					=& \Big( \frac{1}{N} \sum_{i = 1}^N \frac{W_{it}}{\spsi}  \tilde \Lambda_i \tilde \Lambda_i^\T  \Big)^\I \Bigg[ \frac{1}{N} \sum_{i = 1}^N \frac{W_{it}}{\spsi}  H \Lambda_i  \Lambda_i^\T H^\T - \frac{1}{N} \sum_{i = 1}^N \frac{W_{it}}{\spsi}  \tilde \Lambda_i \tilde \Lambda_i^\T   \Bigg]  \Big( \frac{1}{N} \sum_{i = 1}^N \frac{W_{it}}{\spsi}  H \Lambda_i  \Lambda_i^\T H^\T  \Big)^\I  \\
					& \cdot \Big(\frac{1}{N} \sum_{i = 1}^N \frac{W_{it}}{\spsi}  H \Lambda_i  \Lambda_i^\T  F_t +  O_P \Big( \frac{1}{\sqrt{N}} \Big) \Big).
				\end{align*}
				Note that we have the following bound on the weighted difference between the estimated and population loadings
				\begin{align*}
					& \norm{\frac{1}{N} \sum_{i = 1}^N \frac{W_{it}}{\spsi} \tilde \Lambda_i \tilde \Lambda_i^\T  - \frac{1}{N} \sum_{i = 1}^N \frac{W_{it}}{\spsi} H \Lambda_i  \Lambda_i^\T H^\T  }  \leq \frac{1}{N \underline{p}} \sum_{i = 1}^N \norm{\tilde \Lambda_i \tilde \Lambda_i^\T - H \Lambda_i  \Lambda_i^\T H^\T} \\
					\leq& \frac{1}{N} \sum_{i = 1}^N  \norm{\tilde \Lambda_i} \norm{\tilde \Lambda_i - H \Lambda_i } + \frac{1}{N} \sum_{i = 1}^N \norm{\Lambda_i} \norm{\tilde \Lambda_i - H \Lambda_i } \\
					\leq&  \bigg( \frac{1}{N} \sum_{i = 1}^N  \norm{\tilde \Lambda_i}^2  \bigg)^{1/2}  \bigg( \frac{1}{N} \sum_{i = 1}^N \norm{\tilde \Lambda_i - H \Lambda_i }^2\bigg)^{1/2} +  \bigg(  \frac{1}{N} \sum_{i = 1}^N \norm{\Lambda_i}^2 \bigg)^{1/2} \bigg( \frac{1}{N} \sum_{i = 1}^N  \norm{\tilde \Lambda_i - H \Lambda_i }^2 \bigg)^{1/2} = O \bigg( \frac{1}{\dnt} \bigg)
				\end{align*}
				following from Theorem \ref{thm:consistency-same-H},  $\frac{1}{N} \tilde{\Lambda}^\T \tilde{\Lambda} = I_r$ and Assumption \ref{ass:factor-model}.2. This yields
				\[\frac{1}{N} \sum_{i = 1}^N \frac{W_{it}}{\spsi} \tilde \Lambda_i \tilde \Lambda_i^\T  \xrightarrow{p} \frac{1}{N} \sum_{i = 1}^N \frac{W_{it}}{\spsi}  H \Lambda_i  \Lambda_i^\T H^\T.   \]
				This is also equivalent to 
				\[\Big(\frac{1}{N} \sum_{i = 1}^N \frac{W_{it}}{\spsi} \tilde \Lambda_i \tilde \Lambda_i^\T  \Big)^\I \Big(\frac{1}{N} \sum_{i = 1}^N \frac{W_{it}}{\spsi} H \Lambda_i  \Lambda_i^\T H^\T  \Big) \xrightarrow{p} I_k.  \]
				For the term $ \sum_{i = 1}^N \frac{W_{it}}{\spsi} H \Lambda_i  \Lambda_i^\T H^\T - \sum_{i = 1}^N \frac{W_{it}}{\spsi} \tilde \Lambda_i \tilde \Lambda_i^\T   $, we have the decomposition
				\begin{align*}
					&\frac{1}{N} \sum_{i = 1}^N \frac{W_{it}}{\spsi} \tilde \Lambda_i \tilde \Lambda_i^\T - \frac{1}{N}\sum_{i = 1}^N \frac{W_{it}}{\spsi} H \Lambda_i  \Lambda_i^\T H^\T \\
					=& \frac{1}{N}\sum_{i = 1}^N \frac{W_{it}}{\spsi} (\tilde \Lambda_i -  H \Lambda_i) \tilde \Lambda_i^\T +  \frac{1}{N}\sum_{i = 1}^N \frac{W_{it}}{\spsi} H \Lambda_i   (\tilde \Lambda_i -  H \Lambda_i)^\T \\
					=& \frac{1}{N}\sum_{i = 1}^N \frac{W_{it}}{\spsi}  (\tilde \Lambda_i -  H \Lambda_i)  (H \Lambda_i)^\T +  \frac{1}{N}\sum_{i = 1}^N \frac{W_{it}}{\spsi} H \Lambda_i   (\tilde \Lambda_i -  H \Lambda_i)^\T   + \frac{1}{N}\sum_{i = 1}^N \frac{W_{it}}{\spsi}  (\tilde \Lambda_i -  H \Lambda_i) (\tilde \Lambda_i - H \Lambda_i)^\T \\
					=& \underbrace{\frac{1}{N}\sum_{i = 1}^N \frac{W_{it}}{\spsi} (H_i - H)  \Lambda_i \Lambda_i^\T}_{\tilde D^{-1} H\mathbf{X}_t^S}  \cdot H^\T + H \cdot \underbrace{\frac{1}{N}\sum_{i = 1}^N \frac{W_{it}}{\spsi}  \Lambda_i \Lambda_i^\T (H_i - H)^\T }_{(\tilde D^{-1} H\mathbf{X}_t^S)^\T}   \\
					&+ \underbrace{\frac{1}{N}\sum_{i = 1}^N \frac{W_{it}}{\spsi} (\tilde \Lambda_i -  H_i \Lambda_i)  \Lambda_i^\T}_{{O_P \Lp \frac{1}{\delta_{NT}} \Rp \text{ from Lemma \ref{lemma:f-est-error-times-f-and-e-adj}.4 }} }  \cdot H^\T + H \cdot \underbrace{\frac{1}{N}\sum_{i = 1}^N \frac{W_{it}}{\spsi} \Lambda_i   (\tilde \Lambda_i -  H_i \Lambda_i)^\T}_{{O_P \Lp \frac{1}{\delta_{NT}} \Rp \text{ from Lemma \ref{lemma:f-est-error-times-f-and-e-adj}.4 }} }   \\
					&+ \underbrace{\frac{1}{N}\sum_{i = 1}^N \frac{W_{it}}{\spsi}  (\tilde \Lambda_i -  H \Lambda_i) (\tilde \Lambda_i - H \Lambda_i)^\T }_{\text{\RNum{1}}} 
				\end{align*}
				The term \RNum{1} is bounded by
				\begin{align*}
					\norm{\text{\RNum{1}} } &\leq   \frac{1}{N}\sum_{i = 1}^N \frac{W_{it}}{\spsi} \norm{\tilde \Lambda_i -  H \Lambda_i}^2 =  O_P \Lp \frac{1}{\delta_{NT}} \Rp.
				\end{align*}
				The term $\mathbf{X}_t$ is asymptotically normal based on Assumption \ref{ass:mom-clt-conditional}.\ref{ass:asy-normal-add-term-thm-loading} and its convergence rate is $\sqrt{T}$. Hence, the leading term in $\tilde{F}^{S}_t -  \tilde{F}^{S\dagger}_t$ is 
				\begin{align}
					\nonumber\bm{\varepsilon}_{F,t, 3} =&  - \Big(\frac{1}{N} \sum_{i = 1}^N \frac{W_{it}}{\spsi}  H \Lambda_i  \Lambda_i^\T H^\T  \Big)^{-1} \Big(\tilde D^{-1} H  \mathbf{X}_t^S H^\T + H (\tilde D^{-1} H \mathbf{X}_t^S)^\T \Big) \\
					\nonumber & \cdot\Big( \frac{1}{N} \sum_{i = 1}^N \frac{W_{it}}{\spsi} H \Lambda_i  \Lambda_i^\T H^\T  \Big)^\I  \Big(\frac{1}{N} \sum_{i = 1}^N \frac{W_{it}}{\spsi} H \Lambda_i  \Lambda_i^\T  \Big)  F_t \\
					=&  \Big( \frac{1}{N} \sum_{i = 1}^N \frac{W_{it}}{\spsi} H \Lambda_i  \Lambda_i^\T H^\T  \Big)^\I  \Big(\tilde D^{-1} H \mathbf{X}_t^S H^\T + H (\tilde D^{-1} H \mathbf{X}_t^S)^\T \Big) (H^\T)^{-1} F_t,  \label{eqn:F-asy-term3}
				\end{align}
				where $\mathbf{X}_t^S = \frac{1}{N^2} \sum_{i =1}^N \sum_{l = 1}^N  \Lambda_l   \Lambda_l^\T  \Lp \frac{1}{|\tlq_{li}|} \sum_{s \in \tlq_{li}} F_s F_s^\T - \frac{1}{T} \sum_{s=1}^T F_s F_s^\T   \Rp \frac{W_{it}}{\spsi}  \Lambda_i \Lambda_i^\T$. 
				In summary, the asymptotic distribution of $\tilde{F}^{S}_t$ is determined by \eqref{eqn:F-asy-term1}, \eqref{eqn:F-asy-term2} and \eqref{eqn:F-asy-term3}, that is, 
				\begin{align*}
					\dnt (\tilde{F}^{S}_t -  (H^{-1})^\T F_t) =& \dnt (\bm{\varepsilon}_{F,t, 1} + \bm{\varepsilon}_{F,t, 2} + \bm{\varepsilon}_{F,t, 3})   + o_P(1).
				\end{align*}
				Let us first consider the asymptotic distribution of $\bm{\varepsilon}_{F,t, 2} + \bm{\varepsilon}_{F,t, 3}$. Note that by Assumption \ref{ass:mom-clt-conditional}.\ref{ass:asy-normal-add-term-thm-loading} it holds that $\sqrt{T} \tvec(\mathbf{X}_t^S) \xrightarrow{d} \calN (0, \mathbf{\Phi}_t^S)$. Denote $\tilde \Sigma_{\Lambda} := \frac{1}{N} \sum_{i = 1}^N \frac{W_{it}}{\spsi}  \Lambda_i  \Lambda_i^\T$. We can then rewrite $\bm{\varepsilon}_{F,t, 2}$ as
				\begin{align*}
					\bm{\varepsilon}_{F,t, 2} &= (H^\T)^\I \tilde \Sigma_{\Lambda}^\I H^\I \tilde D^{-1} H \Big(  F_t^\T \otimes I_r \Big) \tvec(\mathbf{X}_t^S),
				\end{align*}
				and rewrite $\bm{\varepsilon}_{F,t, 3}$ as
				\begin{align*}
					\bm{\varepsilon}_{F,t, 3} &=  - (H^\T)^{-1} \tilde \Sigma_{\Lambda}^{-1} H^{-1} \Bigg(\tilde D^{-1} H \mathbf{X}_t^S F_t  + H (\mathbf{X}_t^S)^\T H^\T \tilde D^{-1} (H^\T)^{-1}  F_t  \Bigg) \\
					&=  - (H^\T)^{-1} \tilde \Sigma_{\Lambda}^{-1} H^{-1} \Big(\tilde D^{-1} H \big( F_t ^\T \otimes I_r \big)  + H \big(I_r \otimes (H^\T \tilde D^{-1} (H^\T)^{-1}  F_t   )^\T \big)  \Big) \tvec(\mathbf{X}_t^S)
				\end{align*}
				This allows us to derive the following expression for $\bm{\varepsilon}_{F,t, 2} + \bm{\varepsilon}_{F,t, 3}  $  
				\begin{align*}
					& \sqrt{T} \Big( \bm{\varepsilon}_{F,t, 2} + \bm{\varepsilon}_{F,t, 3} \Big)      \\
					=& \sqrt{T} (H^\T)^\I \tilde \Sigma_{\Lambda}^\I H^\I  \Bigg( \tilde D^{-1} H \Big(  F_t^\T \otimes I_r \Big) - \Big(\tilde D^{-1} H  \big(  F_t^\T \otimes I_r \big)  + H \big(I_r \otimes (H^\T \tilde D^\I (H^\T)^{-1} F_t )^\T \big)  \Big) \Bigg) \tvec(\mathbf{X}_t^S)  \\
					=&  -\sqrt{T} (H^\T)^\I \tilde \Sigma_{\Lambda}^\I \Big(I_r \otimes (H^\T \tilde D^\I (H^\T)^{-1} F_t )^\T \Big) \tvec(\mathbf{X}_t^S) .
				\end{align*}
				Lemma \ref{lemma:HDinvHTinv} implies  $H^\T \tilde D^\I (H^\T)^{-1} =  \Big( \frac{\Lambda^\T \tilde \Lambda}{N} \Big)^\I  \Big(\frac{F^\T F}{T} \Big)^\I + O_P \Big( \frac{1}{\sqrt{\delta_{NT}}} \Big) $. Thus, we have 
				\begin{align*}
					\sqrt{T} \Big( \bm{\varepsilon}_{F,t, 2} + \bm{\varepsilon}_{F,t, 3} \Big)      
					=&  -\sqrt{T} (H^\T)^\I \tilde \Sigma_{\Lambda}^\I \Big(I_r \otimes (H^\T \tilde D^\I (H^\T)^{-1} F_t )^\T \Big) \tvec(\mathbf{X}_t^S) \\
					\rightarrow&   \calN \Bigg(0,  Q \Sigma_{\Lambda}^\I  \covIIS_{F,t}   \Sigma_{\Lambda}^\I Q^\T  \Bigg) 
					\quad \mathcal{G}^t-\text{stably},
				\end{align*}
				where $\covIIS_{F,t} = g^S_{t}(F_t)$, and the function $g^S_t(\cdot)$ is defined in Assumption \ref{ass:mom-clt-conditional}.\ref{ass:asy-normal-add-term-thm-loading-conditional}.

				Note that $\bm{\varepsilon}_{F,t, 1}  $ and $\bm{\varepsilon}_{F,t, 2} + \bm{\varepsilon}_{F,t, 3}  $ are asymptotically independent because the randomness of $\bm{\varepsilon}_{F,t, 1}  $ comes from the cross-section average of $ \frac{W_{it}}{\spsi}  \Lambda_i e_{it}$, and the randomness of $\bm{\varepsilon}_{F,t, 2} + \bm{\varepsilon}_{F,t, 3}  $  comes from $\frac{1}{T} \sum_{s=1}^T F_s F_s^\T - \frac{1}{|\tlq_{li}|} \sum_{s \in \tlq_{li}} F_s F_s^\T $.  This leads to
				\begin{align*}
					&	\sqrt{\delta_{NT}} (\tilde F^S_t -  (H^{-1})^\T F_t) \rightarrow \calN \Bigg(0, Q \Sigma_{\Lambda}^\I  \Big[  \problim \Big( \frac{\delta_{NT}}{N} \covI_{F,t}   + \frac{\delta_{NT}}{T} \covIIS_{F,t}  \Big) \Big] \Sigma_{\Lambda}^\I  Q^\T \Bigg) 
					\quad \mathcal{G}^t-\text{stably}.
				\end{align*}
				Left multiplying $\tilde{F}^{S}_t -  (H^{-1})^\T F_t)$ by $H^\T$ and using the delta method, we conclude that 
				\begin{align*}
					&	\sqrt{\delta_{NT}} (H^\T \tilde F^S_t -  F_t) \rightarrow \calN \Bigg(0,  \Sigma_{\Lambda}^\I  \Big[  \problim \Big( \frac{\delta_{NT}}{N} \covI_{F,t}   + \frac{\delta_{NT}}{T} \covIIS_{F,t} \Big)  \Big] \Sigma_{\Lambda}^\I  \Bigg) 
					\quad \mathcal{G}^t-\text{stably}.
				\end{align*}
				
			\end{proof}

			\subsubsection{Proof of Theorem \ref{theorem:asy-normal}.2}

			\begin{proof}[Proof of Theorem \ref{theorem:asy-normal}.2]
				From $\tilde C^S_{jt} = \tilde \Lambda_j^\T \tilde F^S_t$ and $C_{jt} = \Lambda_j^\T F_t$, we have
				\[\tilde C^S_{jt} - C_{jt} = \Lambda_j^\T H^\T  (\tilde F^S_t - (H^\T)^{-1} F_t) + (\tilde \Lambda_j - H \Lambda_j)^\T \tilde F^S_t + o_P(1/\sqrt{\delta_{NT}}). \]
				The second term can be written as 
				\begin{eqnarray*}
					(\tilde \Lambda_j - H \Lambda_j)^\T \tilde F^S_t &=&  (\tilde \Lambda_j - H \Lambda_j)^\T (H^\T)^\I  F_t +  (\tilde \Lambda_j - H \Lambda_j)^\T ( \tilde F^S_t - (H^\T)^\I  F_t)  \\
					&=& (\tilde \Lambda_j - H \Lambda_j)^\T (H^\T)^\I  F_t + o_P(1/\sqrt{\delta_{NT}}).
				\end{eqnarray*}
				Thus, we obtain
				\[\tilde C^S_{jt} - C_{jt} = \Lambda_j^\T H^\T  (\tilde F^S_t - (H^\T)^{-1} F_t)  + (\tilde \Lambda_j - H \Lambda_j)^\T (H^\T)^\I  F_t + o_P(1/\sqrt{\delta_{NT}}). \]
				Denote $X_j = \frac{1}{N} \sum_{l = 1}^N \Lambda_l \Lambda_l^\T  \Big( \frac{1}{|\tlq_{lj}|} \sum_{t \in \tlq_{lj}} F_t F_t^\T - \frac{1}{T} \sum_{t = 1}^T F_t F_t^\T \Big)$, and \\ $\mathbf{X}_t^S = \frac{1}{N^2} \sum_{i =1}^N \sum_{l = 1}^N  \Lambda_l   \Lambda_l^\T  \Lp \frac{1}{|\tlq_{li}|} \sum_{s \in \tlq_{li}} F_s F_s^\T - \frac{1}{T} \sum_{s=1}^T F_s F_s^\T   \Rp \frac{W_{it}}{\spsi}  \Lambda_i \Lambda_i^\T$, which we use in the following expression:
				\begin{align*}
					& \sqrt{\delta_{NT}} \Lambda_j^\T H^\T  (\tilde F^S_t - (H^\T)^{-1} F_t)  \\
					=& \sqrt{\delta_{NT}}   \Lambda_j^\T H^\T H  \bigg( \frac{1}{N}  \sum_{i =1}^N \frac{W_{it}}{\spsi} \Lambda_i  e_{it}  \\
					& \qquad - \frac{1}{N^2} \sum_{i =1}^N \sum_{l = 1}^N \frac{W_{it}}{\spsi}  \Lambda_i \Lambda_i^\T \bigg( \frac{1}{T} \sum_{s=1}^T F_s F_s^\T - \frac{1}{|\tlq_{li}|} \sum_{s \in \tlq_{li}} F_s F_s^\T \bigg) \Lambda_l  \Lambda_l^\T  H^\T \tilde D^{-1}  (H^{-1})^\T F_t \bigg) + o_P(1) \\
					=& \sqrt{\delta_{NT}}   \Lambda_j^\T \Lp \frac{\Lambda^\T \Lambda}{N} \Rp^\I  \bigg( \frac{1}{N}  \sum_{i =1}^N \frac{W_{it}}{\spsi} \Lambda_i  e_{it}  - (\mathbf{X}_t^S)^\T  \Sigma_{\Lambda}^\I \Sigma_{F}^\I  F_t \bigg) + o_P(1), 
				\end{align*}
				where the second equality follows from $H^\T H = \Lp \frac{\Lambda^\T \tilde \Lambda}{N} \Rp^\I +o_P(1)  $ and $H^\T \tilde D^{-1}  (H^{-1})^\T = \Sigma_{\Lambda}^\I \Sigma_{F}^\I  + o_P(1)$ from Lemma \ref{lemma:HDinvHTinv}. Moreover, plugging the decomposition of $\tilde \Lambda_j - H \Lambda_j $ into $\sqrt{\delta_{NT}}  F_t^\T H^\I (\tilde \Lambda_j - H \Lambda_j)$, we obtain 
				\begin{align*}
					&\sqrt{\delta_{NT}}  F_t^\T H^\I (\tilde \Lambda_j - H \Lambda_j) \\
					=& \sqrt{\delta_{NT}}  F_t^\T \Lp \frac{F^\T F}{T} \Rp^\I \Lp \frac{\Lambda^\T \tilde \Lambda}{N} \Rp^\I \tilde D \tilde D^{-1}  H  \bigg( \frac{1}{N} \sum_{i=1}^N \Lambda_i \Lambda_i^\T  \frac{1}{|\tlq_{ij}|} \sum_{t \in \tlq_{ij}} F_t e_{jt} \\
					& \quad  \quad +  \frac{1}{N} \sum_{i = 1}^N \Lambda_i \Lambda_i^\T  \Big( \frac{1}{|\tlq_{ij}|} \sum_{t \in \tlq_{ij}} F_t F_t^\T - \frac{1}{T} \sum_{t = 1}^T F_t F_t^\T \Big) \Lambda_j  \bigg) +o_P(1)  \\
					=& \sqrt{\delta_{NT}}  F_t^\T \Lp \frac{F^\T F}{T} \Rp^\I \Lp \frac{\Lambda^\T  \Lambda}{N} \Rp^\I   \bigg( \frac{1}{N} \sum_{i=1}^N \Lambda_i \Lambda_i^\T  \frac{1}{|\tlq_{ij}|} \sum_{t \in \tlq_{ij}} F_t e_{jt} +  X_j \Lambda_j  \bigg) +o_P(1)  
				\end{align*}
				based on $\Lp \frac{\Lambda^\T \tilde \Lambda}{N} \Rp^\I = H^\T + o_P(1) $ from Lemma \ref{lemma:def-q} and $H^\T H = \Lp \frac{\Lambda^\T \tilde \Lambda}{N} \Rp^\I +o_P(1)  $ from Lemma \ref{lemma:HDinvHTinv}. 
				
				Note that $\mathbf{X}_t^S = \frac{1}{N} \sum_{i = 1}^N \frac{W_{it}}{\spsi}  X_i \Lambda_i \Lambda_i^\T$ and hence $\mathbf{X}_t^S$ and $X_j$ are correlated. However, the other terms in $\tilde F^S_t - (H^\T)^\I F_t$ and $\tilde \Lambda_j - H \Lambda_j$ are asymptotically independent. Using Theorem \ref{theorem:asy-normal-equal-weight}.1 and \ref{theorem:asy-normal}.1, we conclude that
				\begin{align*}
					\sqrt{\delta_{NT}} (\tilde C^S_{jt} - C_{jt})  \rightarrow&  \calN \left( 0,  \problim \Big(\frac{\delta_{NT}}{T}  \Lambda_j^\T \Sigma_\Lambda^\I \covIIS_{F,t} \Sigma_\Lambda^\I \Lambda_j + \frac{\delta_{NT}}{N} \Lambda_j^\T \Sigma_\Lambda^\I  \Gamma^{\obs, S}_{F,t}   \Sigma_\Lambda^\I \Lambda_j \right. \\
					&+ \frac{\delta_{NT}}{T} F_t^\T \Sigma_F^\I \Sigma_\Lambda^\I ( \covI_{\Lambda,j}   +\covII_{\Lambda,j} )  \Sigma_\Lambda^\I \Sigma_F^\I F_t \\
					& \left.  - 2 \cdot \frac{\delta_{NT}}{T}  \Lambda_j^{\top} \Sigma_{\Lambda}^\I \covIIIS_{\Lambda, F, j, t}  \Sigma_\Lambda^\I   \Sigma_F^\I  F_t  \right) 
					\quad \mathcal{G}^t-\text{stably}.
				\end{align*}
				where $\covIIIS_{\Lambda, F, j, t} = g^{\cov,S}_{j, t}(\Lambda_j, F_t)$, and the function $g^{\cov,S}_{j, t}(\cdot,\cdot)$ is defined in Assumption \ref{ass:mom-clt-conditional}.\ref{ass:asy-normal-add-term-thm-loading-conditional}. 
				
			\end{proof}

			\subsection{Proof of Theorem \ref{thm:feasible-estimator}: Feasible Probability Weighted Estimator}
			For notation simplicity, denote $\spsi = P(W_{it} = 1|S_i)$ and $\hatspsi = \hat P(W_{it} = 1|S_i) $. We have the following decomposition for $\hat F_t^S$:
			\begin{align}
				\nonumber \hat F_t^S =& \left( \sum_{i = 1}^N \frac{W_{it}}{\hatspsi}  \tilde  \Lambda_i \tilde  \Lambda_i^\T  \right)^\I  \left( \sum_{i = 1}^N \frac{W_{it} }{\hatspsi} Y_{it} \tilde{\Lambda}_i \right) =  \tilde F_t^S + \underbrace{\left(\frac{1}{N} \sum_{i = 1}^N \frac{W_{it}}{ \spsi}  \tilde  \Lambda_i \tilde  \Lambda_i^\T  \right)^\I }_{\tilde \Sigma_{\Lambda,t}^\I}  \underbrace{\left(\frac{1}{N} \sum_{i = 1}^N  \frac{\spsi-\hatspsi}{\hatspsi  }  \frac{W_{it} }{\spsi} Y_{it} \tilde{\Lambda}_i \right)}_{B}   \\
				&+ \underbrace{\left(\frac{1}{N}\sum_{i = 1}^N \frac{W_{it}}{ \spsi}  \tilde  \Lambda_i \tilde  \Lambda_i^\T  \right)^\I}_{\tilde \Sigma_{\Lambda,t}^\I}  \underbrace{ \left(\frac{1}{N} \sum_{i = 1}^N \frac{\hatspsi - \spsi}{\hatspsi} \frac{W_{it}}{ \spsi}  \tilde  \Lambda_i \tilde  \Lambda_i^\T  \right)  }_C  \underbrace{\left(\frac{1}{N} \sum_{i = 1}^N \frac{W_{it}}{\hatspsi}  \tilde  \Lambda_i \tilde  \Lambda_i^\T  \right)^\I  \left(\frac{1}{N} \sum_{i = 1}^N \frac{W_{it} }{\hatspsi} Y_{it} \tilde{\Lambda}_i \right)}_{\hat F_t^S} \label{eqn:Fhat-decompose}
			\end{align}
			Since $\spsi$ is bounded below from 0 by Assumption \ref{ass:obs}, we have $\max_i \frac{1}{\spsi} = O(1)$. From Assumption \ref{ass:factor-model-conditional}, we have $\tilde{\Sigma}_{\Lambda,t} \xrightarrow{p} \Sigma_{\Lambda,t}$ and $\norm{\tilde{\Sigma}_{\Lambda,t}} = O_P(1)$. 
			
			\subsubsection{Proof of Theorem \ref{thm:feasible-estimator}.2 (a)}
			We can bound terms $B$ and $C$ by 
			\begin{align*}
				\norm{B}^2 &\leq \Lp \frac{1}{N} \sum_{i =1}^N \Big(\frac{\spsi-\hatspsi}{\hatspsi \spsi } W_{it}  \Big)^2  \Rp \Lp  \frac{1}{N} \sum_{i =1}^N \norm{(\Lambda_i^\T F_t  +e_{it}) \tilde \Lambda_i}^2 \Rp \\
				&\leq  \Lp \frac{1}{N} \sum_{i =1}^N \Big(\frac{\spsi-\hatspsi}{\hatspsi \spsi }   \Big)^2  \Rp \underbrace{\Lp  \Big(\frac{\norm{F_t}^4}{N} \sum_{i =1}^N \norm{\Lambda_i}^4    +  \frac{1}{N} \sum_{i =1}^N e_{it}^4 \Big)^{1/2}  \Big( \frac{1}{N} \sum_{i =1}^N \norm{\tilde \Lambda_i}^4 \Big)^{1/2}   \Rp }_{=O_P(1) \text{ following Assumption \ref{ass:factor-model}}}
			\end{align*}
			and 
			\begin{align*}
				\norm{C}^2 &\leq  \Lp \frac{1}{N} \sum_{i =1}^N \Big(\frac{\spsi-\hatspsi}{\hatspsi \spsi }   \Big)^2  \Rp  \Lp \frac{1}{N} \sum_{i =1}^N \norm{\tilde \Lambda_i}^4 \Rp.
			\end{align*}
			
			If $\max_i |\hatspsi - \spsi | = o_P (1)$ as assumed in Theorem \ref{thm:feasible-estimator}.2 (a), then $\frac{1}{N} \sum_{i =1}^N \Big(\frac{\spsi-\hatspsi}{\hatspsi \spsi } \Big)^2 = o_P(1)$ and the factors are estimated consistently pointwise. Hence, the common components are estimated consistently pointwise as well.
			
			Furthermore, if $\frac{1}{N} \sum_{i = 1}^N (\hatspsi - \spsi)^2 = o_P \left( \frac{1}{N} \right)$ as assumed in Theorem \ref{thm:feasible-estimator}.2 (b), then $B = o_P \left( \frac{1}{\sqrt{N}} \right)$.

			\subsubsection{Proof of Theorem \ref{thm:feasible-estimator}.2 (b)}
			In Theorem \ref{theorem:asy-normal}.2, we assume that $\sqrt{N}/T \rightarrow 0$, together with the assumption $\max_i |\hatspsi - \spsi | = o_P \left( \frac{1}{N^{1/4}} \right)$. Therefore, we have $\sqrt{N}/(N^{1/4}\dnt) \rightarrow 0$ and $O_P \Lp\frac{1}{N^{1/4}\dnt} \Rp = o_P \Lp\frac{1}{\sqrt{N}} \Rp $. We are going to use this property extensively in the following proof.
			
			We can decompose $B$ into four terms: 
			\begin{align*}
				B =& \frac{1}{N} \sum_{i = 1}^N  \frac{\spsi-\hatspsi}{\hatspsi  }  \frac{W_{it} }{\spsi} \tilde{\Lambda}_i \Lambda_i^\T F_t  +   \frac{1}{N} \sum_{i = 1}^N  \frac{\spsi-\hatspsi}{\hatspsi  }  \frac{W_{it} }{\spsi} \tilde{\Lambda}_i   e_{it}  \\
				=& \underbrace{ \frac{H}{N} \sum_{i = 1}^N  \frac{\spsi-\hatspsi}{\hatspsi  }  \frac{W_{it} }{\spsi} \Lambda_i \Lambda_i^\T F_t }_{B_1} + \underbrace{ \frac{1}{N} \sum_{i = 1}^N  \frac{\spsi-\hatspsi}{\hatspsi  }  \frac{W_{it} }{\spsi} \big(\tilde{\Lambda}_i - H \Lambda_i \big)  \Lambda_i^\T F_t  }_{B_2} \\
				& + \underbrace{\frac{1}{N} \sum_{i =1}^N \frac{\spsi-\hatspsi}{\hatspsi } \frac{W_{it}}{\spsi} H \Lambda_i  e_{it} }_{B_3} +  \underbrace{\frac{1}{N} \sum_{i =1}^N \frac{\spsi-\hatspsi}{\hatspsi } \frac{W_{it}}{\spsi} (\tilde \Lambda_i -  H \Lambda_i)  e_{it} }_{B_4} 
			\end{align*}
			First, we consider $B_1$:
			\begin{align*}
				\norm{B_1}^2 =&   \Lp \frac{1}{N} \sum_{i =1}^N \Big(\frac{\spsi-\hatspsi}{\hatspsi \spsi }   \Big)^2  \Rp \Lp \frac{1}{N} \sum_{i =1}^N \norm{\Lambda_i}^4    \Rp  \norm{F}^2 \\
				\leq&   \max_i  \Big(\frac{\spsi-\hatspsi}{\hatspsi \spsi }   \Big)^2 \cdot  \Lp \frac{1}{N} \sum_{i =1}^N \norm{\Lambda_i}^4    \Rp \norm{F}^2  = o_P \Lp \frac{1}{\sqrt{N}} \Rp.
			\end{align*}
			This yields $B_1 = o_P \Lp \frac{1}{N^{1/4}} \Rp $.
			
			Next, we consider $B_2$. It holds that 
			\begin{align*}
				\norm{B_2}^2 &\leq \max_i \Big|\frac{\spsi-\hatspsi}{\hatspsi \spsi } \Big|^2 \Big( \frac{1}{N} \sum_{i =1}^N \norm{ H \Lambda_i - \tilde \Lambda_i }^2 \Big) \Big(  \frac{1}{N} \sum_{i = 1}^N \norm{\Lambda_i}^2 \Big) \norm{F_t}^2  = O_P \Lp \frac{1}{N^{1/2} \delta_{NT}} \Rp
			\end{align*}
			and therefore $B_2 = o_P \Lp \frac{1}{N^{1/4} \dnt} \Rp = o_P \Lp \frac{1}{\sqrt{N}} \Rp$. 
			
			Third, we deal with $B_3$. By Assumption \ref{ass:mom-clt-conditional}.\ref{ass:asy-normal-main-term-thm-factor}, it holds that
			\begin{align*}
				B_3 &= \underbrace{\frac{1}{N} \sum_{i =1}^N \frac{\spsi-\hatspsi}{\spsi } \frac{W_{it}}{\spsi} H \Lambda_i  e_{it} }_{B_{3,1} } + \underbrace{ \frac{1}{N} \sum_{i =1}^N \frac{(\spsi-\hatspsi)^2 }{\hatspsi \spsi } \frac{W_{it}}{\spsi} H \Lambda_i  e_{it} }_{B_{3,2} }    
			\end{align*}
			The term $B_{3,2}$ is bounded by 
			\[\norm{B_{3,2}}^2 \leq \norm{H}^2 \Lp \frac{1}{N} \sum_{i=1}^N \frac{(\spsi-\hatspsi)^4 }{(\hatspsi)^2 (\spsi)^2 } \Rp  \Lp  \frac{1}{N} \sum_{i=1}^N \norm{\frac{W_{it}}{\spsi} \Lambda_i  e_{it}  }^2 \Rp  = o\Lp \frac{1}{N} \Rp. \]
			Thus, we have $B_{3,2} =  o\Lp \frac{1}{\sqrt{N}} \Rp$. For the first term in this expression, we denote $\alpha_{it} = \frac{\spsi-\hatspsi}{ \spsi }$ and $b_{it} = \frac{W_{it}}{\spsi} \Lambda_i  e_{it}$, then in $\frac{1}{N^2} \+E \Big[ (\sum_{i = 1}^N \alpha_{it} b_{it}) (\sum_{i = 1}^N \alpha_{it} b_{it})^\T |S \Big] = \frac{1}{N^2} \sum_{i = 1}^N \+E[\alpha_{it}^2 b_{it} b_{it}^\T |S_i]$ \\ $ + \frac{1}{N^2} \sum_{i \neq j}  \+E[\alpha_{it} \alpha_{jt} b_{it} b_{jt}^\T |S_i, S_j]    $, we have $\+E[\alpha_{it} \alpha_{jt} b_{it} b_{jt}^\T |S_i, S_j]   = \+E[\alpha_{it} \alpha_{jt}|S_i, S_j]  \+E[ \frac{W_{it}}{\spsi} \frac{W_{jt}}{p_{jt}} |S_i, S_j] $ \\ $ \+E[(\Lambda_i \+E[e_{it} e_{jt}] \Lambda_j|S_i, S_j] =0$ given $e_{it}$ is cross-sectionally independent. Therefore, we get $B_{3,1} = o_P(\frac{1}{\sqrt{N}})$.

			Third, we consider $B_4$, which is bounded by 
			\begin{align*}
				\norm{B_4}^2 \leq \max_i \Big| \frac{\spsi-\hatspsi}{\hatspsi \spsi} \Big|^2 \Big( \frac{1}{N} \sum_{i =1}^N \norm{\tilde \Lambda_i -  H \Lambda_i}^2  \Big) \Big(  \frac{1}{N} \sum_{i =1}^N e_{it}^2     \Big) = O_P \Lp \frac{1}{N^{1/2} \delta_{NT}} \Rp .
			\end{align*}
			We have $B_4 = O_P \Lp \frac{1}{N^{1/4} \dnt} \Rp = o_P \Lp \frac{1}{\sqrt{N}} \Rp $. In summary, we have 
			\begin{align}
				B =& B_1 + o_P \Lp \frac{1}{\sqrt{N}} \Rp =  \frac{H}{N} \sum_{i = 1}^N  \frac{\spsi-\hatspsi}{\hatspsi  }  \frac{W_{it} }{\spsi} \Lambda_i \Lambda_i^\T F_t  + o_P \Lp \frac{1}{\sqrt{N}} \Rp \label{eqn:B-simplified}
			\end{align}
			Next let us consider the following decomposition of $C$:
			\begin{align*}
				C =&  \frac{1}{N} \sum_{i = 1}^N \frac{\hatspsi - \spsi}{\hatspsi} \frac{W_{it}}{ \spsi}  \tilde  \Lambda_i \tilde  \Lambda_i^\T \\  =& \underbrace{ \frac{1}{N} \sum_{i = 1}^N \frac{\hatspsi - \spsi}{\hatspsi} \frac{W_{it}}{ \spsi} H  \Lambda_i  \Lambda_i^\T  H^\T}_{C_1} +  \underbrace{\frac{1}{N} \sum_{i = 1}^N \frac{\hatspsi - \spsi}{\hatspsi} \frac{W_{it}}{ \spsi}  H  \Lambda_i  \big(\tilde{\Lambda}_i - H\Lambda_i \big)^\T}_{C_2}  \\
				& + \underbrace{ \frac{1}{N} \sum_{i = 1}^N \frac{\hatspsi - \spsi}{\hatspsi} \frac{W_{it}}{ \spsi} \big(\tilde{\Lambda}_i - H\Lambda_i \big)  \Lambda_i^\T  H^\T}_{C_3}  + \underbrace{ \frac{1}{N} \sum_{i = 1}^N \frac{\hatspsi - \spsi}{\hatspsi} \frac{W_{it}}{ \spsi}  \big(\tilde{\Lambda}_i - H\Lambda_i \big)  \big(\tilde{\Lambda}_i - H\Lambda_i \big)^\T}_{C_4} .
			\end{align*}
			$C_1$ is bounded by
			\begin{align*}
				\norm{C_1}^2 =& \norm{H}^4  \Lp \frac{1}{N} \sum_{i =1}^N \Big(\frac{\spsi-\hatspsi}{\hatspsi \spsi }   \Big)^2  \Rp \Lp \frac{1}{N} \sum_{i =1}^N \norm{\Lambda_i}^4    \Rp \\
				\leq&  \norm{H}^4  \cdot \max_i  \Big(\frac{\spsi-\hatspsi}{\hatspsi \spsi }   \Big)^2 \cdot  \Lp \frac{1}{N} \sum_{i =1}^N \norm{\Lambda_i}^4    \Rp  = o_P \Lp \frac{1}{\sqrt{N}} \Rp.
			\end{align*}
			Thus, we have $C_1 = o_P \Lp \frac{1}{N^{1/4}} \Rp $. $C_2$ is bounded by
			\begin{align*}
				\norm{C_2}^2 &\leq \max_i \Big|\frac{\spsi-\hatspsi}{\hatspsi \spsi } \Big|^2 \Big( \frac{1}{N} \sum_{i =1}^N \norm{ \tilde \Lambda_i -  H \Lambda_i }^2 \Big) \Big(  \frac{1}{N} \sum_{i = 1}^N \norm{\Lambda_i}^2 \Big) \norm{H}^2  = O_P \Lp \frac{1}{N^{1/2} \delta_{NT}} \Rp
			\end{align*}
			and therefore $C_2 = o_P \Lp \frac{1}{N^{1/4} \dnt} \Rp = o_P \Lp \frac{1}{\sqrt{N}} \Rp$. Similarly, we can show $C_3 =  o_P \Lp \frac{1}{\sqrt{N}} \Rp$ and $C_4 =  o_P \Lp \frac{1}{\sqrt{N} \dnt} \Rp$. When we multiply $C$ by $\hat{F}^S_t$, we have
			\begin{align}
				\nonumber C \hat{F}^S_t =& C \big(\tilde{F}^S_t +   \tilde{\Sigma}^\I_{\Lambda,t} B +  \tilde{\Sigma}^\I_{\Lambda,t} C  \hat{F}^S_t \big) = C \bigg((H^\I)^\T F_t +  o_P \Big( \frac{1}{N^{1/4} } \Big) \bigg) \\
				\nonumber =& \bigg(C_1 + o_P \Big( \frac{1}{\sqrt{N}}\Big)  \bigg)  \bigg((H^\I)^\T F_t + o_P \Big( \frac{1}{N^{1/4} } \Big) \bigg) \\
				\nonumber =& C_1  (H^\I)^\T F_t  + C_1 \cdot o_P \Big( \frac{1}{\sqrt{N}}\Big) + (H^\I)^\T F_t  \cdot o_P \Big( \frac{1}{\sqrt{N}}\Big)+ o_P \Big( \frac{1}{N^{3/4} } \Big) \\
				\nonumber =& C_1  (H^\I)^\T F_t + o_P \Big( \frac{1}{\sqrt{N}}\Big) \\
				=& \frac{1}{N} \sum_{i = 1}^N \frac{\hatspsi - \spsi}{\hatspsi} \frac{W_{it}}{ \spsi} H  \Lambda_i  \Lambda_i^\T  F_t + o_P \Big( \frac{1}{\sqrt{N}}\Big)  = - B_1 + o_P \Big( \frac{1}{\sqrt{N}}\Big). \label{eqn:C-simplified}
			\end{align}
			Pluggin Eq. \eqref{eqn:B-simplified} and Eq. \eqref{eqn:C-simplified} into Eq. \eqref{eqn:Fhat-decompose}, we conclude that 
			\begin{align}
				\hat F_t^S =&   \tilde F_t^S + \tilde{\Sigma}^\I_{\Lambda,t} B_1 -  \tilde{\Sigma}^\I_{\Lambda,t} B_1 + o_P \Big( \frac{1}{\sqrt{N}}\Big) =  \tilde F_t^S + o_P \Big( \frac{1}{\sqrt{N}}\Big) .
			\end{align}
			Since the leading terms of $\tilde{F}_t^S - (H^\T)^{-1} F_t$ are of the order $\min \Big(\frac{1}{\sqrt{N}}, \frac{1}{\sqrt{T}} \Big)$, the difference between $\hat F_t^S $ and $ \tilde F_t^S $ is of the order $o_P \Big( \frac{1}{\sqrt{N}}\Big)$, which is smaller than the leading terms of $\tilde{F}_t^S - (H^\T)^{-1} F_t$. Hence, the asymptotic distributions of $\hat F_t^S $ and $\tilde F_t^S $ are the same.

			\subsection{Proof of Theorem \ref{theorem:ate-same-factor}: Treatment Tests}
			\begin{proof}[Proof of Theorem \ref{theorem:ate-same-factor}]
				Throughout this proof we use the convention, that if $\Lambda_i$ and $e_{it}$ do not have a superscript, they indicate the loadings and errors on the control panel. 
				
				We can decompose the estimated loadings $\tilde \Lambda_i^\treat$ into the following two terms: 
				\begin{align*}
					\tilde \Lambda_i^\treat =& \bigg( \sum_{t = \Tcontrol+1}^T \tilde F_t \tilde F_t^\T \bigg)^\I \sum_{t = \Tcontrol+1}^T \tilde F_t Y^\treat_{it} \\
					=&\bigg( \sum_{t = \Tcontrol+1}^T \tilde F_t \tilde F_t^\T \bigg)^\I \sum_{t = \Tcontrol+1}^T \tilde F_t F_t^\T  \Lambda_i^\treat + \bigg( \sum_{t = \Tcontrol+1}^T \tilde F_t \tilde F_t^\T \bigg)^\I \sum_{t = \Tcontrol+1}^T \tilde F_t e_{it}^\treat
				\end{align*}
				We first analyze the loading estimator that uses the population factors in the denominator of the regression:
				\begin{align*}
					\tilde \Lambda_i^{\treat \dagger} =&  \bigg( \sum_{t = \Tcontrol+1}^T (H^{-1})^\T F_t  F_t^\T H^{-1} \bigg)^\I \sum_{t = \Tcontrol+1}^T \tilde F_t Y^\treat_{it} \\
					=&H  \Lambda_i^\treat  + \underbrace{H \bigg( \sum_{t = \Tcontrol+1}^T  F_t  F_t^\T \bigg)^\I \sum_{t = \Tcontrol+1}^T  F_t e_{it}^\treat }_{\Delta_1}  \\& + \underbrace{\bigg( \sum_{t = \Tcontrol+1}^T (H^{-1})^\T F_t  F_t^\T H^{-1} \bigg)^\I \sum_{t = \Tcontrol+1}^T (\tilde F_t -(H^{-1})^\T F_t )F_t^\T  \Lambda_i^\treat }_{\Delta_2}  \\
					& +  \bigg( \sum_{t = \Tcontrol+1}^T (H^{-1})^\T F_t  F_t^\T H^{-1} \bigg)^\I \underbrace{\sum_{t = \Tcontrol+1}^T (\tilde F_t - (H^{-1})^\T F_t)e_{it}}_{ \substack{O_P(\frac{1}{\sqrt{\delta_{NT_i}\Ttreat}}) \\ \text{that can be shown similar as in Lemma \ref{lemma:f-est-error-times-f-and-e-adj}.2}} } .
				\end{align*}
				Assumption \ref{ass:add-factor}.2 implies that $\Delta_1$ is asymptotically normal with 
				\[\sqrt{\Ttreat} H \bigg( \sum_{t = \Tcontrol+1}^T  F_t  F_t^\T \bigg)^\I \sum_{t = \Tcontrol+1}^T  F_t e_{it}^\treat \xrightarrow{d} \calN \big(0, (Q^\T)^\I \Sigma_{F}^\I \Sigma_{F,e_i} \Sigma_{F}^\I  Q^\I \big). \]
				In order to deal with the second term $\Delta_2$, recall the following result from the proof of Theorem \ref{theorem:asy-normal}: 
				\[\tilde F_t - (H^{-1})^\T F_t  = (H^\I)^\T \Big(\frac{1}{N}   \sum_{l = 1}^N W_{lt} \Lambda_l  \Lambda_l^\T  \Big)^\I \Big[ \Big(\frac{1}{N}  \sum_{l = 1}^N W_{lt} \Lambda_l  e_{lt}    \Big) - \Big(I_r \otimes (H^\T \tilde D^\I (H^\T)^{-1} F_t )^\T \Big) \tvec(\mathbf{X}_t) \Big] . \]
				{Assumption \ref{ass:add-factor}.2} implies $\sum_{t = \Tcontrol+1}^T \sum_{j = 1}^N W_{jt} \Lambda_j  e_{jt} F_t^\T= O_P \Big( \frac{1}{\sqrt{\delta\Tcontrol}} \Big)$. Denote \\ $U_t =  - \Big(\frac{1}{N}   \sum_{l = 1}^N W_{lt} \Lambda_l  \Lambda_l^\T  \Big)^\I    \mathbf{X}_t^\T H^\T  \tilde D^{-1} (H^\T)^{-1} F_t$ and \\  $\mathbf{U}_i =  \frac{1}{\Ttreat} \sum_{t = \Tcontrol+1}^T U_t F_t^\T \Lambda_i^\treat =-  \frac{1}{\Ttreat} \sum_{t = \Tcontrol+1}^T \Big(\frac{1}{N}   \sum_{l = 1}^N W_{lt} \Lambda_l  \Lambda_l^\T  \Big)^\I                                                                                               \mathbf{X}_t^\T H^\T  \tilde D^{-1} (H^\T)^{-1} F_t F_t^\T  \Lambda_i^\treat$. The leading term in $\Delta_2$ is $H \bigg( \sum_{t = \Tcontrol+1}^T  F_t  F_t^\T  \bigg)^\I  \*U_i$.
				
				Next we analyze the difference between $\tilde \Lambda_i^{\treat} $ and $\tilde \Lambda_i^{\treat \dagger} $.
				\begin{align*}
					\tilde \Lambda_i^{\treat}  - \tilde \Lambda_i^{\treat \dagger} =&  \bigg( \sum_{t = \Tcontrol+1}^T \tilde F_t \tilde F_t^\T \bigg)^\I  \Bigg[  \sum_{t = \Tcontrol+1}^T (H^{-1})^\T F_t  F_t^\T H^{-1} - \sum_{t = \Tcontrol+1}^T \tilde F_t \tilde F_t^\T  \Bigg] \\
					& \cdot \bigg( \sum_{t = \Tcontrol+1}^T (H^{-1})^\T F_t  F_t^\T H^{-1} \bigg)^\I  \sum_{t = \Tcontrol+1}^T \tilde F_t Y^\treat_{it} \\
					=& H \bigg( \sum_{t = \Tcontrol+1}^T  F_t  F_t^\T  \bigg)^\I  H^\T  \Bigg[  \sum_{t = \Tcontrol+1}^T (H^{-1})^\T F_t  F_t^\T H^{-1} - \sum_{t = \Tcontrol+1}^T \tilde F_t \tilde F_t^\T  \Bigg]  H \Lambda_i^\treat + o_P \Big(\frac{1}{\sqrt{\delta_{NT}}}\Big).
				\end{align*}
				For the term $\sum_{t = \Tcontrol+1}^T \tilde F_t \tilde F_t^\T  -  \sum_{t = \Tcontrol+1}^T (H^{-1})^\T F_t  F_t^\T H^{-1} $, it holds that
				\begin{align*}
					& \sum_{t = \Tcontrol+1}^T \tilde F_t \tilde F_t^\T  -  \sum_{t = \Tcontrol+1}^T (H^{-1})^\T F_t  F_t^\T H^{-1} \\
					=&  \sum_{t = \Tcontrol+1}^T (\tilde F_t - (H^{-1})^\T F_t) \tilde F_t^\T  +  \sum_{t = \Tcontrol+1}^T (H^{-1})^\T F_t (\tilde F_t - (H^{-1})^\T F_t)^\T \\
					=&  \sum_{t = \Tcontrol+1}^T (\tilde F_t - (H^{-1})^\T F_t) F_t^\T H^\I  +  \sum_{t = \Tcontrol+1}^T (H^{-1})^\T F_t (\tilde F_t - (H^{-1})^\T F_t)^\T \\
					& \quad + \sum_{t = \Tcontrol+1}^T (\tilde F_t - (H^{-1})^\T F_t) (\tilde F_t - (H^{-1})^\T F_t)^\T.
				\end{align*}
				Hence, $\tilde \Lambda_i^{\treat}  - \tilde \Lambda_i^{\treat \dagger}$ satisfies
				\begin{align*}
					\tilde \Lambda_i^{\treat}  - \tilde \Lambda_i^{\treat \dagger} =& - H \bigg( \sum_{t = \Tcontrol+1}^T  F_t  F_t^\T  \bigg)^\I  \*U_i  - H \bigg( \sum_{t = \Tcontrol+1}^T  F_t  F_t^\T  \bigg)^\I   \sum_{t = \Tcontrol+1}^T  F_t U_t^\T \Lambda_i^\treat + o_P \Big(\frac{1}{\sqrt{\delta_{NT}}}\Big)
				\end{align*}
				This leads to the following distribution result: 
				\begin{align*}
					& \sqrt{T}( (\tilde \Lambda_i^{\treat}  - \tilde \Lambda_i^{\treat \dagger} ) + \Delta_2 )\\ =& - \sqrt{T} H \bigg( \sum_{t = \Tcontrol+1}^T  F_t  F_t^\T  \bigg)^\I   \sum_{t = \Tcontrol+1}^T  F_t U_t^\T \Lambda_i^\treat + O_P \Big(\frac{1}{\sqrt{\delta_{NT}}}\Big) \\
					=&  \sqrt{T} H \bigg( \sum_{t = \Tcontrol+1}^T  F_t  F_t^\T  \bigg)^\I   \sum_{t = \Tcontrol+1}^T  F_t F_t H^\I D^\I H \*X_t  \Big(\frac{1}{N}   \sum_{l = 1}^N W_{lt} \Lambda_l  \Lambda_l^\T  \Big)^\I  \Lambda_i^\treat + O_P \Big(\frac{1}{\sqrt{\delta_{NT}}}\Big) \\
					\rightarrow& \calN \Bigg(0, (Q^\T)^\I \Sigma_{F}^\I  \covIItreat_{\Lambda,i} \Sigma_{F}^\I  Q^\I \Bigg) 
					\quad \mathcal{G}^t-\text{stably}.
				\end{align*}
				where $\covIItreat_{\Lambda,i}  = \Sigma_{\Lambda}^\I \Big[ \frac{1}{\Ttreat^2} \sum_{u,s = \Tcontrol+1}^T   g_{u,s}(\Sigma_{\Lambda,u}^\I \Lambda_i^\treat,\Sigma_{\Lambda,s}^\I \Lambda_i^\treat)  \Big] \Sigma_{\Lambda}^\I   $, and the function $g_{u,s}(\cdot,\cdot)$ is defined in Assumption \ref{ass:add-factor}. Here we use the property that
				\begin{align}
					\nonumber &  \frac{1}{\Ttreat}  \sum_{t = \Tcontrol+1}^T  F_t F_t H^\I D^\I H \*X_t  \Big(\frac{1}{N}   \sum_{l = 1}^N W_{lt} \Lambda_l  \Lambda_l^\T  \Big)^\I  \Lambda_i^\treat \\ =&  \left( \frac{1}{\Ttreat}  \sum_{t = \Tcontrol+1}^T  F_t F_t  \right) \left(\frac{1}{\Ttreat}  \sum_{t = \Tcontrol+1}^T  H^\I D^\I H \*X_t  \Big(\frac{1}{N}   \sum_{l = 1}^N W_{lt} \Lambda_l  \Lambda_l^\T  \Big)^\I  \Lambda_i^\treat   \right)   + o_P \left(\frac{1}{\sqrt{\delta_{NT}}}\right). \label{eqn:property-separate-sum}
				\end{align}
				In the simplified factor model, the component in the asymptotic distribution of \\ $H^\I D^\I H \*X_t  \Big(\frac{1}{N}   \sum_{l = 1}^N W_{lt} \Lambda_l  \Lambda_l^\T  \Big)^\I  \Lambda_i^\treat $ that varies with $t$ is $\Sigma_{\Lambda,t}$ that is independent of $F_t F_t^\T $ (Step 5.3 in the proof of Proposition \ref{prop:simple-assump-imply-general-assump}). We can verify that \eqref{eqn:property-separate-sum} holds in the simplified factor model. In the more general case, the asymptotic distribution of $H^\I D^\I H \*X_t  \Big(\frac{1}{N}   \sum_{l = 1}^N W_{lt} \Lambda_l  \Lambda_l^\T  \Big)^\I  \Lambda_i^\treat $ that varies with $t$ is related to $W$, which is independent of $F_t F_t^\T$. We can verify Proposition \ref{prop:simple-assump-imply-general-assump} holds. The detailed proof is available upon request.

				Hence, we conclude for $\tilde \Lambda_i^{\treat }  - H  \Lambda_i^\treat $:
				\begin{align*}
					&\sqrt{\delta_{NT_i}} (\tilde \Lambda_i^{\treat }  - H  \Lambda_i^\treat  ) 
					\rightarrow \calN\bigg(0, (Q^\T)^\I \Sigma_F^\I \bigg[ \frac{\delta_{NT_i}}{\Ttreat} \Sigma_{F,e_i} +  \frac{\delta_{NT_i}}{T} \covIItreat_{\Lambda,i} \bigg]  \Sigma_{F}^\I  Q^\I \Bigg)  
					\quad \mathcal{G}^t-\text{stably}.
				\end{align*}
				
				Next, we consider $\tilde{C}^\treat_{it} - C^\treat_{it}$: 
				\begin{align*}
					\tilde{C}^\treat_{it} - C^\treat_{it}  =& (\tilde \Lambda_i^{\treat }  - H  \Lambda_i^\treat  )^\T (H^\I)^\T F_t + (\Lambda_i^\treat )^\T H^\T ( \tilde{F}_t -(H^\I)^\T F_t   ) + (\tilde \Lambda_i^{\treat }  - H  \Lambda_i^\treat  )^\T  ( \tilde{F}_t -(H^\I)^\T F_t   ). 
				\end{align*}
				Recall, that we have for $\tilde{F}_t -(H^\I)^\T F_t     $:
				\begin{align*}
					&\tilde{F}_t -(H^\I)^\T F_t         \\
					=& (H^\I)^\T \Big(\frac{1}{N}  \sum_{l = 1}^N W_{lt} \Lambda_l  \Lambda_l^\T  \Big)^\I  \bigg[  \Big(\frac{1}{N}  \sum_{l = 1}^N W_{lt} \Lambda_l  e_{lt}    \Big)  -\Big(I_r \otimes (H^\T \tilde D^\I (H^\T)^{-1} F_t )^\T \Big) \tvec(\mathbf{X}_t) \bigg]  + o_P \Big(\frac{1}{\sqrt{\delta_{NT}}}\Big) .
				\end{align*}
				Therefore, the difference between the estimated and population treated common components equals 
				\begin{align*}
					& \tilde{C}^\treat_{it} - C^\treat_{it}  \\ =& F_t^\T \bigg( \sum_{s = \Tcontrol+1}^T  F_s  F_s^\T \bigg)^\I  \bigg[ \sum_{s = \Tcontrol+1}^T  F_s e_{is}^\treat +  \sum_{s = \Tcontrol+1}^T  F_s F_s^\T H^\I D^\I H \*X_s  \Big(\frac{1}{N}   \sum_{l = 1}^N W_{ls} \Lambda_l  \Lambda_l^\T  \Big)^\I  \Lambda_i^\treat \bigg]  \\
					& +  (\Lambda_i^\treat )^\T \Big(\frac{1}{N}  \sum_{l = 1}^N W_{lt} \Lambda_l  \Lambda_l^\T  \Big)^\I  \bigg[  \Big(\frac{1}{N}  \sum_{l = 1}^N W_{lt} \Lambda_l  e_{lt}^\control    \Big)  -\Big(I_r \otimes (H^\T \tilde D^\I (H^\T)^{-1} F_t )^\T \Big) \tvec(\mathbf{X}_t) \bigg]  + o_P \Big(\frac{1}{\sqrt{\delta_{NT}}}\Big) 
				\end{align*} 
				and we obtain the following distribution
				\begin{align*}
					& \sqrt{\delta_{NT_i}} (\tilde{C}^\treat_{it} - C^\treat_{it} ) 
					\rightarrow \calN \Bigg(0, F_t^\T \Sigma_{F}^\I \bigg[  \problim \Big( \frac{\delta_{NT_i}}{\Ttreat}\Sigma_{F,e_i} +  \frac{\delta_{NT_i}}{T} \covIItreat_{\Lambda,i}  \bigg]   \Sigma_{F}^\I  F_t \\
					& \quad + (\Lambda_i^\treat )^\T \Sigma_{\Lambda, t}^\I \bigg[  \problim \Big( \frac{\delta_{NT_i}}{N} \covI_{F,t} +  \frac{\delta_{NT_i}}{T} \covIIcontrol_{F,t} \Big) \bigg] \Sigma_{\Lambda, t}^\I \Lambda_i^\treat\\
					& \quad -  \problim \,\, 2 \cdot \frac{\delta_{NT_i}}{T} F_t^\T \Sigma_{F}^\I \covIIItreatcontrol_{\Lambda,F,i,t}  \Sigma_{\Lambda, t}^\I \Lambda_i^\treat  \Bigg) 
					\quad \mathcal{G}^t-\text{stably}.
				\end{align*}
				Next, we consider individual treatment effect
				\begin{align*}
					& (\tilde{C}^\treat_{it} - C^\treat_{it} ) - (\tilde{C}^\control_{it} - C^\control_{it} ) \\
					=& (\tilde F_t - (H^\T)^\I F_t)^\T H (\Lambda_i^\treat - \Lambda_i^\control) +  (\tilde \Lambda_i^\treat - H \Lambda_i^\treat)^\T (H^\T)^\I  F_t  -  (\tilde \Lambda_i^\control - H \Lambda_i^\control)^\T (H^\T)^\I  F_t+ o_P \Lp\frac{1}{\sqrt{\delta_{NT}}} \Rp  .
				\end{align*}
				Recall, that the leading terms in  $\tilde \Lambda_i^\control - H \Lambda_i^\control$ are 
				\begin{align*}
					\tilde \Lambda_i^\control - H \Lambda_i^\control =& \tilde D^{-1}  \frac{1}{N} \sum_{l=1}^N H \Lambda_l^\control (\Lambda_l^\control)^\T  \frac{1}{|\tlq_{li}|} \sum_{t \in \tlq_{li}} F_t e_{it} \\
					& + \tilde D^\I  H  \cdot \frac{1}{N} \sum_{l = 1}^N \Lambda_l^\control (\Lambda_l^\control)^\T  \Big( \frac{1}{|\tlq_{li}|} \sum_{t \in \tlq_{li}} F_t F_t^\T - \frac{1}{T} \sum_{t = 1}^T F_t F_t^\T \Big) \Lambda_i^\control + o_P \Lp\frac{1}{\sqrt{\delta_{NT}}} \Rp 
				\end{align*}
				Thus, we have 
				\begin{align*}
					& (\tilde{C}^\treat_{it} - C^\treat_{it} ) - (\tilde{C}^\control_{it} - C^\control_{it} ) \\=&  F_t^\T \bigg( \sum_{s = \Tcontrol+1}^T  F_s F_s^\T \bigg)^\I  \bigg[ \sum_{s = \Tcontrol+1}^T  F_s e_{is}^\treat + \sum_{s = \Tcontrol+1}^T  F_s F_s^\T H^\I D^\I H \*X_s  \Big(\frac{1}{N}   \sum_{l = 1}^N W_{ls} \Lambda_l^\control   (\Lambda_l^\control )^\T  \Big)^\I  \Lambda_i^\treat \bigg] \\
					& +  (\Lambda_i^\treat - \Lambda_i^\control )^\T \Big(\frac{1}{N}  \sum_{l = 1}^N W_{lt} \Lambda_l^\control   (\Lambda_l^\control)^\T  \Big)^\I  \bigg[  \Big(\frac{1}{N}  \sum_{l = 1}^N W_{lt} \Lambda_l^\control   e_{lt}^\control    \Big)  -\Big(I_r \otimes (H^\T \tilde D^\I (H^\T)^{-1} F_t )^\T \Big) \tvec(\mathbf{X}_t) \bigg]  \\
					& -  F_t^\T \Big(\frac{F^\T F}{T} \Big)^\I  \Big(\frac{\Lambda^\T \Lambda}{N} \Big)^\I  
					\bigg[  \frac{1}{N} \sum_{l=1}^N \Lambda_l^\control (\Lambda_l^\control)^\T  \frac{1}{|\tlq_{li}|} \sum_{s \in \tlq_{li}} F_s e_{is} +  \big( (\Lambda_i^\control)^\T \otimes I_r  \big) \tvec(X_i )  \bigg] + o_P \Big(\frac{1}{\sqrt{\delta_{NT}}}\Big) 
				\end{align*}
				and 
				\begin{align*}
					& \sqrt{\delta_{NT_i}} \Big((\tilde{C}^\treat_{it} - C^\treat_{it} ) - (\tilde{C}^\control_{it} - C^\control_{it} )\Big) \\ 
					\rightarrow& \calN \Bigg(0,  \problim \bigg( F_t^\T \Sigma_{F}^\I  \Gamma_{\Lambda,i}^{\textnormal{obs,miss}}   \Sigma_{F}^\I  F_t + \left(\Lambda_i^\treat - \Lambda_i^\control \right)^\T  \Gamma_{F,t}^{\textnormal{obs,miss}}  \left(\Lambda_i^\treat - \Lambda_i^\control \right) \\&
					+ 2 \cdot  F_t^\T \Sigma_{F}^\I  \Gamma_{\Lambda,F,i,t}^{\textnormal{miss,cov,diff}}  \left(\Lambda_i^\treat - \Lambda_i^\control \right)  \bigg) \Bigg)  
					\quad \mathcal{G}^t-\text{stably}.
				\end{align*}
				with 
				$\covIcontrol_{\Lambda,i} $, $\covIIcontrol_{\Lambda,i}$ and $\covIIIcontrolcontrol_{\Lambda, F, i,t}$ given in Theorem \ref{theorem:asy-normal-equal-weight}, and $ \Gamma_{F,t}^{\textnormal{obs,miss}}= \Sigma_{\Lambda, t}^\I \bigg[ \frac{\delta_{NT_i}}{N} \covIcontrol_{F,t} + \frac{\delta_{NT_i}}{T} \covIIcontrol_{F,t}  \bigg]  \Sigma_{\Lambda, t}^\I $, \\
				$ \Gamma_{\Lambda,i}^{\textnormal{obs,miss}} =  \frac{\delta_{NT_i}}{T} \Sigma_{\Lambda}^\I \big[\covIcontrol_{\Lambda,i}   + \covIIcontrol_{\Lambda,i} \big]   \Sigma_{\Lambda}^\I +  \frac{\delta_{NT_i}}{\Ttreat} \covItreat_{\Lambda,i} + \frac{\delta_{NT_i}}{T} \covIItreat_{\Lambda,i}  -  \frac{\delta_{NT_i}}{ T} \big( \covIIItreatcontrol_{\Lambda,\Lambda,i}  + (\covIIItreatcontrol_{\Lambda,\Lambda,i})^\T\big)$, \\   	
				$\Gamma_{\Lambda,F,i,t}^{\textnormal{miss,cov,diff}} = \frac{\delta_{NT_i}}{T} \left( \Sigma_{\Lambda}^\I   \covIIIcontrolcontrol_{\Lambda,F,i,t} -  \covIIItreatcontrol_{\Lambda,F,i,t} \right)\Sigma_{\Lambda, t}^\I $,\\
				$\covIIItreatcontrol_{\Lambda,\Lambda,i} = \Sigma_{\Lambda}^\I \Big[ \frac{\delta_{NT_i}}{\Ttreat } \sum_{s = \Tcontrol+1}^T g^\cov_{i,s}(\Lambda_i^\control,\Sigma_{\Lambda,s}^\I \Lambda_i^\treat )  \Big] \Sigma_{\Lambda}^\I    $, and the function
				and the function $g^\cov_{i,s}(\cdot,\cdot)$ is defined in Assumption \ref{ass:mom-clt}.\ref{ass:asy-normal-add-term-thm-loading}.

				Last but not least, we consider the weighted treatment effect
				\begin{align*}
					& ( \tilde \beta_i^\treat - \beta_i^\treat) - ( \tilde \beta_i^\control - \beta_i^\control) \\
					=& (Z^\T Z)^\I Z^\T\Big( (\tilde C^\treat_{i,(\Tcontrol+1):T} - C^\treat_{i,(\Tcontrol+1):T}) - (\tilde C^\control_{i,(\Tcontrol+1):T} - C^\control_{i,(\Tcontrol+1):T})  \Big) \\
					=& (Z^\T Z)^\I \sum_{t = \Tcontrol+1}^T Z_t \Big((\tilde{C}^\treat_{it} - C^\treat_{it} ) - (\tilde{C}^\control_{it} - C^\control_{it} )\Big) \\
					=& (Z^\T Z)^\I  \bigg( \sum_{t = \Tcontrol+1}^T Z_t  F_t^\T \bigg) \bigg( \sum_{s = \Tcontrol+1}^T  F_s F_s^\T \bigg)^\I  \bigg[ \sum_{s = \Tcontrol+1}^T  F_s e_{is}^\treat  \\
					& \qquad \qquad - \sum_{s = \Tcontrol+1}^T  F_s F_s^\T H^\I D^\I H \*X_s  \Big(\frac{1}{N}   \sum_{l = 1}^N W_{ls} \Lambda_l  \Lambda_l^\T  \Big)^\I  \Lambda_i^\treat \bigg] \\
					& - (Z^\T Z)^\I \sum_{t = \Tcontrol+1}^T Z_t F_t^\T H^\I \tilde D^\I H \*X_t  \Big(\frac{1}{N}  \sum_{l = 1}^N W_{lt} \Lambda_l  \Lambda_l^\T  \Big)^\I  (\Lambda_i^\treat - \Lambda_i^\control )   \\
					& -   (Z^\T Z)^\I \bigg( \sum_{t = \Tcontrol+1}^T Z_t  F_t^\T \bigg) \Big(\frac{F^\T F}{T} \Big)^\I  \Big(\frac{\Lambda^\T \Lambda}{N} \Big)^\I  
					\bigg[  \frac{1}{N} \sum_{l=1}^N \Lambda_l^\control (\Lambda_l^\control)^\T  \frac{1}{|\tlq_{li}|} \sum_{s \in \tlq_{li}} F_s e_{is} +  \big( (\Lambda_i^\control)^\T \otimes I_r  \big) \tvec(X_i )  \bigg] \\
					& + o_P \Big(\frac{1}{\sqrt{\delta_{NT}}}\Big) ,
				\end{align*}
				which results in the distribution
				\begin{align*}
					& \sqrt{\delta_{NT_i}} \Big(( \tilde \beta_i^\treat - \beta_i^\treat) - ( \tilde \beta_i^\control - \beta_i^\control) \Big) \\ 
					\rightarrow& \calN \Bigg(0,   \problim \bigg(  \Sigma_Z^\I  \Sigma_{F,Z} \Sigma_{F}^\I  \Gamma_{\Lambda,i}^{\textnormal{obs,miss}} \Sigma_{F}^\I  \Sigma_{F,Z}^\T \Sigma_Z^\I + \Sigma_Z^\I \covIItreatcontrol_{Z,i}  \Sigma_Z^\I  \\		
					&+  \frac{\delta_{NT_i}}{ T} \Sigma_Z^\I \bigg[ \Sigma_{F,Z} \Sigma_{F}^\I \Sigma_{\Lambda}^\I   \covIIcontroltreatcontrol_{\Lambda,Z,i} +   (\covIIcontroltreatcontrol_{\Lambda,Z,i})^\T \cdot \Sigma_{\Lambda}^\I  \Sigma_{F}^\I \Sigma_{F,Z}^\T \bigg] \Sigma_Z^\I \\
					& - \frac{\delta_{NT_i}}{ T}\Sigma_Z^\I   \bigg[\Sigma_{F,Z} \Sigma_{F}^\I\cdot  \covIItreattreatcontrol_{\Lambda,Z,i} +   (\covIItreattreatcontrol_{\Lambda,Z,i})^\T \cdot \Sigma_{F}^\I  \Sigma_{F,Z}^\T \bigg] \Sigma_Z^\I  \bigg) \Bigg) 
					\quad \mathcal{G}^t-\text{stably}.
				\end{align*}
				with 
				$\covIIcontroltreatcontrol_{\Lambda,Z,i} = \Big[ \frac{1}{\Ttreat} \sum_{s = \Tcontrol+1}^T  g^\cov_{i,s}(\Lambda_i^\control,\Sigma_{\Lambda, s}^\I (\Lambda_i^\treat - \Lambda_i^\control) ) \Big] \Sigma_{\Lambda}^\I  \Sigma_{F}^\I \Sigma_{F,Z}$,  \\
				$\covIItreattreatcontrol_{\Lambda,Z,i} = \Sigma_{\Lambda}^\I \Big[ \frac{1}{\Ttreat^2}  \sum_{u,s = \Tcontrol+1}^T g_{u,s}(\Sigma_{\Lambda,u}^\I \Lambda_i^\treat,\Sigma_{\Lambda, s}^\I (\Lambda_i^\treat - \Lambda_i^\control))    \Big] \Sigma_{\Lambda}^\I  \Sigma_{F}^\I \Sigma_{F,Z}$, \\ and
				$\covIItreatcontrol_{\Lambda,Z,i} = \Sigma_{F,Z}^\T \Sigma_{F}^\I  \Sigma_{\Lambda}^\I \Big[ \frac{1}{\Ttreat^2} \sum_{u,s = \Tcontrol+1}^T g_{u,s}(\Sigma_{\Lambda, u}^\I (\Lambda_i^\treat - \Lambda_i^\control),\Sigma_{\Lambda, s}^\I (\Lambda_i^\treat - \Lambda_i^\control) )    \Big] \Sigma_{\Lambda}^\I  \Sigma_{F}^\I \Sigma_{F,Z}$, $g^\cov_{i,s}(\cdot,\cdot)$ and the functions $g_{u,s}(\cdot,\cdot)$ are defined in Assumptions \ref{ass:mom-clt}.\ref{ass:asy-normal-add-term-thm-loading} and \ref{ass:add-factor}. 
				
				Here we use $\frac{1}{\Ttreat}  \sum_{t = \Tcontrol+1}^T Z_t  F_t^\T \xrightarrow{P} \Sigma_{F,Z}$ and the property that
				\begin{align}
					\nonumber &  \frac{1}{\Ttreat}  \sum_{t = \Tcontrol+1}^T  Z_t F_t H^\I D^\I H \*X_t  \Big(\frac{1}{N}   \sum_{l = 1}^N W_{lt} \Lambda_l  \Lambda_l^\T  \Big)^\I  \Lambda_i^\treat \\ =&  \left( \frac{1}{\Ttreat}  \sum_{t = \Tcontrol+1}^T  Z_t F_t  \right) \left(\frac{1}{\Ttreat}  \sum_{t = \Tcontrol+1}^T  H^\I D^\I H \*X_t  \Big(\frac{1}{N}   \sum_{l = 1}^N W_{lt} \Lambda_l  \Lambda_l^\T  \Big)^\I  \Lambda_i^\treat   \right)   + o_P \left(\frac{1}{\sqrt{\delta_{NT}}}\right). \label{eqn:property-separate-sum-1}
				\end{align}
				Equation \eqref{eqn:property-separate-sum-1} holds for the same reason as equation \eqref{eqn:property-separate-sum}.

				
			\end{proof}

			\subsection{Proof of Proposition \ref{prop:cov}: Feasible Estimator of Asymptotic Variances}
			\subsubsection{Feasible Estimator for Theorem \ref{theorem:asy-normal-equal-weight}.1 }\label{subsec:feasible-loadings}
			
			\begin{lemma}\label{lemma:Lambda-asy-var-main-estimator}
				Assume  we know the set $\mathcal{E}= \{ i,j,s,t : \+E[e_{it}e_{js}] \neq 0 $ and $| \mathcal{E} | = O(NT)$. Under the assumptions of Theorem \ref{theorem:asy-normal-equal-weight}.1, we have
				\[ \widehat \covI_{\Lambda,j}=   \frac{T}{N^2} \tilde D^{-1} \sum_{i=1}^N \sum_{l=1}^N \tilde \Lambda_i \tilde \Lambda_i^\T \Lp  \frac{1}{|\tlq_{ij}| |\tlq_{lj}|}  \sum_{t,s \in \tlq_{ij}}  \tilde F_t \tilde F_s^{\top}    \tilde e_{it} \tilde e_{ls}  \mathbbm{1}_{\{i,k,s,t \in \mathcal{E} \}}  \Rp \tilde \Lambda_l \tilde \Lambda_l^\T \tilde D^{-1} \xrightarrow{P} H \covI_{\Lambda,j} H^\T \]
				with  $\tilde e_{it} = Y_{it} - \tilde \Lambda_i^\T \tilde F_t$ for the observed $Y_{it}$.
			\end{lemma}
			
			\begin{proof}[Proof of Lemma \ref{lemma:Lambda-asy-var-main-estimator}]
				$\tilde e_{it} = Y_{it} - \tilde \Lambda_i^\T \tilde F_t$ is a consistent estimator for $e_{it}$ for $(i,t) \in \{(i,t): W_{it} = 1\}$ because $\tilde F_t$ and $\tilde \Lambda_i$ are consistent estimators for $(H^\T)^\I F_t$ and $H \Lambda_i$ following from Theorem \ref{theorem:asy-normal-equal-weight}. Recall that
				\[\sqrt{T}( \tilde \Lambda_j - H_j \Lambda_j ) = \tilde D^{-1}  \frac{1}{N} \sum_{i=1}^N \sqrt{\frac{T}{|\tlq_{ij}|}} H_i \Lambda_i \Lambda_i^\T  \frac{1}{\sqrt{|\tlq_{ij}|}} \sum_{t \in \tlq_{ij}} F_t e_{jt} + o_P(1).\]
				Note that $X_{it}$ is observed for $t \in \tlq_{ij}$ so $\tilde e_{it}$ is a consistent estimator for $e_{it}$ for $t \in \tlq_{ij}$. Then, for each $i$ and $l$, a consistent estimator for the asymptotic covariance between 
				$\frac{1}{\sqrt{|\tlq_{ij}|}} \sum_{t \in \tlo_{ij}} (H^\T)^\I F_t e_{jt} $ and $\frac{1}{\sqrt{|\tlq_{lj}|}} \sum_{t \in \tlo_{lj}} (H^\T)^\I F_t e_{jt} $ is 
				\[\frac{1}{|\tlq_{ij}| |\tlq_{lj}|} \sum_{s \in \tlq_{ij} , t \in \tlq_{lj}, (s,t) \in \Omega_{e_j}} \tilde F_s \tilde F_t^\T \tilde e_{js} \tilde e_{jt} . \]
				Combing this with the result that $\tilde \Lambda_i$ is a consistent estimator for $H_i \Lambda_i$ and $H \Lambda_i$, we conclude that $\widehat \covI_{\Lambda,j}$ is a consistent estimator for the asymptotic variance of $\sqrt{N}(\tilde \Lambda_i - H_i \Lambda_i)$: 
				\[ \widehat \covI_{\Lambda,j} =   \frac{T}{N^2} \tilde D^{-1} \sum_{i=1}^N \sum_{l=1}^N \tilde \Lambda_i \tilde \Lambda_i^\T \Lp \frac{1}{|\tlq_{ij}| |\tlq_{lj}|}  \sum_{t,s \in \tlq_{ij}}  \tilde F_t \tilde F_s^{\top}    \tilde e_{it} \tilde e_{ls}  \mathbbm{1}_{\{i,k,s,t \in \mathcal{E} \}}  \Rp \tilde \Lambda_l \tilde \Lambda_l^\T \tilde D^{-1}. \]
			\end{proof}

			\begin{lemma}\label{lemma:Lambda-var-corrector}
				Under the assumptions in Corollary \ref{corollary:asy-normal-equal-weight}, the plug-in estimator is consistent, i.e., 
				\[\widehat \covII_{\Lambda, j} =  \big( \tilde \Lambda_j^{\top} \otimes \hat \Sigma_{\Lambda}  \big) \hat \Xi_F \big( \tilde \Lambda_j \otimes \hat \Sigma_{\Lambda} \big)  \xrightarrow{P} H  \covII_{\Lambda,j} H^\T, \]
				where $\hat \Sigma_{\Lambda} =  \frac{1}{N} \sum_{i = 1}^N \tilde \Lambda_i \tilde \Lambda_i^\T$ and $\hat \Xi_F = \frac{1}{T} \sum_{t = 1}^T \tvec( \tilde F_t \tilde F_t^\T ) \tvec( \tilde F_t \tilde F_t^\T )^\T $. 
			\end{lemma}

			\begin{proof}[Proof of Lemma \ref{lemma:Lambda-var-corrector}]
				
				Note that $\covII_{\Lambda,j} = ( \Lambda_j^\T \otimes I_r ) \Phi_j (\Lambda_j \otimes I_r) $, where $\Phi_j$ is the asymptotic variance of $X_j = \frac{1}{N} \sum_{l = 1}^N \Lambda_l \Lambda_l^\T  \Big( \frac{1}{|\tlq_{lj}|} \sum_{t \in \tlq_{lj}} F_t F_t^\T - \frac{1}{T} \sum_{t = 1}^T F_t F_t^\T \Big)$. Under the assumptions in Corollary \ref{corollary:asy-normal-equal-weight}, $\Phi_j$ simplifies to $\Phi_j = (I_r  \otimes  \Sigma_\Lambda )  \Xi_F    (  I_r  \otimes \Sigma_\Lambda)$. We use the plug-in estimators $\hat \Sigma_{\Lambda} =  \frac{1}{N} \sum_{i = 1}^N \tilde \Lambda_i \tilde \Lambda_i^\T$ and $\hat \Xi_F = \frac{1}{T} \sum_{t = 1}^T \tvec( \tilde F_t \tilde F_t^\T ) \tvec( \tilde F_t \tilde F_t^\T )^\T $ for $\Sigma_{\Lambda} $ and $\Xi_F $ respectively.
				The rotation matrices cancel out by the definition of $X_i $. Combining this result with the consistency of $\tilde \Lambda_i$ and $\tilde F_t$, we conclude that $\widehat \covII_{\Lambda,j}$ is consistent. 
			\end{proof}
			
			\subsubsection{Feasible Estimators for Theorem \ref{theorem:asy-normal-equal-weight}.2 and \ref{theorem:asy-normal}.1 } \label{subsec:feasible-factor}

			\begin{lemma}\label{lemma:F-asy-var-main-estimator} 
				Assume  we know the set $\mathcal{E}= \{ i,j,s,t : \+E[e_{it}e_{js}] \neq 0 $ and $| \mathcal{E}_t | = O(N)$.  Under the assumptions in Theorem \ref{theorem:asy-normal},  we have
				\[\widehat \covI_{F,t}  =  \frac{1}{N} \sum_{i=1}^N \sum_{j=1}^N W_{it} W_{jt} \tilde \Lambda_i \tilde \Lambda_j^{\top} \tilde e_{it} \tilde e_{jt} \mathbbm 1_{\{i,j \in \mathcal{E}_t \}} \xrightarrow{P} H \covI_{F,t} H^\T, \]
				where $\tilde e_{it} = \tilde Y_{it} - \tilde \Lambda_i^\T \tilde F_t $  for observed $Y_{it}$.
				Under the assumptions in Theorem \ref{theorem:asy-normal} we have 
				\[\widehat \covIS_{F,t}  =  \frac{1}{N} \sum_{i=1}^N \sum_{j=1}^N \frac{W_{it} W_{lt}}{\hat P(W_{it}=1|S_i) \hat P_t(W_{lt}=1|S_l)}  \tilde \Lambda_i \tilde \Lambda_l^\T  \tilde e^S_{it} \tilde e^S_{jt} \mathbbm 1_{\{i,j \in \mathcal{E}_t \}}\xrightarrow{P} H \covIS_{F,t} H^\T, \]
				$\tilde e^S_{it} = \tilde Y_{it} - \tilde \Lambda_i^\T \tilde F^S_t$ for observed $Y_{it}$ and $\hat P(W_{it}=1|S)$ is a consistent estimate for $P(W_{it}=1|S)$. 
			\end{lemma}
			
			\begin{proof}
				If $Y_{it}$  is observed,  $\tilde e_{it}$ and $\tilde e^S_{it}$ are consistent estimators for $e_{it}$ following from the same reasoning as in Lemma \ref{lemma:Lambda-asy-var-main-estimator}. Combined with the result that $\tilde \Lambda_i$ is a consistent estimator for $H_i \Lambda_i$ and  $\hat P(W_{it}=1|S)$ is a consistent estimate for $P(W_{it}=1|S)$, it follows that Lemma \ref{lemma:F-asy-var-main-estimator} holds.  
			\end{proof}
			
			\begin{lemma}\label{lemma:F-var-corrector} 
				Under the assumptions in Corollary \ref{corollary:asy-normal-equal-weight}, it holds that
				\[\widehat \covII_{F,t} =  \big( I_r \otimes (\tilde F_t^{\top} \hat \Sigma_F^{-1} \hat  \Sigma_{\Lambda}^{-1}) \big) ( \hat  \Sigma_{\Lambda,t}  \otimes  \hat  \Sigma_\Lambda )  \Xi_F    ( \hat  \Sigma_{\Lambda,t} \otimes \hat  \Sigma_\Lambda)  \big( I_r \otimes ( \hat  \Sigma_{\Lambda}^{-1} \hat  \Sigma_F^{-1} \tilde F_t )\big) \xrightarrow{P} H  \covII_{F,t} H^\T; \]
				Under the assumptions in Corollary \ref{corollary:asy-normal}, it holds that
				\[\widehat \covIIS_{F,t} =  \big( I_r \otimes (\tilde F_t^{\top} \hat \Sigma_F^{-1} \hat  \Sigma_{\Lambda}^{-1}) \big) ( \hat  \Sigma_{\Lambda}  \otimes  \hat  \Sigma_\Lambda )  \Xi_F    ( \hat  \Sigma_{\Lambda} \otimes \hat  \Sigma_\Lambda)  \big( I_r \otimes ( \hat  \Sigma_{\Lambda}^{-1} \hat  \Sigma_F^{-1} \tilde F_t )\big) \xrightarrow{P} H  \covIIS_{F,t} H^\T, \]
				where $\hat \Sigma_{\Lambda} =  \frac{1}{N} \sum_{i = 1}^N \tilde \Lambda_i \tilde \Lambda_i^\T$, $\hat \Sigma_{\Lambda,t} = \frac{1}{N} \sum_{i = 1}^N W_{it} \tilde \Lambda_i \tilde \Lambda_i^\T$,  $\hat \Sigma_F = \frac{1}{T} \sum_{t = 1}^T \tilde F_t \tilde F_t^\T$ and \\ $\hat \Xi_F = \frac{1}{T} \sum_{t = 1}^T \tvec( \tilde F_t \tilde F_t^\T ) \tvec( \tilde F_t \tilde F_t^\T )^\T $. 
			\end{lemma}
			
			\begin{proof}[Proof of Lemma \ref{lemma:F-var-corrector}]
				Note that $\covII_{F,t} =  \big( I_r \otimes (F_t^{\top}  \Sigma_F^{-1} \Sigma_{\Lambda}^{-1}) \big) \mathbf{\Phi}_t \big( I_r \otimes (\Sigma_{\Lambda}^{-1} \Sigma_F^{-1} F_t )\big) $, where $\mathbf{\Phi}_t$ is the asymptotic variance of $\mathbf{X}_t = \frac{1}{N^2} \sum_{i = 1}^N \sum_{l = 1}^N  W_{it} \Lambda_l \Lambda_l^\T  \Big( \frac{1}{|\tlq_{li}|} \sum_{t \in \tlq_{li}} F_t F_t^\T - \frac{1}{T} \sum_{t = 1}^T F_t F_t^\T \Big) \Lambda_i \Lambda_i^\T$. Under the assumptions in Corollary \ref{corollary:asy-normal-equal-weight}, $\mathbf{\Phi}_t$ simplifies to $\mathbf{\Phi}_t = ( \Sigma_{\Lambda,t}  \otimes  \Sigma_\Lambda )  \Xi_F    (  \Sigma_{\Lambda,t} \otimes \Sigma_\Lambda)$. We use plug-in estimators $\hat \Sigma_{\Lambda} =  \frac{1}{N} \sum_{i = 1}^N \tilde \Lambda_i \tilde \Lambda_i^\T$, $\hat \Sigma_{\Lambda,t} = \frac{1}{N} \sum_{i = 1}^N W_{it} \tilde \Lambda_i \tilde \Lambda_i^\T$,  $\hat \Sigma_F = \frac{1}{T} \sum_{t = 1}^T \tilde F_t \tilde F_t^\T$ and $\hat \Xi_F = \frac{1}{T} \sum_{t = 1}^T \tvec( \tilde F_t \tilde F_t^\T ) \tvec( \tilde F_t \tilde F_t^\T )^\T $ for $\Sigma_{\Lambda} $, $\Sigma_{\Lambda,t} $, $ \Sigma_F $, and $\Xi_F $ respectively.
				The rotation matrices cancel out by the definition of $\mathbf{X}_t $. Combining these results with the consistency of $\tilde \Lambda_i$ and $\tilde F_t$, and the non-singularity of $\Sigma_F$ and $\Sigma_\Lambda$, we conclude that $\widehat \covII_{F,t}$ and $\widehat \covIIS_{F,t}$ are consistent. 
			\end{proof}

			For the other terms in Theorem \ref{theorem:asy-normal-equal-weight}.3, Theorem \ref{theorem:asy-normal}.2 and Theorem \ref{theorem:ate-same-factor}, we can use similar arguments as in Section \ref{subsec:feasible-loadings} and \ref{subsec:feasible-factor} to prove that the plug-in estimators of the asymptotic covariances are consistent. By Slusky's theorem, the  asymptotic statements in the respective theorems continue to hold with the estimated covariance matrices.

	\onehalfspacing

\end{document}